 \pdfoutput=1
\documentclass[phd]{scsthesis}

\usepackage{craig_phd}

\makeatletter
\renewenvironment{abstract}{%
  \if@twocolumn
    \section*{\abstractname}
  \else
    \small
    \begin{center}%
      {\bfseries \abstractname\vspace{-.5em}\vspace{\z@}}%
    \end{center}%
    \quotation
  \fi}
  {\if@twocolumn\else\endquotation\fi}
\makeatother

\newenvironment{acknowledgements}[1][Acknowledgments]{
\begin{center}\textbf{#1}\end{center} 
}

\title{External Memory Algorithms For Path Traversal in Graphs}
\author{Craig Dillabaugh}
\firstthesissupervisor{Anil Maheshwari}
\secondthesissupervisor{}
\director{Douglas Howe}

\begin{document}

\beforepreface
\checkoddpage
  \ifoddpage
  \else
      \newpage\mbox{}
  \fi
\begin{abstract}
  This thesis will present a number of results related to path traversal
  in trees and graphs.
  In particular, we focus on data structures which allow such traversals
  to be performed efficiently in the external memory setting. 
  In addition, for trees and planar graphs the data structures
  we present are succinct. 
  Our tree structures permit efficient bottom-up path traversal in 
  rooted trees of arbitrary degree and efficient top-down path traversal
  in binary trees.
  In the graph setting, we permit efficient traversal of an arbitrary path
  in bounded degree planar graphs.
  Our data structures for both trees and graphs match or slightly improve
  current best results for external memory path traversal in these 
  settings while at the same time improving space bounds due to the 
  succinct nature of our data structures.
  Employing our path traversal structure for bounded degree
  planar graphs, we describe a number of useful applications of this technique
  for triangular meshes in $\plane$.
  As an extension of the $\plane$ representation for triangular meshes we
  also present an efficient external memory representation for well-shaped
  tetrahedral meshes in $\threeD$. 
  The external memory representation we present is based on a partitioning 
  scheme that matches the current best-known results for well-shaped tetrahedral
  meshes.
  We describe applications of path traversal in tetrahedral 
  meshes which are made efficient in the external memory setting using our
  structure.
  Finally, we present a result on using jump-and-walk point location in 
  well-shaped meshes in both $\plane$ and $\threeD$. 
  We demonstrate that, given an approximate nearest neighbour from among the 
  vertices of a mesh, locating the simplex (triangle or tetrahedron) 
  containing
  the query point involves a constant length walk (path traversal) in the
  mesh.
  This final result does not rely on the data structures described above, but
  has applications to the mesh representations.
\end{abstract}
\begin{acknowledgements}

I would like to start by thanking my supervisor Anil Maheshwari. 
Anil has shown great patience as he guided me through my research.
He must have worried at times that I might be his student
forever, yet he never made a discouraging comment.
Also thanks to all the professors in the Computational Geometry group
including Jit Bose, Pat Morin, and Michiel Smid.
I was lucky enough to take at least one course with each of them, and
I have enjoyed our lunch-time meetings and our open-problem sessions.

I would like to thank the members of my thesis committee;
Anil, Vida Dujmovic, Paola Flocchini, J{\"o}rg-R{\"u}diger Sack, and
Jan Vahrenhold.

Thank you to my lab mates in the Computational Geometry
Lab at Carleton.
There are too many people to name, but I have very much enjoyed the 
associations with my fellow students.
I would especially like to thank Rosen Atanassov, Mathieu Couture, 
and Stefanie Wuhrer, who helped me find my feet after I arrived in the 
School of Computer Science with my Geography degree, with little
experience in theoretical Computer Science.

Thanks to the staff of the Carleton School of Computer Science 
including Linda Pfeiffer, Sharmila Antonipillai, Anna Riethman, 
Edina Storfer, and Claire Ryan.
I appreciate your helpfulness and all the work you do to make life
easier for the students.

Thank you to my co-authors Anil, Jean-Lou De Carufel, Meng He, and 
Norbert Zeh.
I have enjoyed working with, and learned much from, each of you. 
I was very fortunate to work with such excellent researchers.

For financial support I would like to express appreciation to Anil, 
The National 
Sciences and Research Council of Canada, to Carleton University, 
and to the former Nortel Networks.  

Finally, and most importantly, I would like to thank my family.
Thank you to my father Derlin, and mother Sondra.
Congratulations to Mom in particular for living to see the day when
all her sons had finished school!
Thank you to my wonderful children Eden and Vaughn.  
Eden and Vaughn likely did more to slow down
work on my thesis than to speed it up, but I don't regret a single 
second of time
spent putting together legos are telling stories - when I could have 
been working on my thesis.
Thank you to my wife Kristie, for your patience and support over the 
many years it has taken me to complete this work.

\end{acknowledgements}

\afterpreface


\chapter{Introduction}

Many problems in computing can be posed as those of finding or following a path. 
For example, most people will be familiar with web applications that let a person 
find the fastest route between two destinations. 
This problem involves a shortest path search in a graph, which represents the road 
network. 
There are also problems that involve path traversal, which are less obvious, 
like searching for a word or phrase in a text database. 
One solution to this problem involves following a path through a text-indexing 
structure such as a suffix tree.

There are numerous application domains in which processing such large datasets 
is a significant challenge.
The examples of finding a route on a map, and text searching, are important 
problems in the realms of Geographic Information Systems (GIS) and text indexing 
respectively. 
In both of these application domains, dealing with huge datasets containing 
gigabytes or even terabytes of data is not uncommon. 

Two important strategies that have emerged in Computer Science for dealing with 
large data volumes are external memory algorithms and succinct data structures.  
Each strategy approaches the challenge of handling data in very different ways.  
External memory algorithms and data structures present a model of the computer in 
which the computer has limited internal memory (RAM), and unlimited external 
memory (usually a magnetic disk). 
The problem of dealing with large datasets is handled by using the unlimited 
external memory to ensure that any size of problem can be processed. 
In practice, however, the speed of operations on data in RAM is orders of magnitude 
faster than operations on data on disk. 
Thus, external memory techniques revolve around keeping a working subset of the data 
in RAM, where it can be efficiently processed, while attempting to minimize the 
expensive transfer of data between disk and RAM. 

Succinct data structures try to solve the problem by using a different strategy. 
Here the goal is to reduce the size of the problem by reducing the space that 
the structural component of the data structures uses, and make the problem 'smaller' 
in this manner. 
It is also important that, while the data structures are stored using less memory, 
the operations on the data structures need to remain efficient.

In this thesis we address problems related to path traversal in very large datasets. 
We present algorithms for tree and graph data structures that permit 
efficient path traversal in the external memory setting. 
In addition to developing efficient external memory techniques, we have examined how 
we can incorporate concepts from the field of succinct data structures to reduce the memory 
footprint of these data structures. 
As we are interested in dealing with huge datasets, integrating these two techniques seems 
to be a natural fit.

\section{Summary of Results}\label{sec:res_summary}

In this section we summarize the results of this thesis and give
an overview of the main results arising from this research.
Chapters \ref{chp:leaf-root} and \ref{chp:succinct_graphs} describe
data structures that are both succinct and efficient in the 
external memory model.
Specifically, the data structures are designed to make path 
traversal efficient.
In Chapter \ref{chp:leaf-root} we present succinct representations
for rooted trees that permit efficient traversal of root to leaf 
paths, and paths from an arbitrary node in the tree to the root.
In Chapter \ref{chp:succinct_graphs} we develop succinct data 
structures for representing planar graphs that are efficient
in the external memory setting.
Chapters \ref{chp:mesh_trav} and \ref{chp:jump_and_walk}
are an extension of this work and focus on topics related to path
traversal;
Chapter ~\ref{chp:mesh_trav} presents a representation for tetrahedral
meshes which is efficient in the external memory setting, but which
is not succinct, and
Chapter~\ref{chp:jump_and_walk} presents results on using the
jump-and-walk point location paradigm in well-shaped triangular or tetrahedral
meshes. 
This technique has application to the structures
we describe for representing meshes in the earlier chapters.

\subsection{Succinct and External Memory Path Traversal in Trees}

In Chapter~\ref{chp:leaf-root} we present results on data structures that are 
both succinct and efficient for performing path traversal in trees in the 
external memory setting. 

\begin{enumerate}
  \item {
    We show how a tree $T$, built on $N$ nodes, can be blocked in a succinct 
    fashion such 
    that a bottom-up traversal requires $O(K / B)$ I/Os using only 
    $2N + \frac{\epsilon N }{\log_B{N}} + o(N)$ bits to store $T$, where $K$ 
    is the path length and $0 < \epsilon < 1$.
    Here $\epsilon$ is a user-tunable constant, which can increase the 
    I/O-efficiency at the cost of greater space usage.
    This technique is based on~\cite{hutch_mesh_zeh_2003}, and achieves an 
    improvement on the space bound by a factor of $\lg{N}$.
  }
  \item {
    We show that a binary tree, with keys of size $q=O(\lg{N})$ 
    bits, can be stored using $(3+q)N + o(N)$ bits so that a root-to-node 
    path of length $K$ can be reported with: 
    (a) $O \left( \frac{K}{\lg(1+(B \lg N)/q)} \right)$ I/Os, when 
    $K = O(\lg N)$; 
    (b) $O \left( \frac{\lg N}{\lg (1+\frac{B \lg^2 N}{qK} )} \right)$ 
    I/Os, when $K=\Omega(\lg N)$ and $K=O \left( \frac{B \lg^2 N}{q} \right)$; 
    and (c) $\BigOh{\frac{qK}{B \lg N}}$ I/Os, when 
    $K = \Omega \left( \frac{B \lg^2 N}{q} \right)$. 
    This result achieves a $\lg{N}$ factor improvement on the previous 
    space cost in~\cite{dem_iac_lan_2004} for the tree structure. 
    We further show that when the size, $q$, of a key is constant that 
    we improve the I/O efficiency for the case where $K=\Omega(B \lg{N})$ 
    and $K=O(B \lg^2{N})$ from $\Omega(\lg{N})$ to $O(\lg{N})$ I/Os. 
  }
\end{enumerate}

This work has been published in \emph{Algorithmica} 
\cite{DBLP:journals/algorithmica/DillabaughHM12}.
The significance of the $\lg{N}$-bit space reduction achieved in this 
chapter and the next is twofold. 
First, it reduces storage requirements for massive datasets where the
amount of storage needed may be a substantial consideration.
Second, as demonstrated by our second result for Chapter~\ref{chp:leaf-root},
the increase in block size can potentially improve I/O-efficiency.

\subsection{Data Structures for Path Traversal in Bounded-Degree Planar Graphs}

In Chapter~\ref{chp:succinct_graphs} we present succinct data structures for 
representing bounded-degree planar graphs such that paths can be traversed 
efficiently in the external memory setting.

\begin{enumerate}
  \item{
    We describe a data structure for bounded-degree planar graphs that allows traversal of 
    a path of length $K$ (where a path may be defined as a sequence of $K$ edge adjacent
    vertices) using $O(\frac{K}{\lg B})$ I/Os. 
    For a graph with $N$ vertices, this matches the I/O complexity of 
    Agarwal~\etal~\cite{DBLP:conf/soda/AgarwalAMVV98}
    but improves the space complexity from $O(N \lg{N}) + 2Nq$ bits to $O(N) + Nq + o(Nq)$ bits, 
    which is considerable in the context of huge datasets. 
    (For example, with keys of constant size, the improvement in the space is by a factor of $\lg{N}$.)
  }
  \item {
    We adapt our structure to represent triangulations. 
    If storing a point requires $\bitsPerPoint$ bits, we are able to store the triangulation in
    $N \bitsPerPoint + \OhOf{N} + \ohOf{N \bitsPerPoint}$ bits so that any path crossing $K$
    triangles can be traversed using $\OhOf{K / \lg B}$ I/Os.
    Again, the I/O efficiency of our structure matches that
    of~\cite{DBLP:conf/soda/AgarwalAMVV98} with a similar space improvement
    as for bounded-degree planar graphs.
  }
  \item {
    We show how to
    augment our triangulation representation with $\ohOf{N\bitsPerPoint}$ bits of extra
    information in order to support point location 
    queries using $\OhOf{\log_B N}$ I/Os.
    Asymptotically, this does not change the space requirements.
  }
  \item {
  We describe several
  applications that make use of our representation for triangulations.
  We demonstrate that reporting terrain profiles and trickle paths takes
  $\BigOh{K / \lg B}$ I/Os.
  We show that connected component queries, that is, reporting a set of
  triangles that share a common attribute and induce a connected
  subgraph in the triangulation's dual, can be performed using
  $\BigOh{K / \lg B}$ I/Os when the component being reported is convex and consists
  of $K$ triangles.
  For non-convex regions with holes, we achieve a query bound of
  $\BigOh{K / \lg B + h \log_B h}$, where $h$ is the number
  of edges on the component's boundary.
  To achieve this query bound, the query procedure uses
  $\BigOh{h \cdot (\bitsPerKey + \lg h)}$ extra space.
  Without using any extra space, the same query can be answered using
  $\BigOh{K / \lg B + h' \log_B h'}$ I/Os, where $h'$ is the number of
  triangles that are incident on the boundary of the component.
  }
\end{enumerate}

A preliminary version of this research was published in 2009 at the
\emph{20\textsuperscript{th} International Symposium on Algorithms and
Computation} \cite{DBLP:conf/isaac/DillabaughHMZ09}.
A journal version of the paper has been submitted 
to \emph{Algorithmica}.

\subsection{Path Traversal in Well-Shaped Tetrahedral Meshes}

Chapter~\ref{chp:mesh_trav} presents results on efficient path traversal 
in the external memory setting for well-shaped meshes in $\threeD$.
The data structures presented in this chapter are not succinct.
A \emph{well-shaped} mesh is composed of simplicies for which the aspect
ratio is bounded (see definition in Section \ref{ssec:geom_seperators}).

\begin{enumerate}
  \item{
    We show how to represent a tetrahedral mesh with $N$ tetrahedra in 
    $\BigOh{N/B}$ blocks 
    such that an arbitrary path of length $K$ can be traversed 
    in $\BigOh{\frac{K}{\lg{B}}}$ I/Os.
    Unlike our earlier results, this structure is not succinct.
  }
  \item{
    For a convex mesh stored using our representation, we demonstrate that 
    box queries which report $K$ tetrahedra can be performed in 
    $\BigOh{\frac{K}{\lg{B}}}$ I/Os.
  }
  \item {
    Given an arbitrarily-oriented plane which intersects at some point a convex 
    tetrahedral mesh, report the intersection of all tetrahedra which intersect 
    the plane. 
    For a plane that intersects $K$ tetrahedra this would require 
    $\BigOh{\frac{K}{\lg B}}$ I/Os.
   }
\end{enumerate}

These results were previously published at the 2010
\emph{Canadian Conference on Computational Geometry}
\cite{DBLP:conf/cccg/Dillabaugh10}.

\subsection{Jump-and-Walk Search in Well-Shaped Meshes}

In Chapter~\ref{chp:jump_and_walk} we present results on point location queries 
in $\plane$ and $\threeD$ using the \emph{Jump-and-walk} technique. 
Jump-and-walk is a simple two-step strategy for solving point location queries;
the first step involves a jump to some point near the query point, followed
by a walk step through the subdivision or mesh being searched.
We present no data structures, either succinct or efficient in the external 
memory setting, but rather we investigate a technique that has application to path 
traversal in both triangular and tetrahedral meshes. 
Our key results are as follows:

\begin{enumerate}
  \item {
    Given a well-shaped mesh $\mathcal{M}$ in $\plane$ or $\threeD$, 
    jump-and-walk search can be performed in the time required to perform a 
    nearest neighbour search on the vertices of $\mathcal{M}$, plus 
    {\em constant time} for the walk-step to find the triangle/tetrahedron 
    containing the query point.
  }
  \item{
    Given a well-shaped mesh $\mathcal{M}$ in $\plane$ or $\threeD$, 
    jump-and-walk search can be performed in the time required to perform an
    \emph{approximate} nearest neighbour (see Definition~\ref{def:ANN}) 
    search on the vertices of $\mathcal{M}$, plus {\em constant time} 
    for the walk-step to find the triangle/tetrahedron containing the 
    query point.
  }
  \item {
    We implement our jump-and-walk method in $\plane$, and present 
    experimental results for walks starting with the nearest
    and approximate nearest neighbours.
  }
\end{enumerate}

A preliminary version of this work was published at the 2011
\emph{Canadian Conference on Computational Geometry}
\cite{DBLP:conf/cccg/CarufelDM11}.
A journal version of the paper has been submitted to the 
 \emph{Journal of Computational Geometry}.

\section{Organization}

The remainder of this thesis is organized in the following
manner.
Chapter~\ref{chp:background} presents the relevant theoretical
background to the topics disussed in the later chapters.
Chapters~\ref{chp:leaf-root} through
\ref{chp:jump_and_walk} each present the results of an 
individual research paper, as outlined in 
Section~\ref{sec:res_summary}.
Each chapter is self-contained and includes a summary
of results, previous work, background, etc.
Separate chapters do not, however, include individual reference 
sections.
Given this layout there is some overlap between 
Chapter~\ref{chp:background} and the background sections
contained in the specific chapters.
However, the thesis background section provides more in-depth
coverage of topics only lightly touched upon in the
later chapters.
Finally, Chapter~\ref{chp:conclusions} summarizes the thesis, and 
identifies open problems and topics for future research
arising from this work.
\chapter{Background}\label{chp:background}

This chapter presents background information on the key structures and
techniques used in this thesis. 
The chapters presenting thesis results (Chapters 3 through 6) are intended to be 
self-contained; as a result, there may be some minor overlap with the current 
chapter's content.

\section{Trees and Graphs}

We begin by providing a few basic definitions related to graphs and trees. 
The description of graphs is largely based on \cite{diestel_GT_2000},
\cite{harary1969}, and
\cite{lipton_tarjan_pst1979}, while the description for trees is based on 
\cite{cormen_book_2004} and \cite{Rosen:2002:DMA:579402}.

A \emph{graph} $G = (V,E)$ consists a of set of \emph{vertices}, $V$, and a set of 
\emph{edges}, $E$.
An edge in the graph is defined by an unordered vertex pair $\{v,w\}$, in which 
case we say that $v$ and $w$ are \emph{adjacent}.
The edge $\{v,w\}$ is said to be \emph{incident} on its endpoints $v$ and $w$.
This definition of graphs describes what may be more specifically called
a \emph{simple} graph.
Other types of graphs may allow directed edges (in which case
the edge pairs are ordered), loops (where $w = v$), or multiple edges connecting
the same pair of vertices.
This later class of graph is known as a \emph{multigraph}.
In this thesis any reference to a graph will refer to a simple graph, unless 
stated otherwise.
For a vertex $v$ we call any vertex adajcent to $v$ a \emph{neighbour}, and
collectively refer to the set of neighbours as the \emph{neighbourhood} 
of $v$.  
The cardinality of the neighbourhood of $v$ is the degree of $v$, 
denoted $d(v)$.

For graph $G = (V,E)$ the \emph{induced subgraph} is the subset $G' =(V',E')$, including
vertices $V' \subset V$, and edges $E' = \{ (v,w) | (v,w) \in E$~and~$v, w \in V'\}$.
A \emph{path} in a graph is a sequence of edges connecting adjacent vertices.
A path is said to traverse this sequence of edges.
For a simple graph, the path can also be represented by a sequence of adjacent
veritices, as this sequence will identify the set of edges traversed along
the path.
A graph is \emph{connected} if for every pair of vertices in the graph there is a 
path connecting those vertices.
A \emph{cycle} in a graph is a path which begins and ends at the same vertex
without visiting any edge more than once.

A graph is said to be \emph{embedded} in a surface $S$ when it is drawn on $S$ 
such that no two edges intersect.
If the graph \emph{can} be embedded in the plane ($\plane$), then it is said to
be a \emph{planar graph}, while a graph that \emph{has been} embedded in the
plane is called a \emph{plane graph}.
A plane graph subdivides the plane into regions called \emph{faces}, including
an unbounded region called the \emph{exterior} face.
The collection of edges bounding a face forms a cycle. 
In this thesis, unless explicitly stated otherwise, we will assume each face of 
a plane graph is empty, such that it contains no other face, edge, or vertex.

A \emph{tree} $T$ is a connected graph with no cycles. 
The vertices of a tree are frequently referred to as \emph{nodes}. 
A rooted tree is a tree for which we identify a special node called the 
\emph{root}; any reference to a tree in this paper can be assumed to 
refer to a rooted tree, unless explicitly stated otherwise.
The \emph{depth} of a node $v \in T$, denoted $\funccase{depth}(v)$, is the number
of edges along the path from $v$ to $\funccase{root}(T)$. 
For the node $v$, we call the next node on the path from $v$ to $\funccase{root}(T)$ 
the \emph{parent} of $v$. 
If $p$ is the parent of $v$, then we say $v$ is a child of $p$, and if $v$ and $m$ 
are children of $p$ then $v$ and $m$ are \emph{siblings}. 
A node with no children is a \emph{leaf} node. 
The \emph{height} of a tree, normally denoted $h$, is defined as the maximum 
depth of any leaf in the tree. 
All nodes at depth $i$ in the tree $T$ are said to be at the 
$i$\textsuperscript{th} level of $T$. 

A tree is a \emph{binary} tree if nodes in the tree have at most 
two children. 
It is sometimes helpful to distinguish between the child nodes in a binary tree 
where a child node may be labeled as a left child or a right child.
The order of the children is significant, such that if a node has only one 
child it is still known whether this is the left or right child. 
A binary tree labeled in this manner is a subclass of \emph{ordinal} trees in which 
the order of children is significant and is preserved. 
In a \emph{cardinal} tree the children of a node are still ordered, but they do
not have specific positions. 
In a cardinal binary tree, if a node has only a single child, there is no 
way of distinguishing if it is the left or right 
child \cite{DBLP:conf/swat/FarzanM08}.

There are a number of ways in which all the nodes of a tree can be 
traversed. 
We will describe two common traversals that we employ in our algorithms, 
both of which start with the root of the tree. 
In a \emph{preorder} traversal we report the current node, and then report its 
children in the order in which they appear. 
A preorder traversal on a tree is analagous to a depth-first traversal in a 
graph. 
In a \emph{level order} traversal we visit first all nodes at depth $0$ in $T$ 
(the root), followed by all nodes at depth $1$ and so forth to depth $h$.
The level order traversal is analagous to breadth-first search in a graph.

\subsection{Notation}
In this thesis we use the term $\lg{N}$ to represent $\log_2{N}$.
For a set $S$ we use $|S|$ to denote the cardinality of the set. 
The cardinality of a graph $G$ is defined as the sum of
the cardinalities of its vertex and edges sets, therefore $|G| = |V|+|E|$. 
For the classes of graphs with which we deal in this thesis 
$|E| = \BigOh{|V|}$ and therefore we will generally use $|V|$ to specify 
the size of a graph.

\section{Planar Subdivisions, Triangulations, and Meshes}

A \emph{planar subdivision} is a decomposition of the plane into polygonal
faces, induced by the embedding of a planar graph in $\plane$, with 
graph edges mapped to straight-line segments.
For the purposes of this thesis we will assume that faces are simple
polygons that do not contain any other component of the graph.
Two faces in such a graph are adjacent if they share one or more edges. 
We can form the \emph{dual graph} of a planar subdivision by representing each
face of the subdivision by a vertex, and then connecting all vertices which
represent adjacent faces by an edge.

A \emph{triangulation}, $\triang$ is a special case of a planar graph in 
whieach face, excluding the outer face, has exactly three edges.
A triangulation on $N$ vertices has at most $3N - 6$ edges.

A \emph{mesh} is a subdivision of space into regularly-shaped simplices.
A triangulation embedded in $\mathbb{R}^2$, with vertices connected by 
straight-line segments is a triangular mesh. 
A \emph{tetrahedral mesh} is an extension of this idea into $\threeD$.
Such a mesh is a subdivision of a volume into tetrahedra. 
The vertices of such a mesh are points in $\threeD$ and each tetrahedron 
consists of six edges and four triangular faces.
Two tetrahedra which share a face are said to be adjacent.

\section{Model of Computation}
\label{sec:model_of_comp}

Our results are presented in one of two models of computation; namely,
the Word RAM model and the \emph{external memory} (EM) model. 
Here we give brief descriptions of both models. 

\subsection{Word RAM Model}
In the classical \emph{random access machine} or RAM model 
\cite{CookReckhow_RAMModel}, 
the memory is assumed to be composed of an infinite number of cells with integer 
addresses $0, 1, \ldots, \infty$ . 
It is assumed that each cell can hold an arbitrarily large integer value. 
The set of basic operations available in this model includes arithmetic operators, 
loading and storing values in cells, and indirect addressing (where the value of one 
cell is used to address another cell).  

The Word RAM model \cite{DBLP:conf/stacs/Hagerup98} differs from the classical 
RAM model, primarily in that the range of integers than can be stored in a cell 
is restricted to the range ${0, 1, \ldots, 2^w-1}$, where $w$ is the 
\emph{word size}.
This model adds some operations to the classic RAM model that are natural on 
integers represented as strings of $w$ bits.  
These operations include bit shifting and bitwise boolean operations AND, OR, 
and NOT. 
Under this model, all of these operations can be performed in constant time.

In this thesis we will assume that $w = \BigTheta{\lg N}$, where $N$ is the 
maximum problem size, measured by the number of data elements to be processed.

\subsection{External Memory Model}
\label{ssec:em_model}

In the \emph{external memory} model (EM), we divide the memory available to the process
into two levels; the internal memory, which has some finite capacity, and the
external memory which is assumed to have infinite capacity.
This model is somewhat more representative of the real situation for modern 
computer systems, which have a relatively limited amount of 
RAM and a much larger, though not infinite, amount of on disk memory.
In the external memory literature, the terms external memory, and disk, are often
used interchangeably.

This model of computation is intended to deal with 
problems in which the number of objects to process is greater than the computer's 
internal memory.
While the disk has infinite capacity in the model, access to an object stored
on disk is much slower (up to one million times) than the fastest components of
the 
computer's internal memory \cite{DBLP:journals/fttcs/Vitter06}. 
We refer to the process of reading data from disk, or writing data to disk, as an 
Input/Output operation (I/O). 
In this research, where  our data structures are static, we can assume 'read' 
operations are being performed. 
Most of the time required during an I/O to access an object on disk is for the 
positioning of the read/write head, after which reading contiguous objects 
is quite fast. 
Thus whenever possible, when accessing disk it is typically preferable to read 
many objects during a single I/O operation. 
In order to take advantage of this situation, when an I/O is performed, many objects 
are read from disk to internal memory in a single transfer. 
We refer to the collection of objects transfered between the disk and internal
memory as a \emph{block}.

External memory algorithms are evaluated with respect to the number of I/Os,
or block transfers, that are required to solve the problem.
In order to represent a particular instance of an external memory algorithm, the following 
parameters are commonly used (where all sizes are measured in terms of the number 
of objects being processed):

\begin{center}
  \begin{tabular}{lll}
    N &=& problem size \\
    M &=& size of internal memory \\
    B &=& block size \\
  \end{tabular}
\end{center}

For any interesting problem instance we can assume that $N > M \ge B \ge 1$.  
In some cases, we wish to refer to the number of blocks involved in a computation,
thus we refer to the problem size as $n=N/B$ and the memory size as $m=M/B$.

There is a large body of research on EM data structures and algorithms.
Agarwal and Vitter \cite{DBLP:journals/cacm/AggarwalV88} addressed 
several fundamental algorithms in the EM setting, including sorting.
In the EM literature, the efficiency of algorithms is often evaluated
in comparison to known bounds on fundamental operations in this setting,
including: searching $Search(N) = \BigTheta{\log_B{N}}$; scanning
$Scan(N) = \BigTheta{N/B}$; and sorting 
$Sort(N) = \BigTheta{ \frac{N}{B} \log_{\frac{M}{B}} \frac{N}{B} }$.
An excellent survey of the EM literature can be found in 
\cite{DBLP:journals/fttcs/Vitter06}.

External memory algorithms and data structures are designed to use locality of 
reference to ensure good I/O performance.
The general idea is that since objects are loaded into memory in block-size
chunks, we attempt to arrange the data so as to take advantage of this.
Therefore, when a block is loaded into memory, many of the objects it contains 
will be used in processing before we are required to load another block.
A fundamental concept in our research is that of \emph{blocking} a data structure.
This is the process of taking the objects that are represented in our data
structure, and determining in which disk block (or blocks) they will be placed.
The blocking of graphs, which is central to Chapters~\ref{chp:succinct_graphs} 
and \ref{chp:mesh_trav} of this work, relies on the division of the graph using 
\emph{separators}.
Separators form the topic of the Section (\ref{sec:separators}), 
after which we discuss the blocking of trees and graphs. 
It is worth noting that the best blocking scheme for a data structure is often
dependent on the specific application, or expected access pattern.

\section{Separators}\label{sec:separators}

Let  $G$ be a graph on $N$ vertices\footnote{Technically, for a 
separator theorem to apply,
the graph $G$ must belong to a class of graphs closed under the subgraph relation.
That is to say, if $S$ is such a class of graphs and $G \in S$, and $G_1$ is
a subgraph of $G$, then $G_1 \in S$. 
Consider, for example, planar graphs; a
subgraph of a planar graph remains planar, and as such the class of planar graphs
is closed under the subgraph relation. The class of connected graphs would not be
considered closed under the subgraph relation, since a subgraph of a connected
graph could be disconnected.}; we say that a $f(N)$-vertex separator exists for 
constants $\delta < 1$ and $\beta > 0$ if the vertices of $G$ can be partitioned 
into sets $A$, $B$, and $C$ such that:

\begin{enumerate}
 \item There is no edge connecting a vertex in $A$ with a vertex in $B$.
 \item Neither $A$, nor $B$, contains more than $\delta \cdot N$ vertices.
 \item $C$ contains at most $\beta \cdot f(N)$ vertices.
\end{enumerate}

The function $f(N)$ is a function of $N$ that is asymptotically smaller than $N$; 
for example, $\sqrt{N}$ or $N^{2/3}$.
Given the $\delta$ value for a separator, we state that the separator $\delta$-splits 
the graph $G$, and that the separator size is $\beta \cdot f(N)$.
Weighted versions of separators also exist.
Depending on the class of graphs, different strategies for identifying
separators are required. 
For the class of planar graphs, Lipton and Tarjan~\cite{lipton_tarjan_pst1979} 
gave a separator theorem with $\beta = 2 \sqrt{2}$, $f(N) = \sqrt{N}$ and 
$\delta = 2/3$. 
They provide a constructive proof which leads to a linear time algorithm
to find such a separator.
Miller~\cite{miller_1986} studied the problem of finding \emph{simple cycle 
separators} for embedded planar graphs. 
A simple cycle separator is either a vertex, or a simple cycle of bounded size 
which $2/3$-splits $G$ with a simple cycle of size at most 
$2 \sqrt{2 \left\lfloor \frac{c}{2} \right\rfloor N}$, where $c$ is the maximum
face size in the graph.
In this setting $\delta = 2/3$, 
$\beta = 2 \sqrt{ \left\lfloor \frac{c}{2} \right\rfloor }$ and $f(N) = \sqrt{N}$.
This construction also requires linear time.

Frederickson~\cite{Frederickson87} developed a scheme for subdividing 
bounded-degree planar graphs into bounded size regions by recursively applying the 
planar separator theorem of~\cite{lipton_tarjan_pst1979}. 
Given a graph $G$ with vertex set $V$, $V$ is partitioned into overlapping 
sets called \emph{regions}. 
Vertices contained in just one region are \emph{interior} to that region 
while vertices shared between multiple regions are termed \emph{boundary} 
vertices. 
His result is summarized in Lemma \ref{lem:bg_fred_graph_sep}.

\begin{lemma} [\cite{Frederickson87}] \label{lem:bg_fred_graph_sep}
A planar graph with $N$ vertices can be subdivided into $\Theta(N/r)$ 
regions, each consisting of at most $r$ vertices, with $\OhOf{N/\sqrt{r}}$
boundary vertices in total.
\end{lemma}

We also give the following bound on the space required by the recursive 
application of the separator algorithm of Miller.

\begin{lemma}[\cite{miller_1986}]\label{lem:cycle_sep_size}
For an embedded planar graph $G$, on $N$ vertices, with faces of bounded maximum 
size $c$, recursively partitioned into regions of size $r$ using the cycle 
separator algorithm of Miller, the total number of boundary vertices is
bounded by $\BigOh{N/\sqrt{r}}$.
\end{lemma}

\begin{proof}
Miller's simple cycle separator on a graph of $N$ vertices has size at most 
$2 \sqrt{ 2 \left\lfloor \frac{c}{2} \right\rfloor N}$, which is 
$\BigOh{\sqrt{N}}$ for constant $c$. 
At each step in the partitioning, $G$ is subdivided into a separator 
$|S| = \BigOh{\sqrt{N}}$, plus subsets $N_1$ and $N_2$ each containing at most 
$\frac{2}{3}N$ vertices, plus the separator. 
Thus if $\frac{1}{3} \le \epsilon \le \frac{2}{3}$, we can characterize the size 
of the resulting subsets by 
$|N_1| = \epsilon N + \BigOh{\sqrt{N}}$ and $|N_2| = (1-\epsilon)N + \BigOh{\sqrt{N}}$. 
Following the proof of Lemma 1 in \cite{Frederickson87}, the total separator 
size required to split $G$ into regions of size at most $r$
becomes $\BigOh{\frac{N}{\sqrt{r}}}$.
\end{proof}

\subsection{Geometric Separators}\label{ssec:geom_seperators}

Miller~\etal~\cite{DBLP:journals/jacm/MillerTTV97} studied the problem of 
finding separators for sphere packings. 
More specifically, they developed a technique for finding geometric separators 
for $k$-ply neighbourhood systems, which we now define.

\begin{definition}[\cite{DBLP:journals/jacm/MillerTTV97}] 
A \emph{$k$-ply neighbourhood system} in $d$ dimensions is a set 
${B_1, \ldots, B_N}$, of $N$ closed  balls in $\mathbb{R}^d$, such that no point
in $\mathbb{R}^d$ is strictly interior to more than $k$ balls.
\end{definition}

The primary result of their research is summarized in 
Theorem~\ref{thm:miller_kply_sphere_separator}. 
As part of their research they show that there is a randomized linear time 
algorithm that can find such a separator.

\begin{theorem}[\cite{DBLP:journals/jacm/MillerTTV97}]\label{thm:miller_kply_sphere_separator}
Let $\Gamma = {B_1,\ldots,B_N}$ be a collection of balls that form a $k$-ply 
system in ${\mathbb{R}^d}$ space. 
There is a sphere, $S$, that partitions 
$\Gamma$ into three sets: $A$, the set of balls contained entirely in $S$; $B$, 
the set of balls entirely outside $S$; and $C$, the set of balls intersected by 
$S$ such that $|C| = \BigOh{k^{1/d} N^{1-1/d}}$ and 
$|A|,|B| \le \left( \frac{d+1}{d+2} \right) N$.
\end{theorem}

Based on this result, the authors were able show how separators could be found on
several concrete classes of graphs that can be represented as $k$-ply 
neighbourhood systems. 
These include intersection graphs and $k$-nearest neighbour graphs. 
Interestingly, the result for intersection graphs leads to a 
geometric proof of the planar separator theorem of Lipton and 
Tarjan~\cite{lipton_tarjan_pst1979}.
This separator also generalizes planar separators to higher dimensions, 
as it is applicable to any fixed dimensional space for dimension $d$.

In \cite{mttv_1998} the authors extend this work to the class of 
\emph{$\alpha$-overlap graphs}.  For a neighbourhood system  it is possible to 
associate an overlap graph with that system. The $\alpha$-overlap graph for a 
given neighbourhood system is defined as follows:

\begin{definition}[\cite{mttv_1998}]
Let $\alpha \ge 1$ be given, and let $(B_1,\ldots,B_n)$ be a collection of balls
that form a 1-ply neighbourhood system. 
If ball $B$ has radius $r$, then ball $\alpha \cdot B$ is the ball with the same
centre as $B$ with radius $\alpha \cdot r$.
The \emph{ $\alpha$-overlap graph} for this neighbourhood system is the undirected 
graph with vertices $V={1,\ldots,n}$ and edges:
\begin{equation}
E = {(i,j): B_i \cap (\alpha \cdot B_j) \ne \emptyset \mbox{ and } (\alpha \cdot B_i) 
\cap B_j \ne \emptyset}
\end{equation}
\end{definition}

For the class of  $\alpha$-overlap graphs the author's main result is summarized
in Theorem \ref{thm:alpha_overlap_separator}.

\begin{theorem}[\cite{mttv_1998}]\label{thm:alpha_overlap_separator}
For fixed dimension $d$ let $G$ be an  $\alpha$-overlap graph. Then $G$ has an
\begin{equation}\nonumber
\BigOh{\alpha \cdot n^{(d-1)/d}+q(\alpha,d))}
\end{equation}
separator that $(d+1)/(d+2)$-splits $G$. 
\end{theorem}

By demonstrating that finite-element simplical meshes, which satisfy a shape 
criterion, are overlap graphs, the authors show how geometric separators can be 
applied to meshes. 
Here we provide an overview of how well-shapedness is defined, as well as the proof 
that the geometric separator can be applied to finite-element simplical meshes 
(such as triangulations in $\mathbb{R}^2$ and tetrahedral meshes in 
$\mathbb{R}^3$).

Let $t$ be a simplex in a finite-element mesh.  The \emph{aspect ratio} of $t$,
denoted $\rho_t$,
may be defined as the ratio of the radius of the smallest sphere that contains 
$t$ to the radius of the largest sphere that can be inscribed in $t$.
We denote by $\rho$ the bound on the aspect ratio of all simplicies within a
the mesh.
These radii may be denoted by $R(t)$ (smallest containing sphere) and $r(t)$ 
(largest inscribed sphere). The aspect ratio $\rho_t$ of $t$ is then 
$\rho_t = R(t) / r(t)$. 

A sketch of this proof is as follows:
Let $\myMesh$ be a well-shaped mesh and let $G$ be a graph built on the interior 
vertices of $\myMesh$. 
Around each vertex $v$ in $G$, a ball is placed of radius $\frac{1}{2} \cdot r(t')$, 
where $t'$ is the simplex adjacent to $v$ for which the size of its maximum 
inscribed ball ($r(t')$) is minimized. None of these balls overlaps, so this 
forms a $1$-ply system.  
Because the aspect ratio of the simplices is bounded, a corresponding bound on 
the number of simplices adjacent to $v$ can be proven. 
Let $q$, which is a function of $\rho$ and $d$, be this bound. 
There is a bounded number of simplices neighbouring $v$, all connected under a 
neighbour relation, and of bounded aspect ratio. 
This leads to a bound on the largest value for $R(t)$ for a simplex adjacent 
to $v$ which is $\rho^q \cdot r(t')$. 
Since the radius of $B_v$ is $\frac{1}{2} \cdot r(t')$, selecting 
$\alpha = 2 \rho^q$ ensures that $\alpha B_v$ intersects the balls 
of all neighbours of $v$ in $G$.

\subsection{Well-Shaped Meshes}\label{ssec:ws_mesh}

Well-shaped meshes play an important role in the research presented in this
thesis.
When the aspect ratio of all simplicies within is bounded by a constant, $\rho$,
then we say the mesh is well-shaped, with bounded aspect ratio.
In this thesis, two slightly different definitions of aspect ratio are used.
In each case, $r(t)$ defines the radius of the incircle(sphere) of a simplex, but 
the meaning of $R(t)$ is changed.
In Chapter~\ref{chp:mesh_trav} we use the definition of \cite{mttv_1998} given
above, where $R(t)$ represents the radius of the smallest enclosing sphere.
This definition is used here because in Chapter~\ref{chp:mesh_trav} our structures are based on
Geometric separators from \cite{mttv_1998}.
In Chapter~\ref{chp:jump_and_walk} we alter the definition of $R(t)$ to be
the circumcircle(sphere) of the simplex.
This second definition was used in order to allow us to make the precise calculations 
regarding the number of simplicies visited on a straight-line walk through a
well-shaped mesh.
While the two definitions lead to different values of $\rho$ for a given mesh,
they are conceptually the same.

\section{Blocking Trees and Graphs}

\subsection{Blocking Trees}\label{ssec:tree_blocking}

There are existing tree structures, particularly the B-tree and its numerous 
variants, that are designed to support efficient operations in external memory. 
However, there are problems for which trees may be layed out in external memory,
or blocked, in a way the gives better I/O performance than the corresponding 
queries on a standard B-tree. 
Furthermore, there are instances in which the 
topology of the tree is important, and as such we are not free to alter the 
structure of the tree in order to improve I/O performance. 
As a result, our interest is in covering
the nodes of the tree with a set of blocks (again potentially allowing a storage 
blow-up) so that queries may be performed efficiently.

We are particularly interested in two problems:
\begin{enumerate}
 \item Blocking a static tree so that bottom-up traversals can be performed. 
 In this case, we are given a node in the tree and we wish to report the path 
 (of length $K$) from the node to the root.
 \item Blocking a static, binary tree so that we can perform top-down traversals,
 starting at the root of the tree, and report the path to a node at depth $K$ 
 (in other words, a path of length $K$).
\end{enumerate}

For the first query, Hutchinson~\etal~\cite{hutch_mesh_zeh_2003} gave a blocking 
strategy that uses linear space and reports bottom-up queries in $\BigOh{K/B}$ 
I/Os.
The \emph{giraffe-tree} of Brodal and Fagerberg~\cite{brodal_fagerberg_06} also 
reports such queries with $\BigOh{K/B}$ I/Os using linear space, and it also works 
in the cache-oblivious model. 
As our bottom-up query structure is based on the strategy of Hutchinson~\etal, 
we present their strategy here with the simplification that we assumed the 
tree is blocked on a single disk. 

Let $T$ be a tree rooted. 
Given a vertex $v$ at depth $K$ in $T$, we wish to 
report the path from $v$ to $\funccase{root}(T)$. 
The scheme works by splitting the tree into layers of height $\tau B$, where
$\tau$ is a constant such that $0 < \tau < 1$. 
Within each layer the tree is split into a set of subtrees rooted at the top
level of the layer.
The layers are then blocked in such a way that for each vertex $v$ in the layer
there is a block that includes the full path from $v$ to the root of its subtree 
within the layer. This guarantees that with a single I/O we are able to traverse 
$\tau B$ levels towards $\funccase{root}(T)$.

Let $h$ denote the height of $T$, and let $h'=\tau B$ be the maximum height of 
each layer; $T$ is then divided into ${\left\lceil h/h'\right\rceil}-1$ layers. 
We use $L_i$ to denote the layer from level $i \cdot h'$ to level 
$(i+1) \cdot h' - 1$. 
The set of subtrees produced for layer $L_i$ is rooted at level $i \cdot h'$.
Starting at the root of the leftmost subtree rooted at level 
$i \cdot h'$, the vertices of level $L_i$ are numbered according to their 
preorder labels from $1$ to $|L_i|$.

We next divide the vertices of $L_i$ into sets $V_0, \ldots, V_s$ of size 
$t = B - \tau B$ vertices, where set $V_j$ contains the vertices with preorder 
labels ${jt, \ldots, (j+1)t - 1}$. For each set $V_j$, let $A_j$ be the set of 
all vertices in layer $L_i$ not in $V_j$ but on a path from a vertex in $V_j$ 
to level $i \cdot h'$ ($L_i$'s top level). We store $L_i$ on disk in $|L_i| / s$ 
blocks, where block $B_j$ stores the subtrees of $T$ induced on vertices 
$V_j + A_j$.

In order for the blocking scheme described above to work, it must be the case that 
$|A_j| \leq \tau B$, so that $V_j + A_j \leq B$. 
Here we provide a alternate proof to that of \cite{hutch_mesh_zeh_2003} that 
this property holds. 

\begin{claim}
$|A_j| \leq \tau B$.
\end{claim}

\begin{proof}
Since a layer contains $\tau B$, levels it is sufficient to prove that $A_j$ 
contains no two vertices on the same level. 
Assume that this were not the case, and let $q,r \in A_j$ be vertices on the same 
level of $T$ with preorder numbers $p(q)$ and $p(r)$. 
Assume that $p(q) < p(r)$. 
Consider the subtrees rooted at $q$ and $r$, denoted $T(q)$ and $T(r)$. $T(q)$ 
and $T(r)$ both include vertices from $V_j$, otherwise they would not be 
included in $A_j$. 
In a preorder labeling, all vertices in $T(q)$ are labeled 
prior to $T(r)$, including $r$. However, $T(r)$ contains vertices from $V_j$, 
but the vertices in $V_j$ are numbered consecutively; therefore either $T(r)$ contains 
no vertices of $V_j$ or $r$ belongs to $V_j$. In either case, $r$ is not in $A_j$.
\end{proof}

By packing non-full blocks together on disk it can be shown that the total 
space required in the above blocking scheme is $O(N/B)$ blocks. 
A path of length $K$ from $v$ to $\funccase{root}(T)$ will visit at 
most $K/ \tau B$ layers, and the subpath within each layer can be traversed 
with a single I/O. 
This leads to the following lemma.

\begin{lemma}[\cite{hutch_mesh_zeh_2003}]\label{lem:bot_up_hutch}
A rooted tree $T$ can be stored in $O(N/B)$ blocks on disk such that a bottom-up 
path of length $K$ in $T$ can be traversed in $O(K/\tau B)$ I/Os, where 
$0 < \tau < 1$ is a constant.
\end{lemma}

We have just described a blocking scheme for blocking arbitrary trees for
bottom-up traversal.
We now describe a blocking scheme for top-down traversal in binary
trees.
Demaine \etal~\cite{dem_iac_lan_2004} described an optimal blocking technique 
for binary trees that bounds the number of I/Os in terms of the depth of a 
node within $T$. 
We provide an overview of this scheme, as we also employ it 
in the data structures described in Chapter~\ref{chp:leaf-root}.
The blocking is performed in two phases. 

In the first phase, the highest $c \lg N$ levels of $T$, where $c > 0$ is a 
constant,  are blocked as if they contained a perfect binary tree. 
(Again, we can think of splitting the top $\BigOh{\lg N}$ levels of $T$ into 
layers, ${L_0,\ldots,L_{c \lg{N}/ \lfloor \lg{(B+1)} \rfloor}}$ of height 
$h = \lfloor \lg{(B+1)} \rfloor$. 
Layer $L_0$ includes a single subtree rooted at 
$\funccase{root}(T)$, of height $h$, while layer $L_i$ is blocked into $B^i$ 
subtrees). 

Conceptually removing the nodes blocked in the Phase 1 blocking results in a 
set of subtrees root at level $\lfloor c \cdot \lg{N} \rfloor + 1$.
In Phase 2 we block these subtrees. 
Let $w(v)$ be the size of the subtree rooted at vertex $v$ and let $l(v)$ and 
$r(v)$ represent $v$'s left and right children, respectively. A null child has 
a weight of zero. 

Let $A$ be the number of vertices from the current subtree that can be placed 
in the current block without the block size exceeding $B$. 
Start with an empty block so that initially $A=B$. 
Given the root $v$ of a subtree of $T$, the blocking algorithm recursively selects 
a set of nodes to add to the current block as follows:

\begin{itemize}
 \item If $A = 0$, then create a new empty block and repeat the recursion 
 starting with $v$ and with $A=B$.
 \item If $A > 0$, then add $v$ to the current block and divide the remaining 
 $A-1$ capacity of the current block between $l(v)$ and $r(v)$ as follows: 
 \begin{eqnarray*}
   A_{l(v)} &=& (A-1) \cdot \frac{w(l(v))}{w(v)} \\
   A_{r(v)} &=& (A-1) \cdot \frac{w(r(v))}{w(v)}
 \end{eqnarray*}
\end{itemize}

Given a binary tree layed out according to this blocking scheme the authors
prove the following result:

\begin{lemma}[\cite{dem_iac_lan_2004}]\label{lem:top_down_demaine}
For a binary tree $T$, a traversal from the root to a node of depth $K$ 
requires the following number of I/Os:
\begin{enumerate}
   \item $\Theta( K/\lg(1+B))$, when $K = O(\lg N)$,
   \item $\Theta( \lg N/ (\lg (1+B \lg N/K)))$, when $K = \Omega(\lg N)$ and 
   $K = O(B \lg N)$, and 
   \item $\Theta(K/B)$, when $K = \Omega(B \lg N)$. 
\end{enumerate}
\end{lemma}

\subsection{Blocking Graphs}

The model of graph blocking with which we work was described by 
Nodine~\etal~\cite{ngv_1996}. 
We are interested in performing external searching in a graph, $G=(V,E)$, that is
too large to fit into memory. 
The model is generic in that we assume we have some algorithm which generates 
a \emph{path} through the graph. 
A path is defined as a sequence of edge-adjacent vertices that are visited 
during the execution of the algorithm. 
We evaluate the effectiveness of algorithms with respect to the length of the 
path to be traversed. 
The length of the path is measured by the number of vertices visited in the 
sequence. 
We denote the path length by $K$.  
There is no restriction on the number of times a vertex may be visited, therefore 
it is possible that $K > N$.

There are a number of algorithms that may be at work. 
These include a blocking algorithm that involves the assignment of vertices to 
``blocks''. 
This can be viewed as a preprocessing step. 
The second algorithm at play is a processing algorithm that generates the path 
to be followed. 
Finally, there is a paging algorithm that determines which block should be loaded 
or evicted when the path leaves the current block.

There are a number of key properties of this model:
\begin{enumerate}
 \item The data associated with each vertex is stored with the vertex, and is of 
 fixed size.
\item A block stores at most $B$ vertices. 
\item Vertices may be present in more than one block.
\item Any block containing $v$ is said to \emph{serve} $v$.
An I/O operation is incurred when the next step in the processing algorithm will 
extend the path to a vertex $u$ that is not currently served.
\end{enumerate}

For the algorithms that we develop, we are not generally concerned with the paging 
algorithm used to load/evict blocks. 
Rather, we will generally assume that memory
is limited, so that $M=B$, and we can only hold a single block in memory at any
one time.
In Nodine's model, blocking schemes are evaluated with respect to two attributes; 
the \emph{blocking speed-up} denoted $\alpha(B)$ and the \emph{storage blow-up} 
denoted $s$. 
The blocking speed up, $\alpha(B)$, is the worst case number of vertices that 
can be traversed along the path before a vertex is encountered that is not served 
by the current block, incurring an I/O. One manner in which blocking may be 
improved is replicating / duplicating vertices so that they appear in multiple 
blocks. 
The storage blow-up measures the extent to which vertices are duplicated. 
If the blocked size of a graph is $S$, then the storage blow up $s$ is calculated
by $s = S/N$. 
If, for some blocking strategy, $s$ is a constant, then the storage requirement 
for $G$ is $\BigOh{N/B}$.

Nodine~\etal  studied the graph-blocking problem for a wide range of trees and 
graphs. 
Of greatest interest to our research, is a lower bound found on the blocking 
speed up of $\BigOmega{ \log_d B }$ when $s=2$ for $d$-ary trees (where $d$ is 
the maximum degree of a vertex). 
This result also implies an equivalent lower bound of 
$\BigOmega{ \log_d B }$ on the blocking speed up for degree 
$d$ planar graphs.

Agarwal~\etal~\cite{DBLP:conf/soda/AgarwalAMVV98} presented an algorithm that 
achieved such a speed up of $\BigOmega{ \log_d B }$ for planar graphs with $s=2$. 
As their algorithm is one of the building blocks of our graph blockings, we 
describe their technique in some detail. 

Let $G=(V,E)$ be a graph on $N$ vertices. 
The set $(V_1,\ldots,V_n)$ is a 
covering of the vertices of $V$ such that $V_i \subset V$ and 
$\cup_{i=1}^n V_i = V$. 
Each subset, $V_i$, is refered to as a region. A vertex $v$ is interior to region 
$V_i$ if all 
neighbours of $v$ also belong to $V_i$, otherwise we call $v$ a \emph{boundary} 
vertex. 
The boundary vertices are those vertices which are present in two or more regions. 

A $B$-division of $G$ is a covering $(V_1,\ldots,V_i)$ of $V$ by 
$T=\BigOh{N/B}$ blocks such that:
\begin{itemize}
 \item Each region has $\le B$ vertices.
 \item Any vertex $v$ is either a boundary or interior vertex. 
 \item The total number of boundary vertices over all regions is $\BigOh{N/\sqrt{B}}$.
\end{itemize}

The graph-partitioning scheme of Frederickson~\cite{Frederickson87} recursively 
applies the planar separator algorithm of Lipton and 
Tarjan~(\cite{lipton_tarjan_pst1979}, \cite{DBLP:journals/siamcomp/LiptonT80}), 
and for an $N$ vertex planar graph divides the graph into $\BigOh{N/r}$ regions 
of $r$ vertices and a total of $\BigOh{N/\sqrt{r}}$ boundary vertices.  
Setting $r = B$ results in a suitable $B$-division of $G$.

Let $S$ be the set of boundary nodes. 
By Lemma \ref{lem:bg_fred_graph_sep} above 
we have $S=\BigOh{N/\sqrt{r}}$. 
For each $v \in S$ grow a breadth-first search tree rooted at $v$, until there 
are $\sqrt{B}$ vertices in the tree. We call the subgraph of $G$ over this set 
of vertices a $\sqrt{B}$-neighbourhood.
Since $G$ is of bounded degree $d$, the $\sqrt{B}$-neighbourhood around a vertex 
$v$ contains all vertices of $G$ at distance less than $\frac{1}{2} \log_d B$ 
from $v$. 

The graph $G$ is then blocked on disk as follows:
the $B$-division is generated and each region $V_i$ of at most $B$ vertices
is stored in a single block on disk. The $\BigOh{N/\sqrt{B}}$ boundary vertices
of $S$ are grouped into $\BigOh{N/B}$ subsets such that the subsets contain at 
most $\sqrt{B}$ boundary nodes. 
The $\sqrt{B}$-neighbourhoods of all boundary vertices within a subset are 
stored in a single block on disk. 
The total storage for both the regions and the $\sqrt{B}$-neighbourhoods is 
$\BigOh{N/B}$.

The paging algorithm for path traversal works in the following manner. 
Assume the path is currently at vertex $v$ which is interior to some region 
$V_i$.
We follow the path until we reach a vertex $u \in V_i$ which is a boundary 
vertex. 
The $\sqrt{B}$-neighbourhood of $u$ is then loaded into memory and traversal 
continues. 
When the path leaves the $\sqrt{B}$-neighbourhood of $u$ at a vertex $v'$, 
the region containing $v'$, $V_j$, is loaded into memory, and so on.  
In the worst case, each time a new region is loaded we are unlucky and we load 
a boundary vertex, therefore we must load a new $\sqrt{B}$-neighbourhood. 
However, within the $\sqrt{B}$-neighbourhoods we are guaranteed to progress 
at least $\BigOh{\log_d B}$ steps along the path. 
Thus the blocking speed-up of this algorithm is 
$\alpha(B) = \BigOh{\log_d B}$; in other words, we can traverse a path of
length $K$ in $\BigOh{K/ \log_d B}$ steps. 
To summarize we have the following lemma:

\begin{lemma}[\cite{DBLP:conf/soda/AgarwalAMVV98}]\label{lem:agarwal_EM_graph_blocking}
A planar graph with bounded degree $d$ can be stored using $O(N/B)$ blocks so 
that any path of length $K$ can be traversed using $\BigOh{K/\log_d B}$ I/Os.
\end{lemma}

Baswana and Sen~\cite{DBLP:journals/algorithmica/BaswanaS02} 
describe a revision of the blocking scheme of Agarwal that yields
an improved I/O efficiency for some graphs. 
As with the blocking strategy of Agarwal~\etal~\cite{DBLP:conf/soda/AgarwalAMVV98}, 
this new strategy recursively identifies a separator on the nodes 
of the graph $G$, until $G$ is divided into regions of size $r$ which can fit 
on a single block (in this context, region is used generically to refer to the 
size of partitions in the graph, as in Lemma \ref{lem:bg_fred_graph_sep}). 
However, rather than store $\alpha$-neighbourhoods of size $\sqrt{r}$ for each 
boundary vertex, larger $\alpha$-neighbourhoods are stored for a subset of the 
boundary vertices. 
In order for this scheme to work, it is necessary that boundary vertices be 
packed close together, so that the selected $\alpha$-neighbourhoods cover a
sufficiently large number of boundary vertices. 
This closeness of boundary vertices is not guaranteed using the technique of 
\cite{DBLP:conf/soda/AgarwalAMVV98}, but if the alternate separator algorithm 
of Miller \cite{miller_1986} is incorporated, then region boundary vertices 
are sufficiently packed.

Before presenting the main result of 
\cite{DBLP:journals/algorithmica/BaswanaS02} we define a few terms. 
For graph $G=(V,E)$ let $r^{-}(s)$ be the minimum depth of a breadth-first 
search tree of size $s$ built on any vertex $v \in V$. Let $c$ denote the 
maximum face size in $G$, and let $s$ be the size of the $\alpha$-neighbourhoods
constructed around a subset of the region boundary vertices. 
For regions with $r = B$ we have the following result based on 
\cite{DBLP:journals/algorithmica/BaswanaS02}:

\begin{lemma}\label{lem:Baswana_path_traversal}
Planar graph $G$ of maximum face size $c$ can be stored in 
$\BigOh{\frac{N}{B}}$ blocks such that a path of length $K$ can be 
traversed using $\BigOh{\frac{K}{ r^{-}(s) }}$ I/O operations where 
$s = \textbf{min}\left(r^{-}(s)\sqrt{\frac{B}{c}}, B \right)$.
\end{lemma}

\section{Succinct Data Structures}
\label{sec:bckgrnd-succint-ds}

Our results in Chapters \ref{chp:leaf-root} and 
\ref{chp:succinct_graphs}, in
addition to being efficient in the EM setting, are succinct.
We now introduce succinct data structures, and some 
of the methods used in creating such data structures.

Consider a binary tree that perhaps stores a single ASCII
character per node. 
In a typical representation of such a tree, an individual
node would consist of the character and two pointers, one 
to each child node.
The bulk of the space required to store this representation
of a binary tree will be used to store pointers, not the 
character \emph{data} stored in the tree.
A similar situtation can arise for many data structures
in which the structural component of the data structure
accounts for much, or even most, of the overall space 
requirements.

Now say we have such a data structure and we want to 
reduce the space requirements.
We could save disk space by compressing the 
data structure using any one of the many compression
schemes available. 
However, this approach has two significant drawbacks.
Firstly, the compressed data structure is unusable in 
its compressed state, therefore we incur the additional computational
cost of decompressing it.
Secondly, the in-memory footprint of the decompressed
data structure results in no savings in space in terms of
RAM when the time comes to actually use the data structure.

The goal of succinct data structures is to
compress the \emph{structural} part of a data structure as
much as possible, but to leave it in a state where certain
desired operations can still be performed efficiently. 
More specifically, for a data structure to be succinct,
it must have a space complexity that matches the
information-theoretic lower bounds for the data structure
whilst permitting a constant time\footnote{or at least efficient 
when compared to similar operations in non-succinct structures}
set of desired 
operations~\cite{DBLP:conf/swat/FarzanM08}.

Two fundamental, and closely related, data structures used as building blocks in 
succinct data structures are bit vectors (arrays) and balanced parenthesis 
sequences.

A bit vector in this context is an array $\vvec[] [1 \ldots N]$, on $N$ bits, that 
supports the operations \emph{rank} and \emph{select} in constant time. 
These operations are defined as follows:

\begin{itemize}
 \item $\rankop[1]{\vvec,i}$ $(\rankop[0]{\vvec,i})$ returns the number of $1$s ($0$s) 
 in $\vvec[] [1 \ldots i]$.
 \item $\selop[1]{\vvec,r}$ $(\selop[0]{\vvec,r})$ returns the position of the 
 $r$\textsuperscript{th} occurrence of $1$ ($0$) in $\vvec$. 
 In other words, it returns the position $i$ such that 
 $\rankop[1]{\vvec,i}$ $(\rankop[0]{\vvec,i})$ returns $r$.
\end{itemize}

Jacobson \cite{jac_1989} described data structures for $\rankop$ and $\selop$ 
which required $N + \LittleOh{N}$ bits storage, but supported $\selop$ in 
suboptimal $\BigOh{\lg N}$ time under the word RAM model. 
Clark \cite{clark_96t} described how to achieve constant-time performance for 
the $\selop$ operation, again in $N + \LittleOh{N}$ space. 
Finally, Raman~\etal~\cite{DBLP:journals/talg/RamanRS07} presented a structure 
on $\lg\binom{N}{R} + O(N \lg \lg N / \lg N)$ bits to support the access to each 
bit, as  well as $\rankop$ and $\selop$ operations in $O(1)$, where $R$ is the 
number of $1$s in $A$. 
When $R$ is $\LittleOh{N}$, the space for this structure 
becomes $\LittleOh{N}$.

The following lemma summarizes the results of Jacobson~\cite{jac_1989} and 
Clark and Munro~\cite{clark_96t} (part (a)), and 
Raman~\etal~\cite{DBLP:journals/talg/RamanRS07} (part (b)):

\begin{lemma}
  \label{lem:rank_select}
  A bit vector $\vvec$, of length $N$, and containing $R$ $1$s can be represented 
  using either (a) $N + o(N)$ bits or (b) $\lg\binom{N}{R} + O(N \lg \lg N / \lg N)$ 
  bits to support the access to each bit, as well as $\rankop$ and $\selop$ 
  operations in $O(1)$ time (or $O(1)$ I/Os in external memory).
\end{lemma}

Closely related to bit vectors are balanced-parentheses sequences.
Let $S$ be a sequence composed of $2N$ matching \emph{open} '(' and 
\emph{close} ')' parentheses (or alternately $N$ pairs of parentheses).  
The relation to bit vectors can easily be seen by replacing the open parentheses
with $1$ bits and the close parentheses with $0$ bits. 
For a parenthesis sequence on $N$ pairs, Munro and Raman~\cite{munro_raman_01} 
describe $\LittleOh{N}$ bit auxiliary dictionary structures that enable 
several useful operations in constant time. 
Given a parenthesis sequence 
$P$, and an open parenthesis at position $P[m]$, we define two such operations, 
$\funccase{findclose}(P,m)$ and $\funccase{enclose}(P,m)$.

\begin{enumerate}
	\item $\funccase{findclose}(P,m)$ returns the position in $P$ of the 
	closing parenthesis matching the open parenthesis at $P[m]$.
	\item $\funccase{enclose}(P,m)$ for the parentheses pair with open 
	parenthesis $m$ return the position in $P$ of the open parenthesis of 
	the nearest enclosing pair. 
\end{enumerate}

Munro and Raman~\cite{munro_raman_01} demonstrated that both of these operations 
require only constant time with their data structures.

\subsection{Tree Representations}

Here we present two succinct representations for trees used in our research. 
The first is the \emph{level-order binary marked} representation (LOBM) for 
binary trees of Jacobson~\cite{jac_1989}, the second is the 
balanced-parenthesis sequence used by Munro and Raman~\cite{munro_raman_01} for trees 
of arbitrary degree.

A LOBM tree is represented as a bit vector. 
Consider the set of nodes in the original tree;
these nodes form the set of \emph{interior} nodes. 
An interior node which is a leaf has zero children, while other interior nodes 
may have one or two children. 
For all interior nodes that are missing either their left or right child, 
insert a new \emph{external} node in place of the missing child. 
The bit vector is then generated as follows 
(see Fig \ref{fig:lobm_tree_representation} ).
\begin{enumerate}
 \item Mark each interior node in $T$ with a $1$ bit, and each external node 
 with a $0$ bit.
 \item Traverse $T$ in level order, visiting the nodes at each level from 
 left to right, and record the bit value of each node in the bit vector, in 
 the order visited.
\end{enumerate}

For a tree $T$ let $\vvec_T$ be the LOBM bit vector representing the tree. 
$\vvec_T$ contains $N$ one bits, and $N+1$ zero bits, for a 
total of $2N + 1$ bits. 
It is even possible to shave off $3$ bits for any non-empty tree, since the bit 
vector encoded using the LOBM technique always begins with a $1$ bit and ends 
with two $0$ bits. 
Let $m$ be the position of some node in $T$ in the bit array $\vvec_T$. 
We can navigate in $T$ using the following functions:


\begin{itemize}
 \item $\funccase{left-child}(m) \leftarrow 2 \cdot \rankop[1]{\vvec_T, m}$
 \item $\funccase{right-child}(m) \leftarrow 2 \cdot \rankop[1]{\vvec_T, m} + 1$
 \item $\funccase{parent}(m) \leftarrow \selop[1]{\vvec_T, \lfloor m/2 \rfloor }$
\end{itemize}

Now we describe the balanced-parentheses encoding of ordinal trees, where nodes may 
be of any degree, given in~\cite{munro_raman_01} and illustrated in 
Fig.~\ref{fig:balance-parens-encoding}. 

Let $P$ denote a balanced parentheses sequence encoding the tree $T$.
We construct $P$ by performing a preorder traversal on $T$.
The first time we encounter any node during this traversal we output an
open parenthesis '('. 
The last time we encounter any node during this traversal we output a close
parenthesis ')'. 
This process completes the parentheses sequence construction.
A node $m$ in $T$ is represented by the open parenthesis '(' output when $m$
was first encountered.

To navigate in $T$, we employ the $\funccase{findclose}()$ and $\funccase{enclose}()$
functions defined above.
$P[m]$ is the open parenthesis representing node $m \in T$. 
To find the parent node of $m$ we execute $\funccase{enclose}(P,m)$.
To find the children of $m$, we will define the function 
$\funccase{list-children}(P,m)$, which may be implemented as follows:
\vspace{6pt}

\protect \framebox[0.95\linewidth]
{
\begin{minipage}{0.90\linewidth}

\textbf{function $\funccase{list-children}(P,m)$} \\
1. \hspace{3pt} $c \leftarrow m$ \\
2. \hspace{3pt} \textbf{while} $P[c+1]='('$ \\
3. \hspace{18pt} report $P[c+1]$ as child of $m$ \\
4. \hspace{18pt} $c = \funccase{findclose}(P,c+1)$
\end{minipage}
}

\begin{figure}[ht]
	\centering
	\includegraphics[trim=0cm 18cm 0cm 0cm, clip=true,width=1.0\textwidth]{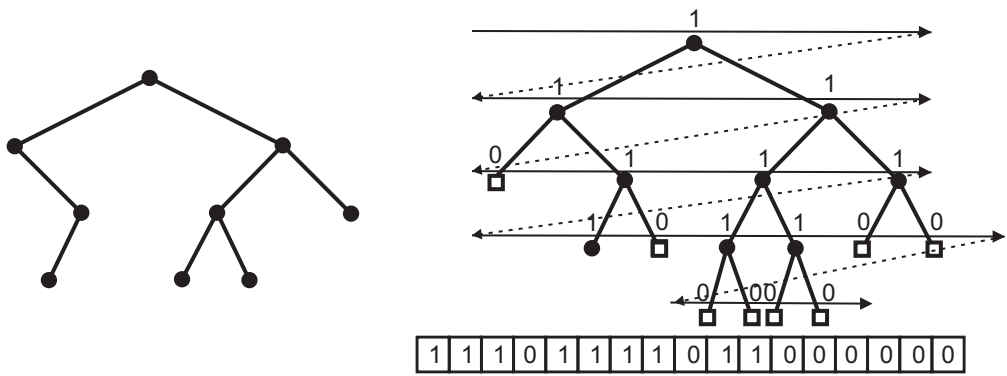}
	\caption{
	  The \emph{level-order binary marked} representation for binary trees.  
	  A binary tree is shown on the right-hand side. 
	  On the left-hand side the same tree is shown with external nodes (hallow 
	  squares) added and the representation as a bit vector. Arrows indicate the order 
	  in which nodes are visited in producing the bit vector.
	}
	\label{fig:lobm_tree_representation}
\end{figure}

\begin{figure}[ht]
	\centering
		\includegraphics[trim=0cm 6cm 0cm 2cm, clip=true,width=0.8\textwidth]{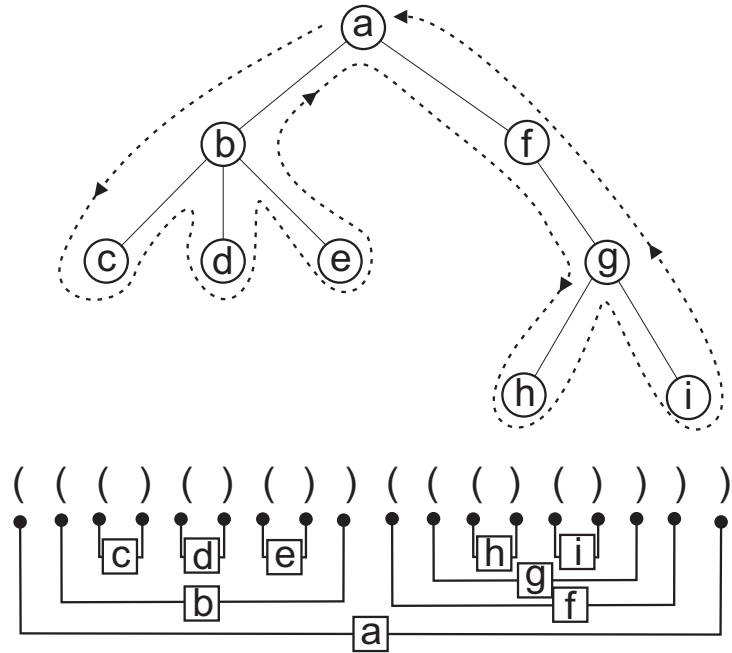}
	\caption{Balanced-parentheses encoding of an ordinal tree.}
	\label{fig:balance-parens-encoding}
\end{figure}


\subsection{Representing Graphs}

Since the data structures that we have developed for graphs in 
Chapter~\ref{chp:succinct_graphs} do not rely on the implementation details
of a specific succinct graph representation, we do not give a detailed 
description of the various representations used here.
A survey of the current state of research in the field of succinct graphs
is given in Section~\ref{sec:background} of that chapter.

\section{Memoryless Subdivision Traversal}

Consider a triangulation in $\plane$. We are given the problem of visiting 
every triangle in the triangulation and reporting some information. 
Given a starting triangle, a natural way to perform this operation is to execute 
a depth-first search from the start triangle, reporting the desired 
information when we first visit a triangle. 
The standard way of facilitating such a traversal is to store, with each triangle, 
a mark bit that will indicate if it has previously been visited. 
This marking prevents different branches of the traversal from revisiting 
the same triangle. 
However, there are instances where storing such a mark bit is problematic. 
For example, in Chapter~\ref{chp:succinct_graphs} we develop data structures 
and algorithms for  I/O efficient traversal of triangulations. 
The I/O efficiency is effectively achieved by storing duplicate copies of 
individual triangles in different blocks. 
Using mark bits in this structure would necessitate marking the visited triangle,
plus all duplicates in other blocks. 
This requirement would defeat the purpose of the blocking scheme which is to 
minimize the number of blocks that must be visted.

An algorithm that avoids the use of mark bits in the traversal of triangulated 
subdivisions was developed by Gold~\etal~\cite{gold_cormack86}, \cite{gold_maydell_1978}. 
De Berg~\etal~\cite{DBLP:journals/gis/BergKOO97} extended this approach to general 
planar subdivisions, and to traversing the surface of convex polyhedra. 

The algorithms presented for mesh traversals in Chapters 
\ref{chp:succinct_graphs} and \ref{chp:mesh_trav} use modifications of these
memoryless traversal algorithms.
Therefore, we present here an overview of how these algorithms work. 
A reference point in $\plane$ is selected from which to start the traversal. 
Typically, this point will be interior to the first face to be visited by the 
traversal. 
From this reference point, a unique entry edge can be selected for each face in 
the subdivision. 
How the entry edge can be uniquely identified is described in 
Gold~\etal~\cite{gold_cormack86}, \cite{gold_maydell_1978} for triangles, and
in De Berg~\etal~\cite{DBLP:journals/gis/BergKOO97} for simple polygons.

For the polygonal subdivision $S$, let $\dual{S}$ denote its dual graph.
We define the dual graph of $S$ as a directed multigraph defined as
follows.
We associated with each face in $S$ a vertex in $\dual{S}$. 
If $v$ and $w$ are vertices in $\dual{S}$, then we draw a directed edge from
$v$ to $w$, and from $w$ to $v$, for each edge they share. 
Now prune the edges of $\dual{S}$ as follows:
for node $v$ remove all outgoing edges except the one corresponding to unique
entry edge on the face corresponding to $v$ in $S$.
The resulting pruned dual graph forms a tree, rooted at the face containing 
the reference point, with all edges directed toward the 
root (for proof see \cite{DBLP:journals/gis/BergKOO97}).
 Unlike general graphs, traversal in a tree can be performed without the need 
for any mark bits, therefore the traversal of $S$ can be performed on this tree 
in order to report the faces of the subdivision. 
The entry face and the parent-child relations in the tree can be determined 
when a face is visited, so there is no need to precompute the tree in the dual;
rather, everything can be handled algorithmically.

In a polygonal subdivision, the determination of the entry edge in a memoryless 
environment involves testing each edge of the polygon in order to determine if 
it is the entry edge. De Berg~\etal accomplished this check by walking the boundary of each face 
encountered. 
Due to this requirement, the running time for this algorithm on a subdivision 
with $|E|$ edges is $\BigOh{|E|^2}$. 
This algorithm was improved by Bose and 
Morin~\cite{DBLP:conf/isaac/BoseM00} to $\BigOh{|E| \lg{|E|} }$. 
Their technique defined a total order on the edges of a face, relative to the 
reference point.
The edge of minimal value in this total order was selected 
as the entry edge for the face. 
When checking if some edge was the entry edge for its polygon, a 'both-ways' 
search\footnote{This two-ways search is similar to leader election in ring 
networks in distributed computing~\cite{santoro2006design}.} was employed, 
resulting in the reduced algorithm complexity. 

In the case of triangular subdivisions, note that the entry edge can always be 
found in constant time, thus the traversal time becomes $\BigOh{N}$ on a 
subdivision on $N$ triangles.


\chapter[Traversal in Trees]{Succinct and I/O Efficient Data Structures for 
  Traversal in Trees}\label{chp:leaf-root}
\chaptermark{Traversal in Trees}

\section{Introduction}

Many operations on graphs and trees can be viewed as the traversal of 
a path. 
Queries on trees, for example, typically involve traversing a path from 
the root to some node, or from some node to the root. 
Often, the datasets represented in graphs and trees are too large to fit 
in internal memory, and traversals must be performed efficiently in external 
memory (EM). 
Efficient EM traversal in trees is important for structures such as suffix 
trees, as well as being a building block in graph searching and shortest path 
algorithms. 
In both cases huge datasets are often dealt with.  
Suffix trees are frequently used to index very large texts or collections 
of texts, while large graphs are common in numerous applications such as 
Geographic Information Systems. 

Succinct data structures were first proposed by Jacobson~\cite{jac_1989}. 
The idea was to represent data structures using space as near the 
information-theoretical lower bound as possible while allowing efficient 
navigation. 
Succinct data structures, which have been studied largely outside the 
external memory model, also have natural applications to large data sets.

In this chapter, we present data structures for traversal in trees that are 
both efficient in the EM setting, and that encode the trees succinctly. 
We are aware of only the work of Clark and Munro~\cite{clark_96} and 
Chien \etal~\cite{DBLP:conf/dcc/ChienHSV08} on succinct full-text indices 
supporting 
efficient substring search in EM, that follows the same track. 
Our contribution is the first such work on general trees that bridges these 
two techniques. 
The novelty of our approach lies in the combination of known blocking schemes
and succinct encodings to create I/O-efficient succinct structures.
The key challenges we address in doing so are keeping memory requirements as 
small as possible in spite of the need to duplicate information, and maintaining
succinct mappings between blocks.

\subsection{Previous Work}
The results in this chapter are presented in the I/O model, whose details 
are given in Chapter~\ref{chp:background} (Section~\ref{ssec:em_model}).
In the I/O model the parameters $B$, $M$, and $N$ are used, respectively, 
to represent the size (in terms of the number of data elements) of a 
block, internal memory, and the problem instance. 
\emph{Blocking} of data structures in the I/O model has reference to the 
partitioning of the data into individual blocks that can subsequently be 
transferred with a single I/O.

Nodine \etal~\cite{ngv_1996} studied the problem of blocking graphs and 
trees for efficient traversal in the I/O model. 
Specifically, they looked at trade-offs between I/O efficiency and space 
when redundancy of data was permitted. 
The authors arrived at matching upper and lower bounds for complete $d$-ary 
trees and classes of general graphs with close to uniform average vertex degree. 
Among their main results they presented a bound of $\Theta(\log_d{B})$ for 
d-ary trees where, on average, each vertex may be represented twice. 
Blocking of bounded degree planar graphs, such as Triangular Irregular Networks
 (TINs), was examined by Agarwal \etal~\cite{DBLP:conf/soda/AgarwalAMVV98}. 
The authors showed how to store a planar graph of size $N$, and of bounded 
degree $d$, in $O(N/B)$ blocks so that any path of length $K$ can be traversed 
using $O(K/ \log_d{B})$ I/Os.

Hutchinson \etal~\cite{hutch_mesh_zeh_2003} examined the case of bottom-up traversal 
where the path begins with some node in the tree and proceeds to the root. 
The authors gave a blocking that supports bottom-up traversal in $O(K/B+1)$ I/Os when the 
tree is stored in $O(N/B)$ blocks (see Lemma~\ref{lem:bot_up_hutch}).  
The case of top-down traversal has been studied more extensively. 
Clark and Munro \cite{clark_96} described a blocking layout that yields a logarithmic 
bound for root-to-leaf traversal in suffix trees. 
Given a fixed independent probability on the leaves, 
Gil and Itai~\cite{gil_itai_1999} presented a blocking layout that yields the minimum 
expected number of I/Os on a root to leaf path.  
In the cache-oblivious model where the sizes of blocks and internal memory are unknown, 
Alstrup \etal~\cite{abdfrt_2004} gave a layout that yields a minimum worst case, 
or expected number of I/Os, along a root-to-leaf path, up to constant factors. 
Demaine \etal~\cite{dem_iac_lan_2004} presented an optimal blocking strategy that 
yields differing I/O complexity depending on the length of the path (see 
Lemma~\ref{lem:top_down_demaine}). 
Finally, Brodal and Fagerberg~\cite{brodal_fagerberg_06} described the giraffe-tree
which, similarly, permits a $O(K/B+1)$ root-to-leaf tree traversal with $O(N)$ space in 
the cache-oblivious model.

\subsection{Our Results} 
Throughout this chapter we assume that $B=\Omega(\lg{N})$, where $N$ is the number of nodes 
of the tree given (i.e., the disk block is of reasonable size)\footnote{In this chapter, 
$\lg{N}$ denotes $\log_2{N}$. 
When another base is used, we state it explicitly (e.g., $\log_B{N}$). }. 
We also assume that $B \le N$, as otherwise a tree on $N$ nodes can be trivially stored 
in a constant number of blocks, and any path on this tree can be traversed in $O(1)$ I/Os. 
Regarding the size of machine words in our model, we adopt the assumption that each word has 
$w = \Theta(\lg N)$ bits. 
This assumption is commonly used in papers on succinct data structures in internal 
memory~\cite{DBLP:journals/talg/RamanRS07}, and we adopt it for the external memory model we use. 

We present two main results:
\begin{enumerate}
\item In Section \ref{sec:bottom-up}, we show how a tree $T$ on $N$ nodes 
with $q$-bit keys, where $q = O(\lg N)$, can be blocked in a succinct fashion 
such that a bottom-up traversal requires $O(K / B+1)$ I/Os using only 
$(2+q)N + q \cdot \left[ \frac{2 \tau  N  (q + 2 \lg B)}{w} +  o(N) \right]  + \frac{8\tau N \lg B}{w}$ 
bits  for any constant $0 < \tau <1$ to store $T$, where $K$ is the 
path length and $w$ is the word size. 
Our data structure is succinct since the above space cost is at most 
$(2 + q) N + q\cdot (\eta N + o(N))$ bits for any constant $0 < \eta < 1$, as 
shown in the discussion after the proof of Theorem~\ref{thm:node-root}. 
When storing keys with tree nodes is not required (this problem is still valid 
since any node-to-root path is well-defined without storing any keys), we can 
represent $T$ in $2N + \frac{\epsilon N \lg B}{w} + o(N)$ bits for any constant 
$0 < \epsilon < 1$ while proving the same support for queries\footnote{The 
numbers of I/Os required to answer queries for the two data structures presented 
in this paragraph are in fact $O(K/(\tau B))$ and $O(K/(\epsilon B))$, respectively. 
We simplify these results using the fact that $\tau$ and $\epsilon$ are both 
constant numbers.}. 
Our technique is based on~\cite{hutch_mesh_zeh_2003} which only considers trees 
in which nodes do not store keys, and achieves an improvement on the space bound 
by a factor of $\lg{N}$.

\item In Section \ref{sec:top_down}, we show that a binary tree on $N$ nodes, with 
keys of size $q=O(\lg{N})$ bits, can be stored using $(3+q)N + o(N)$ bits so that 
a root-to-node path of length $K$ can be reported with: (a) 
$O \left( \frac{K}{\lg(1+(B \lg N)/q)} \right)$ I/Os, when $K = O(\lg N)$; 
(b) $O \left( \frac{\lg N}{\lg (1+\frac{B \lg^2 N}{qK} )} \right)$ I/Os, when 
$K=\Omega(\lg N)$ and $K=O \left( \frac{B \lg^2 N}{q} \right)$; and (c) 
$O \left( \frac{qK}{B \lg N} \right)$ I/Os, when 
$K = \Omega \left( \frac{B \lg^2 N}{q} \right)$. 
This result achieves a $\lg{N}$ factor improvement on the previous space 
cost in~\cite{dem_iac_lan_2004}. 
We further show that (see the discussion after Corollary~\ref{cor:top_down_fixed_keys}), 
when $q$ is constant, our approach not only reduces storage space but also 
improves the I/O efficiency of the result in~\cite{dem_iac_lan_2004}.
\end{enumerate}

\section{Preliminaries}\label{sec:prelim}

\subsection{Bit Vectors} 

The basic definitions for, and operations on, bit vectors have
been described in Chapter~\ref{chp:background}, 
Section~\ref{sec:bckgrnd-succint-ds}.
We repeat here a key lemma which summarizes results on
bit vectors used in the current Chapter.
Part (a) is from Jacobson~\cite{jac_1989} and Clark and 
Munro~\cite{clark_96}
while part (b) is from Raman~\etal~\cite{DBLP:journals/talg/RamanRS07}:

\begin{lemma}\label{lem:rank_select}
A bit vector $V$ of length $N$ can be represented using either: 
(a) $N + o(N)$ bits, or (b) $\lceil\lg{ N \choose R}\rceil + O(N \lg \lg N / \lg N)$ bits, 
where $R$ is the number of $1$s in $V$, to support the access to each bit, $\rankop$ and $\selop$ in 
$O(1)$ time (or $O(1)$ I/Os in external memory). 
\end{lemma}

\subsection{Succinct Representations of Trees}\label{sec:succincttrees} 
As there are  ${2N \choose N} / (N+1)$ different binary trees (or ordinal trees)
on $N$ nodes, the information-theoretic lower bound of representing a binary 
tree (or ordinal tree) on $N$ nodes is $2N-O(\lg n)$. 
Thus, various approaches~\cite{jac_1989,munro_raman_01,benoit_et_al_2005} have 
been proposed to represent a binary tree (or ordinal tree) in $2N+o(N)$ bits
while supporting efficient navigation. 
Jacobson~\cite{jac_1989} first presented the \emph{level-order binary marked} 
(LOBM) structure for binary trees, which can be used to encode a binary tree 
as a bit vector of $2N$ bits. 
He further showed that operations such as retrieving the left child, the right 
child, and the parent of a node in the tree can be performed using $\rankop$ and 
$\selop$ operations on bit vectors. 
We make use of his approach in order to encode tree structures in Section~\ref{sec:top_down}. 

Another approach that we use in this work is based on the isomorphism between 
\emph{balanced parenthesis sequences} and ordinal trees. 
The balanced parenthesis sequence of a given tree can be obtained by performing 
a depth-first traversal and outputting an opening parenthesis the first time 
a node is visited, and a closing parenthesis after we visit all its descendants. 
Based on this, Munro and Raman~\cite{munro_raman_01} designed a succinct representation
of an ordinal tree of $N$ nodes in $2N+o(N)$ bits, 
which supports the computation of the parent, the depth, and the number of descendants 
of a node in constant time, and the $i$\textsuperscript{th} child of a node in 
$O(i)$ time. 
Lu and Yeh~\cite{DBLP:journals/talg/LuY08} further extended this representation 
to support the computation of the $i$\textsuperscript{th} child of a node and 
several other operations in $O(1)$ time. 

\subsection{I/O Efficient Tree Traversal} Hutchinson \etal  
\cite{hutch_mesh_zeh_2003} presented a blocking technique for rooted trees in 
the I/O model that supports bottom-up traversal. 
Their result is summarized in the following lemma:

\begin{lemma}[\cite{hutch_mesh_zeh_2003}]\label{lem:bot_up_hutch}
A rooted tree $T$ on $N$ nodes can be stored in $O(N/B)$ blocks on disk such 
that a bottom-up path of length $K$ in $T$ can be traversed in $O(K/ B+1)$ I/Os. 
\end{lemma}

\noindent Their data structure involves cutting $T$ into layers of height 
$\tau B$, where $\tau$ is a constant ($0 < \tau < 1$). 
A forest of subtrees is created within each layer, and the subtrees are stored in blocks. 
If a subtree needs to be split over multiple blocks, the path to the top of the layer is stored for that block. 
This ensures that the entire path within a layer can be read by performing a single I/O.

To support top-down traversal, Demaine \etal~\cite{dem_iac_lan_2004} described an optimal blocking technique for binary trees that bounds the number of I/Os in terms of the depth of the node within $T$. The blocking has two phases. The first phase blocks the top $c \lg N$ levels of the tree, where $c$ is a constant, as if it were a complete tree. In the second phase, nodes are assigned recursively to blocks in a top-down manner. 
The proportion of nodes in either child's subtree assigned to the current block is determined by the sizes of the subtrees.  
The following lemma summarizes their results:

\begin{lemma} [\cite{dem_iac_lan_2004}]\label{lem:top_down_demaine}
A binary tree $T$ on $N$ nodes can be stored in $O(N/B)$ blocks on disk such that a traversal from the root to a node of depth $K$ requires the following number of I/Os:
\begin{enumerate}
   \item $\Theta( K/\lg(1+B))$, when $K = O(\lg N)$,
   \item $\Theta( \lg N/ (\lg (1+B \lg N/K)))$, when $K = \Omega(\lg N)$ and $K = O(B \lg N)$, and 
   \item $\Theta(K/B)$, when $K = \Omega(B \lg N)$. 
\end{enumerate}
\end{lemma}

\section{Bottom-Up Traversal}\label{sec:bottom-up}

In this section, we present a set of data structures that encode a tree $T$ 
succinctly so that the number of I/Os performed in traversing a path from a given node 
to the root is asymptotically optimal. 
We wish to maintain a key with each node, and each key can be encoded with 
$q = O(\lg N)$ bits. 
Given the bottom-up nature of the queries, there is no need to check a node's 
key value while traversing, 
since the path always proceeds to the current node's parent. 
Thus, after we present our results on representing trees on nodes with keys, 
we also show how to encode trees whose nodes do not store key values in a corollary. 

In this section a node, $v$, of $T$ is uniquely identified by the number 
of the layer containing $v$ (the notion of layer is to be defined later) and by 
the preorder number of $v$ in its layer. 
Thus traversing a path means returning the identifiers of the nodes along the path together with the associated keys if keys are maintained. 

\subsection{Blocking Strategy} 
\label{sec:blocking}

Our blocking strategy is inspired by \cite{hutch_mesh_zeh_2003}; we have modified their technique and introduced new notation. 
Firstly, we give an overview of our approach. 
We partition $T$ into layers of height $\tau B$ where $0 < \tau < 1$. 
We permit the top layer and the bottom layer to contain fewer than $\tau B$ levels,
as doing so provides the freedom to partition the tree into layers with a desired distribution of nodes. 
See Figure~\ref{fig:layers} for an example. 
We then group the nodes of each layer into {\em tree blocks}
\footnote{A tree block is a portion of the tree, which is different from
 the notion of disk block. 
 A tree block is not necessarily stored in a disk block, either. 
 In the rest of the chapter, when the context is clear, we may use 
 the term block to refer to either a tree block or a disk block.}, 
 and we store with each block a \emph{duplicate path} which is defined 
 later in this section. 
In order to bound the space required by block duplicate paths, we further
 group blocks into \emph{superblocks}. 
 The duplicate path of a superblock's first block is the 
 \emph{superblock duplicate path} for that superblock. 
 By loading at most the disk block containing a node, along with its
 associated duplicate path and the superblock duplicate path, we 
 demonstrate that a layer can be traversed with at most $O(1)$ I/Os. 
 A set of bit vectors, to be described later, that maps the nodes at 
 the top of one layer to their parents in the layer above is used 
 to navigate between layers.

\begin{figure}[ht]
	\centering
		\includegraphics{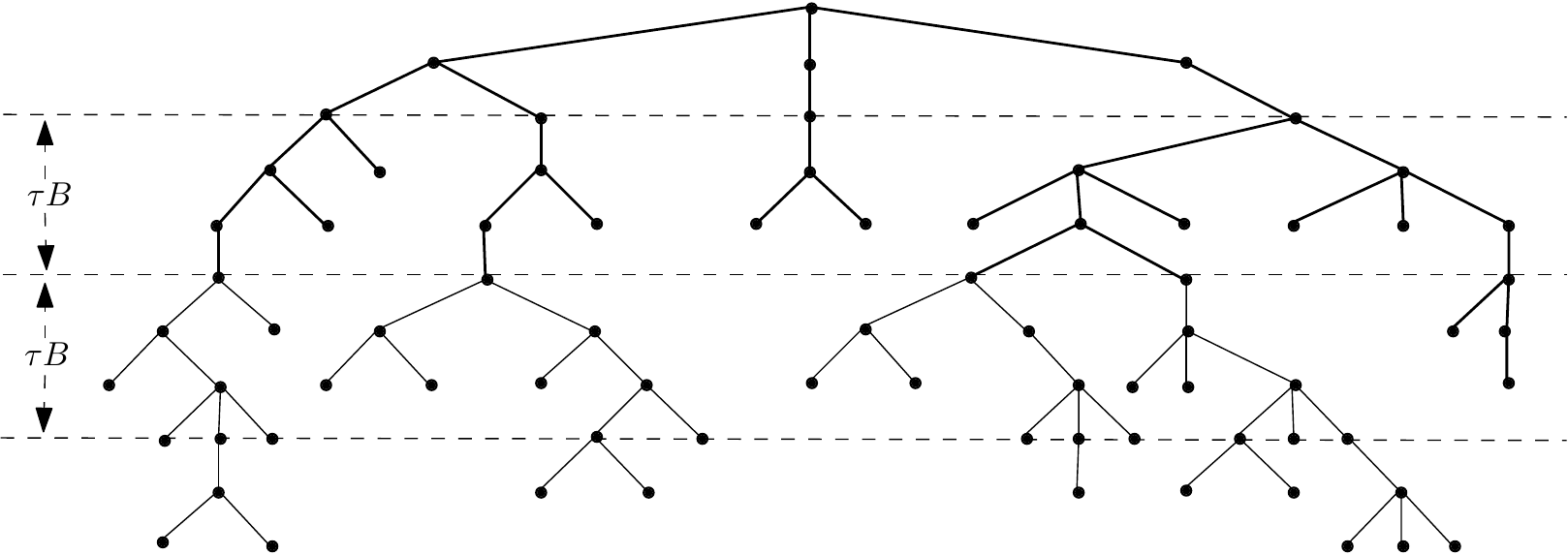}
	\caption{An example of partitioning a tree into layers}
	\label{fig:layers}
\end{figure}
 
Layers are numbered starting at $1$ for the topmost layer. 
As stated in the previous paragraph, we have the flexibility to choose an arbitrary level within the first $\tau B$ levels as the top of the second layer. 
Given this flexibility, we can prove the following lemma:
\begin{lemma}\label{lem:chooselayer}
There exists a division of $T$ into layers such that the total number of nodes on the top level of layers is bounded by $\lceil N /(\tau B)\rceil$. 
\end{lemma}
\begin{proof}
There are $\tau B$ different ways to divide $T$ into layers. 
Let $s_i$ denote the total number of nodes on the top level of layers under the $i$\textsuperscript{th} way of dividing $T$, and let $S_i$ denote the set of such nodes. 
We observe that given a node of $T$ that is not the root, there is only one value of $i$ such that this node is in $S_i$, while the root of $T$ appears once in each $S_i$. 
Therefore, we have $\sum_{i=1}^{\tau B} s_i = N - 1 + \tau B$. 
Then $\min{s_i} \le \lfloor (N-1+\tau B) / (\tau B) \rfloor = \lfloor N/(\tau B) + 1 - 1/(\tau B) \rfloor \le \lceil N/(\tau B) \rceil$. 
\end{proof}

Thus, we pick the division of $T$ into layers that makes the total number 
of nodes on the top level of layers bounded by $\lceil N /(\tau B)\rceil$. 
Let $L_i$ be the $i$\textsuperscript{th} layer in $T$. 
The layer is composed of a forest of subtrees whose roots are all at the 
top level of $L_i$.  
We now describe how the blocks and superblocks are created within $L_i$.
We number $L_i$'s nodes in preorder, starting from $1$ for the leftmost 
subtree, and number the nodes of the remaining subtrees from left to right. 
Once the nodes of $L_i$ are numbered, they are grouped into tree blocks 
of consecutive preorder number. 
The exact number of nodes in each tree block is to be determined later. 
We term the first tree block in a layer the \emph{leading block}, and 
the remaining tree blocks in a layer {\em regular blocks}. 
Each superblock, excepting possibly the first one in a layer (which we term 
the \emph{leading superblock}), contains exactly $\lfloor \lg B \rfloor$ 
tree blocks (see  Figure \ref{fig:block_group}). 
We term each of the remaining superblocks a {\em regular superblock}. 

\begin{figure}[t]
	\centering
		\includegraphics{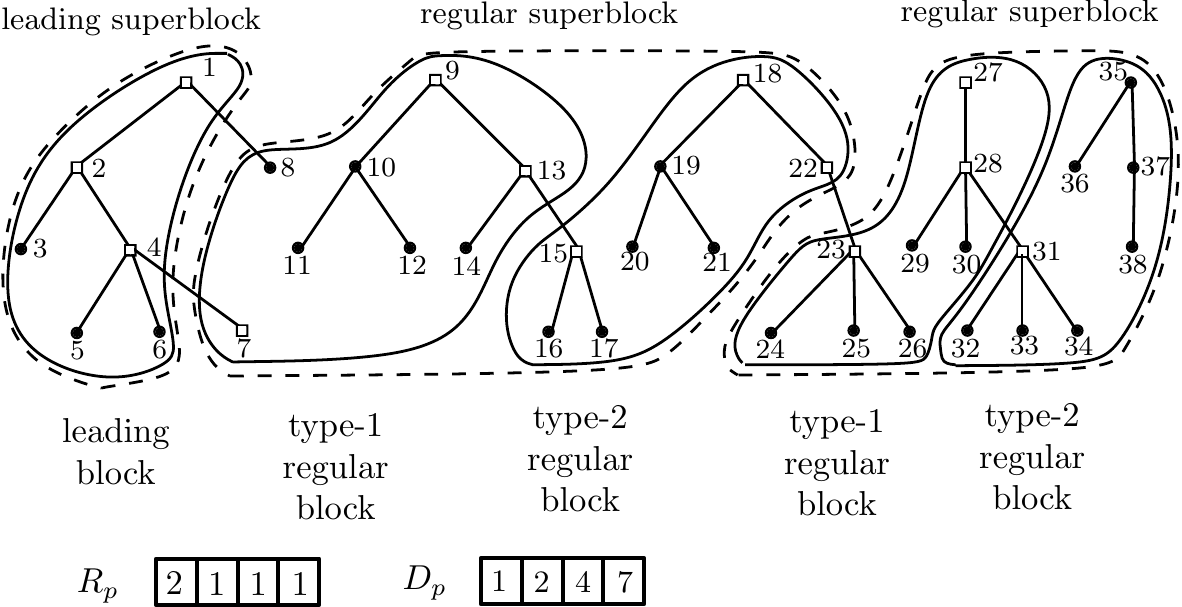}
	\caption[Labelling nodes in a tree]{Blocking within a layer 
(Layer $L_3$ in Figure~\ref{fig:layers}) is shown, 
along with block (solid lines) and superblock (dashed lines) boundaries. 
Numbers shown are preorder values in the layer. 
Nodes on the duplicate paths are indicated by a hollow square. 
The key values stored with nodes are not shown.
}
	\label{fig:block_group}
\end{figure}

We select as a superblock's {\em duplicate path} the path from the node with 
minimum preorder number in the superblock to the layer's top level. 
Similarly, a tree block's duplicate path is defined as the path from the node 
with minimum preorder number in the block to the lowest ancestor of this node 
whose parent is outside the superblock, if such an ancestor exists, or to the 
root otherwise. 
In the example in Figure~\ref{fig:block_group}, the duplicate path of the second 
block consists of nodes $7$, $4$, $2$, and $1$. 
A duplicate path has at most $\tau B$ nodes and satisfies the following property (this is analogous to Property 3 in Hutchinson~\etal~\cite{hutch_mesh_zeh_2003}):

\begin{property}\label{prop:dup-path}
Consider a superblock $Y$. 
For any node $x$ in $Y$, there exists either a path from $x$ to the top of its 
layer consisting entirely of nodes in $Y$, or a path from $x$ to a node on the 
duplicate path of $Y$ consisting entirely of nodes in $Y$ plus one node on the 
duplicate path. 

Similarly, for a tree block $Y'$ and a node $x'$ in $Y'$, there exists either a 
path, consisting entirely of nodes in $Y'$, from $x'$ to the root or a node 
whose parent is outside the superblock containing $Y'$, or a path from $x'$ to 
a node on the duplicate path of $Y'$ consisting entirely of nodes in $Y'$ plus 
one node on the duplicate path. 
\end{property}

\begin{proof}
Let $v$ be the node with the minimum preorder number in $Y$, and in the forest 
stored in $Y$ let $T_v$ be the subtree that contains $v$. 
For example, in Figure~\ref{fig:block_group}, if $Y$ is the fourth block, 
then $T_v$ is the subtree consisting of nodes $23$, $24$, $25$ and $26$. 
As the preorder number of $v$ is smaller than that of any other node in $T_v$, 
we conclude that node $v$ is the root of $T_v$. 
For the case in which $x \in T_v$, because the path from $x$ to $v$ is entirely 
within $T_v$, and $v$ is on the duplicate path of $Y$, our claim is true. 

The second case is $x \notin T_v$. In the forest stored in $Y$, let $T_x$ be the 
subtree that contains $x$, and let node $y$ be its root. 
If $y$ is at the top level of its layer, as the path from $x$ to $y$ contains 
only nodes in $Y$, the lemma follows directly.  
Thus, we need only consider the case in which $y$ is not at the top level of 
this layer, and it suffices to prove that the parent, $z$, of $y$ is on the 
duplicate path of $Y$. 
Assume to the contrary that $z$ is not (i.e., $z \ne v$ and $z$ is not $v$'s 
ancestor). 
As the preorder number of $z$ is smaller than that of $y$, $z$ is in a 
superblock, $Z$, that is to the left of $Y$. 
Therefore, the preorder number of $z$ is smaller than that of $v$. 
As $v$ is not a descendant of $z$, by the definition of preorder traversal, 
the preorder number of $v$ is larger than any node in the subtree rooted at 
$z$ including $y$, which is a contradiction. 

The second claim in this property can be proved similarly.
%
\end{proof}

The sizes of tree blocks are dependent on the approach we use to store them on disk. 
When storing tree blocks in external memory, we treat leading blocks and regular 
blocks differently. 
We use a disk block to store a regular block along with the representation of its
duplicate path, or the superblock duplicate path if it is the first tree block in 
a superblock. 
In such a disk block, we refer to the space used to store the duplicate path as 
\emph{redundancy} of the disk block. 
For simplicity, such space is also referred to as the redundancy of the regular tree block stored in this disk block. 
To store the tree structure and the keys in our succinct tree representation, 
we require $2 + q$ bits to represent each node in the subtrees of $L_i$. 
Therefore, if a disk block of $Bw$ bits (recall that $w = \Theta(\lg N)$ 
denotes the word size) has redundancy $W$, the maximum number of nodes 
that can be stored in it is: 

\begin{equation}\label{equ:redundancy}
   \succblksize = \left\lfloor \frac{B w - W}{2+q}\right\rfloor 
\end{equation}

Layers are blocked in such a manner that when a regular block is stored in a disk block, 
it has the maximum number of nodes as computed above, and the leading block is 
the only block permitted to have fewer nodes than any regular block. 
There are two types of regular blocks: a {\em type-1 regular block} is the first
block in a regular superblock while a {\em type-2 regular block} is a regular 
block that is not the first block in its superblock. 
In our representation, the redundancy in each disk block that stores a type-1 or type-2 regular block is fixed. 
Therefore, the number, $\succblksize_1$, of nodes in a type-1 regular block and the number, $\succblksize_2$, of nodes in a type-2 regular block are both fixed. 
In order to divide a layer into tree blocks it suffices to know the values of 
$\succblksize_1$ and $\succblksize_2$ (we give these values in Lemma~\ref{lem:noofnodes}).  
More precisely, we first compute the number, $s$, of nodes in a regular superblock using $s = \succblksize_1+(\lfloor \lg B \rfloor-1)\succblksize_2$. 
Let $l_i$ be the number of nodes in layer $L_i$. 
We then put the last $s\lfloor l_i/s\rfloor$ nodes into $\lfloor l_i/s\rfloor$ superblocks of which each can be easily divided into tree blocks. 
Finally, the first $l_i'=l_i \bmod s$ nodes are in the leading superblock in 
which the first $l_i'\bmod \succblksize_2$ nodes form the leading block. 
As the sizes of leading blocks can be arbitrarily small, we pack them into a sequence of disk blocks.

\subsection{Data Structures}\label{sec:bottomup_ds}

Each tree block is encoded by five data structures:

\begin{enumerate}

\item The block keys, $B_k$, is an $\succblksize$-element array which encodes the keys of the tree block. 

   \item An encoding of the tree structure, denoted $B_t$. The subtree(s) contained 
within the block are encoded as a sequence of balanced parentheses (see Section~\ref{sec:succincttrees}). 
Note that in this representation, the $i$\textsuperscript{th} opening  parenthesis corresponds to 
the $i$\textsuperscript{th} node in preorder in this block. 
More specifically, a preorder traversal of the subtree(s) is performed (again from 
left to right for blocks with multiple subtrees). 
At the first visit to a node, an opening parenthesis is output.
When a node is visited for the last time (going up), a closing parenthesis is output. 
Each matching parenthesis pair represents a node while the parentheses in 
between represent the subtree rooted at that node. 
For example, the fourth block in Figure~\ref{fig:block_group} is encoded as $(()()())((()()))$. 

   \item The duplicate path array, $D_p[1..\tau B]$. 
Entry $D_p[j]$ stores the preorder number of the node at the 
$j$\textsuperscript{th} level in the layer on the duplicate path, and a $0$ if 
such a node does not exist (recall that preorder numbers begin at 1, so the 0 value effectively flags an entry as invalid). 
In order to identify each node in $D_p$, there are three cases (let $v$ be the node with the smallest preorder number in the block): 
\begin{enumerate}
\item The block is a leading block. 
Consequently, the node $v$ is the only node on the duplicate path. Thus, for a leading block, we do not store $D_p$. 
\item The block is a type-1 regular block. For such a block, $D_p$ stores the preorder numbers of the nodes on the superblock duplicate path with respect to the preorder numbering in the layer. 
For example, in Figure~\ref{fig:block_group}, the duplicate path array for the 
second block stores $1, 2, 4$, and $7$. 
Note that it may be the case that $v$ is not at the $\tau B$\textsuperscript{th} level of the layer. In this case, the last one or more entries of $D_p$ are set to $0$. 
\item The block is a type-2 regular block. In this case, $D_p$ stores the duplicate path of this block, and each node in $D_p$ is identified by its preorder number with respect to the block's superblock. 
Then the duplicate path array for the third block in Figure~\ref{fig:block_group} 
stores $3, 7, 9$, and $0$. 
Note that in addition to the possibility that the last one or more entries of 
$D_p$ can be set to $0$s as described in case (b), the first one or more entries 
can also possibly be set to $0$s, 
in the case that the block duplicate path does not reach the top of the layer.
\end{enumerate}

\item The duplicate path key array, $D_k[1..\tau B]$, storing the keys of the nodes in $D_p$. More precisely, $D_k[i]$ stores the key value of node $D_p[i]$ for $ 1 \le i \le j$. 

   \item The root-to-path array, $R_p[1..\tau B]$, for each regular block.  
A regular block may include subtrees whose roots are not at the top level of the layer. 
Among them, the root of the leftmost subtree is on the duplicate path, and by 
Property~\ref{prop:dup-path}, the parents of the roots of the rest of all such 
subtrees are on the duplicate path of the block.  
$R_p$ is constructed to store such information, in which $R_p[j]$ stores the number 
of subtrees whose roots are either on the duplicate path somewhere between, and 
including, level $j$ and level $\tau B$ of the layer, or have parents on the 
duplicate path at one of these levels of the layer. 
The number of subtrees whose roots are children of the node stored in $D_p[j]$ 
can be calculated by evaluating $R_p[j] - R_p[j+1]$, if $D_p[j]$ is not the node 
with the smallest preorder number in the regular block. 
For example, in Figure~\ref{fig:block_group}, the second block has three subtrees. 
The root of the leftmost subtree is node $7$, and it is on the duplicate path of 
this block at level $4$. 
The parent of the root of subtree in the middle is node $1$ which is on the 
duplicate path at level $1$. The root of the rightmost subtree is at the top level. 
Thus, the content of the root-to-path array for the second block is $2, 1, 1, 1$. 
\end{enumerate}


For an arbitrary node $v \in T$, let $v$'s layer number be $\ell_v$ and its preorder number within the layer be $p_v$.  Each node in $T$ is uniquely represented by the pair $(\ell_v,p_v)$. Let $\pi$ define the lexicographic order on these pairs. Given a node's $\ell_v$ and $p_v$ values, we can locate the node by navigating within the corresponding layer. The challenge is how to map between the roots of one layer and their parents in the layer above.  
Consider the set of $N$ nodes in $T$. 
We define the following data structures which facilitate mapping between layers:

\begin{enumerate}

\item { Bit vector $\mathcal{V}_{first}[1..N]$, where $\mathcal{V}_{first}[i]=1$ 
iff the $i$\textsuperscript{th} node in $\pi$ is the first node within its layer.}


\item{ Bit vector $\mathcal{V}_{parent}[1..N]$, where $\mathcal{V}_{parent}[i]=1$ 
iff the $i$\textsuperscript{th} node in $\pi$ is the parent of some node at 
the top level of the layer below. }

\item{ Bit vector $\mathcal{V}_{first\_child}[1..N]$, where 
$\mathcal{V}_{first\_child}[i]=1$ iff the $i$\textsuperscript{th} node in $\pi$ 
is a root in its layer and no root in this layer with a smaller preorder number 
has the same parent. }
\end{enumerate}

Figure \ref{fig:bottom_up_bitvectors} demonstrates how the three bit vectors 
represent the mapping between nodes in different layers.

\begin{figure}[t]
	\centering
		\includegraphics{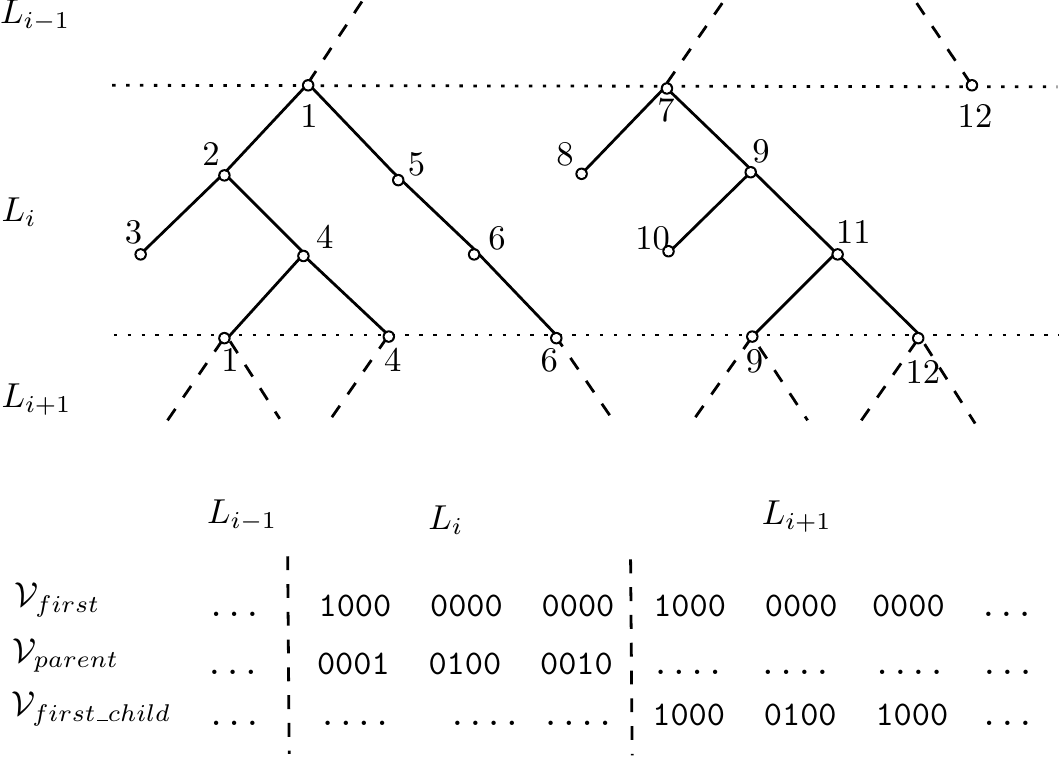}
	\caption[Scheme for mapping node labels between tree layers]{ Scheme for mapping between layers. 
The dashed horizontal lines indicate the top level of each layer. The bottom part of the figure shows the corresponding portions of the bit vectors used to maintain the mapping between layer $L_i$ and its neighbouring layers. 
}
	\label{fig:bottom_up_bitvectors}
\end{figure}

All leading blocks are packed together on disk. 
Note that leading blocks do not require a duplicate path or root-to-path array, 
therefore only the tree structure and keys need be stored for these blocks. 
Due to the packing, a leading block may overrun the boundary of a block on disk. 
We use the first $\lceil\lg{(Bw)}\rceil$ bits of each disk block to store an offset that indicates the position of the starting bit of the first leading block inside this disk block. This allows us to skip any overrun bits from a leading block stored in the previous disk block. 

We store two bit arrays to aid in locating blocks. The first indexes the leading blocks, and the second indexes regular blocks. 
Let $x$ be the number of layers on $T$, and let $z$ be the total number of regular blocks over all layers. The bit vectors are:

\begin{enumerate}

\item Bit vector $\mathcal{B}_l[1..x]$, where $\mathcal{B}_l[i]=1$ iff the $i$\textsuperscript{th} leading block resides in a different disk block than the $(i-1)$\textsuperscript{st} leading block.

\item Bit vector $\mathcal{B}_r[1..(x+z)]$ that encodes the number of regular blocks in each layer in unary. More precisely, $\mathcal{B}_r[1..(x+z)] = 0^{l_1}10^{l_2}10^{l_3}1 \ldots$, where $l_i$ is the number of regular blocks in layer $i$.
\end{enumerate}


With the details of data structures given, we can now determine the number of nodes in a type-1 or type-2 regular block. 

\begin{lemma}\label{lem:noofnodes}
To use the approach in Section~\ref{sec:blocking} to divide a layer into blocks, it is sufficient to choose:
\begin{enumerate}
\item $\succblksize_1=\lfloor \frac{ B w -  \tau B (\lceil \lg{N}\rceil + q +  \lg{B} +  \lg{w} )}{2+q} \rfloor$, and 
\item $\succblksize_2=\lfloor \frac{Bw- \tau B ( 2 \lg{B}  + 2  \lg w  +  \lg{ \lceil \lg{ B } \rceil}  + q)}{2+q} \rfloor$. 
\end{enumerate}
\end{lemma}

\begin{proof}
By Equation~\ref{equ:redundancy}, in order to choose appropriate values for 
$\succblksize_1$ and $\succblksize_2$ we need only compute the 
redundancy of a type-1 regular block and a type-2 regular block, respectively. 
Thus, we consider the space required to store $D_p$, $D_k$ and $R_p$. 

The array $D_k$ uses $\tau B q$ bits. 
To compute the space required for $D_p$, there are two cases: 
\begin{enumerate}

\item For a type-1 regular block, $D_p$ stores preorder values 
with respect to the layer preorder numbering. 
There can be as many as $N$ nodes in a layer, therefore each 
entry requires $\lceil \lg{N} \rceil$ bits. The total space for $D_p$ is thus $\tau B \lceil \lg{N} \rceil$.

\item For each node in a duplicate path of a type-2 regular block, $D_p$ stores 
its preorder number in the superblock. 
Since each superblock has $ \lceil \lg {B} \rceil$ blocks and thus has 
$B \lceil \lg {B} \rceil w $ bits, and in order to encode the tree structure 
and keys, we use $2+q$ bits per node. 
There are at most $B \lceil \lg {B} \rceil w /(2+q) \le B \lceil \lg {B} \rceil w /2$ nodes in a superblock. 
Therefore, $\lceil \lg{(B \lceil \lg {B} \rceil w /2)} \rceil$ bits are sufficient to store the preorder number of each node on the path.
As the duplicate path has $\tau B$ entries, the array $D_p$ requires the following number of bits:
\begin{eqnarray}
   \tau B \left\lceil \lg{\left(\frac{B \lceil\lg{B}\rceil w }{2}\right)} 
   \right\rceil 
   &\leq&
   \tau B \left( \lg{\left( \frac{B \lceil\lg{B}\rceil w}{2} \right)}  +1\right)
   \nonumber\\
&=& \tau B (  \lg{B} +   \lg{ \lceil \lg{ B } \rceil }  +  \lg{ w } )
\end{eqnarray}

\end{enumerate}

As a block may have at most $B w/(2+q) \le Bw /2$ nodes, each entry in $R_p$ 
can be encoded in at most $\lceil\lg (Bw /2)\rceil \le \lg{B}  +  \lg{w}$ bits. 
Thus the space per block for this array is at most $\tau B ( \lg{B}  +  \lg{w} )$ 
bits. This value holds whether the corresponding duplicate path is associated 
with a type-1 or type-2 regular block. 

The redundancy of a type-1 regular block includes space required to encode $D_p$, 
$D_k$, and $R_p$, which is $\tau B (\lceil \lg{N}\rceil + q +  \lg{B} +  \lg{w} )$ bits.

For a type-2 regular block, the number of bits used to store both $D_p$, $D_k$ and $R_p$ is:

\begin{eqnarray}
   && \tau B (  \lg{B}
      +  \lg{ \lceil \lg{ B } \rceil}   
      +  \lg w  + q 
      +  \lg{B}  
      +  \lg w) \nonumber\\
   &=& \tau B ( 2 \lg{B} 
      + 2  \lg w
      +  \lg{ \lceil \lg{ B } \rceil}  + q )
 \label{eqn:block_bits}
\end{eqnarray}

The results in the lemma can then be proved by substituting the redundancy $W$ in Equation~\ref{equ:redundancy} by the redundancy of a type-1 or type-2 regular block as computed above.
\end{proof}

To analyze the space costs of our data structures, we have the following lemma: 

\begin{lemma}\label{lem:bottom_up_space}
For sufficiently large $N$, there exists a positive constant $c$ such that the data structures described in this section occupy $(2+q)N + q \cdot \left[ \frac{2 \tau  N  (q + 2 \lg B)}{w} +  o(N) \right]  + \frac{8\tau N \lg B}{w} $ bits when $0 < \tau < c \le 1$. 
\end{lemma}

\begin{proof}

First consider the number of bits used to store the actual tree structure of 
$T$ (i.e., the total space used by all the $B_t$s). 
The balanced parentheses encoding requires $2N$ bits, and each node of $T$ is 
contained in one and only one block. 
Hence the structure of $T$ is encoded using $2N$ bits. 

Next consider the space for all the block keys (i.e. the total space used by all 
the $B_k$s. 
Since the key of each node is encoded exactly once in all the $B_k$s, the total 
space cost is $Nq$ bits. 

We now consider the total space required for all the duplicate paths, duplicate 
path keys, and root-to-path arrays. 
Since there is one type-1 regular block and $\lceil\lg B\rceil - 1$ type-2 regular blocks in a regular superblock, by the proof of Lemma~\ref{lem:noofnodes}, the sum of the redundancy of all the blocks in a regular superblock is at most:  

\begin{eqnarray}
 && \tau B (\lceil \lg{N}\rceil + q +  \lg{B} +  \lg{w} ) \nonumber \\
  && + (\lceil \lg{B} \rceil - 1) \tau B ( 2 \lg{B}  + 2  \lg w  +  \lg{ \lceil \lg{ B } \rceil}  + q)
\end{eqnarray}

Dividing the above expression by $\lceil\lg B\rceil$, which is the number of 
blocks in a regular superblock, we get an upper bound of the average redundancy 
per block in a regular superblock:

\begin{eqnarray}
   \overline{W} &=& \frac{\tau B q + (\lceil \lg{B} \rceil - 1)\tau B q}{\lceil \lg B \rceil} + \frac{\tau B  \lceil \lg{N} \rceil}{\lceil \lg B \rceil} +
 \frac{\tau B ( \lg{B}  +  \lg w)}{\lceil \lg B \rceil} \nonumber \\
 && + \frac{(\lceil \lg{B} \rceil - 1) \tau B ( 2 \lg{B}  + 2  \lg w  +  \lg{ \lceil \lg{ B } \rceil} )}{ \lceil \lg B \rceil} \nonumber\\
&<& \tau B q + \frac{\tau B  (\lg{N}+1)}{\lceil\lg B\rceil} +
 \frac{\tau B ( \lg{B}  +  \lg w)}{\lceil\lg B \rceil} 
 +\tau B ( 2  \lg{B} 
 			+ 2  \lg w
 			+   \lg{\lg{B}} ) \nonumber \\
 && - \frac{\tau B ( 2  \lg{B}  
 			+ 2  \lg w
 			+   \lg{\lceil \lg{B} \rceil} )}{ \lceil \lg B \rceil} \nonumber\\
 	&<& \tau B  (\log_B{N}+q)  + \tau B ( 2  \lg{B} 
 			+ 2  \lg w
 			+   \lg{\lceil \lg{B} \rceil} )
   \label{eqn:avg_block_redundancy}
\end{eqnarray}

A regular block in a leading superblock is a type-2 regular block, and its 
redundancy is given by Equation~\ref{eqn:block_bits} which is smaller 
than $\overline{W}$. 
Therefore, we use $\overline{W}$ to bound the average redundancy of a regular 
block. 
%
The total number of blocks required to store $T$ is then at most 
$\frac{N}{\lfloor (B w - \overline{W})/(2+q) \rfloor}<\frac{N}{(B w - \overline{W})/(2+q) -1} = \frac{(2+q)N}{B w - \overline{W} - (2+q)}$. 
Then the total size, $R$, of the redundancy for $T$ is:

\begin{equation}
  R < (2+q)N \cdot \frac{\overline{W}}{B w - \overline{W}-(2+q)} < (2+q)N \cdot \frac{2\overline{W}+(2+q)}{B w} 
\label{eqn:red}
\end{equation}

when $\overline{W} <  B w - \overline{W} - (2+q)$. 
%
By Inequality~\ref{eqn:avg_block_redundancy}, this condition is true if: 

\begin{eqnarray}
   B w &>& 2 \tau B \log_B{N}  + 2\tau B q + 4 \tau B  \lg{B} + 4 \tau B   \lg w 
 			+  2 \tau B  \lg{\lceil \lg{B} \rceil}  + 2 + q      \nonumber \\
  \iff w &>& \tau( 2 \log_B{N}  + 2 q + 4 \lg{B} + 4 \lg w 
 			+  2 \lg{\lceil \lg{B} \rceil}  + \frac{2 + q}{\tau B}  )
   \label{eqn:tau_value}
\end{eqnarray}

Since we assume $B = \Omega(\lg N)$, $B \le N$, $q = O(\lg N)$, and $w = \Theta(\lg N)$, the expression inside the brackets on the right-hand side of Inequality~\ref{eqn:tau_value} is $O(w)$. 
By the definition of Big-Oh notation, for sufficiently large $N$ there exists a constant, $c'$, such that:

$$2 \log_B{N}  + 2 q + 4 \lg{B} + 4 \lg w +  2 \lg{\lceil \lg{B} \rceil}  + \frac{2 + q}{\tau B} \le c' w.$$ 

Thus for any $0 < \tau < c$, where $c = \min(1, \frac{1}{c'})$, Inequality~\ref{eqn:tau_value} holds. 

We then substitute for $\overline{W}$ in Inequality~\ref{eqn:red}, to obtain the following:

\begin{eqnarray}
   R&<& \frac{(2qN + 4N)\tau B\log_B{N}}{B w}  
  + \frac{(2qN + 4N)\tau B q }{B w} \nonumber \\
   &&+ \frac{(4qN + 8N)\tau B \lg{B}}{B w} 
   + \frac{(4qN + 8N)\tau B \lg w}{B w} \nonumber \\
      &&+ \frac{(2qN + 4N)\tau B \lg{\lceil \lg{B} \rceil}}{B w} 
   + (2+q)N \cdot \frac{2+q}{B w} \nonumber \\
  &<& \frac{(2qN + 4N)\tau\lg N}{ w\lg{B}} + 
    \frac{(2qN + 4N)\tau q }{w} \nonumber \\ 
   &&+ \frac{(4qN + 8N)\tau (\lg B ) }{ w} + 
   \frac{(4qN + 8N)\tau \lg{w}}{w} \nonumber \\ 
      &&+ \frac{(2qN + 4N)\tau \lg{\lceil \lg{B} \rceil}}{w} 
   + (2+q)N \cdot \frac{2+q}{B w}
\end{eqnarray}

By our assumptions that $B = \Omega(\lg N)$, $B \le N$, $q = O(\lg N)$ and 
$w = \Theta(\lg N)$, 
 all terms except the second and third are $q \cdot o(N)$, therefore we can 
 summarize the space complexity of the redundancy as:

\begin{equation}
   R < q \cdot \left[ \frac{2 \tau  N  (q + 2 \lg B)}{w} +  o(N) \right]  + \frac{8\tau N \lg B}{w}
\end{equation}

When packed on the disk, each leading block requires an offset of 
$\lceil\lg{(Bw)}\rceil$ bits. 
As the number of leading blocks is $\lceil N/{\tau B}\rceil$, such information requires $o(N)$ bits in total.

Now we consider the space required to store $\mathcal{V}_{first}$, 
$\mathcal{V}_{parent}$, and $\mathcal{V}_{first\_child}$. 
Each vector must index $N$ bits. 
However, using Lemma~\ref{lem:rank_select}b, we can do better than $3N + o(N)$ 
bits of storage, if we consider that the number of $1$s in each bit vector is small. 

For $\mathcal{V}_{first}$, the total number of $1$s is $\lceil N/(\tau B)\rceil$ 
in the worst case, when $T$ forms a single path. 
The number of $1$s appearing in $\mathcal{V}_{parent}$ is bounded by the number of 
roots appearing on the top level of the layer below which is bounded by 
$\lceil N/(\tau B)\rceil$ as shown in Lemma~\ref{lem:chooselayer}. 
The bit vectors $\mathcal{V}_{first\_child}$ has at most as many $1$ bits as $\mathcal{V}_{parent}$, 
and as such the number of $1$ bits in each of the three vectors is bounded 
by $\lceil N/(\tau B) \rceil$. 
By Lemma \ref{lem:rank_select}b, each of the three bit vectors requires at most $\left\lceil\lg{ N \choose N/(\tau B)} \right\rceil + O(N \lg \lg N / \lg N) = o(N)$ bits. 

Finally, we consider the bit vectors $\mathcal{B}_l$ and $\mathcal{B}_r$.  $\mathcal{B}_l$ stores a single bit for each leading block.  There are as many leading blocks as there are layers, so  this bit vector has at most $\lceil N/(\tau B) \rceil$ bits.  
Thus, it can be represented using $o(N)$ bits. 
The number of $1$s in $\mathcal{B}_r$ is equal to the number of layers, which is $\lceil N/(\tau B) \rceil$. 
The number, $a$, of $0$s in $\mathcal{B}_r$ is equal to the total number of regular blocks. 
As there are fewer than $2N$ nodes in all the regular blocks and each regular 
block contains a non-constant number of nodes, 
we have $a = o(N)$. 
Therefore, the length of $\mathcal{B}_r$ is $o(N)$, which can be stored in $o(N)$ bits using Lemma \ref{lem:rank_select}a. 


Summing up, our data structures require $(2+q)N + q \cdot \left[ \frac{2 \tau  N  (q + 2 \lg B)}{w} +  o(N) \right]  + \frac{8\tau N \lg B}{w} $ bits in total when $0 <\tau < c = \min(1, \frac{1}{c'})$ for sufficiently large $N$.
\end{proof}

\subsection{Navigation} The algorithm for reporting a node-to-root path, including layer preorder numbers of the nodes on the path and their keys, is given by algorithms $ReportPath(T,v)$ (see Figure~\ref{fig:report_path}) and $ReportLayerPath(\ell_v, p_v)$ (see Figure \ref{fig:find_root}). 
Algorithm $ReportPath(T,v)$ is called with $v$ being the number of a node in $T$ given by $\pi$. $ReportPath$ handles navigation between layers, and calls $ReportLayerPath$ to perform the traversal within each layer. 
For $ReportLayerPath$, the parameters $\ell_v$ and $p_v$ are the layer number and the preorder value of node $v$ within the layer, as previously described. $ReportLayerPath$ returns the preorder number, within layer $\ell_v$, of the root of path reported from that layer. In $ReportLayerPath$ we find the block $b_v$ containing node $v$ using the algorithm $FindBlock(\ell_v, p_v )$ described in Figure \ref{fig:find_block}. 

The above algorithms operate on the data structures defined in Section~\ref{sec:bottomup_ds}. 
There are two types of data structures; the first type includes the data 
structures that are stored in individual tree blocks such as $B_t$ and $D_p$. 
After the corresponding trees blocks are loaded into internal memory, 
information stored in these data structures can be retrieved by performing 
operations such as linear scan without incurring additional I/Os. 
The second type are those that occupy a non-constant number of disk blocks such as $\mathcal{V}_{first}$ and $\mathcal{V}_{first\_child}$. 
We do not load them entirely into internal memory;
instead, we perform operations such as $\rankop$ and $\selop$ to retrieve information from them.

\begin{figure}[t]
\framebox[\linewidth]{
\begin{minipage}{0.95\linewidth}
\normalsize
   \textbf{Algorithm $ReportPath(T,v)$}
   \begin{enumerate}
      \item Find $\ell_v$, the layer containing $v$, using 
      $\ell_v = \rankop[1]{\mathcal{V}_{first},v}$.
      \item Find $\alpha_{\ell_v}$, the position in $\pi$ of $\ell_v$'s first 
      node, using $\alpha_{\ell_v} = \selop[1]{\mathcal{V}_{first},\ell_v}$.
      \item Find $p_v$, $v$'s preorder number within $\ell_v$, using 
      $p_v = v - \alpha_{\ell_v}$.
      \item Repeat the following steps until the top layer has been reported.
      \begin{enumerate}
         \item Let $r = ReportLayerPath(\ell_v, p_v)$ be the preorder number of the root of the path in layer $\ell_v$ (this step also reports the path within the layer).
         \item Find $\alpha_{(\ell_v-1)}$, the position in $\pi$ of the first node at the next higher layer, using $\alpha_{(\ell_v-1)} = \selop[1]{\mathcal{V}_{first},\ell_v-1}$.
         \item Find $\lambda$, the rank of $r$'s parent among all the nodes 
	  in the layer above that have children in $\ell_v$, using 
	  $\lambda = \rankop[1]{\mathcal{V}_{first\_child}, \alpha_{\ell_v} + r - 1} - 
	    \rankop[1]{\mathcal{V}_{first\_child}, \alpha_{\ell_v} - 1}$.
         \item Find which leaf $\delta$ at the next higher layer corresponds to 
	    $\lambda$, using  $\delta = \selop[1]{\mathcal{V}_{parent}, \rankop[1]{\mathcal{V}_{parent}, \alpha_{(\ell_v-1)}} -1 + \lambda}$.
         \item Update $\alpha_{\ell_v} = \alpha_{(\ell_v-1)}$; $p_v = \delta - \alpha_{(\ell_v - 1)}$; and $\ell_v = \ell_v - 1$.
      \end{enumerate}
   \end{enumerate}
\end{minipage}
}
\caption{Algorithm for reporting the path from node $v$ to the root of $T$}
\label{fig:report_path}
\end{figure}

\begin{figure}[t]
\framebox[\linewidth]{
\begin{minipage}{0.95\linewidth}
\normalsize
   \textbf{Algorithm $ReportLayerPath(\ell_v,p_v)$}
   \begin{enumerate}
      \item Load block $b_v$ containing $p_v$ by calling $FindBlock(\ell_v, p_v)$. Scan $B_t$ (the tree's representation) to locate $p_v$. 
Let $SB_v$ be the superblock containing $b_v$, and load $SB_v$'s first block if $b_v$ is not the first block in $SB_v$. Let $\lowestop(D_p)$ be the preorder number with respect to $SB_v$ of the lowest node on $b_v$'s duplicate path, i.e. the last non-zero entry in $D_p$ (let $\lowestop(D_p)=1$ if $b_v$ is a leading block), and let $\lowestop(SB_{D_p})$ be the preorder number in with respect to layer $\ell_v$ of the lowest node on the superblock duplicate path (again, if $b_v$ is a leading block, let $\lowestop(SB_{D_p})=1$).
      \item Traverse the path from $p_v$ to a root in $B_t$. If $r$ is the preorder number (within $B_t$) of a node on this path, report $(r-1) + (\lowestop(D_p)-1) + \lowestop(SB_{D_p})$ as its preorder number in layer $\ell_v$ and $B_k[r]$ as its key. This step terminates at a root in $B_t$. Let $r_k$ be the rank of this root in the set of roots of $B_t$.
      \item Scan the root-to-path array, $R_p$ from $\tau B$ to $1$ to find the largest $i$ such that $R_p[i] \ge r_k$. If $r_k \ge R_p[1]$, then $r$ is on the top level in the layer, so return $(r-1) + (\lowestop(D_p)-1) + \lowestop(SB_{D_p})$ as its preorder number in layer $\ell_v$ and $B_k[r]$ as its key. We then terminate.
      \item Set $j = i - 1$.
      \item $\whileop( j \geq 1$ and $D_p[j] \neq 0)$ report $D_p[j] + \lowestop(SB_{D_p}) - 1$, the preorder number in layer $\ell_v$ of the node on the duplicate path at level $j$ of this layer, and $B_k[D_p[j]]$, the key stored in this node, and then set $j = j - 1$.
      \item If $j \geq 1$ then report $SB_{D_p}[j]$, the preorder number in layer $\ell_v$ of the node on the superblock duplicate path at level $j$ of this layer, and the key stored in this node (retrieved from $SB_v$'s first block), and set $j = j - 1$  $\untilloop(j < 1)$.
   \end{enumerate}
\end{minipage}
}
\caption[Algorithm $ReportLayerPath$]{Steps to execute traversal within a layer, $\ell_v$, starting at the node with preorder number $p_v$. This algorithm reports the nodes visited and returns the layer preorder number of the root at which it terminates.}
\label{fig:find_root}
\end{figure}

\begin{figure}[t]
\framebox[\linewidth]{
\begin{minipage}{0.95\linewidth}
\normalsize
   \textbf{Algorithm $FindBlock(\ell_v,p_v)$}
   \begin{enumerate}
      \item Find $\sigma$, the disk block containing $\ell_v$'s 
	leading block using $\sigma = \rankop[1]{\mathcal{B}_l,\ell_v}$.
      \item Find $\alpha$, the rank of $\ell_v$'s leading block within 
	$\sigma$, using the identity $\alpha = \ell_v - \selop[1]{B_l, \sigma} + 1$. 
      \item Scan $\sigma$ to find, and load, the data for $\ell_v$'s leading block (may required loading the next disk block). Note the size $\delta$ of the leading block.
      \item If $p_v \le \delta$ then $p_v$ is in the already loaded leading block, terminate.
      \item Calculate $\omega$, the rank of the regular block containing $p_v$ within the $\selop[1]{\mathcal{B}_r, \ell_v+1} - \selop[1]{\mathcal{B}_r, \ell_v}$ 
      regular blocks in this layer, using the values of $\succblksize_1$ and $\succblksize_2$. 
      \item Load the disk block containing the ($\rankop[0]{\mathcal{B}_r, \ell_v} + \omega$)\textsuperscript{th} regular block and terminate.
   \end{enumerate}
\end{minipage}
}
\caption[Algorithm $Findblock$]{$FindBlock$ algorithm.}
\label{fig:find_block}
\end{figure}

We now have the following lemma: 

\begin{lemma}\label{lem:bottom_up_IOs}
The algorithm ReportPath traverses a path of length $K$ in $T$ in $O( K / \tau B + 1)$ I/Os.
\end{lemma}

\begin{proof}
In each layer we progress $\tau B$ steps toward the root of $T$. 
In order to do so, we must load the disk block containing the current node and 
possibly also the block storing the superblock duplicate path. 
When we step between layers we must account for the I/Os involved in mapping the 
layer level roots to their parents in the predecessor layer. 
This involves a constant number of \emph{rank} and \emph{select} operations 
which may be done in $O(1)$ I/Os.

The $FindBlock$ algorithm involves a scan of the disk blocks storing leading 
blocks, but this may generate at most 2 I/Os. 
The remaining operations in $FindBlock$ use a constant number of \emph{rank} and \emph{select} calls, and therefore require $O(1)$ I/Os. 

As a path of length $K$ has nodes in $\lceil K / \tau B\rceil$ layers, as in 
order to traverse the path the number of I/Os required in each layer and between 
two consecutive layers is constant (as shown above), we conclude that it 
takes $O(\lceil K / \tau B\rceil) = O( K / \tau B + 1)$ I/Os to traverse the path.
\end{proof}

Lemmas \ref{lem:bottom_up_space} and \ref{lem:bottom_up_IOs} lead to the following theorem: 

\begin{theorem}\label{thm:node-root}
A tree $T$ on $N$ nodes with $q$-bit keys, where $q = O(\lg N)$, 
can be represented in 
$(2+q)N + q \cdot \left[ \frac{2 \tau  N  (q + 2 \lg B)}{w} +  o(N) \right]  + \frac{8\tau N \lg B}{w} $ 
bits such that given a node, the path from this node to the root of $T$ can be 
reported in $O(K/\tau B + 1)$ I/Os, where $K$ is the length of the node-to-root path 
reported, $\tau$ is a constant such that $0 < \tau < 1$, and $w$ is the word size. 
\end{theorem}

\begin{proof}
Lemmas \ref{lem:bottom_up_space} and \ref{lem:bottom_up_IOs} guarantee that 
this theorem is true for sufficiently large $N$ when $\tau$ is a constant 
such that $0 < \tau < c$. 

To argue the mathematical correctness of our theorem when the above conditions 
are removed we first claim that when $c \le \tau < 1$ our theorem is still true 
for sufficiently large $N$. 
In this case, we set the height of the layers to be $\lambda = c/2 $ when constructing our data structures. 
The data structures constructed would occupy 
$(2+q)N + q \cdot \left[ \frac{2 \lambda N  (q + 2 \lg B)}{w} +  o(N) \right]  + \frac{8\lambda N \lg B}{w}$ bits, 
which is smaller than the space cost in our theorem, since $\lambda < \tau$. 
We can then pad our data structures with bits of $0$s to achieve the claimed 
space bound. 
These bits are never used and it is therefore unnecessary to add them in practice; 
they are there to guarantee that our theorem be mathematically correct. 

In order to further remove the condition that $N$ is sufficiently large, we observe from the proof of Lemma~\ref{lem:bottom_up_space} that the condition ``for sufficiently large $N$'' comes from order notation. 
Thus this condition means that there exists a positive constant number $N_0$ such that for any $N \ge N_0$, this theorem is true. 
When $N < N_0$, the tree $T$ has a constant number of nodes, and thus we can easily store it in external memory in a constant amount of space, say $S_0$ bits, while supporting the report of any node-to-root path in $O(1)$ I/Os. 
Thus, for any $N$ we can use at most 
$\max(S_0, (2+q)N + q \cdot \left[ \frac{2 \tau  N  (q + 2 \lg B)}{w} +  o(N) \right]  + 
\frac{8\tau N \lg B}{w}) < S_0 + (2+q)N + q \cdot \left[ \frac{2 \tau  
N  (q + 2 \lg B)}{w} +  o(N) \right]  + \frac{8\tau N \lg B}{w}$ bits to store $T$ while achieving the I/O bounds claimed in this theorem. 
Since $S_0$ is a constant, it can be absorbed in the term $q\cdot o(\lg N)$. 
\end{proof}

To show that our data structure is space-efficient, it suffices to show that the space 
cost minus the $(2+q)N$ bits required to encode the tree structure and the keys is small. 
This extra space cost is shown to be 
$q \cdot \left[ \frac{2 \tau  N  (q + 2 \lg B)}{w} +  o(N) \right]  + \frac{8\tau N \lg B}{w}$ 
bits in the above theorem. 
Since $q \cdot \left[ \frac{2 \tau  N  (q + 2 \lg B)}{w}  \right]  + \frac{8\tau N \lg B}{w} = \tau \cdot q \cdot O(N)$,
we can select $\tau$ such that this expression is at most $\eta N$ for any constant 
$0 < \eta < 1$ when $N$ is sufficiently large.
Thus, our data structure occupies at most $(2 + q) N + q\cdot (\eta N + o(N))$ bits for any constant $0 < \eta < 1$ for sufficiently large $N$. 
The strategy used in the proof of Theorem~\ref{thm:node-root} can also be used here to remove the condition that $N$ is sufficiently large. 

For the case in which we do not maintain a key with any node, we have the following corollary: 

\begin{corollary}\label{cor:node-root}
A tree $T$ on $N$ nodes can be represented in $2N + \frac{\epsilon N \lg B}{w} + o(N)$ bits such that given a node-to-root path of length $K$, the path can be reported in $O(K/B+1)$ I/Os, for any constant number $\epsilon$ such that $0 < \epsilon < 1$.
\end{corollary} 

\begin{proof}
When we prove Lemma~\ref{lem:bottom_up_space}, Lemma~\ref{lem:bottom_up_IOs}, 
and Theorem~\ref{thm:node-root}, the only places where we make use of the fact 
that $q$ is positive is where we use $q\cdot o(N)$ to absorb the term $o(N)$ in space analysis. Thus in the above theorem and lemmas, we can set $q = 0$ and add a term of $o(N)$ to the space bound to achieve the following result: A tree $T$ on $N$ nodes can be represented in $2N + \frac{8\tau N \lg B }{w} + o(N)$ bits such that given a node-to-root path of length $K$, the path can be reported in $O(K/B+1)$ I/Os, for any constant number $\tau$ such that $0 < \tau < 1$. 

In order to simplify our space result, we define one additional term $\epsilon = 8 \tau$, and then the claim in this corollary is true for any constant $0 < \epsilon < 8$, which includes the special case in which $0 < \epsilon < 1$.
\end{proof}

It is clear that the space cost of our data structure in the above corollary is very close to the information-theoretic lower bound of representing an ordinal tree on $N$ nodes, which is $2N-O(\lg N)$ as stated in Section~\ref{sec:succincttrees}.

\section{Top-Down Traversal}\label{sec:top_down}

Given a binary tree $T$, in which every node is associated with a key, we wish to traverse a top-down path of length $K$ starting at the root of $T$ and terminating at some node $v \in T$. Let $\succblksize$ be the maximum number of nodes that can be stored in a single block, and let $q = O(\lg N)$ be the number of bits required to encode a single key. 

\subsection{Data Structures} 

We begin with a brief sketch of our data structures. 
A tree $T$ is partitioned into subtrees, where each subtree $T_i$ is laid out 
into a \emph{tree block}. 
Each tree block contains a succinct representation of $T_i$ and the set of 
keys associated with the nodes in $T_i$. 
The edges in $T$ that span a block boundary are not explicitly stored within 
the tree blocks. 
Instead, they are encoded through a set of bit vectors (detailed later in 
this section) that enable navigation between tree blocks.

To introduce our data structures, we give some definitions. 
If the root node of a tree block is the child of a node in another block, 
then the first block is a \emph{child} of the second. 
There are two types of tree blocks: \emph{internal} blocks that have one 
or more \emph{child} blocks, and {\em terminal} blocks that have no \emph{child} blocks.  
The \emph{block level} of a block is the number of blocks along a path from the 
root of this block to the root of $T$.

We number the internal blocks in the following manner. 
Firstly, we number the block containing the root of $T$ as 1, and we number its 
child blocks consecutively from left to right. 
We then consecutively number the internal blocks at each successive block 
level (see Figure \ref{fig:root-node-block-arrange}). 
The internal blocks are stored on the disk in an array $I$ of disk blocks, such 
that the tree block numbered $j$ is stored in entry $I[j]$.

\begin{figure}[t]
	\centering
		\includegraphics{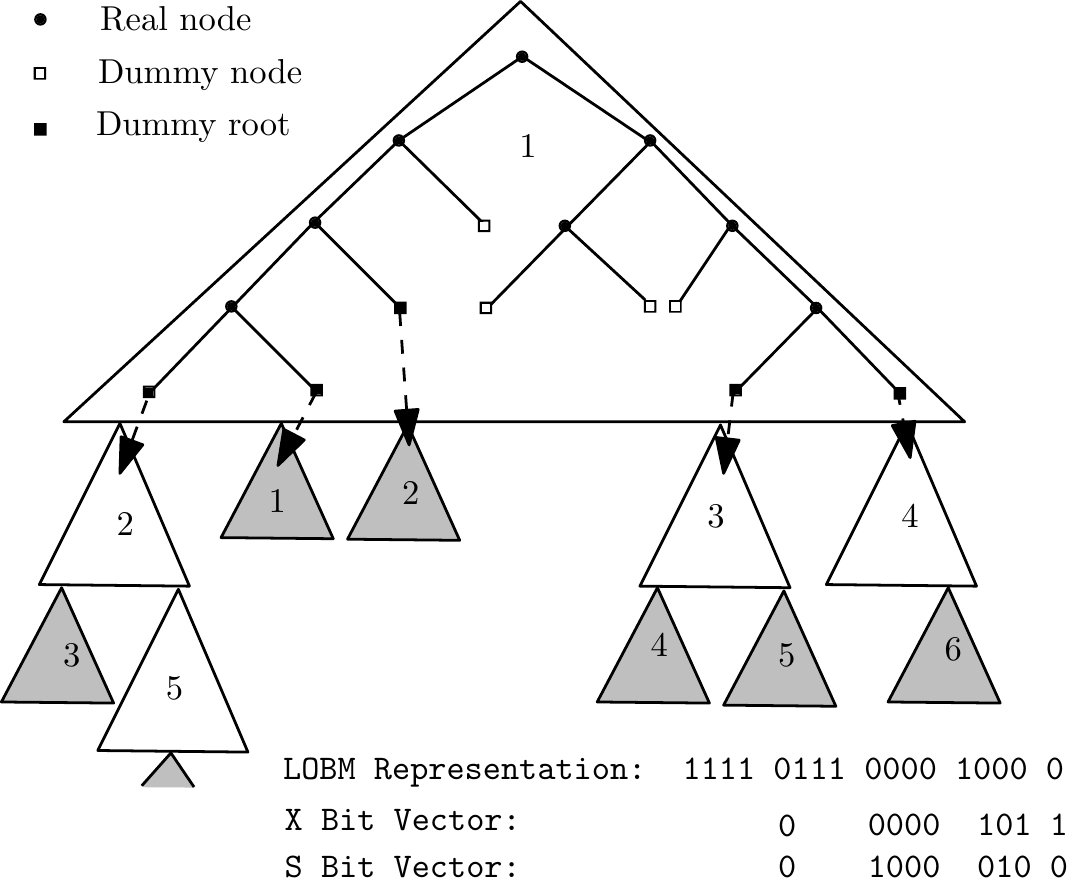}
	\caption[Tree block representation]{Numbering of internal (hallow triangles) and terminal 
(shaded triangles) blocks for $T$.  
The structure of $T$ within internal block 1 is also shown. 
The dashed arrows indicate the parent-child relationship between dummy 
roots in internal block 1 and their child blocks.  
Finally, the LOBM representation for internal block 1 and the corresponding 
bits in bit vectors X and S are shown at the bottom.  
Bits in bit vectors X and S have been spaced such that they align with their 
corresponding 0 bits (the dummy nodes/roots) in the LOBM representation.}
	\label{fig:root-node-block-arrange}
\end{figure}

Terminal blocks are numbered and stored separately. 
Starting again at 1, they are numbered from left to right at each block level. 
Terminal blocks are stored in the array $Z$. As terminal blocks may vary in size, there is no one-to-one correspondence between disk and tree blocks in $Z$; rather, the tree blocks are packed into $Z$ to minimize wasted space.  At the start of each disk block $j$, a $\lceil\lg{(Bw)}\rceil$-bit \emph{block offset} is stored which indicates the position of the starting bit of the first terminal block stored in $Z[j]$. Subsequent terminal blocks are stored immediately following the last bits of the previous terminal blocks. If there is insufficient space to record a terminal block within disk block $Z[j]$, the remaining bits are stored in $Z[j+1]$.

We now describe how an individual internal tree block is encoded. Consider the block of subtree $T_j$; it is encoded using the following structures:
\begin{enumerate}
\item The block keys, $B_k$, is an $\succblksize$-element array which encodes the keys of $T_j$ in level order.
\item The tree structure, $B_s$, is an encoding of $T_j$ using the LOBM sequence of 
Jacobson~\cite{jac_1989}. 
More specifically, we define each node of $T_j$ as a \emph{real} node.  
$T_j$ is then augmented by adding \emph{dummy} nodes as the left and / or right 
child of any real node that does not have a corresponding real child node in $T_j$. 
The dummy node may, or may not, correspond to a node in $T$, but the corresponding
node is not part of $T_j$. 
We then perform a level order traversal of $T_j$ and output a 1 each time we 
visit a real node, and a 0 each time we visit a dummy node. 
If $T_j$ has $\succblksize$ nodes, the resulting bit vector has $\succblksize$ 1s for real nodes and $\succblksize+1$ 0s for dummy nodes. Observe that the first bit is always 1, and the last two bits are always 0s, so it is unnecessary to store them explicitly. Therefore, $B_s$ can be represented with $2\succblksize-2$ bits.
\item The \emph{dummy offset}, $B_d$. Let $\Gamma$ be a total order over the 
set of all dummy nodes in internal blocks. 
In $\Gamma$ the order of dummy node, $d$, is determined first by its block number, and second by its position within $B_s$. The dummy offset records the position in $\Gamma$ of the first dummy node in $B_s$.
\end{enumerate}

 The encoding for terminal blocks is identical to internal blocks except that the dummy offset is omitted, and the last two $0$s of $B_s$ are encoded explicitly.

We now define a \emph{dummy root}. Let $T_j$ and $T_k$ be two tree blocks where $T_k$ is a child block of $T_j$. Let $r$ be the root of $T_k$, and $v$ be $r$'s parent in $T$. When $T_j$ is encoded, a dummy node is added as a child of $v$ which corresponds to $r$. Such a dummy node is termed a dummy root.

Let $\ell$ be the number of dummy nodes over all internal blocks.  We create three bit arrays:
\begin{enumerate}
\item $X[1..\ell]$ stores a bit for each dummy node in internal blocks. Set $X[i] = 1$ iff the $i$\textsuperscript{th} dummy node in $\Gamma$ is the dummy root of an internal block.
\item $S[1..\ell]$ stores a bit for each dummy node in internal blocks. Set $S[i] = 1$ iff the $i$\textsuperscript{th} dummy node in $\Gamma$ is the dummy root of a terminal block.
\item $S_B[1..\ell']$, where $\ell'$ is the number of 1s in $S$. Each bit in this array corresponds to a terminal block. Set $S_B[j] = 1$ iff the corresponding terminal block is stored starting in a disk block of $Z$ that differs from the one in which terminal block $j-1$ starts.
\end{enumerate}

\subsection{Block Layout} 
We have yet to describe how $T$ is split up into \emph{tree} blocks. 
This is achieved using the two-phase blocking strategy of 
Demaine \etal~\cite{dem_iac_lan_2004} (see Section~\ref{ssec:tree_blocking} for details). 
Phase one blocks the first $c \lg {N}$ levels of $T$, where $0 < c < 1$. 
Starting at the root of $T$, the first $\lfloor \lg{(\succblksize+1) \rfloor}$ levels are placed in a block. Conceptually, if this first block is removed, we are left with a forest of $O(\succblksize)$ subtrees. The process is repeated recursively until $c \lg {N}$ levels of $T$ have thus been blocked.

In the second phase we block the rest of the subtrees by the following recursive procedure. The root, $r$, of a subtree is stored in an empty block. The remaining $\succblksize-1$ capacity of this block is then subdivided, proportional to the size of the subtrees, between the subtrees rooted at $r$'s children. During this process, if at a node the capacity of the current block is less than $1$, a new block is created. To analyze the space costs of our structures, we have the following lemma.

\begin{lemma}\label{lem:top_down_space}
The data structures described above occupy $(3 + q)N + o(N)$ bits.
\end{lemma}

\begin{proof}
We first determine the maximum block size $\succblksize$. 
In our model a block stores at most $B w$ bits. 
The encoding of the subtree $T_j$ requires $2 \succblksize$ bits. 
We also need $\succblksize q$ bits to store the keys, 
and $\lceil\lg N\rceil + \lceil \lg (Bw)\rceil$ bits to store the dummy offset.  
We therefore have the following equation: $2\succblksize + \succblksize q + 
\lceil\lg N\rceil + \lceil \lg (Bw)\rceil \le B w$. 
Thus, number of nodes stored in a single block satisfies:
\begin{eqnarray}
\succblksize \le \frac{Bw - \lceil\lg N\rceil - \lceil \lg (Bw)\rceil}{q+2}
\end{eqnarray}

Therefore, we choose $\succblksize = \lfloor\frac{Bw - \lceil\lg N\rceil - \lceil \lg (Bw)\rceil}{q+2}\rfloor$ to partition the tree. Thus:
\begin{eqnarray}
   \succblksize =  \Theta \left( \frac{ B \lg {N} }{q} \right)
\end{eqnarray}

During the first phase of the layout, a set of non-full internal blocks may be 
created.  However, the height of the phase 1 tree is bounded by $c \lg{N}$ levels, 
so the total number of wasted bits in these blocks is bounded by $o(N)$.

The arrays of blocks $I$ and $Z$ store the structure of $T$ using the LOBM 
succinct representation which requires $2N$ bits. 
The dummy roots are duplicated as the roots of child blocks, but as the first 
bit in each block need not be explicitly stored, the entire tree structure still 
requires only $2N$ bits. We also store $N$ keys which require $N \cdot q$ bits. 
The block offsets stored in $Z$ and the dummy offsets stored for internal blocks require $o(N)$ bits in total. 
The bit vectors $S$ and $S_B$ have size at most $N+1$, but in both cases the number of $1$ bits is bounded by $N/\succblksize$. By Lemma \ref{lem:rank_select}b, we can store these vectors in $o(N)$ bits. 
The number of $1$ bits in $X$ is bounded by $N/2$. By lemma \ref{lem:rank_select}a it can be encoded by $N + o(N)$ bits. 
The total space is thus $(3+q)N + o(N)$ bits.
\end{proof}

\subsection{Navigation} 

Navigation in $T$ is summarized in Figures \ref{fig:top-down-search-alg} 
and \ref{fig:term-block-search-alg} which show the algorithms 
$Traverse(key,i)$ and $TraverseTerminalBlock(key,i)$, respectively. 
During the traversal, the function \linebreak $\compop(key)$ compares the 
value \emph{key} to the key of a node in order to determine which branch of the 
tree to traverse. 
The parameter $i$ is the number of a \emph{disk} block. 
Traversal is initiated by calling $Traverse(key,1)$.

\begin{figure}[t]
\framebox[\linewidth]{
\begin{minipage}{0.95\linewidth}
\normalsize
   \textbf{Algorithm $Traverse(key, i)$}
   \begin{enumerate}
      \item Load block $I[i]$ to main memory. Let $T_i$ denote the subtree stored in $I[i]$.
      \item Scan $B_s$ to navigate within $T_i$. At each node $x$, use $\compop(key, B_k[x])$ to determine which branch to follow until a dummy node $d$ with parent $p$ is reached.
      \item Scan $B_s$ to determine $j = \rankop[0]{B_s,d}$.
      \item Determine the position of $j$ with respect to $\Gamma$ by adding the \emph {dummy offset} to calculate $\lambda = B_d + j - 1$.
      \item If $X[\lambda]=1$, then set $i= \rankop[1]{X,\lambda}$ and call $Traverse(key,i)$.
      \item If $X[\lambda] =0$ and $S[\lambda]=1$, then set $i= \rankop[1]{S,\lambda}$ and call $TraverseTerminalBlock(key,i)$.
      \item If $X[\lambda] = 0$ and $S[\lambda]=0$, then $p$ is the final node on the traversal, so the algorithm terminates.
   \end{enumerate}
\end{minipage}
}
\caption{Top-down searching algorithm for a blocked tree}
\label{fig:top-down-search-alg}
\end{figure}

\begin{figure}[t]
\framebox[\linewidth]{
\begin{minipage}{0.95\linewidth}
\normalsize
   \textbf{Algorithm $TraverseTerminalBlock(key,i)$}
   \begin{enumerate}
      \item Load disk block $Z[\lambda]$ containing terminal block $i$, where $\lambda = \rankop[1]{S_B, i}$.
      \item Let $B_d$ be the offset of disk block $Z[\lambda]$.
      \item Find $\alpha$, the rank of terminal block $i$ within $Z[\lambda]$, by scanning from $S_B[i]$ backwards to find the largest $j \le i$ such that $S_B[j]=1$. Set $\alpha = i-j+1$.
      \item Starting at $B_d$, scan $Z[\lambda]$ to find the start of the $\alpha$\textsuperscript{th} terminal block. 
      Recall that each block stores a bit vector $B_s$ in the LOBM encoding, so that we can determine when we have reached the end of one terminal block as follows:
      \begin{enumerate}
         \item Set two counters $\mu = \beta = 1$; $\mu$ records the number of $1$ bits (\emph{real} nodes) encountered in $B_s$ during the scan, while $\beta$ records the excess of $1$ to $0$ bits  (\emph{dummy} nodes) encountered. Both are initially set to $1$ as the root node of the block is implied.
         \item Scan $B_s$. When a $1$ bit is encountered increment $\mu$ and $\beta$. When a $0$ bit is encountered decrement $\beta$. Terminate the scan when $\beta < 0$ as the end of $B_s$ has been reached.
         \item Now $\mu$ records the number of nodes in the terminal block, so calculate the length of array $B_k$ needed to store the keys and jump ahead this many bits. This will place the scan at the start of the next terminal block.
      \end{enumerate} 
      \item Once the $\alpha$\textsuperscript{th} block has been reached, the 
      terminal block can be read in (process is the same as scanning the 
      previous blocks). It may be the case that this terminal block overruns 
      the \emph{disk} block $Z[\lambda]$ into $Z[\lambda+1]$. 
      In this case, skip the first $\lceil \lg{(Bw)} \rceil$ bits of 
      $Z[\lambda+1]$ which stores the block offset, and continue reading in the 
      terminal block.
      \item With the terminal block in memory, the search can be concluded in 
      a manner analogous to that for internal blocks, except that once a dummy 
      node is reached, the search terminates.
   \end{enumerate}
\end{minipage}
}
\caption{Performing search for a terminal block}
\label{fig:term-block-search-alg}
\end{figure}

\begin{lemma}\label{lem:top-down-IOs}
For a tree $T$ laid out in blocks and represented by the data structures described 
above, a call to $TraverseTerminalBlock$ can be performed in $O(1)$ I/Os, while $Traverse$ can be executed in $O(1)$ I/Os per recursive call.
\end{lemma}

\begin{proof}

Internal blocks are traversed by the algorithm $Traverse(key,i)$ in Figure~\ref{fig:top-down-search-alg}. Loading the block in step 1 requires a single I/O, while steps 2 through 4 are all performed in main memory. Steps 5 and 6 perform lookups and call $\rankop$ on $X$ and $S$, respectively. This requires a constant number of I/Os. Step 7 requires no additional I/Os.

The algorithm $TraverseTerminalBlock(key,i)$ is executed at most once per traversal. 
The look-up and $\rankop$ require only a single I/O. 
The only step that might cause problems is step 3 in which the bit array $S_B$ 
is scanned.  
Note that each bit in $S_B$ corresponds to a terminal block stored in $Z$. 
The terminal block corresponding to $i$ is contained in $Z[\lambda]$, 
and the terminal block corresponding to $j$ also starts in $Z[\lambda]$. A terminal block is represented by at least $2 + q$ bits. As blocks in $S_B$ are of the same size as in $Z$, we cross at most one block boundary in $S_B$ during the scan.
\end{proof}

The I/O bounds are then obtained directly by substituting our succinct block size $\succblksize$ for the standard block size $B$ in Lemma \ref{lem:top_down_demaine}. Combined with Lemmas \ref{lem:top_down_space} and  \ref{lem:top-down-IOs}, this gives us the following result:

\begin{theorem}\label{thm:root-node}
A rooted binary tree, $T$, of size $N$, with keys of size $q=O(\lg{N})$ bits,  can be stored using $(3+q)N + o(N)$ bits so that a root to node path of length $K$ can be reported with:
\begin{enumerate}
\item $O \left( \frac{K}{\lg(1+(B \lg N)/q)} \right)$ I/Os, when $K = O(\lg N)$,
\item $O \left( \frac{\lg N}{\lg (1+\frac{B \lg^2 N}{qK} )} \right)$ I/Os, when $K=\Omega(\lg N)$ and $K=O \left( \frac{B \lg^2 N}{q} \right)$, and 
\item $O \left( \frac{qK}{B \lg N} \right)$ I/Os, when $K = \Omega \left( \frac{B \lg^2 N}{q} \right)$.
\end{enumerate}
\end{theorem}

When key size is constant, the above result leads to the following corollary.

\begin{corollary}\label{cor:top_down_fixed_keys}
Given a rooted binary tree, $T$, of size $N$, with keys of size  $q = O(1)$ bits, $T$ can be stored using $(3+q)N + o(N)$ bits in such a manner that a root to node path of length $K$ can be reported with:
\begin{enumerate}
\item $O \left( \frac{K}{\lg(1+(B \lg N))} \right)$ I/Os when $K = O(\lg N)$,
\item $O \left( \frac{\lg N}{\lg (1+\frac{B \lg^2 N}{K} )} \right)$ I/Os when $K=\Omega(\lg N)$ and $K=O \left( B \lg^2 N \right)$, and 
\item $O \left( \frac{K}{B \lg N} \right)$ I/Os when $K = \Omega (B \lg^2 N)$.
\end{enumerate}
\end{corollary}

Corollary \ref{cor:top_down_fixed_keys} shows that, in the case where the number 
of bits required to store each search key is constant, our approach not only 
reduces storage space but also improves the I/O efficiency. 
In particular, when $K=\omega(B \lg{N})$ and $K=O(B \lg^2{N})$, the number of I/Os required in Lemma \ref{lem:top_down_demaine} is $\Theta(K/B) = \omega(\lg{N})$ while that required in Corollary~\ref{cor:top_down_fixed_keys} is $O \left( \frac{\lg N}{\lg (1+\frac{B \lg^2 N}{K} )} \right) = O(\lg{N})$. 
For all the other values of $K$, it is clear that our I/O bounds are better. 

\section{Conclusions}
 
We have presented two new data structures that are both I/O efficient 
and succinct for bottom-up and top-down traversal in trees. 
Our bottom-up result applies to trees of arbitrary degree, while our 
top-down result applies to binary trees.  
In both cases the number of I/Os is asymptotically optimal.

Our results lead to several open problems.  
Our top-down technique is valid for binary trees only. 
Whether this approach can be extended to trees of larger degrees is an 
open problem.  
For the bottom-up case it would be interesting to see if the asymptotic 
bound on I/Os can be improved from $O(K/B+1)$ to something closer to 
$O(K/\succblksize+1)$ I/Os, where $\succblksize$ is the number of nodes 
that can be represented succinctly in a single disk block.  
In both the top-down and bottom-up cases, several $\rankop$ and 
$\selop$ operations are required to navigate between blocks. 
These operations use only a constant number of I/Os, and it would be 
useful to reduce this constant factor. 
This might be achieved by reducing the number of $\rankop$ and 
$\selop$ operations used in the algorithms, or by demonstrating how 
the bit arrays could be interleaved to guarantee a low number of I/Os 
per block.

%
%
%
%
%
%

\chapter[Traversal in Planar Graphs]{I/O-Efficient Path Traversal in Planar 
  Graphs}\label{chp:succinct_graphs}
\chaptermark{Traversal in Planar Graphs}

\section{Introduction}

External memory (EM) data structures and succinct data structures both
address the problem of representing very large data sets.
In the EM model, the goal is to structure data that are too large to fit into
internal memory in a way that minimizes the transfer of data between
internal and external memory when answering certain queries.
For succinct data structures, the aim is to encode the structural
component of the data structure using as little space as is
theoretically possible while still permitting efficient navigation of
the structure.
Thus, EM data structures deal with the I/O bottleneck
that arises when the data are too large to fit into memory, while
succinct data structures help to avoid this bottleneck as they allow
more data to be stored in memory.
Succinct EM data structures maximize the amount of data that fits on a disk of
a given size or in a disk block.
The former is important because an increasing number of
large-scale applications find themselves limited by the amount of data
that fits on a disk.
The latter helps to reduce the I/O bottleneck further, as more data
can be swapped between memory and disk in a single I/O operation.

In this chapter, we develop a succinct EM data structure for path
traversal in planar graphs.
Given a bounded-degree planar graph $G$,
our goal is to simultaneously minimize the amount of space used to
store $G$ on disk as well as the number of I/O operations required to report
a path of length $K$ in $G$.
As practical applications of our
structure, we show how it can be used to answer a range of important
queries on triangular irregular network (TIN) models of terrains.

\subsection{Background}\label{sec:background}

In the \emph{external memory} (EM) model \cite{DBLP:journals/cacm/AggarwalV88}, the
computer is assumed to be equipped with a two-level memory hierarchy
consisting of \emph{internal} and (disk-based) \emph{external memory}.
The external memory is assumed to have infinite size, but accessing data
elements in external memory is several orders of magnitude slower than accessing
data in internal memory.
Conversely, while internal memory permits efficient operations, its size is limited 
to $M$ data elements.
Data is transferred between internal and external memory by means of
\emph{I/O operations} (\emph{I/Os} for short), each of which transfers
a block of $B$ consecutive data items. 
The efficiency of a data structure in the EM model is measured in terms of the
space it uses and the number of I/Os required to answer certain queries.
In order to take advantage of blockwise disk accesses, it is necessary that $M
\ge B$.
Throughout this work, we assume $M \ge 2B$ and $B = \OmegaOf{\lg N}$,
where $N$ denotes the input size. Furthermore, we assume that the machine
word size is $\wsize = \Theta(\lg N)$ bits, an assumption commonly used 
in papers on succinct data structures in internal 
memory~\cite{DBLP:journals/talg/RamanRS07}.

Nodine~\etal~\cite{ngv_1996} first explored the problem of blocking
graphs in external memory for efficient path traversal. 
They stored the
 graph in disk blocks (possibly with duplication of
vertices and edges), so that any path in the graph could be traversed
using few I/Os relative to the path's length.
The efficiency of the blocking is expressed in terms
of the \emph{blocking speed-up}, which is the minimum ratio between
the length of a traversed path and the number of I/Os required to
traverse it, taken over all paths that require more than $c$ I/Os to
traverse, for some constant $c$, and assuming that only one block can
be held in memory at any point in time.
The authors identified the optimal bounds for the worst-case blocking speed-up
for several classes of graphs.
Agarwal \etal~\cite{DBLP:conf/soda/AgarwalAMVV98} proposed a
blocking of bounded-degree planar graphs such that any path of length
$K$ can be traversed using $\OhOf{K / \lg B}$ I/Os.

Succinct data structures represent their structural
components using space as near the information-theoretic lower bound as
possible while still permitting efficient operations.
These were originally proposed by Jacobson
\cite{jac_1989}, who designed succinct representations of trees and
graphs.
To represent graphs, Jacobson relied on the technique of book
embeddings by Bernhart and Kainen
\cite{DBLP:journals/jct/BernhartK79}.
A $k$-page book embedding of a graph is an ordering of the graph's vertices,
along with a partition of its edges into $k$ ``pages'',
each of which is a subset of edges that induces an outerplanar embedding of
the graph consistent with the chosen vertex ordering.
Yannakakis \cite{DBLP:conf/stoc/Yannakakis86}
demonstrated that $k=4$ is necessary and sufficient to partition
planar graphs.
Jacobson's data structure embeds each page of a graph
on $N$ vertices using $\OhOf{N}$ bits, and can thus represent
$k$-page graphs using $\OhOf{kN}$ bits and planar graphs using
$\OhOf{N}$ bits.
The structure supports adjacency queries using
$\OhOf{\lg N}$ bit probes, and the listing of the neighbours of a
vertex $v$ of degree $\deg{v}$ using $\OhOf{\deg{v} \lg N + k}$ bit
probes.

Jacobson's result was improved upon by Munro and Raman
\cite{DBLP:conf/focs/MunroR97} under the word-RAM model.
They showed how to represent a $k$-page graph with $N$ vertices and $M$ edges
using $2kN + 2M + \ohOf{Nk + M}$ bits such that adjacency and
vertex degree queries can be answered in $\OhOf{k}$ time.
The neighbours of a vertex $v$ can be reported in $\OhOf{d(v)+k}$ time.
For planar graphs, this result translates to an $(8N + 2M + \ohOf{N+M})$-bit
representation that answers adjacency and degree queries in constant
time, and lists neighbours in $\OhOf{\deg{v}}$ time.
Gavoille and Hanusse \cite{DBLP:journals/dmtcs/GavoilleH08} proposed an encoding
for $M$-edge $k$-page-embeddable graphs that allows multiple edges and
loops.
For a graph with no isolated vertices, their structure uses $2M
\lg k + 4M$ bits, which is an improvement over Munro and Raman's
structure whenever $M \le kN / (2\lg k)$.
By adding an auxiliary table of $\ohOf{M \lg k}$ bits, this encoding
supports calculating vertex degrees in constant time, answering adjacency
queries in $\OhOf{\lg k}$ time (constant for planar graphs), and
accessing all neighbours of a vertex in $\OhOf{\deg{v}}$ time.

An alternate approach to book embeddings of planar graphs, which uses
canonical orderings of the graph's vertices and edges, was presented
by Chuang \etal~\cite{chuang_et_al_1998}.
Their solution represents a planar graph using
$2M + (5 + 1/\epsilon)N + \ohOf{M + N}$ bits, for any $\epsilon > 0$ 
(with $\epsilon = \OhOf{1}$), and supports constant-time adjacency 
and degree queries.
If the graph is simple, this bound becomes
$\frac{5}{3}M + (5 + 1/\epsilon)N + \ohOf{N}$.
For triangulated planar graphs, they reduced the space
bound to $2M + 2N + \ohOf{N}$ bits.
The space bound for general planar graphs was reduced by Chiang
\etal~\cite{DBLP:journals/siamcomp/ChiangLL05}, who showed that a
graph with no multiple edges or self loops can be represented using
$2M + 2N + \ohOf{M + N}$ bits to support the same set of
operations.
For the case of triangulated planar graphs, Yamanaka and
Nakano~\cite{DBLP:conf/walcom/YamanakaN08} presented an encoding that
uses $2N + \ohOf{M}$ bits and supports adjacency, degree, and
clockwise neighbour queries in constant time.

Barbay \etal~\cite{DBLP:conf/isaac/BarbayAHM07} presented several
results with respect to both planar graphs, triangulations, and
$k$-page graphs, including the first results for labeled graphs and
triangulations.
For planar triangulations, they added support for
rank/select queries of edges in counterclockwise order using $2M
\lg 6 + \ohOf{M}$ bits.
For plane graphs, their structures support
standard queries, as well as rank/select queries in counterclockwise
order, using $3N(2 \lg 3 + 3 + \epsilon) + \ohOf{N}$
bits (for $0 < \epsilon < 1$).
For $k$-page graphs
with large values of $k$, they proposed a representation using $N + 2M
\lg k + M \cdot \ohOf{\lg k} + \OhOf{M}$ bits and supporting
adjacency queries in $\OhOf{\lg k \lg\lg k}$ time, degree queries in
constant time, and the listing of neighbours in $\OhOf{\deg{v} \lg\lg k}$ time.
An alternative representation uses
$N + (2 + \epsilon)M \lg k + M \cdot \ohOf{\lg k} + \OhOf{M}$ bits and
supports adjacency queries in $\OhOf{\lg k}$ time, degree queries in
constant time, and the listing of neighbours in $\OhOf{\deg{v}}$ time.

A third strategy for succinct graph representations is based on graph
partitions.
Aleardi \etal~\cite{DBLP:conf/wads/AleardiDS05}
introduced a succinct representation for triangulations with a
boundary based on a hierarchical partition of the triangulation.
The combinatorial structure of the triangulation is partitioned into small
sub-triangulations which are further partitioned into tiny
sub-triangulations.
The representation stores the connectivity
information with asymptotically 2.175 bits per triangle, and supports
navigation between triangles in constant time.
The same authors presented an improved result of 1.62 bits per triangle for
triangulations without a boundary, and also demonstrated that
3-connected planar graphs can be represented with 2 bits per edge
\cite{DBLP:journals/tcs/AleardiDS08}.
For both planar triangulations and 3-connected planar graphs, this
representation permits constant-time navigation using additional
$\OhOf{N \lg \lg N \mathop{/} \lg N}$ bits of storage.
A partitioning approach was also
employed by Blandford \cite{DBLP:conf/soda/BlandfordBK03}, who showed
that graphs with small separators can be represented using $\OhOf{N}$
bits such that adjacency and degree queries can be answered in
constant time and listing the neighbours of a vertex takes
$\OhOf{\deg{v}}$ time.

Farzan and Munro \cite{DBLP:conf/esa/FarzanM08} considered
the case of succinct representations of arbitrary graphs.
They described an encoding that requires
$(1 + \epsilon) \lg \parensXL{\genfrac{}{}{0pt}{}{N^2}{M}}$ bits,
for an arbitrarily small constant $\epsilon > 0$,
and supports adjacency and degree queries in constant
time, and the listing of neighbours in $\OhOf{\deg{v}}$ time.

Little work has focused on obtaining succinct EM data structures.
The only results of which we are aware focus on text indexing 
\cite{DBLP:conf/dcc/ChienHSV08,  clark_96}, and on path traversals in 
trees \cite{DillabaughHM08}.

\subsection{Our Contributions}

In this chapter we present the following results:

\begin{enumerate}
\item In Section \ref{sec:graph_rep}, we present a data structure that
  uses $Nq + \OhOf{N} + \ohOf{Nq}$ bits to represent a planar graph with $N$
  vertices, each with a label of size $q$, and which allows the traversal of
  any path of length $K$ using $O(K / \lg B)$ I/Os.
  This path traversal cost matches that achieved by the data structure of
  Agarwal \etal~\cite{DBLP:conf/soda/AgarwalAMVV98}, but the latter
  uses $\ThetaOf{N \lg N} + 2Nq$ bits to store the graph.
  In the context of large datasets, this space saving represents a
  considerable improvement
  (e.g., with keys of constant size, the space bound improves by
  a factor of $\lg N$).
\item In Section \ref{sec:tins}, we apply our structure to store
  triangulations.
  If storing a point requires $\bitsPerPoint$ bits, we are able to store a triangulation in
  $N \bitsPerPoint + \OhOf{N} + \ohOf{N \bitsPerPoint}$ bits so that any path crossing $K$
  triangles can be traversed using $\OhOf{K / \lg B}$ I/Os.
  Again, the I/O efficiency of our structure matches that
  of~\cite{DBLP:conf/soda/AgarwalAMVV98} with a similar space improvement
  as for bounded-degree planar graphs.
\item In Section \ref{sec:point_location}, we show how to
  augment our triangulation representation with $\ohOf{N\bitsPerPoint}$ bits of extra
  information in order to support point location 
  queries using $\OhOf{\log_B N}$ I/Os.
  Asymptotically, this does not change the space requirements.
\item In Section \ref{sec:applications}, we describe several
  applications that make use of our representation for triangulations 
  from Section~\ref{sec:tins}.
  We demonstrate that reporting terrain profiles and trickle paths takes
  $\BigOh{K / \lg B}$ I/Os.
  We show that connected-component queries---that is, reporting a set of
  triangles that share a common attribute and induce a connected
  subgraph in the triangulation's dual---can be performed using
  $\BigOh{K / \lg B}$ I/Os when the component being reported is convex and consists
  of $K$ triangles.
  For non-convex regions with holes, we achieve a query bound of
  $\BigOh{K / \lg B + h \log_B h}$, where $h$ is the number
  of edges on the component's boundary.
  In order to achieve this query bound, the query procedure uses
  $\BigOh{h \cdot (\bitsPerKey + \lg h)}$ extra space.
  Without using any extra space, the same query can be answered using
  $\BigOh{K / \lg B + h' \log_B h'}$ I/Os, where $h'$ is the number of
  triangles that are incident on the boundary of the component.
\end{enumerate}

\section{Preliminaries}

\paragraph{Rank and select queries on bit vectors.}

The basic definitions for, and operations on, bit vectors have
been described in Chapter~\ref{chp:background}, 
Section~\ref{sec:bckgrnd-succint-ds}.
We repeat here a key lemma which summarizes results on
bit vectors used in the current Chapter.
Part (a) is from Jacobson~\cite{jac_1989} and Clark and 
Munro~\cite{clark_96}
while part (b) is from Raman~\etal~\cite{DBLP:journals/talg/RamanRS07}:

\begin{lemma}
  \label{lem:rank_select}
  A bit vector $S$ of length $N$ and containing $R$ $1$s can be
  represented using either (a) $N + \ohOf{N}$ bits or (b) $\lg \binom{N}{R}
  + \OhOf{N \lg \lg N / \lg N}$ bits to support the access to each bit, as
  well as $\rankopsym$ and $\selopsym$ operations, in $\OhOf{1}$ time (or
  $\OhOf{1}$ I/Os in external memory).\footnote{Note that $\log \binom{N}{R} +
    \OhOf{N \log \log N / \log N} = \ohOf{N}$ as long as $R = \ohOf{N}$.}
\end{lemma}

\paragraph{Planar graph partitions.}

Frederickson \cite{Frederickson87} introduced the notion of an
\emph{$r$-partition} of an $n$-vertex graph $G$, which is a collection of
$\ThetaOf{n / r}$ subgraphs of size at most $r$ such that for every
edge $xy$ of~$G$ there exists a subgraph that contains both $x$ and $y$.
A vertex $x$ is \emph{interior} to a subgraph if no other
subgraph in the partition contains $x$, otherwise $x$ is a
\emph{boundary vertex}.
Frederickson showed the following result for planar graphs.

\begin{lemma}[\cite{Frederickson87}]
  \label{lem:fred_graph_sep}
  Every $n$-vertex planar graph of bounded degree has an $r$-partition
  with $\OhOf{n / \sqrt{r}}$ boundary vertices.
\end{lemma}

\section{Notation}\label{sec_notation}

The following tables provide a summary of the various notations used in this chapter,
in particular in Section~\ref{sec:graph_rep}.

\begin{table}[H]
\centering
\begin{tabular}{l | l}
  Data Structure & Description \\ \hline
  $\fstofreg$ & A bit-vector of length $\dnumv$ the $k$'th bit of which \\
              & marks that $\vvec[k]$ as the first element of its region. \\
  $\fstofsubreg$ & A bit-vector of length $\dnumv$ the $k$'th bit of which marks that\\
              &  $\vvec[k]$ is the first element of its subregion. \\
  $\idxvec$ & Bit vector marking the number of subregion boundary vertices per \\
              & region. \\
  $\bdvec$ & Bit vector of length $n'$ which marks if vertex $\vvec[k]$ is a subregion \\
              &  boundary vertex. \\
  $\regvec$ & Vector each mapping subregion boundary vertex to its defining \\
              & subregion and label within that subregion. The length of this\\
			  & vector is the sum of the subregion boundary vertices for each \\
			  & region. \\
  $\subregvec$ & Vector recording the region label of all subregion boundary \\
			  & vertices within the sub-regions. The length of this vector is \\
			  & the sum of the subregion boundary vertices for each subregion \\
			  & (Due to duplicates, $|\subregvec| > |\regvec|$.). \\ \hline
\end{tabular}
\caption{Data structures used to convert between graph, region, and subregion labels}
\label{tab:ds_label_conv}
\end{table}

\begin{table}[H]
	\centering
		\begin{tabular}{ l | l}
			Notation & Description \\ \hline
			$\perm[i,j]$ & Order of vertices within a subregion. \\
			$B$ & Block size in number of elements. \\
			$\succblksize$ & Block in number of elements for a succinctly-encoded block. \\
			$G$ & A planar graph of bounded degree on N vertices. \\
			$\glbl{x}$ & Unique graph label of vertex $x$.\\
			$\reglbl[i]{x}$ & Unique region label of vertex $x$ in region $R_i$. \\
			$\subreglbl[i,j]{x}$ & Unique region label of vertex $x$ in sub-region $R_{i,j}$. \\
			$N_\alpha(x)$ & The $\alpha$-neighbourhood of vertex $x$. \\
			$q$ & Key size in bits. \\
			$q_i$ & The number of subregions in region $R_i$. \\
			$t$ & The number of regions in $G$. \\
			$\numv[i]$ & Denotes the number of vertices in region $\reg[i]$. \\
			$\dnumv[i]$ & Denotes the total number of vertices in all sub-regions of~$\reg[i]$ (This \\
			  & differs from $\numv[i]$ in that duplicates may be counted several times.).\\
			$\dnumv$ & Denotes the total number of vertices in all subregions of $G$.\\
			$\reg[i]$ & Region $i$, a subgraph of $G$ in the top level partition. \\
			$\subreg[i,j]$ & Subregion $i,j$, a subgraph of $G$, and subset of region $R_i$. \\
			$\numv[i,j]$ & Denotes the number of vertices in subregion $\subreg[i,j]$. \\
			$s_R$ & Upper bound on the number of bits required to encode a subregion \\
			& $\alpha$-neighbourhood. \\
			$s_S$ & Upper bound on the number of bits required to encode a region. \\		
			$S^{R}_\alpha$ & Set of region boundary vertices selected to build \\
			& $\alpha$-neigbhourhoods for (Section~\ref{sec:alt_block_scheme}). \\
			$S^{SR}_\alpha$ & Set of subregion boundary vertices selected to build \\
			& $\alpha$-neigbhourhoods for (Section~\ref{sec:alt_block_scheme}). \\ 
			$\vvec$ & Vector formed by concatenating all subregions of $G$.\\
			$\vvec[i]$ & Vector formed by concatenating all regions of $\reg[i]$.\\ 
			$\vvec[i,j]$ & Vector on all vertices of $\subreg[i,j]$.\\
			$\wsize$ & Word size in bits. \\ \hline
		\end{tabular}
	\caption{Notations used in describing graph representation}
	\label{tab:notation}
\end{table}

\begin{table}[H]
	\centering
		\begin{tabular}{ l | l}
			Data Structures & Description \\ \hline
			$\B_R$ & Bit vector of length $N$ marking region boundary vertices. \\
			$\B_S$ & Bit vector of length $N$ marking sub-region boundary vertices. \\
			$\S$ & Array storing all subregions, packed. \\
			$\F$ & Per subregion bit-vector marking if subregion is first sub-region in its region. \\
			$\D$ & Per subregion bit-vector marking if subregion starts in different block of $\S$ \\
			& than the preceeding subregion. \\
			$\N_S$ & Array storing $\alpha$-neighbourhoods of all sub-region boundary vertices. \\
			$\N_R$ & Array storing $\alpha$-neighbourhoods of all region boundary vertices. \\
			$\L_S$ & Vector recording the minimum graph label of interior vertices of each sub-region. \\
			$\L_R$ & Vector recording the minimum graph label of interior vertices of each region. \\
			$N_{i,j}$ & The number of vertices stored in subregion $i,j$ (part of  \\
			 & subregion encoding).\\
			$\G_{i,j}$ & A succinct encoding of the graph structure of subregion $R_{i,j}$ (part of  \\
			 & sub-region encoding).\\
			$\B_{i,j}$ & A bit-vector marking vertices in a subregion as either region or subregion \\
			& boundaries (part of sub-region encoding).\\
			$\K_{i,j}$ & An array storing the $q$ bit-keys for each vertex within a subregion \\
			 & (part of sub-region encoding). \\
			$\G_x$ & A succinct encoding of the graph structure of $N_\alpha(x)$ \\
			 & (part of $\alpha$-neighbourhood encoding).\\
			$\T_x$ & A bit vector marking the terminal vertices in an  $\alpha$-neighbourhood. \\
			 & (part of $\alpha$-neighbourhood encoding). \\
			$p_x$ & Position of $x$ in $\pi_x$ (permutation of vertices in $N_\alpha(x)$\\
			 & (part of $\alpha$-neighbourhood encoding) \\
			$\K_x$ & An array storing keys associated with vertices in  $N_\alpha(x)$ . \\
			 & (part of $\alpha$-neighbourhood encoding) \\		
			$\L_x$ & An array storing labels of the vertices in $N_\alpha(x)$. \\
			& (part of $\alpha$-neighbourhood encoding) \\ 
			$\regptable$ & $\alpha$-neighbourhood pointer table for region boundary \\
			& vertices. \\
			$\subregptable$ & $\alpha$-neighbourhood pointer table for subregion boundary\\
			& vertices. \\ \hline
		\end{tabular}
	\caption{Data structures used to represent graph components}
	\label{tab:data_structs}
\end{table}

\section{Succinct Representation of Bounded-Degree Planar Graphs}
\label{sec:graph_rep}

Our main result considers planar graphs of degree $d = \OhOf{1}$
where the degree of a graph is the maximum degree of its vertices.
Let $G$ be such a graph.
Each vertex $x$ in $G$ stores a $\bitsPerKey$-bit key
to be reported when visiting $x$.
We assume that $\bitsPerKey = \OhOf{\lg N}$, but
$\bitsPerKey$ is not necessarily related to the size of the graph.
It may take on a small constant value, e.g., the labels may be colours
chosen from a constant-size set.

Our data structure is based on a two-level partition of $G$, similar to
that used in \cite{DBLP:journals/talg/BoseCHMM12}.
More precisely, we compute an
$\regsz$-partition of $G$, for some parameter $\regsz \ge B \lg^2 N$ to be
defined later.
We refer to the subgraphs in this partition as \emph{regions}, and
denote them by $\reg[1], \reg[2], \dots, \reg[t]$.
Each region $\reg[i]$ is further
divided into smaller subgraphs, called \emph{subregions}, that form
an $\subregsz$-partition of $\reg[i]$, for some parameter
$B \le \subregsz < \regsz$.
We use $q_i$ to denote the number of subregions of $\reg[i]$, and $\subreg[i,1],
\subreg[i,2], \dots, \subreg[i,q_i]$ to denote the subregions themselves.
According to the definitions from the previous section, we classify
every vertex in a region $\reg[i]$ as a \emph{region-interior} or
\emph{region boundary vertex}, and every vertex in a subregion
$\subreg[i,j]$ as a \emph{subregion-interior} or \emph{subregion boundary
  vertex}.
Every region boundary vertex is also considered to be a
subregion boundary vertex in any subregion that contains it
(see Fig.~\ref{fig:partitioned_graph} for an illustration).
Note that if a vertex is a region boundary vertex of one region, it is a region
boundary vertex for all regions that contain it.
The same is true for subregion boundary vertices.
Thus, we can refer to a vertex as a region or subregion boundary or interior
vertex without reference to a specific region or subregion.

\begin{figure}[t]
  \centering
  \includegraphics[width=0.6\textwidth]{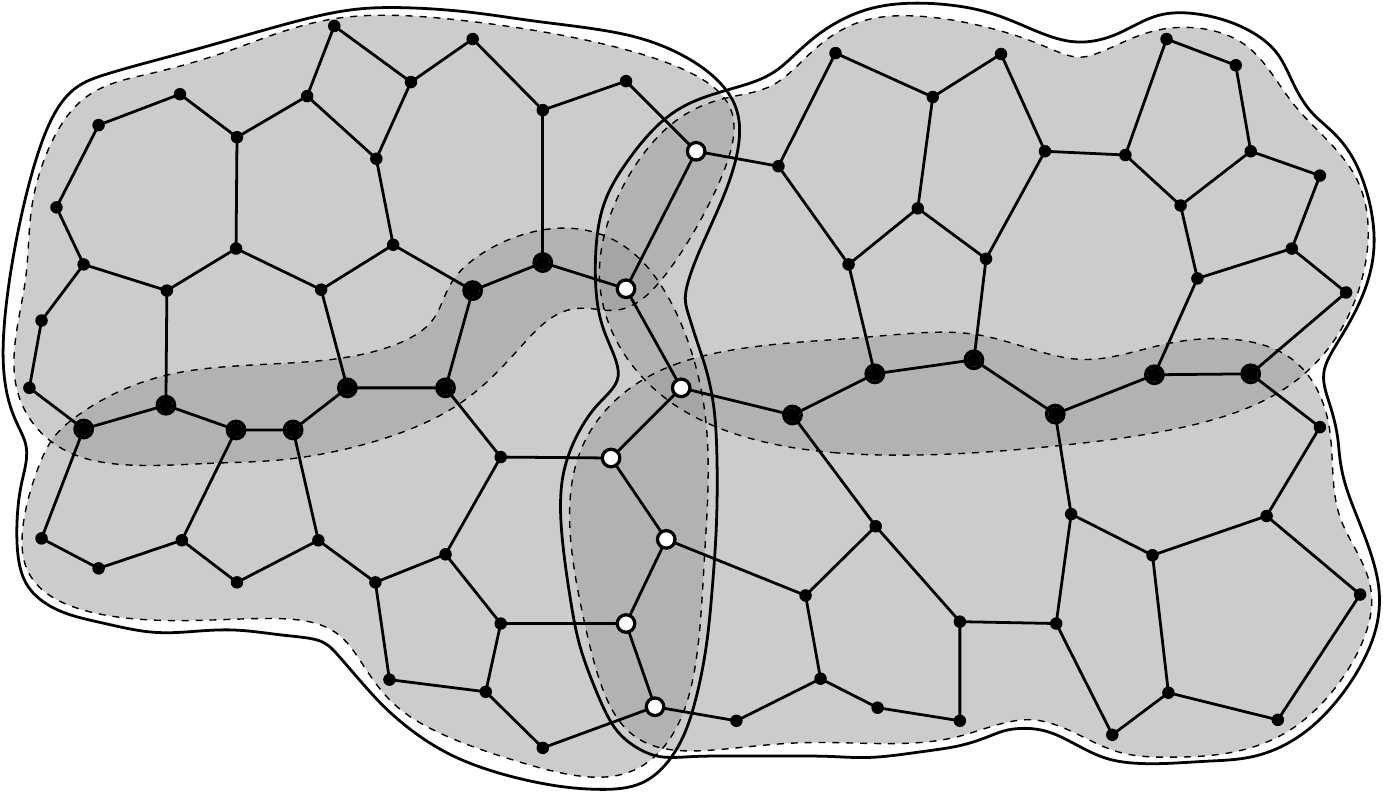}
  \caption[Recursively-partitioned graph]{A graph $G$ partitioned into two 
	regions (delineated by solid lines), which are further partitioned into two 
	subregions each (delineated by dashed lines).
    Region boundary vertices are shown as large hollow disks while subregion 
	boundary vertices interior to their regions are shown as large solid disks.
    Vertices interior to their subregions are shown as small solid disks.}
  \label{fig:partitioned_graph}
\end{figure}

In our data structure, all subregions are stored explicitly while
regions are represented implicitly as unions of their subregions.
In addition, we store the \emph{$\alpha$-neighbourhoods} of all (region
and subregion) boundary vertices.
These neighbourhoods are defined
as follows~\cite{DBLP:conf/soda/AgarwalAMVV98}.
For a vertex $x$, consider a breadth-first search in $G$ starting at $x$.
Then the $\alpha$-neighbourhood $\nb[\alpha]{x}$ of $x$ is the subgraph of $G$
induced by the first $\alpha$ vertices visited by the search.
For a region boundary vertex, we store its entire $\alpha$-neighbourhood.
For a subregion boundary vertex interior to a region $\reg[i]$, we store
only those vertices of its $\alpha$-neighbourhood that belong to
$\reg[i]$.
A vertex $y \in \nb[\alpha]{x}$ is an \emph{interior vertex} of
$N_\alpha(x)$ if all its neighbours belong to $\nb[\alpha]{x}$;
otherwise $y$ is \emph{terminal}
(see Fig.~\ref{fig:alpha_neighbourhoods} for an illustration of the
definitions pertaining to $\alpha$-neighbourhoods).

\begin{figure}[t]
  \centering
  \includegraphics{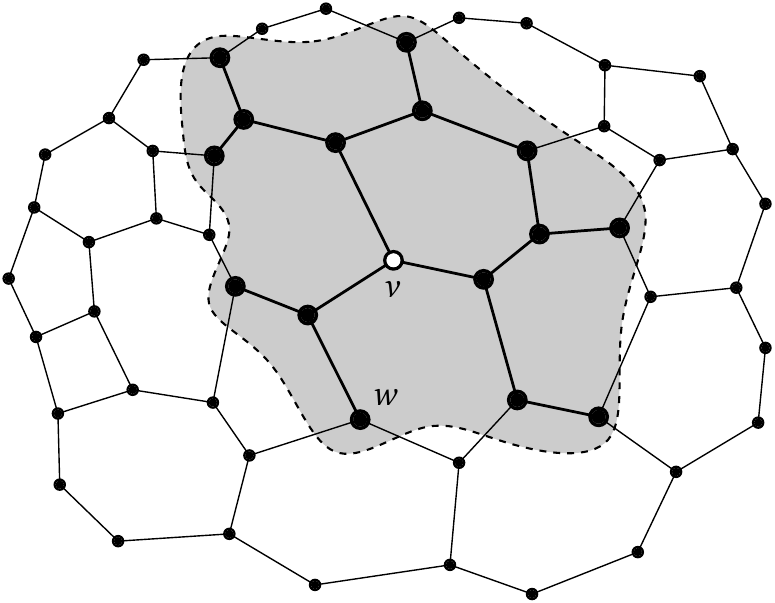}
  \caption[$\alpha$-neighbourhood of a boundary vertex]{The 
	$\alpha$-neighbourhood $\nb[\alpha]{v}$ of a vertex $v$,
    for $\alpha = 16$.
    Vertex $v$ is interior to $\nb[\alpha]{v}$ while vertex $w$ is terminal.}
  \label{fig:alpha_neighbourhoods}
\end{figure}

We refer to subregions and $\alpha$-neighbourhoods---that is, to the
subgraphs of $G$ we store explicitly---as \emph{components} of $G$.
We store each component succinctly so as to allow for the efficient
traversal of edges in the component.
The traversal of a path in $G$ must be able to visit different components.
In order to facilitate this, we assign a unique \emph{graph label} to each vertex of $G$,
and provide mechanisms to (1) identify a component containing a vertex $x$ and
locate $x$ in that component, given its graph label, and (2) determine
the graph label of any vertex in a given component.

The rest of this section is organized as follows.
In Section~\ref{sec:graph_labelling}, we define a number of labels
assigned to the vertices of $G$, including the graph labels just
mentioned, and provide $\ohOf{N}$-bit data structures that allow us to
convert any such label into any other label using $\OhOf{1}$ I/Os.
In Section~\ref{sec:datastructs}, we discuss the succinct
representations we use to store subregions and
$\alpha$-neighbourhoods.
In Section~\ref{sec:navigation},
we discuss how to traverse paths I/O efficiently using this
representation.
In Section \ref{sec:alt_block_scheme}, we describe an 
alternative scheme
for blocking planar graphs that can improve I/O efficiency in 
some instances.

\subsection{Graph Labeling}

\label{sec:graph_labelling}

In this section, we describe the labeling scheme that underlies our
data structure.
Our scheme is based on the one used in
Bose~\etal~\cite{DBLP:journals/talg/BoseCHMM12}.
It assigns three labels to each vertex.
As already mentioned, the \emph{graph label} $\glbl{x}$ identifies each vertex
$x \in G$ uniquely.
The \emph{region label} $\reglbl[i]{x}$ of a vertex $x$ in a region $\reg[i]$
identifies $x$ uniquely among the vertices in $\reg[i]$, and the
\emph{subregion label} $\subreglbl[i,j]{x}$ of a vertex $x$ in a subregion
$\subreg[i,j]$ identifies $x$ uniquely
among the vertices in $\subreg[i,j]$.
(Note that a region boundary vertex appears in more than one region and,
thus, receives one region label per region that contains it.
Similarly, every region or subregion boundary vertex receives multiple
subregion labels.)

A standard data structure would identify vertices by their graph labels, and
store these labels for every vertex.
Since there are $N$ vertices, every graph label must use at least
$\lg N$ bits, and such a representation uses at least $N \lg N$ bits of space.
In our structure, we store graph labels for only a small subset of the vertices, 
region labels for yet another subset of the vertices, and subregion labels for 
a third subset. 
Many vertices in the graph have no label explicitly stored.
Since regions and subregions are small, and region and subregion labels have to
be unique only within a region or subregion, they can be stored using
less than $\lg N$ bits each.
Next we define these labels.

In the data structure described in the next section, the vertices of
each subregion $\subreg[i,j]$ are stored in a particular order $\perm[i,j]$.
We use the position of a vertex $x$ in $\perm[i,j]$ as its
subregion label $\subreglbl[i,j]{x}$ for subregion~$\subreg[i,j]$.

For a region $\reg[i]$ and a vertex $x \in \reg[i]$, we say an occurrence of
$x$ in a subregion $\subreg[i,j]$ \emph{defines} $x$ if there is no
subregion $\subreg[i,h]$ with $h < j$ that contains~$x$; otherwise the
occurrence is a \emph{duplicate}.
If the occurrence of $x$ in a subregion $\subreg[i,j]$ defines $x$, we also call
$\reg[i,j]$ the \emph{defining subregion} of $x$.
Clearly, all occurrences of subregion interior vertices are
defining.
We order\footnote{This ordering, as with the ordering used for assigning 
graph labels, is employed strictly for labeling purposes. It does not imply
any arrangement of vertices within the structures used to represent the graph.}
the vertices of $\reg[i]$ so that all subregion boundary vertices
precede all subregion-interior vertices, and the vertices in each of these two
groups are sorted primarily by their defining subregions and secondarily by
their subregion labels within these subregions.
The region labels of the vertices in $\reg[i]$ are now obtained by numbering the
vertices in this order.

Graph labels are defined analogously to region labels.
Again, we say an occurrence of a vertex $x$ in a region $\reg[i]$ defines
$x$ if $x \not\in \reg[h]$, for all $h < i$.
We then order the vertices of $G$ so that the region boundary vertices precede
the region-interior vertices, and the vertices in each of these two groups are
sorted primarily by their defining regions and secondarily by their region
labels inside their defining regions.
The graph labels of the vertices in $G$ are
then obtained by numbering the vertices in order.

Observe that this labeling scheme ensures that all vertices interior
to a region $\reg[i]$ receive consecutive graph labels, and all vertices
interior to a subregion $\subreg[i,j]$ receive consecutive graph and
region labels.
The labeling of subregion vertices at the region level is depicted in 
Figure~\ref{fig:region_labeling}.
Generating graph labels is an analogous process.

\begin{figure}[t]
  \centering
  \includegraphics[width=0.8\textwidth]{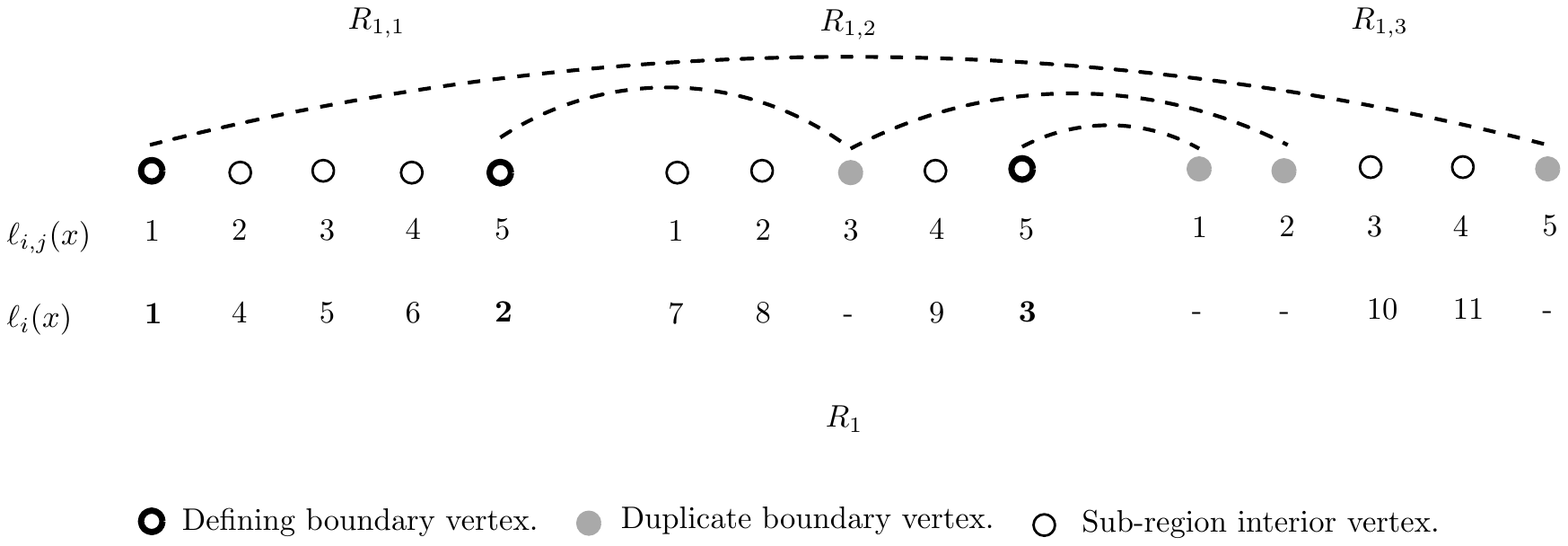}
  \caption[Assigning region labels]{
    The generation of region labels for region $\reg[1]$ is depicted,
    where $\reg[1]$ is composed of subregions $\subreg[1,1]$, $\subreg[1,2]$,
    and $\subreg[1,3]$.
    The dashed lines along the top indicate equivalence between vertices.
    Each subregion has five vertices, assigned subregion labels 1 through
    5.
    Assigned region labels are indicated below subregion labels. 
    Region labels of boundary vertices are bold while the dash indicates
    a vertex is a duplicate, and therefore not assigned an additional
    region label.
  }
  \label{fig:region_labeling}
\end{figure}

Our data structure for representing the graph $G$ will make use of
the following result, which extends a result from \cite{DBLP:journals/talg/BoseCHMM12}.
In particular, parts~\ref{item:subregion-to-region}
and~\ref{item:region-to-graph} were proved in \cite{DBLP:journals/talg/BoseCHMM12} for a slightly
different vertex labeling than we use here while
parts~\ref{item:graph-to-region} and~\ref{item:region-to-subregion} are new.
For completeness, we prove all parts of the lemma here.

\begin{lemma}
  \label{lem:label_ops}
  There is a data structure of $\ohOf{N}$ bits that allows the following
  conversion between vertex labels of $G$ to be performed using
  $\OhOf{1}$ I/Os.
  \begin{enumerate} 
  \item Given the graph label $\glbl{x}$ of a vertex $x$, compute the defining
    region $\reg[i]$ of $x$ and the region label $\reglbl[i]{x}$ of $x$ in
    $\reg[i]$.\label{item:graph-to-region}
  \item Given the region label $\reglbl[i]{x}$ of a vertex $x$ in a region
    $\reg[i]$, compute the defining subregion $\subreg[i,j]$ of $x$ and the
    subregion label $\subreglbl[i,j]{x}$ of $x$ in
    $\subreg[i,j]$.\label{item:region-to-subregion}
  \item Given the subregion label of a vertex $x$ in a subregion
    $\subreg[i,j]$ of $\reg[i]$, compute the region label $\reglbl[i]{x}$
    of $x$ in $\reg[i]$.\label{item:subregion-to-region}
  \item Given the region label of a vertex $x$ in a region $\reg[i]$,
    compute the graph label $\glbl{x}$ of $x$.\label{item:region-to-graph}
  \end{enumerate}
\end{lemma}

\begin{proof}
  We construct separate $\ohOf{N}$-bit data structures in order to accomplish parts
  \ref{item:region-to-subregion} and~\ref{item:subregion-to-region}, and to
  accomplish parts \ref{item:graph-to-region} and~\ref{item:region-to-graph}.
  Since these data structures are very similar, we discuss those for
  parts \ref{item:region-to-subregion} and~\ref{item:subregion-to-region} in detail,
  and only provide a space analysis for the second structure.
  We use $\numv[i]$ to denote the number of vertices in region $\reg[i]$,
  $\numv[i,j]$ to denote the number of vertices in subregion $\subreg[i,j]$,
  $\dnumv[i] := \sum_{j=1}^{q_i} \numv[i,j]$ to denote the total number of
  vertices in all subregions of~$\reg[i]$, where $q_i$ is the number of
  subregions of region $\reg[i]$, and $\dnumv := \sum_{i=1}^t \dnumv[i]$ to
  denote the total number of vertices in all subregions.

  Our data structure for converting between region and subregion labels consists
  of a number of vectors.
  We introduce them as required for the operations we discuss.
  Most of these vectors are bit vectors representing a vector $\vvec$
  of the vertices in all subregions of $G$.
  We define this vector first.
  $\vvec$ is the concatenation of $t$ subvectors $\vvec[1], \vvec[2],
  \dots, \vvec[t]$ storing the vertices in the regions of $G$.
  Each such vector $\vvec[i]$ is itself the concatenation of
  subvectors $\vvec[i,1], \vvec[i,2], \dots, \vvec[i,q_i]$ representing
  the subregions of $\reg[i]$.
  Each such vector $\vvec[i,j]$ stores the vertices in $\subreg[i,j]$
  sorted by their subregion labels.
  Note that our data structure does not store $\vvec$.
  We only introduce it here as a reference for discussing the vectors our
  data structure does store.
  Throughout this proof, when we define a bit vector, we store it using
  the second part of Lemma~\ref{lem:rank_select} to support constant-I/O
  accesses to its elements, as well as constant-I/O $\rankopsym$ and $\selopsym$
  operations.

\begin{table*}[ht]
	\centering
		\begin{tabular}{ l | l}
			Operator & Description \\ \hline
			$\rankop[x]{S, i}$ & Returns the number of $x$'s in bit-vector $S$ up to
				position $S[1..i]$, where $x \in \{0,1\}$.\\
			$\selop[x]{S, r}$ & Returns the $r$\textsuperscript{th} $x$ in bit-vector $S$
				where $x \in \{0,1\}$.\\
			$\regpos{i}$ & Returns the first element of region $\reg[i]$ in $\vvec$. \\
			$\subregpos{i,j}$ & Returns the first element of sub-region $\subreg[i,j]$
				in $\vvec$. \\
			$\possubreg{i,p}$ & Returns the sub-region index $j$ from $\vvec$ where
				$p \in \reg[i]$. \\
			$\bdcount{i}$ & Calculates in $\vvec$ the number of sub-region boundary vertices
				in $\reg[i]$. \\ \hline
		\end{tabular}
	\caption{Summary of bit-vector operations used in this section. }
	\label{tab:succ_operators}
\end{table*}

  \paragraph{Locating subvectors.}

  The first two vectors that our data structure stores are bit vectors
  $\fstofreg$ and $\fstofsubreg$ that
  allow us to identify, for a region $\reg[i]$, the index of the first
  element of $\vvec[i]$ in $\vvec$ and, for a subregion $\subreg[i,j]$,
  the index of the first element of $\vvec[i,j]$ in $\vvec$.
  Both vectors have length $\dnumv = \size{\vvec}$.
  The $k$th bit in $\fstofreg$ is $1$ if and only if $\vvec[][k]$ is the
  first element of $\vvec[i]$, for some region $\reg[i]$,
  and the $k$th bit of $\fstofsubreg$ is $1$
  if and only if $\vvec[][k]$ is the first element of $\vvec[i,j]$, for some
  subregion $\subreg[i,j]$.
  Using these two bit vectors, we can implement the following three operations
  using $\OhOf{1}$ I/Os.
  \begin{itemize}
  \item $\regpos{i}$ returns the index of the first element of $\vvec[i]$
    in $\vvec$ and is implemented as $\regpos{i} := \selop[1]{\fstofreg,i}$.
  \item $\subregpos{i,j}$ returns the index of the first element of $\vvec[i,j]$
    in $\vvec$ and is implemented as
    $\subregpos{i,j} := \selop[1]{\fstofsubreg,\rankop[1]{\fstofsubreg,\regpos{i}-1}+j}$.
  \item $\possubreg{i,p}$ returns the index $j$ of the subregion $\subreg[i,j]$
    of $\reg[i]$ such that $\vvec[][p]$ is an element of  $\vvec[i,j]$,
    provided $\vvec[][p]$ is an element of $\vvec[i]$.
    This operation can be implemented as
    $\possubreg{i,p} := \rankop[1]{\fstofsubreg,p} - \rankop[1]{\fstofsubreg,
      \regpos{i} - 1}$.
  \end{itemize}

  \paragraph{Identifying subregion boundary vertices.}

  The next two vectors we store, $\idxvec$ and $\bdvec$, help us decide whether
  a vertex is a subregion boundary vertex based only on its region or subregion
  label.
  This is useful because the conversion between region and subregion
  labels is done differently for subregion boundary and subregion-interior
  vertices.
  Vector $\idxvec$ is the concatenation of $t$ bit vectors
  $\idxvec[1], \idxvec[2], \dots, \idxvec[t+1]$.
  For $1 \le i \le t$, the subvector $\idxvec[i]$ has length equal to the
  number of subregion boundary vertices in $\reg[i]$.
  Its first bit is $1$, all other bits are $0$.
  The last subvector $\idxvec[t+1]$ stores a single $1$ bit.
  Vector $\bdvec$ is a bit vector of length $n'$.
  $\bdvec[][k] = 1$ if and only if $\vvec[][k]$ is a subregion boundary
  vertex.

  Now observe that a vertex $x$ with region label $\reglbl[i]{x}$ in region
  $\reg[i]$ is a subregion boundary vertex if and only if
  $\reglbl[i]{x}$ is no greater than the number of subregion boundary vertices
  in $\reg[i]$ because, in the ordering used to define
  the region labels of the vertices in $\reg[i]$, all subregion boundary
  vertices precede all subregion-interior vertices.
  In order to test this condition, we need to compute the number of subregion
  boundary vertices in $\reg[i]$, which we can do using $\OhOf{1}$ I/Os
  as $\bdcount{i} := \selop[1]{\idxvec,i+1} - \selop[1]{\idxvec,i}$.
 
  For a vertex $x$ with subregion label $\subreglbl[i,j]{x}$ in subregion
  $\subreg[i,j]$, observe that the subvector $\vvec[i,j]$ of $\vvec$
  stores the vertices of $\subreg[i,j]$ sorted by their subregion labels.
  Thus, $x$ appears at position $\subreglbl[i,j]{x}$ in $\vvec[i,j]$,
  which is position $k := \subregpos{i,j} + \subreglbl[i,j]{x} - 1$ in $\vvec$.
  The corresponding bit $\bdvec[][k]$ in $\bdvec$ is $1$ if and only if
  $x$ is a subregion boundary vertex.
  Thus, testing whether $x$ is a subregion boundary vertex requires calculating
  $\subregpos{i,j}$, which takes $\OhOf{1}$ I/Os followed by a single lookup
  in the bit vector $\bdvec$.

  \paragraph{Subregion-interior vertices.}

  For the set of subregion-interior vertices, the vectors we have already
  defined suffice to convert between region
  and subregion labels.

  Given a subregion-interior vertex $x$ in $\reg[i]$ and its region label
  $\reglbl[i]{x}$ in $\reg[i]$, observe that the subregion-interior vertices in
  $\reg[i]$ appear in $\vvec[i]$ sorted by their region labels, and
  that these are exactly the vertices in $\vvec[i]$ whose corresponding
  bits in $\bdvec$ are $0$.
  Thus, since subregion-interior vertices have greater region labels than
  subregion boundary vertices, we can find the position $k$ of the only
  occurrence of $x$ in $\vvec$ as
  $k := \selop[0]{\bdvec, \rankop[0]{\bdvec, \regpos{i} - 1} + \reglbl[i]{x} 
  - \bdcount{i}}$, which takes $\OhOf{1}$ I/Os to compute.
  Vertex $x$ belongs to the subregion $\subreg[i,j]$ such that
  $\vvec[][k]$ belongs to $\vvec[i,j]$.
  The index $j$ of this subregion is easily computed using $\OhOf{1}$ I/Os, as
  $j := \rankop[1]{\fstofsubreg,k} - \rankop[1]{\fstofsubreg,\regpos{i}-1}$.
  Since vertices in $\vvec[i,j]$ are sorted by their subregion labels in
  $\subreg[i,j]$, the subregion label of $x$ is the index of $\vvec[][k]$
  in $\vvec[i,j]$, which can be computed using $\OhOf{1}$ I/Os as
  $\subreglbl[i,j]{x} := k - \subregpos{i,j} + 1$.

  To compute the region label in $\reg[i]$ for a subregion-interior
  vertex $x$ in $\subreg[i,j]$, we perform this conversion in the reverse
  direction.
  Firstly, we find the position $k$ of $x$ in $\vvec$ as
  $k := \subregpos{i,j} + \subreglbl[i,j]{x} - 1$.
  Secondly, we compute the region label of $x$ in $\reg[i]$ as
  $\reglbl[i]{x} := \rankop[0]{\bdvec,k} - \rankop[0]{\bdvec,\regpos{i}-1} + \bdcount{i}$,
  which is correct by the argument from the previous paragraph.
  Thus, the conversion from subregion labels to region labels also takes
  $\OhOf{1}$ I/Os.

  \paragraph{Subregion boundary vertices.}

  For subregion boundary vertices in $\reg[i]$, the vectors we have defined so
  far cannot be used to convert between region and subregion labels, mainly
  because these vertices by definition appear more than once in $\vvec[i]$.
  On the other hand, there are only a few such vertices, which allows us to
  store their region and subregion labels explicitly.
  In particular, we store two vectors $\regvec$ and~$\subregvec$.
  Vector $\regvec$ is the concatenation of subvectors $\regvec[1], \regvec[2],
  \dots, \regvec[t]$, one per region $\reg[i]$ of $G$.
  The length of $\regvec[i]$ is the number of subregion boundary vertices
  in $\reg[i]$, and $\regvec[i][k]$ is the pair $(j,k')$ such that the
  subregion $\subreg[i,j]$ is the defining subregion of the vertex $x$
  with region label $k$, and $k'$ is $x$'s subregion label in $\subreg[i,j]$.
  We represent the index $j$ using $\lg \regsz$ bits, and the subregion
  label $k'$ using $\lg \subregsz$ bits.
  The length of vector $\subregvec$ equals the number of $1$s in vector
  $\bdvec$.
  $\subregvec[][k]$ is the region label of the vertex in $\vvec$ corresponding
  to the $k$th $1$ in $\bdvec$.
  We store each such region label in $\subregvec$ using $\lg \regsz$ bits.

  Now consider a subregion boundary vertex $x$ with region label
  $\reglbl[i]{x}$ in $\reg[i]$.
  The index of the entry in $\regvec$ storing the index of the defining
  subregion of $x$ and $x$'s subregion label in this subregion is 
  $k := \selop[1]{\idxvec,i} + \reglbl[i]{x} - 1$, which takes
  $\OhOf{1}$ I/Os to compute.
  Another I/O suffices to retrieve $\regvec[][k]$.

  Given a subregion boundary vertex $x$ with subregion label
  $\subreglbl[i,j]{x}$ in subregion $\subreg[i,j]$, the index of $x$ in
  $\vvec$ is $k := \subregpos{i,j} + \subreglbl[i,j]{x} - 1$.
  The corresponding bit $\bdvec[][k]$ in $\bdvec$ is $1$, and the
  region label of $x$ can be retrieved from
  $\subregvec[][\rankop[1]{\bdvec,k}]$, which takes $\OhOf{1}$ I/Os.

  \paragraph{Space bound.}

  The vectors $\fstofreg$ and $\fstofsubreg$ have length
  $n' = N + \OhOf{N / \sqrt{\subregsz}} = \OhOf{N}$ because their length
  equals that of $\vvec$; only subregion boundary vertices appear more
  than once in $\vvec$, and the total number of occurrences of subregion
  boundary vertices in $\vvec$ is $\OhOf{N / \sqrt{\subregsz}}$ by
  Lemma~\ref{lem:fred_graph_sep}.
  Every $1$-bit in such a vector corresponds to the start of a region
  or subregion.
  Thus, both vectors contain $\OhOf{N / \subregsz} = \ohOf{N}$ $1$s.
  By Lemma~\ref{lem:rank_select}(b), this implies that both vectors can
  be represented using $\ohOf{N}$ bits.

  The vector $\idxvec$ has length $\OhOf{N / \sqrt{\subregsz}} = \ohOf{N}$
  because every bit in $\idxvec$ corresponds to a unique subregion boundary
  vertex.
  Thus, by Lemma~\ref{lem:rank_select}(a), we can store $\idxvec$ using
  $\ohOf{N}$ bits.

  The vector $\bdvec$ has length $n' = \OhOf{N}$ and contains
  $\OhOf{N / \sqrt{\subregsz}}$ $1$s, one per occurrence of a subregion boundary
  vertex in a subregion.
  Hence, by Lemma~\ref{lem:rank_select}(b), we can store $\bdvec$ using
  $\ohOf{N}$ bits.

  Finally, the number of entries in vectors $\regvec$ and $\subregvec$ is
  bounded by the number of occurrences of subregion boundary vertices in
  subregions, which is $\OhOf{N / \sqrt{\subregsz}}$.
  Each entry in such a vector is represented using
  $\OhOf{\lg \regsz} = \OhOf{\lg\subregsz + \lg \lg N}$ bits.
  Thus, the total space used by these two vectors is
  $\OhOf{(N / \sqrt{\subregsz})(\lg \subregsz + \lg \lg N)}$,
  which is $\ohOf{N}$ because $\subregsz = B = \OmegaOf{\lg N}$.
  To summarize, the vectors we use to convert between region and subregion
  labels use $\ohOf{N}$ space and allow us to convert between these labels
  using $\OhOf{1}$ I/Os.
  This proves parts (b) and (c) of the lemma.

  \paragraph{Conversion between region and graph labels.}

  The data structure for conversion between region and graph labels is
  identical, with the following modifications.
  \begin{itemize}
  \item The vector $\vvec$ is the concatenation of $t$ subvectors
    $\vvec[1], \vvec[2], \dots, \vvec[t]$.
    Vector $\vvec[i]$ stores the vertices in region $\reg[i]$ sorted
    by their region labels.
  \item The $k$th bit in $\bdvec$ is one if and only if the vertex
    $\vvec[k]$ is a region boundary vertex.
  \item There is no need for vectors $\fstofsubreg$ and $\idxvec$.
    Instead, we store a single $(\lg N)$-bit integer counting the number
    of region boundary vertices of $G$.
  \item The entries in $\regvec$ are pairs of indices of defining
    regions and region labels in these regions, and the entries in
    $\subregvec$ are graph labels.
  \end{itemize}
  Using these modified vectors, the conversion between graph and region labels
  can be accomplished using straightforward modifications of the procedures
  described in this proof.

  Vectors $\fstofreg$ and $\bdvec$ have length
  $N + \OhOf{N / \sqrt{\regsz}} = \OhOf{N}$ and contain $\OhOf{N / \regsz}$ and
  $\OhOf{N / \sqrt{\regsz}}$ $1$s, respectively, which are both $\ohOf{N}$.
  Hence, by Lemma~\ref{lem:rank_select}(b), they can both be stored using
  $\ohOf{N}$ bits.
  The number of region boundary vertices is stored using $\lg N = \ohOf{N}$
  bits.
  Each entry in vectors $\regvec$ and $\subregvec$ uses $\OhOf{\lg N}$ bits
  of space, and the number of these entries is $\OhOf{N / \sqrt{\regsz}}$.
  Hence, the total size of these vectors is $\OhOf{(N / \sqrt{\regsz}) \lg N}$,
  which is $\ohOf{N}$ because $\regsz \ge B \lg^2 N = \OmegaOf{\lg^3 N}$.
  This proves parts (a) and (d) of the lemma.
\end{proof}

\subsection{Data Structures}\label{sec:datastructs}

Our representation of $G$ consists of three parts: (1) succinct
representations of the subregions, packed into an array $\ds$, (2)
succinct representations of the $\alpha$-neighbourhoods of boundary
vertices, packed into two arrays $\regnbvec$ and $\subregnbvec$, and (3) two
vectors $\regminvec$ and $\subregminvec$ recording the minimum graph labels of
the interior vertices in each region or subregion.
Throughout this section, we use total orders $\reglt$ and $\subreglt$ of the
regions and subregions of~$G$, respectively.
For two regions $\reg[i]$ and $\reg[j]$, we define $\reg[i] \reglt \reg[j]$ if
$i < j$.
For two subregions $\subreg[i,j]$ and $\subreg[i',j']$, we define
$\subreg[i,j] \subreglt \subreg[i',j']$ if $i < i'$ or $i = i'$ and $j < j'$.

\paragraph{Minimum graph labels.}

The first part of our data structure consists of two arrays $\regminvec$ and
$\subregminvec$.
The $k$\textsuperscript{th} position of $\regminvec$ stores the minimum graph
label of the interior vertices in region $\reg[k]$.
The $k$\textsuperscript{th} position in $\subregminvec$ stores the minimum graph label of the
interior vertices in the subregion $\subreg[i,j]$ at position $k$ in the
sequence of subregions defined by $\subreglt$.

\paragraph{Succinct encoding of subregions.}

We represent each subregion $\subreg[i,j]$ using four data structures:
\begin{enumerate}
\item Its number of vertices $\numv[i,j]$ stored in $\lg N$
  bits.\footnote{$\lg \subregsz$ bits would suffice; however it simplifies the
    analysis to use $\lg N$ bits.}
\item The graph structure of $\subreg[i,j]$ encoded succintly as
	$\graphrep[i,j]$, using the encoding of 
	\cite{DBLP:journals/siamcomp/ChiangLL05}.
  This encoding involves a permutation $\perm[i,j]$ of
  the vertices in $\subreg[i,j]$.
\item A bit vector $\bdvec[i,j]$ of length $\size{\subreg[i,j]}$ with
  $\bdvec[i,j][i] = 1$ if and only if the corresponding vertex in $\perm[i,j]$
  is a  (region or subregion) boundary vertex.
\item An array $\keyvec[i,j]$ of length $\size{\subreg[i,j]}$ that stores the
  $\bitsPerKey$-bit key for each vertex.
  These labels are stored according to the permutation $\perm[i,j]$.
\end{enumerate}
The representation of $\subreg[i,j]$ is the concatenation of $\numv[i,j]$,
$\graphrep[i,j]$, $\bdvec[i,j]$, and $\keyvec[i,j]$.

Let $\succblksize$ be the maximum integer such that a subregion with $\succblksize$
vertices can be encoded in this manner. 
We choose the parameters of our two-level partition as $\regsz := \succblksize \log^3
N$ and $\subregsz := \succblksize$.
This ensures that every subregion fits in a disk block.
However, a given subregion may be much smaller, and storing
each subregion in its own disk block may be wasteful.
Instead, we pack the representations of all subregions into an array $\ds$,
without gaps, using a construction proposed 
in~\cite{DBLP:journals/algorithmica/DillabaughHM12} as follows.
The subregion representations in $\ds$ are ordered according to $\subreglt$
and the array $\ds$ is stored in consecutive disk blocks; the first
$\lg N$ bits of each block are reserved for a \emph{block offset}
value.
Note that the representation of each region $\subreg[i,j]$ is
distributed over at most two consecutive blocks in $\ds$.
If $\subreg[i,j]$ is distributed over two blocks $B_1$ and $B_2$,
the block offset of the second block, $B_2$, records the number of bits
of $\subreg[i,j]$ stored in $B_2$.
Any other block stores a block offset of~$0$.
In other words, the block offset of each block stores the starting
position of the first region $\subreg[i,j]$ in this block that is not
stored partially in the previous block.
Since we use $\lg N$ bits per block to store the block offset, each subregion
of size $\succblksize$ must be stored in at most $B \wsize - \lg N$ bits.
This representation of $\ds$ is illustrated in Figure~\ref{fig:block-packing}.

\begin{figure}[t]
  \centering
  \includegraphics{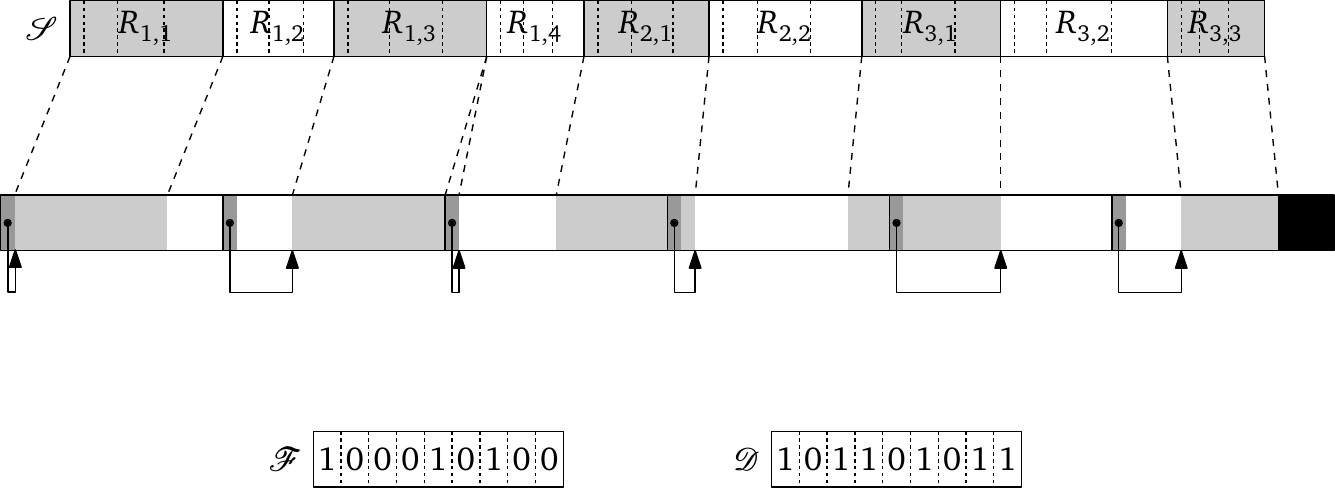}
  \caption[Packing graph blocks in memory]{Arrays $\ds$, $\fst$, and $\diff$, 
	and the packing of $\ds$ into blocks.
    The dark portion at the beginning of each block stores the block offset.
    The black portion of the last block is unused.}
  \label{fig:block-packing}
\end{figure}

In order to facilitate the efficient lookup of subregions in $\ds$, we augment
it with two bit vectors $\fst$ and $\diff$, each represented to support
constant-time $\rankop$ and $\selop$ queries (see Lemma~\ref{lem:rank_select}).
If $Q$ is the total number of subregions in~$G$, each of the two vectors has
length $Q$.
Their entries are defined as follows:
\begin{enumerate}
\item $\fst[k] = 1$ if and only if the $k$th subregion $\subreg[i,j]$
  is the first subregion in $\reg[i]$, that is, if $j = 1$.
\item $\diff[k] = 1$ if and only if the $k$th subregion starts in
  a different block from the $(k-1)$st subregion; that is, if
  the first bits of these two regions belong to different blocks in
  $\ds$.
\end{enumerate}

\paragraph{\boldmath Succinct encoding of $\alpha$-neighbourhoods.}

The third piece of our data structure consists of succinct encodings
of the $\alpha$-neighbourhoods of (region and subregion) boundary
vertices, for $\alpha = (\succblksize)^{1/3}$.
Recall that the $\alpha$-neighbourhood of a subregion boundary vertex interior
to its region includes only vertices inside this region.
Each $\alpha$-neighbourhood $\nb[\alpha]{x}$ is represented using the
following structures:
\begin{enumerate}
\item A succinct encoding $\graphrep[x]$ of the graph structure of
  $\nb[\alpha]{x}$.
  This encoding involves a permutation $\perm[x]$ of the vertices in
  $\nb[\alpha]{x}$.
\item A bit vector $\term[x]$ of length $\size{\nb[\alpha]{x}}$ for which
  $\term[x][k] = 1$ if and only if the $k$th vertex in $\perm{x}$ is a terminal
  vertex of $\nb[\alpha]{x}$.
\item A $(\lg \alpha)$-bit integer $\pos[x]$ that records the position of
  $x$ in $\perm[x]$.
\item An array $\keyvec[x]$ of length $\size{\nb[\alpha]{x}}$ that stores the
  key associated with each vertex in $\nb[\alpha]{x}$.
  These keys are stored according to the permutation $\perm[x]$.
\item An array $\lblvec[x]$ of length $\size{\nb[\alpha]{x}}$ that stores
  labels of the vertices in $\nb[\alpha]{x}$.
  The specific labels stored in $\lblvec[x]$ depend on whether $x$ is a region
  boundary vertex or a subregion boundary vertex interior to a region.
  \begin{itemize}
  \item If $x$ is a region boundary vertex, $\lblvec[x][k]$ is the graph
    label of the $k$th vertex in $\perm[x]$.
    Each such graph label is stored using $\lg n$ bits.
  \item If $x$ is a subregion boundary vertex interior to a region
    $\reg[i]$, the label stored in $\lblvec[x][k]$ depends on the type of the
    $k$\textsuperscript{th} vertex $y$ in $\perm[x]$.
    If $y$ is an interior vertex of $\nb[\alpha]{x}$---that is,
    $\term[x][k] = 0$---then $\lblvec[x][k]$ is the difference between the graph
    label of $y$ and the minimum graph label of the interior vertices in
    $\reg[i]$,  which is stored in $\regminvec$.
    We call this the \emph{region offset} of $y$.
    If $y$ is a terminal vertex of $\nb[\alpha]{x}$---that is,
    $\term[x][k] = 1$---then $\lblvec[x][k]$ is the region label of $y$.
    We note that each of these two types of labels can be stored using
    $\lg \regsz$ bits.
    For region labels this is obvious, as each region has at most $\regsz$
    vertices.
    For region offsets, we observe that every interior vertex of
    $\nb[\alpha]{x}$ must also be interior to $\reg[i]$ because every
    boundary vertex of $\reg[i]$ has at least one neighbour outside of
    $\reg[i]$ and, hence, must be terminal for $\nb[\alpha]{x}$ (if it is
    contained in $\nb[\alpha]{x}$) because $\nb[\alpha]{x}$ contains only
    vertices in $\reg[i]$.
    The bound on the number of bits required to store region offsets now
    follows because $\reg[i]$ has at most $\regsz$ vertices, and all interior
    vertices of $\reg[i]$ receive consecutive graph labels.
  \end{itemize}
\end{enumerate}

Given that each $\alpha$-neighbourhood contains at most $\alpha$
vertices, there exist upper bounds, $\regub$ and $\subregub$, on the number of
bits required to encode the $\alpha$-neighbourhood of any region or
subregion boundary vertex, respectively.
We pad every $\alpha$-neighbourhood $\nb[\alpha]{x}$ to size $\regub$ or
$\subregub$ using extra bits, depending on whether $x$ is a region boundary
vertex.
We store the $\alpha$-neighbourhoods of all region boundary vertices in
an array $\regnbvec$, and the $\alpha$-neighbourhoods of all subregion boundary
vertices interior to their regions in an array $\subregnbvec$.
The $\alpha$-neighbourhoods in each array are ordered according to the graph
labels of their defining boundary vertices.

In order to be able to locate the $\alpha$-neighbourhoods of boundary vertices
in $\regnbvec$ and $\subregnbvec$, we store two bit vectors,
$\regbdvec$ and $\subregbdvec$, of length $N$ that support $\rankop$
and $\selop$ queries;
$\regbdvec[i] = 1$ if and only if the vertex with graph label $i$ is a region
boundary vertex; $\subregbdvec[i] = 1$ if and only if the vertex is a subregion
boundary vertex interior to its region.

\paragraph{Space analysis.}

Before discussing how to traverse a path in $G$ using this
representation, we prove a bound on its size.

\begin{lemma}
  \label{lem:main_space}
  The above representation of a bounded-degree planar graph $G$ with
  $N$ vertices, each with an associated $\bitsPerKey$-bit key, uses 
  $N \bitsPerKey + \OhOf{N} + \ohOf{N \bitsPerKey}$ bits.
\end{lemma}

\begin{proof}
  First we determine $\succblksize$.
  Let $c = \OhOf{1}$ be the number of bits per vertex used in the succinct
  representation $\graphrep[i,j]$ of the graph structure of each subregion
  $\subreg[i,j]$.
  Then the representation of each subregion uses at most 
  $\lg N + (c + \bitsPerKey + 1)\succblksize$ bits.
  Since a block can hold $B \wsize$ bits, with the first $\lg N$ bits used
  for the block offset in $\ds$, we have
  \begin{equation}
    \label{eqn:block_size}
    \succblksize = \frac{B \wsize - 2 \lg N}{c + \bitsPerKey + 1},
  \end{equation}
  assuming for the sake of simplicity that $c + \bitsPerKey + 1$ divides 
  $B \wsize - 2 \lg N$.
  Now we analyze the space consumption of the different parts of our structure.

  \textit{Minimum graph labels.}  The arrays $\regminvec$ and $\subregminvec$
  together store $\OhOf{N / \succblksize}$ graph labels, each of size $\lg N$.
  By (\ref{eqn:block_size}) and because $\bitsPerKey = \OhOf{\lg N}$, 
  $\succblksize = \OmegaOf{B} = \OmegaOf{\lg N}$.
  Hence, these two arrays use $\OhOf{N}$ bits of space.

  \textit{Subregions.}  Storing every vertex once uses $N(c + \bitsPerKey + 1)
  = N \bitsPerKey + \OhOf{N}$ bits.
  However, boundary vertices are stored once per subregion in which they are 
contained.
  By Lemma~\ref{lem:fred_graph_sep}, there are $\OhOf{N / \sqrt{\regsz}}$ region
  boundary vertices and $\OhOf{N / \sqrt{\subregsz}}$ subregion boundary
  vertices, counting each occurrence separately.
  Since $\regsz \ge \subregsz = \succblksize = \OmegaOf{\lg N}$, storing boundary
  vertices therefore requires
  \begin{equation*}
    \OhOfV{\frac{(c + \bitsPerKey + 1)N}{\sqrt{\succblksize}}} =
    \OhOfV{\frac{cN}{\sqrt{\succblksize}}} +
    \OhOfV{\frac{\bitsPerKey N}{\sqrt{\succblksize}}} = \ohOf{N} + \ohOf{N \bitsPerKey}
  \end{equation*}
  bits.
  In addition, we store a $(\lg N)$-bit size $\numv[i,j]$ for each
  subregion $\subreg[i,j]$.
  These numbers use $\OhOf{N \lg N / \succblksize} = \OhOf{N}$
  bits, as there are $\OhOf{N / \succblksize}$ subregions.
  In total, the size of array $\ds$ is $T := N \bitsPerKey + \OhOf{N} + \ohOf{N \bitsPerKey}$,
  excluding block offsets.
  There is one block offset of size $\lg N$ per block of $B \wsize =
  \OmegaOf{\lg^2 N}$ bits.
  Thus, the total space used by block offsets is $\OhOf{T / \lg N} = \ohOf{N \bitsPerKey}$.

  The final part of the representation of subregions is the vectors
  $\fst$ and $\diff$.
  Each vector stores one bit per subregion.
  Thus, their total size is $\OhOf{N / \succblksize} = \ohOf{N}$, and they can be stored
  in $\ohOf{N}$ space to support $\rankop$ and $\selop$ operations.
  In summary, we use $N \bitsPerKey + \OhOf{N} + \ohOf{N \bitsPerKey}$ bits to represent all
  subregions.

  \textit{$\alpha$-Neighbourhoods.}  We bound the sizes of the
  $\alpha$-neighbourhoods of region and subregion boundary vertices
  separately.
  By Lemma~\ref{lem:fred_graph_sep}, there are $\OhOf{N / \sqrt{\regsz}}$
  region boundary vertices.
  The representation of the $\alpha$-neighbourhood of each such vertex uses
  \begin{equation*}
    \lg \alpha + \alpha(c + 1 + \bitsPerKey + \lg N) = \OhOfV{(\succblksize)^{1/3} \lg N}
  \end{equation*}
  bits.
  Hence, the total size of the $\alpha$-neighbourhoods of all
  region boundary vertices stored in $\regnbvec$ is:
  
  \begin{equation*}
    \OhOfV{\frac{N (\succblksize)^{1/3} \lg N}{\sqrt{\regsz}}} =
    \OhOfV{\frac{N (\succblksize)^{1/3} \lg N}{\sqrt{\succblksize \lg^3 N}}} =
    \ohOf{N}.
  \end{equation*}
  The $\alpha$-neighbourhood of every subregion boundary vertex is
  represented using
  \begin{equation*}
    \lg \alpha + \alpha(c + 1 + \bitsPerKey + \lg \regsz) = \OhOfV{(\succblksize)^{1/3} (\bitsPerKey + \lg
      \regsz)}
  \end{equation*}
  bits.
  Since there are $\OhOf{N / \sqrt{\subregsz}} = \OhOf{N / \sqrt{\succblksize}}$
  subregion boundary vertices, the total size of their
  $\alpha$-neighbourhoods stored in $\subregnbvec$ is therefore
  \begin{align*}
    \OhOfV{\frac{N (\succblksize)^{1/3} (\bitsPerKey + \lg \regsz)}{\sqrt{\succblksize}}} &=
    \OhOfV{\frac{N (\succblksize)^{1/3} \bitsPerKey}{\sqrt{\succblksize}}} +
    \OhOfV{\frac{N (\succblksize)^{1/3} \lg (\succblksize \lg^3 N)}{\sqrt{\succblksize}}}\\
    &= \ohOf{N \bitsPerKey}.
  \end{align*}
  The last equality follows because $\succblksize = \OmegaOf{\lg N}$.
  Together, arrays $\regnbvec$ and $\subregnbvec$ use $\ohOf{N \bitsPerKey}$ space.

  The bit vectors $\regbdvec$ and $\subregbdvec$ have length $N$ each and
  contain $\ohOf{N}$ 1 bits each.
  Thus, by Lemma \ref{lem:rank_select}(b), they
  can be stored in $\ohOf{N}$ bits.

  In summary, we use $\ohOf{N \bitsPerKey}$ bits to store all $\alpha$-neighbourhoods.
  By summing the sizes of the three parts of
  our data structure, we obtain the space bound claimed in the lemma.
\end{proof}

\subsection{Navigation}

\label{sec:navigation}

To traverse a path $\CDpath$ in $G$, we use a strategy similar to that used 
by Agarwal \etal~\cite{DBLP:conf/soda/AgarwalAMVV98},
alternately using subregions and $\alpha$-neighbourhoods in order to
progress along $\CDpath$.
We assume we are given a function, $\stepop$, that takes the graph label
and key of the current vertex, $x$, as an argument, and either reports that $x$
is the last vertex of $\CDpath$ (in which case the traversal should stop), or outputs
the graph label and key of the successor of $x$ in $\CDpath$.
To make this decision, $\stepop$ may
use state information computed during the traversal of $\CDpath$ up to $x$
and needs access to the set of neighbours of $x$.
Therefore, the task of our traversal procedure is to ensure that, during each step
along $\CDpath$, the neighbours of the current vertex are in memory.

Assume, without loss of generality, that the start vertex of $\CDpath$ is 
interior to a
subregion~$\subreg[i,j]$.
This is easy to ensure initially by loading the
representation of $\subreg[i,j]$ into memory.
Then we follow the path $\CDpath$ by repeated application of the $\stepop$ function
and by invoking the following paging procedure for each visited vertex before
applying $\stepop$ to it.
\begin{itemize}
\item If the current vertex $x$ is interior to the current component
  (subregion or $\alpha$-neighbourhood) in memory, we do not do
  anything, as all of $x$'s neighbours are in memory.
\item If the current component is a subregion $\subreg[i,j]$ and $x$ is a
  boundary vertex of $\subreg[i,j]$, or the current component is an
  $\alpha$-neighbourhood $\nb[\alpha]{y}$ and $x$ is a terminal vertex of
  $\nb[\alpha]{y}$, then $x$ has neighbours outside the current
  component.
  We load a component containing $x$ and all its neighbours into memory:
  \begin{itemize}
  \item If $x$ is a region or subregion boundary vertex, we load its
    $\alpha$-neighbourhood $\nb[\alpha]{x}$ into memory.
  \item If $x$ is interior to a subregion $\subreg[i',j']$, we load
    $\subreg[i',j']$ into memory.
  \end{itemize}
\end{itemize}

Our paging strategy ensures that, for every vertex $x \in \CDpath$, we have
a succinct representation of a component of $G$ containing all
neighbours of $x$ in memory when $x$ is visited.
To show that this allows us to traverse any path of length $K$
using $\OhOf{K / \lg B}$ I/Os, we need to show that
(1) given the graph label of a vertex $x$, we can identify and load its
$\alpha$-neighbourhood, or the subregion $x$ is interior to, using $\OhOf{1}$
I/Os; (2) the graph labels of the vertices in the component currently
in memory and the graph structure of this component
can be determined without performing any I/Os; and (3) for every
$\lg B$ steps along path $\CDpath$, we load only $\OhOf{1}$ components into memory.

Before proving these claims, we prove that we can efficiently convert between
graph, region, and subregion labels of a vertex $x$, and identify the region and
subregion to which $x$ is interior, or, if $x$ is a region or subregion boundary vertex,
the region or subregion which defines $x$.
The following lemma states this formally.
Its proof is the same as in~\cite{DBLP:journals/talg/BoseCHMM12}.
We include it here for completeness.

The results of the previous lemma lead to the following lemma.

\begin{lemma}
\label{lem:component_ios}
When the traversal algorithm encounters a terminal or boundary
vertex $v$, the next component containing $v$ in which the traversal
may be resumed can be loaded in $O(1)$ I/O operations.
\end{lemma}

\begin{proof}
Consider firstly the case of a boundary vertex, $v$, for a subregion.
The vertex may be a region or subregion boundary.
By Lemmas \ref{lem:label_ops}(c) and \ref{lem:label_ops}(d), the graph
label of $v$ can be determined in $O(1)$ I/O operations.
If $v$ is a region boundary vertex, this graph label serves as a direct 
index into the array of region boundary vertex $\alpha$-neighbourhoods 
(recall that all region boundary vertices are labeled with consecutive 
graphs labels).
Loading the $\alpha$-neighbourhood requires an additional I/O operation.

If $v$ is a subregion boundary vertex, the region and region
label can be determined by Lemma \ref{lem:label_ops}(a). 
By Lemma \ref{lem:label_ops}(d) we can determine the graph label. 
For the subregion boundary vertex, $v$, with graph label $\ell_G$, we can
determine the position of the $\alpha$-neighbourhood of $v$ in
$\N_S$ in $O(1)$ I/Os by 
$\rankop[1]{\bdvec, \ell_G}$.
In order to report the graph labels of vertices in the
  $\alpha$-neighbourhood, we must know the graph label of the first
  interior vertex of the region, which we can read from
  $\mathcal{L}_R$ with at most a single additional I/O operation.

  Further,  consider the case of a terminal node in a
  $\alpha$-neighbourhood. If this is the $\alpha$-neighbourhood of a
  subregion boundary vertex, the region $R_i$ and region label are
  known, so by Lemma \ref{lem:label_ops}(d) we can determine the graph
  label, and load the appropriate component with $O(1)$ I/O
  operations. Likewise, for $\alpha$-neighbourhoods of a region
  boundary vertex, the graph label is obtained directly from the vertex
  key.

  Finally, we show that a subregion can be loaded in $O(1)$
  I/Os. Assume that we wish to load subregion $R_{i,j}$. Let $r =
  \selop[1]{\mathcal{B}_R,i}$ mark the start of the region $i$. We can
  then locate the block, $b$, in which the representation for
  subregion $j$ starts by $b=\rankop[1]{\mathcal{B}_S, r+j}$. 
  There may be other subregions stored entirely in $b$ prior to
  $R_{i,j}$. 
  We know that if $\mathcal{B}_S[r+j] = 0$, then $R_{i,j}$ is the
  first subregion stored in $b$ that has its representation start in
  $b$.  
  If this is not the case, the result of
  $\selop[1]{\rankop[1]{\mathcal{B}_S,r+j}}$ will indicate the position
  of $r$ among the subregions stored in $b$.  
  We can now read subregion $R_{i,j}$ as follows: 
  we load block $b$ into memory and
  read the block offset which indicates where the first subregion in
  $b$ starts. 
  If this is $R_{i,j}$, we read the subregion offset in order to
  determine its length, and read $R_{i,j}$ into memory, possibly
  reading into block $b+1$ if there is an overrun. 
  If $R_{i,j}$ is not
  the first subregion starting in block $b$, then we note how many
  subregions we must skip from the start of $b$, and by reading only
  the subregion offsets we can jump to the location in $b$ where
  $R_{i,j}$ starts. 
  The $\rankop$ and $\selop$ operations require
  $O(1)$ I/Os, and we read portions of at most two blocks $b$ and
  $b+1$ to read a subregion, therefore we can load a subregion with $O(1)$
  I/Os.
\end{proof}

\begin{lemma}
  \label{lem:report_int_labels}
  Given a subregion or $\alpha$-neighbourhood, the graph labels of
  all interior (subregions) and internal ($\alpha$-neighbourhoods)
  vertices can be reported without incurring any additional I/Os
  beyond what is required when the component is loaded to main memory.
\end{lemma}

\begin{proof}
  The encoding of a subregion induces an order on all vertices, both
  interior and boundary, of that subregion. Consider the interior
  vertex at position $j$ among all vertices in the subregion. 
  The position of this vertex among all interior vertices may be
  determined by the result of $\rankop[0]{\mathcal{B}, j}$. 
  Recall that graph labels assigned to all interior vertices in a subregion are
  consecutive; therefore, by adding one less this value to the
  graph label of the first vertex in the subregion, the graph label
  of interior vertex at position $j$ is obtained.

  For the $\alpha$-neighbourhood of a region boundary vertex, the graph
  labels from $\mathcal{L}_\alpha$ can be reported directly.  
  For the $\alpha$-neighbourhoods of subregion boundary vertices, we can
  determine from $\mathcal{L}_R$ the graph label of the first vertex
  in the parent region. 
  This potentially costs an I/O; however, we can pay for this when the 
component is loaded. 
  We can then report graph
  labels by adding the offset stored in $\mathcal{L}_{\alpha}$ to the
  value from $\mathcal{L}_R$.
\end{proof}

\begin{lemma}
  \label{lem:main_IO_bound}
  Using the data structures and navigation scheme described above, a
  path of length $K$ in graph $G$ can be traversed in $O \left(
    \frac{K}{ \lg{\succblksize} } \right)$ I/O operations.
\end{lemma}

\begin{proof}
  The $\alpha$-neighbourhood components are generated by performing a
  breadth-first traversal, from the boundary vertex $v$, which
  defines the neighbourhood. A total of $\succblksize^{1/3}$ vertices are added
  to each $\alpha$-neighbourhood component. Since the degree of $G$ is
  bounded by $d$, the length of a path from $v$ to the terminal
  vertices of the $\alpha$-neighbourhood is
  $\log_d{\succblksize^{1/3}}$. 
  However, for subregion boundary vertices, the
  $\alpha$-neighbourhoods only extend to the boundary vertices of the
  region, such that the path from $v$ to a terminal node may be less
  than $\log_d{\succblksize^{1/3}}$. 
  In the later case the terminal vertex
  corresponds to a region boundary vertex.

  Without loss of generality, assume that traversal starts with an
  interior vertex of some subregion. Traversal will continue in the
  subregion until a boundary vertex is encountered, at which time the
  $\alpha$-neighbourhood of that vertex is loaded. In the worst case
  we travel one step before arriving at a boundary vertex of the
  subregion. If the boundary vertex is a region boundary, the
  $\alpha$-neighbourhood is loaded, and we are guaranteed to travel at
  least $\log_d{\succblksize^{1/3}}$ steps before another component must be
  visited.  If the boundary vertex is a subregion boundary, then the
  $\alpha$-neighbourhood is loaded, and there are again two possible
  situations. In the first case, we are able to traverse
  $\log_d{\succblksize^{1/3}}$ steps before another component must be visited. In
  this case, by visiting two components, we have progressed a minimum
  of $\log_d{\succblksize^{1/3}}$ steps. In the second case, a terminal vertex in
  the $\alpha$-neighbourhood is reached before $\log_d{\succblksize^{1/3}}$ steps
  are taken. This case will only arise if the terminal vertex
  encountered is a region boundary vertex. Therefore, we load the
  $\alpha$-neighbourhood of this region boundary vertex, and progress
  at least $\log_d{\succblksize^{1/3}}$ steps along the path before another I/O
  will be required.

  Since we traverse $\log_d{\succblksize^{1/3}}$ vertices with a constant number
  of I/Os, we visit $O \left( \frac{K}{ \lg{\succblksize} } \right)$ components
  to traverse a path of length $K$. By Lemma \ref{lem:component_ios},
  loading each component requires a constant number of I/O operations,
  and by Lemma \ref{lem:report_int_labels} we can report the graph
  labels of all vertices in each component without any additional
  I/Os. Thus, the path may be traversed in $O \left( \frac{K}{ \lg \succblksize }
  \right)$ I/O operations.
\end{proof}

\begin{theorem}
  \label{thm:planar_graph}
  Given a planar graph $G$ of bounded degree, where each vertex stores
  a key of $\bitsPerKey$ bits, there is a data structure that represents $G$ in
  $N\bitsPerKey +O(N) + o(N\bitsPerKey)$ bits that permits traversal of a path of length
  $K$ with $O \left( \frac{K}{ \lg{B} } \right)$ I/O operations.
\end{theorem}

\begin{proof}
  Proof follows directly from Lemmas \ref{lem:main_space} and
  \ref{lem:main_IO_bound}. We substitute $B$ for $\succblksize$, using
  Eq.~\ref{eqn:block_size} ($\succblksize = \Omega(B)$), as this is standard for
  reporting results in the I/O model.
\end{proof}

Due to the need to store keys with each vertex, it is impossible to
store the graph with fewer than $N\bitsPerKey$ bits.  The space savings in our
data structure are obtained by reducing the space required to store
the actual graph representation.
Agrawal~\etal~\cite{DBLP:conf/soda/AgarwalAMVV98} do not attempt to
analyze their space complexity in bits but, assuming they use
$\lg{N}$ bit pointers to represent the graph structure, their
structure requires $\OhOf{N \lg{N}}$ bits for any size $\bitsPerKey$. If $\bitsPerKey$ is
a small constant, our space complexity becomes $\OhOf{N} +
\ohOf{N}$, which represents a savings of $\lg{N}$ bits compared to
the $\OhOf{N \lg{N}}$ space of of Agarwal~\etal. In the worst case
when $\bitsPerKey = \ThetaOf{\lg{N}}$, our space complexity is 
$\ThetaOf{N \lg{N}}$, which is asymptotically equivalent to that of
Agarwal~\etal. However, even in this case, our structure can save a
significant amount of space due to the fact that we store the actual
graph structure with $\OhOf{N}$ bits (the $N\bitsPerKey$ and $\ohOf{N\bitsPerKey}$
terms in our space requirements are related directly to space required
to store keys), compared to $\OhOf{N \lg{N}}$ bits in their
representation.

\subsection{An Alternate Blocking Scheme}\label{sec:alt_block_scheme}

In this section, we describe a blocking scheme based on that
of Baswana and Sen \cite{DBLP:journals/algorithmica/BaswanaS02}, which
allows us to improve the I/O efficiency of our data structure for some
planar graphs. 
As with the blocking strategy of
\cite{DBLP:conf/soda/AgarwalAMVV98} employed previously, this new
strategy recursively identifies a separator on the nodes of the graph
$G$, until $G$ is divided into regions of size $r$ which can fit on a
single block (in this context, region is used generically to refer to
the size of partitions in the graph as in Lemma
\ref{lem:fred_graph_sep}). However, rather than store
$\alpha$-nieghbourhoods of size $\sqrt{r}$ for each boundary vertex,
larger $\alpha$-neighbourhoods are stored for a subset of the boundary
vertices. In order for this scheme to work, it is necessary that
boundary vertices be packed close together, in order that the selected
$\alpha$-nieghbourhoods cover a sufficiently large number of boundary
vertices. This closeness of boundary vertices is not guaranteed using
the technique of \cite{DBLP:conf/soda/AgarwalAMVV98}, but if the
alternate separator algorithm of Miller~\cite{miller_1986} is
incorporated, then region boundary vertices are sufficiently packed.

For an embedded graph $G=(V,E)$ with maximum face size $c$,
the separator result of Miller~\cite{miller_1986} is 
stated in Lemma~\ref{lem:miller_cycle_sep}.

\begin{lemma}[Miller~\cite{miller_1986}]\label{lem:miller_cycle_sep}
If $G$ is an embedded planar graph consisting of $N$ vertices, then
there exists a balanced separator which is a single vertex, or a simple cycle
of size at most $2 \sqrt{2 \left\lfloor \frac{c}{2} \right\rfloor N}$,
where $c$ is the maximum face size.
\end{lemma}

Baswana and Sen~\cite{DBLP:journals/algorithmica/BaswanaS02}
construct a graph permitting {I/O}-efficient traversal by recursively
applying this separator.
If the separator is a single vertex $v$, then an $\alpha$-neighbourhood
is constructed for $v$ with $\alpha = B$. 
If the separator is a cycle, then the following rule is applied.
Let $r^{-}(k)$ be the minimum depth of a
breadth-first search tree of size $k$ built on any vertex $v \in
V$.
From this cycle every $\frac{r^{-}(s)}{2}$\textsuperscript{th}
vertex is selected, and an $\alpha$-neighbourhood is constructed 
by performing a breadth-first search from each of the selected 
vertices.
In this case, we let $\alpha=s$, where $s$ is
a number in the range $(\sqrt{B},B)$, a precise value that will be
defined shortly.
Vertices on the cycle which are not selected are associated with the
$\alpha$-neighbourhood constructed for the nearest selected vertex $v$.

Let $S(N)$ be the total space required to block 
the graph in this fashion.
This value includes both the space for blocks storing the
regions, and for those storing $\alpha$-neighbourhoods.
The value for $S(N)$ following recursive application of the 
separator defined in Lemma~\ref{lem:miller_cycle_sep} is
given by

\begin{equation}
  S(N) = c_1N + c_2\frac{\sqrt{c}N}{\sqrt{B}r^{-}(s)} s
\end{equation}

where $c_1$ and $c_2$ are constants independent of $s$.
In order to maximize $s$ while maintaining $S(N) = \OhOf{N}$,
select for $s$ the largest value $z \le B$ such that 
$z \le r^{-}(z)\sqrt{\frac{B}{c}}$.
This is achieved by choosing 
$s = \textbf{min}\left(r^{-}(s)\sqrt{\frac{B}{c}}, B \right)$.

This construction is summarized in the following lemma.

\begin{lemma}[\cite{DBLP:journals/algorithmica/BaswanaS02}]
  A planar graph $G$ of size $N$ with maximum face size $c$ 
  can be stored in
  $\BigOh{ \frac{N}{B} }$ blocks so that a path of length $K$ can
  be traversed using $\BigOh{ \frac{K}{ r^{-}(s) } }$ I/O operations
  where $s = \textbf{min}\left( r^{-}(s)\sqrt{\frac{B}{c}}, B \right)$.
\end{lemma}

Our succinct graph representation relies on a two-level partitioning of 
the graph.
Before determining a bound on the space for such a succinct representation,
it is useful to bound the total number of boundary vertices 
resulting from the recursive application of Miller's separator; this is 
done in Lemma~\ref{lem:cycle_sep_size}.

\begin{lemma}
  \label{lem:cycle_sep_size}
  For an embedded planar graph $G$ on $N$ nodes, with faces of
  bounded maximum size $c$, recursively partitioned into regions of
  size $r$ using the cycle separator algorithm of Miller
  \cite{miller_1986}, the separator $S$ has $\OhOf{N/\sqrt{r}}$
  vertices.
\end{lemma}

\begin{proof}
  Miller's simple cycle separator on a graph of $N$ nodes has size at
  most $2 \sqrt{ 2 \left\lfloor \frac{c}{2} \right\rfloor N}$, which
  is $\OhOf{\sqrt{N}}$ for constant $c$. At each step in the
  partitioning, $G$ is subdivided into a separator $|S| =
  \OhOf{\sqrt{N}}$, plus subsets $N_1$ and $N_2$, each containing at
  most $\frac{2}{3}N$ nodes plus the separator. Thus if $\frac{1}{3}
  \le \epsilon \le \frac{2}{3}$, we can characterize the size of the
  resulting subsets by $|N_1| = \epsilon N + \OhOf{\sqrt{N}}$ and
  $|N_2| = (1-\epsilon)N + \OhOf{\sqrt{N}}$. Following the proof of
  Lemma 1 in \cite{Frederickson87}, the total separator size require to
  split $G$ into regions of size at most $N$ then becomes
  $\BigOh{\frac{N}{\sqrt{r}}}$.
\end{proof}

We now describe how this technique is applied in order to achieve a succinct graph
encoding. 
Let $S_\alpha$ be the set of separator vertices selected as boundary
vertices.
We distinguish between the region separator vertices, 
$S^{R}_\alpha$, and the sub-region separator vertices $S^{SR}_\alpha$.
Select every $\frac{r^{-}(s)}{2}$\textsuperscript{th} separator vertex 
for which to build an $\alpha$-neighbourhood.
Each boundary vertex, $v$, is assigned to a single $\alpha$-neighbourhood, namely
the $\alpha$-neighbourhood centred nearest $v$ (or possibly at
$v$, if $v$ was an element of $S^{R}_\alpha$ or $S^{SR}_\alpha$). 
The graph is then encoded using the encoding described in 
Section~\ref{sec:datastructs}, with two small additions.
We add two arrays; one for region boundaries, which we denote 
$\regptable$, and the second for subregion boundaries,
which we denote $\subregptable$.
The length of these arrays is equal to the number of region boundary 
vertices and the number of sub-region boundary vertices, respectively.
Each of $\regptable$ (and $\subregptable$) is an array of tuples, where
the first tuple element is an index into $\regnbvec$ ($\subregnbvec$), 
and the second tuple element records the position of the vertex within
its $\alpha$-neighbourhood. 
By adding these structures to our existing succinct encoding, we have the
 following lemma.

\begin{lemma}
\label{lem:ptable_ios}
When the traversal algorithm encounters a terminal or boundary
vertex $v$, the next component containing $v$ in which the traversal
may be resumed can be loaded in $O(1)$ I/O operations.
\end{lemma}

\begin{proof}
The proof of this lemma is the same as that in 
Lemma~\ref{lem:component_ios}, with the addition that we must
use $\regptable$ and $\subregptable$ to lookup the correct 
$\alpha$-neighbourhood and the position of boundary vertex $v$
within it.

We begin with the case of a region boundary vertex.
In Lemma~\ref{lem:component_ios} it was shown that for a region
boundary vertex, we can locate its graph label in $\OhOf{1}$ I/O
operations.
The graph label serves as a direct index into $\regnbvec$.
In the current case, we let the graph label of a region boundary
vertex serve as an index into $\regptable$.
The first element in the corresponding tuple is a pointer to the
$\regnbvec$, and the second element of the tuple indicates from where
to resume traversal.
Using $\regptable$ incurs at most a single additional I/O.
Thus, lookup for region boundary vertices is $\OhOf{1}$.

For subregion boundary vertices, we apply the same indirection
to locate the proper $\alpha$-neighbourhood.
The result of the lookup $\rankop[1]{\bdvec, \ell_G}$, rather
than serving as a direct index into $\subregnbvec$, is used 
as an index into $\subregptable$.
The pointer stored in the first element of the returned
tuple can then be used to locate the correct $\alpha$-neighbourhood
in $\subregnbvec$. 
Again, this extra array lookup results in only a single 
additional I/O.
\end{proof}

\begin{lemma}\label{lem:cycle_sep_space}
  The representation described above of a bounded-degree planar graph $G$, 
  with $n$ vertices, each with an associated $\bitsPerKey$-bit key, uses 
  $n\bitsPerKey + \OhOf{n} + \ohOf{n\bitsPerKey}$ bits.
\end{lemma}

\begin{proof}
By Lemma \ref{lem:cycle_sep_size}, our two-level partitioning results
in $\OhOf{N/\sqrt{\succblksize \lg^{3}N}}$ region boundary vertices, and
$\OhOf{N/\sqrt{\succblksize}}$ subregion boundary vertices. From the separators
we select every $\frac{r^{-}(s)}{2}$'th vertex to add to $S_{\alpha}$,
so that we have $|S^{R}_{\alpha}| = \BigOh{ \frac{N}{\sqrt{\succblksize \lg^{3}N}
    r^{-}(s)} }$ and $|S^{SR}_{\alpha}| = \BigOh{\frac{N}{\sqrt{\succblksize}
    r^{-}(s)} }$.

We now demonstrate that the succinct encoding of the data structures
required to represent this partitioning requires the same amount of
storage as our previous result. 
We now store fewer, but larger,
$\alpha$-neighbourhoods, and we must store the tables $\subregptable$ and
$\regptable$.
For $\subregptable$, the first pointer in the tuple points to a block 
associated with an $\alpha$-neighbourhood in
$S^{SR}_{\alpha}$, and therefore requires 
$\lg{ \left( \BigOh{\frac{N}{\sqrt{\succblksize} r^{-}(s)}} \right)}$ 
bits, while the
vertex pointer requires $\lg{\succblksize}$ bits. The total space requirement in
bits for $\subregptable$ is thus:

\begin{equation}
  \OhOf{N/\sqrt{\succblksize}} \cdot \left( \lg{\left( \OhOf{\frac{N}{\sqrt{\succblksize}
            r^{-}(s)}} \right)} + \lg{\succblksize}  \right) = \ohOf{N}
\end{equation}

For $\regptable$, the respective pointer sizes are
$\lg{\left( \BigOh{\frac{N}{\sqrt{\succblksize \lg^3{N} } r^{-}(s)}} \right)}$
and $\lg{\succblksize}$ bits and, similarly, the total space is $\ohOf{N}$
bits.

Finally, we must account for the space required for the
$\alpha$-neighbourhoods. Before doing so, we must determine the actual
size of the $\alpha$-neighbourhoods, which we will denote $s$. Let $s
= \textbf{min}\left(r^{-}(s)\frac{\succblksize^{1/3}}{\sqrt{c} }, \succblksize \right)$; we
then have the following space complexities for region boundary
vertices (the term $k$ represents the constant from the $O()$
notation):

\begin{equation}
  \label{eqn:space_region_bound_scs}
  \frac{k\sqrt{c}N}{\sqrt{\succblksize \lg^3{N}} r^{-}(s)} \cdot
  r^{-}(s)\frac{\succblksize^{1/3}}{\sqrt{c} } \cdot ( \lg{N} + \bitsPerKey + c ) = o(N)
\end{equation}


\begin{eqnarray}
\frac{k\sqrt{c}N}{\sqrt{\succblksize} r^{-}(s)} \cdot r^{-}(s)\frac{\succblksize^{1/3}}{\sqrt{c} } \cdot ( \lg{(\succblksize \lg^3{N})} + \bitsPerKey + c ) &=& \frac{kN}{\succblksize^{\frac{1}{6}}} (\lg{\succblksize} + 3 \lg{\lg{N}}) + \frac{kN}{\succblksize^{\frac{1}{6}}}(\bitsPerKey+c) \nonumber \\
&=& \OhOf{N} + \ohOf{N\bitsPerKey} \label{eqn:space_sub_reg_bound}
\end{eqnarray}

\end{proof}

These results match the space requirements obtained for region and
subregion $\alpha$-neighbourhoods using our original partitioning
scheme (see Eqs. \ref{eqn:space_region_bound_scs} and
\ref{eqn:space_sub_reg_bound}), and lead to the following theorem.

\begin{theorem}
  Given a planar graph $G$ of bounded degree and with face size
  bounded by $c$, where each vertex stores a key of $\bitsPerKey$ bits, there is
  a data structure representing $G$ in $N\bitsPerKey +O(N) + o(N\bitsPerKey)$ bits that
  permits traversal of a path of length $K$ with $O \left( \frac{K}{
      r^{-}(s) } \right)$ I/O operations, where $s =
  \textbf{min}\left(r^{-}(s)\frac{\succblksize^{1/3}}{\sqrt{c} }, \succblksize \right)$.
\end{theorem}

\begin{proof}
The succinct structure and navigation technique is the same as
described in Theorem~\ref{thm:planar_graph}.
The only change is the addition of $\subregptable$ and $\regptable$.
Lemma~\ref{lem:ptable_ios} demonstrates that the new structures do not alter the
bounds on I/Os, while Lemma~\ref{lem:cycle_sep_space} establishes that the 
space complexity is unchanged.
\end{proof}

\section{Triangulations Supporting Efficient Path Traversal}
\label{sec:tins}

In this section, we describe how our data structures may be used to 
represent triangulations in a manner that permits efficient path
traversal in the EM setting.
In Section \ref{ssec:triangulation_rep}, we describe this representation in the 
general, non-EM setting, giving a compact representation for triangulations 
based on succinct representations of planar graphs.
Later, in Section \ref{ssec:compact_tin_external}, we detail how our succinct
and I/O-efficient planar graph representation can be applied to this representation 
to yield a compact representation for 
triangulations permitting I/O-efficient path traversal.

\subsection{Representing Triangulations}
\label{ssec:triangulation_rep}

Let $\pointset$ be a set of points in $\plane$, then the triangulation $\triang$
of $\pointset$ is the planar graph with $\pointset$ vertices and all faces 
(except the outer face) of size $3$.
The set of all edges adjacent to the outer face of the triangulation is the
convex hull of $\pointset$. 
Let $\triang = (\pointset,E,T)$ be the triangulation with vertices $\pointset$, 
edges $E$, and 
triangles (faces) $T$. 
The dual graph of $\triang$ is $\dual{\triang} = (T^*,E^*,\pointset^*)$, it has a node
 $t^* \in T^*$ for each triangle $t \in T$, a face $p^* \in \pointset^*$ for each 
vertex $p \in \pointset$, and an edge $e^* \in E^*$ connecting vertices $t^*_1$ and
 $t^*_2$, if and only if triangles $t_1$ and $t_2$ share an edge in $\triang$.

We further define the \emph{augmented dual graph} 
 $\augdual{\triang} = (\augdual{V} = (\pointset \cup \dual{T}), \augdual{E},
\augdual{F})$.
Figure \ref{fig:aug_tri} shows an example of a triangulation, its dual, 
and its augmented dual graph.
The vertex set of $\augdual{\triang}$ is formed by the union of the vertex set 
$P$ of $\triang$ and the nodes $\dual{T}$ of $\dual{\triang}$. 
We will distinguish between two types of vertices in $\augdual{\triang}$, 
depending on from which set the vertex is drawn. 
The vertices taken from the dual graph, denoted $\augdual{T}$, 
are referred to as the \emph{triangle nodes}. 
The vertices drawn from the set of vertices in the primal, $\pointset$, 
are refereed   
to as the \emph{point vertices} and are denoted $\augdual{\pointset}$.

For the edges of $\augdual{\triang}$, we again distinguish two types
 of edges. 
The \emph{triangle edges} are exactly the edge set in the dual, $E^*$. 
A \emph{point edge} is added between a triangle node and a point vertex
 if the point vertex is one of the corresponding triangle's three vertices in 
 the primal.  
We denote this set $\augdual{E}$, and define it formally as:\\

\begin{centering}
$\augdual{E} = \{\augdual{e} | \left( e^* \in E^* \right)
 \cup \left( \augdual{e}=( \augdual{p} \in \augdual{\pointset}, 
  \augdual{t} \in \augdual{T} ) \text{ and } p 
  \text{ adjacent to } t \text{ in } \triang \right) \} $
\end{centering}
\vspace{12pt}

\noindent where $p$ and $t$ are the vertex and triangle, respectively, corresponding to 
$\augdual{p} \in \augdual{\pointset}$ and $t^{\augdualsym} \in T^{\augdualsym}$. 
$\augdual{\triang}$ also contains a set of faces $F^{\augdualsym}$, but these faces 
do not figure in our discussion.

For purposes of quantifying bit costs, we denote by $\bitsPerPoint = \OhOf{\lg{N}}$ 
the number of bits required to represent a point in our data structures. 
Many applications for triangulations will not require that a key be stored 
with the triangles, therefore we assume that there is no $q$ bit key associated 
with each triangle node. 
We show that the space used by keys in our graph structure is effectively 
the same as the space used by the point set in our triangulation, but that 
if we wish to maintain keys we can do so without significantly increasing 
the space used.

\begin{figure}
\centering
	\includegraphics[width=0.8\textwidth]{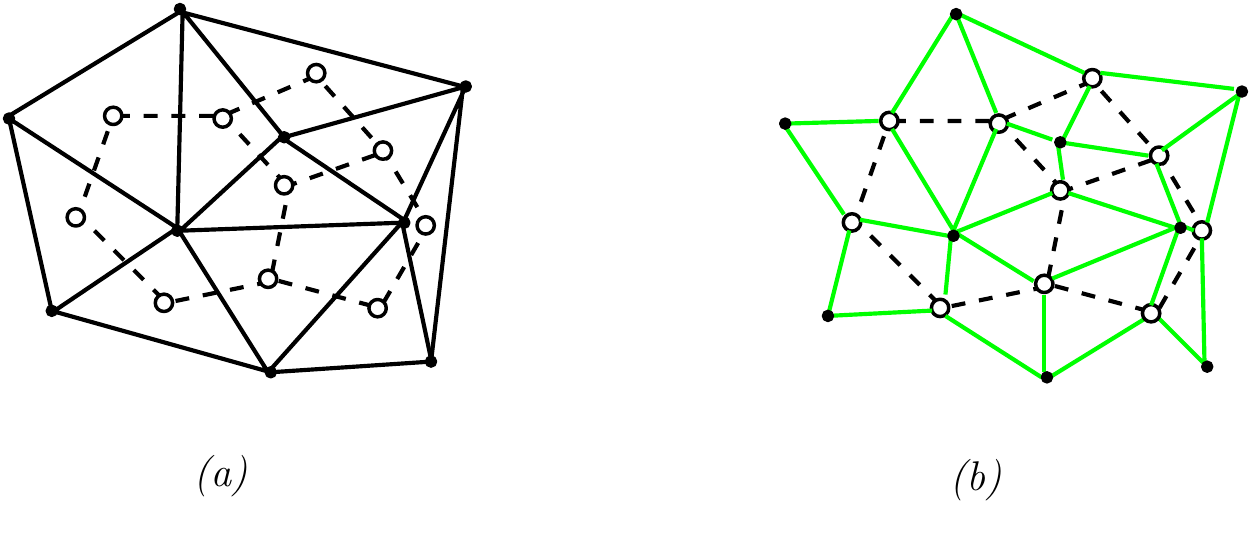}
\caption[Representation of a triangulation as a planar graph]{Representation 
	of a triangulation as a planar graph. 
  (a) The dual graph $\dual{\triang}$ of a triangulation, with vertices as hallow 
circles and edges as dashed lines. 
Edges in the triangulation $\triang$ are shown as solid lines. 
(b) The augmented graph $\augdual{T}$ is shown;
hallow points are triangle vertices, solid points are point 
vertices, black dashed lines are triangle edges, and solid green lines are point 
edges. }\label{fig:aug_tri}
\end{figure}

We begin with the following lemma. 

\begin{lemma}\label{lem:points_for_dual}
Given the dual graph $\dual{\triang}$, of triangulation $\triang$ with 
$N$ triangles, the augmented dual graph $\augdual{\triang}$ has at most $2N+2$ 
vertices.
\end{lemma}

\begin{proof}
The vertex set $\augdual{V}$ includes the nodes $\dual{T}$ for which 
$|\dual{T}| = N$, and the vertex set $\pointset$. 
We prove by induction that $|P| \le N+2$, thus $|\augdual{V}| \le 2N + 2$. 
For the base case, we have a terrain $\triang$ with a single triangle in which case 
$N = 1$, and $|P| = 3  = N + 2$. 
Now assume that $|P| \le N + 2$ holds for all terrains of $N$ triangles. 
Let $\dual{\triang}_N$ be the dual graph of a triangulation $\triang$ with $N$ 
triangles and $|P|$ vertices. 
The dual graph $\dual{\triang}_{N+1}$ is created by adding a single triangle 
to $\triang$. 
This new triangle, $t_1$, has three adjacent points; however, 
since $\dual{\triang}_{N+1}$ is connected, at least one dual edge is added 
connecting a vertex $t^*_2 \in \dual{\triang}_N$ with the new vertex 
$t^*_1 \in \dual{\triang}_{N+1}$. 
In $\triang$, this dual edge represents the fact that $t_1$ and $t_2$ are 
adjacent, and two of the vertices from $P$ adjacent to $t_2$ are already 
in $\augdual{V}$. Therefore, adding a new triangle adds at most one 
additional point to $\augdual{V}$ and $|\augdual{V}_{N+1}| \le (N+1) + 2$.
\end{proof}

We encode $\augdual{\triang}$ using a succinct planar graph data structure.
The various encodings all involve a permutation of the vertices of 
$\augdual{\triang}$.
Let $\ell(v)$, the \emph{augmented graph label} of $v \in \augdual{V}$ be 
the position of a vertex $v$ in this permutation. 
For every point vertex $p \in \augdual{P} \subset \augdual{V}$ in this set, 
we store the point coordinates in an array $\mathcal{P}$ ordered by the 
augmented graph label, $\ell(p)$. 
Finally, we create a bit vector $\pi$ of length $|\augdual{V}|$, where 
$\pi[v] = 0$ if $v$ is a triangle vertex and $\pi[v] = 1$ if $v$ is a 
point vertex. 
To summarize this structure we have the following lemma.

\begin{lemma}\label{lem:terrain_space}
The data structures described above can represent a terrain $\triang$ composed 
of $N$ triangles with $\bitsPerPoint = O(\lg{N})$ bit point coordinates, using 
$N\bitsPerPoint + O(N)$ bits, such that given the label of a triangle, the adjacent
 triangles and points can be reported in $O(1)$ time.
\end{lemma}

\begin{proof}
Since the augmented planar graph is simple, we can encode it with 
$2|\augdual{E}| + 2|\augdual{V}| + o(|\augdual{V}|)$ bits 
using the encoding of~\cite{DBLP:journals/siamcomp/ChiangLL05}. 
By Lemma \ref{lem:points_for_dual}, if $|\dual{T}| = N$ then 
$|\augdual{V}| \le 2N+2$, which bounds the number of vertices. 
Each triangle node is connected by an edge to at most three
other triangles, thus there are at most $\frac{3}{2} N$ triangle edges. 
Additionally, there are $3N$ point edges connecting the 
triangle vertices with point vertices, therefore in total there are no 
more than $\frac{9}{2}N$ edges. 
Thus the augmented graph can be encoded using at most $13N + o(N)$ bits. 

In~\cite{DBLP:journals/siamcomp/ChiangLL05} adjacency and degree queries
 can be performed in $O(1)$. 
We must still demonstrate that given a label, we can identify the corresponding 
triangle vertex, and show that for a vertex we can distinguish between point 
vertex and triangle node neighbours. 
We assign to each triangle a unique graph label as follows. 
Consider the set of triangles in $\triang$, the graph label of each 
triangle corresponds to that of its dual vertex $t^* \in T^*$. 
In $\augdual{\triang}$, each of these vertices has an augmented graph label.  
The graph label of dual node $\dual{t} \in \dual{T}$, and therefore the 
corresponding triangle is equal to the rank of the corresponding 
triangle vertex $\augdual{t}$'s augmented graph label among all 
vertices from the set $\augdual{V}$. 

Given the graph label of a triangle node $\augdual{t} \in \augdual{T}$,
 the augmented graph label is calculated as $\ell(t^{\augdualsym}) = 
\selop[0]{\pi, \augdual{t}}$.
Conversely, given $\ell(\augdual{t})$ for a triangle vertex 
$\augdual{t} \in \augdual{T}$, we can report the graph label 
of $\augdual{t}$ by $\rankop[0]{\pi,\ell(\augdual{t})}$. 
In order to report the triangles adjacent to $\augdual{t}$, we check in 
the succinct encoding for $\augdual{\triang}$ to find all neighbours of 
$\augdual{t}$. 
Since $d \le 6$, this operation takes constant time. 
Let $\augdual{v} \in \augdual{V}$ be a neighbour of $\augdual{t}$ in 
$\augdual{\triang}$. 
The value of $\pi[\augdual{v}]$ indicates whether $\augdual{v}$ is a 
triangle node or point vertex. 
We can recover the coordinates of a point vertex $\augdual{v}$ by 
$\mathcal{P}[\rankop[1]{\pi,\augdual{v}}]$.

The array $\mathcal{P}$ stores at most $N+2$ points. 
If each point requires $\bitsPerPoint$ bits, then $\mathcal{P}$ 
requires $N\bitsPerPoint$ bits. 
Finally, by Lemma \ref{lem:rank_select}(b) we can store the bit array 
$\pi$ such that $\rankop$ and $\selop$ can be performed in constant time 
with $N + o(N)$ bits.
\end{proof}

\subsection{Compact External Memory Representation}
\label{ssec:compact_tin_external}

In this section, we extend our data structures for I/O-efficient traversal in 
bounded-degree planar graphs (Section \ref{sec:graph_rep}) to triangulations. 
We thereby obtain a triangulation that is compact, but which also 
permits efficient traversal. 
We represent $\triang$ by its dual graph $\dual{\triang}$. 
Since the dual graph of the triangulation (and subsequently each component) is a 
bounded-degree planar graph, it can be represented with the data structures 
described in Section \ref{sec:graph_rep}. 
Each component is a subgraph of $\dual{\triang}$, for which we generate the 
augmented subgraph, as described in Section~\ref{ssec:triangulation_rep}.

\begin{lemma}\label{lem:terrain_components}
The space requirement, in bits, to store a component ($\alpha$-neighbourhood
 or sub-region) representing a triangulation is within a constant factor of the 
space required to store the triangulation's dual graph. 
\end{lemma}

\begin{proof}
We first consider the case of $\alpha$-neighbourhoods. 
To store a triangulation, we require additional space to store: the points in $P$ 
adjacent to the component's triangles; the augmented 
graph, which includes more vertices and edges; and the bit-vector $\pi$. 
By Lemma~\ref{lem:points_for_dual}, the points array is of size at most 
$(\alpha + 2)\bitsPerPoint$ bits. 
Since the $\alpha \bitsPerKey$-bit array of keys is no longer needed, storing the 
points incurs no additional space over that used to store the dual graph 
in our construction in Section \ref{sec:graph_rep}. 
However, if the keys are needed, we require 
$\alpha (\bitsPerPoint + \bitsPerKey)$ bits to 
store both the points and the key values.  
For the graph representation, again by Lemma \ref{lem:points_for_dual}, in 
the augmented graph we at most double the size of the vertex set, and add 
no more than $3 \alpha$ edges. 
The graph encoding is a linear function of the number of edges and 
vertices. 
Therefore, the number of bits required for the graph encoding 
increases by a constant factor. 
Finally, for $\pi$ we require fewer than two bits per vertex, thus we can add 
this cost to the constant cost of representing the graph.

For subregions, the analysis is more complex. 
A subregion may be composed of multiple connected components, thus we cannot 
assume that there will be at most $\succblksize+2$ point vertices in the 
augmented subgraph.  
By Lemma \ref{lem:fred_graph_sep}, there are no more than 
$\Theta(N/ \sqrt{\succblksize})$ boundary vertices in $\dual{\triang}$. 
Since each connected component in a subregion must have at least one boundary vertex, 
this bounds the number of distinct connected components in all subregions. 
Each subregion is composed of at most $\Theta(\sqrt{\succblksize})$ individual 
components, therefore, in the worst case there are no more than 
$\sqrt{\succblksize}( \sqrt{\succblksize} + 2) < 2 \succblksize$ point vertices. 
Thus, the total number of vertices in the augmented graph is less than 
$3\succblksize$, which adds at most an additional $6\succblksize$ edges. 
As demonstrated with the $\alpha$-neighbourhoods, the additional storage, 
for all data structures used to represent a subregion, 
increases by only a constant factor.
\end{proof}

One important feature of our representation is that each triangle stores 
very limited topological information. 
Consider the information available to some triangle 
$\augdual{t} \in \augdual{T}$ in the augmented dual graph. 
We can determine the - up to three - triangles adjacent to $\augdual{t}$, which 
correspond to the adjacent triangle vertices, as well as the three 
adjacent point vertices, but we have no information about how these are related.  
The only way to acquire this information is to actually visit each of 
$\augdual{t}$'s neighbours, and thereby construct the topological 
information. 
Fortunately, the need to visit the neighbours of a triangle in order to reconstruct 
its topological information will not increase the I/O costs of path traversal. 
If $\augdual{t}$ corresponds to an interior vertex (in a subregion), or an 
internal vertex (in an $\alpha$-neighbourhood), then all of its neighbours are 
represented in the current component. 
If $\augdual{t}$ corresponds to a boundary (subregion) or terminal 
($\alpha$-neighbourhood) vertex, the traversal algorithm already requires 
that a new subregion, or $\alpha$-neighbourhood, be loaded. 
If we must load an $\alpha$-neighbourhood, we are ensured that all of 
$\augdual{t}$'s neighbours are in the new component. 
In the worst case, $\augdual{t}$ is a terminal vertex in a 
$\alpha$-neighbourhood, and a boundary vertex in the newly loaded region. 
This forces us to load a second $\alpha$-neighbourhood; however, observe that 
the same sequence of I/Os is necessary even if we were not required to 
reconstruct $\augdual{t}$'s local topology. 
Thus, at no point during a traversal is it necessary to load a component 
merely in order to construct the topology of a triangle in an adjacent, or overlapping, 
component.

Due to Theorem~\ref{thm:planar_graph} and Lemma~\ref{lem:terrain_components} we
have the following theorem.

\begin{theorem}\label{thm:terrain_traversal}
Given triangulation $\triang$, where each point coordinate may be stored in 
$\bitsPerPoint$ bits, there is a data structure that represents $\triang$ in 
$N\bitsPerPoint + O(N) +o(N\bitsPerPoint)$ bits, that permits traversal of a path 
which crosses $K$ faces in $\triang$ with 
$O \left( \frac{K}{ \lg{B} } \right)$ I/O operations.
\end{theorem} 

For the case in which we wish to associate a $q$ bit key with each 
triangle we also have the following theorem.

\begin{theorem}\label{thm:terrain_keys_traversal}
Given triangulation $\triang$, where each point coordinate may be stored 
in $\bitsPerPoint$ bits, and where each triangle has associated with it a 
$\bitsPerKey$ bit key, there is a data structure that represents 
$\triang$ in 
$N(\bitsPerPoint + \bitsPerKey) + O(N) + o(N(\bitsPerPoint + \bitsPerKey))$ bits, that 
permits traversal of a path which crosses $K$ faces in $\triang$ 
with $O \left( \frac{K}{ \lg{B} } \right)$ I/O operations.
\end{theorem}

\begin{proof}
In our proof of Lemma \ref{lem:terrain_components}, we demonstrated 
that the number of triangles (dual nodes with an associated $\bitsPerKey$ 
bit key) and points (of $\bitsPerPoint$ bits) is within a constant factor 
of each other in the worst case. 
Assuming this worst case does occur, we are effectively assuming each 
triangle is associated with a $\bitsPerKey$ bit key in our data structures. 
This yields the desired space bound.
\end{proof}

\section{Point Location for Triangulations}\label{sec:point_location}

A common operation in any partition of space is point location, or more
specifically in $\plane$ planar point location.
In Section \ref{sec:applications}, we describe a number of queries
that require, as part of the query, identification of the triangle 
containing a point. 
For example, in order to perform a rectangular window query, we start by
identifying the triangle containing one of the rectangle's four
corner points.
In this section, we describe how to extend our data structure
for triangulations (see Section \ref{sec:tins}) in order to answer point 
location queries efficiently.

Our point location structure is based on that of 
Bose~\etal~\cite{DBLP:journals/talg/BoseCHMM12}.
One important difference between our technique and theirs is that we use
a different data structure as the basis for performing point location.
Their construction relies on the point location structure of 
Kirkpatrick~\cite{DBLP:journals/siamcomp/Kirkpatrick83} for which there is no
known external memory version.  
We use the structure of Arge~\etal~\cite{DBLP:conf/alenex/ArgeDT03},
which uses linear space, and answers vertical ray shooting queries in 
$\OhOf{\log_B N}$ I/Os.

We want to design a point location structure which, given a query point $p_q$, 
will return the label of the sub-region containing $p_q$. 
Since the subregion fits in internal memory, point location can be carried
out within the subregion at no additional I/O cost.
We use a two-level search structure.  
The first level allows us to locate the region containing $p_q$ while the 
second allows us to locate the subregion containing $p_q$.
As the structure is effectively the same at each level, we will only 
describe the structure for locating the region in detail.

Consider the set of region boundary vertices used to partition $\dual{\triang}$. 
These correspond to a subset of the triangles in $\triang$.
We select all the edges of this subset of the triangles of $\triang$,
and insert these into a search structure which we denote $\N$.
With each edge in $\N$, we associate the label of the region
immediately below it.
If the triangle below an edge is shared by more than one region, as
will often be the case, we can arbitrarily select one. 
Our search structure $\N$ is built using the persistent B-tree structure of 
Arge~\etal~\cite{DBLP:conf/alenex/ArgeDT03} which answers vertical
ray shooting queries.
We repeat this process for each region, creating a search structure
$\N_i$ for each region $\reg[i]$.

In order to determine to which region the query point $p_q$ belongs, we perform a vertical 
ray shooting query from $p_q$ to report the first edge encountered in $\N$
above that point. 
Since we associate with each edge the region below it, the result of this query 
yields the region, $\reg[i]$ containing $p_q$.
We then search in $\N_i$ to locate the sub-region $j$ containing $p_q$.
Finally, we perform an exhaustive search on the triangles of $\subreg[i,j]$
to locate the triangle containing $p_q$.
It is possible, in the event that $p_q \notin \triang$, that the search yields
no result at the cost of searching $\N$, $\N_i$, plus the
the exhaustive search of $\subreg[i,j]$.
This is of course the same cost incurred by a successful search in the worst
case.

There is one important addition yet to make to the search structure $\N$.
Consider the situation depicted in Fig.~\ref{fig:pl_structure}(a), where
we wish to find the triangles containing points $p$ and $q$. 
The ray-shooting query from $p$ correctly returns an edge of the separator
vertically above it; however, the ray starting from $q$ never intersects a 
segment from the separator.

\begin{figure}
  \includegraphics[width=\textwidth]{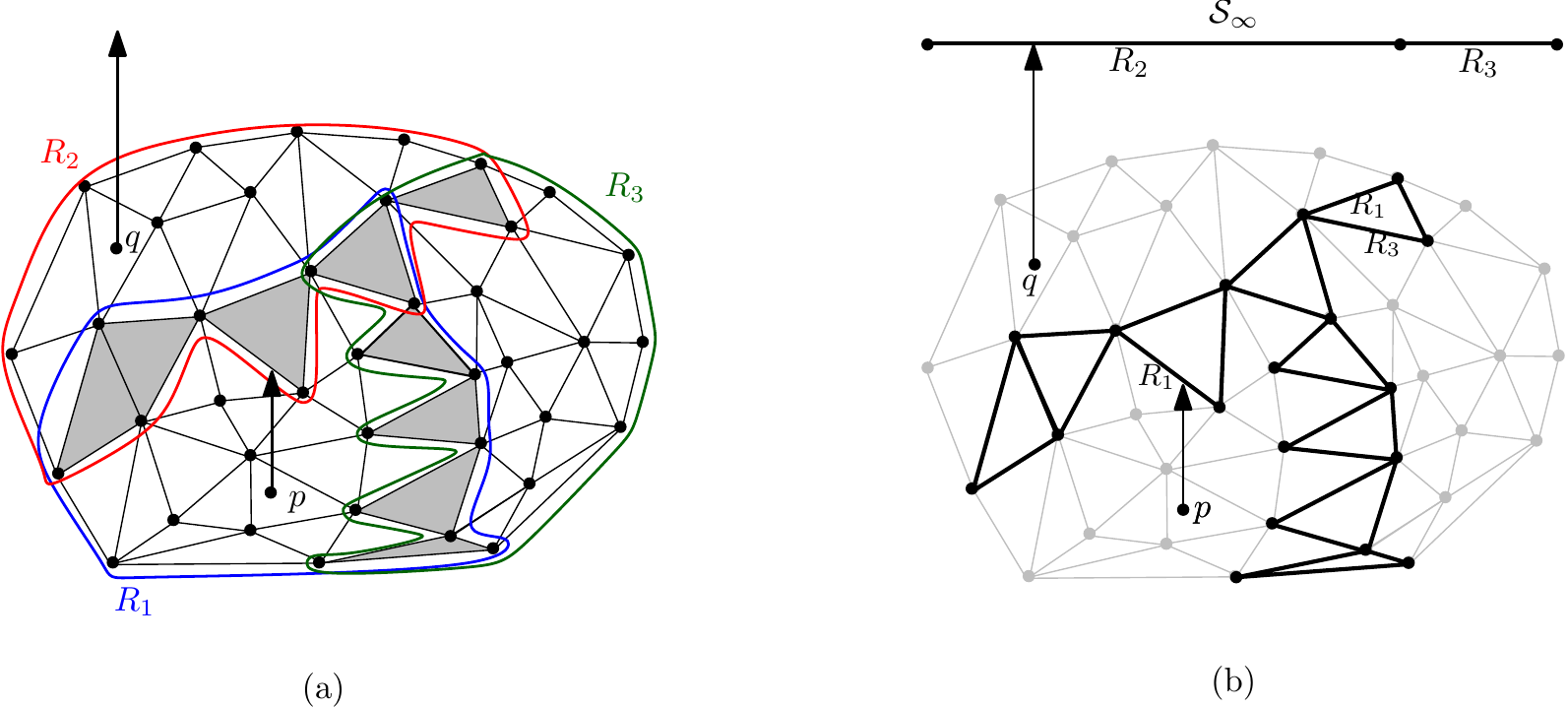}
  \caption[Point location by vertical ray shooting]{In (a), a 
    partitioned triangulation is depicted along with two vertical
    ray shooting queries $p$ and $q$.
    Separator triangles between regions are shaded.
    Recall that the separator is calculated on the dual graph, thus
    the jagged appearance of the separator when drawn in the 
    triangulation.
    Inserting just the edges of boundary triangles into $\N$ is 
    sufficient to handle query $p$ but query $q$ fails.
    In (b), the complete search structure is shown.
    The segments added to $\N$ (drawn as black and heavy) 
    include all boundary triangle edges and the segments generated
    by partitioning $S_{\infty}$.
    Region labels for a small subset of $\N$ are indicated just below
    their respective segments.
    Note that any query falling within $\triang$, and even some outside,
    will return a region to search.
  }
  \label{fig:pl_structure}
\end{figure}
              
In order to ensure that the ray shooting query always reports a valid region, we 
extend our search structure $\N$ as follows.
Let $S_{\infty}$ be a line segment with endpoints at the min and max 
$x$ coordinates of $\triang$, and with $y$ coordinate $+\infty$.
We scan $\triang$ from left to right and split $S_{\infty}$ at any
point where the region immediately below the upper hull changes,
and record for each such segment the region immediately below.
If we insert all the newly-generated segments in $\N$, all vertical ray 
shooting queries within $\triang$ now return a valid result.

\begin{lemma}\label{lem:size_s_infty}
The number of segments added to $\N$ by partitioning $S_{\infty}$
is bounded by the size of the region separator for the subdivision
at the region level. Likewise, the number of segments added 
to all $\N_i$ is bounded by the size of the subregion separators.
\end{lemma}

\begin{proof}
We begin with proving this claim for $\N$.
By definition, $S_{\infty}$ is only split when the region immediately
below the upper hull of $\triang$ changes.
Any such change must occur at a vertex in $\triang$.
Let $v$ be this vertex and an let $t_i$ be a triangle in $\reg[i]$ 
immediately to its left, and $t_j$ a triangle in $\reg[j]$ immediately
to its right, where $t_i$ and $t_j$ are on the upper hull of
$\triang$ (see Fig.~\ref{fig:pl_s_infty}).
Assume for the sake of contradiction that none of the triangles
adjacent to $v$ is a boundary triangle. 
In our construction, no two edge adjacent triangles in $\triang$
can be in different regions unless at least one of them is
a boundary triangle.
Since $t_i$ and $t_j$ are in separate regions if we walk 
from $t_i$ to $t_j$ around $v$, we must arrive at two 
triangles which do not have the same region, since the
region label changes at some point along this walk.
This contradicts the claim that no triangle adjacent to $v$ is 
a boundary triangle.

The proof on the bound for the sizes of each $\N_i$ is 
analogous to that for $\N$.
There is one matter that must be addressed;
the separator algorithm of \cite{Frederickson87} does
not guarantee contiguous regions in $\triang$.
Within a region, this may introduce discontinuities 
that increase the complexity of $S_{\infty}$.
However, the number of such discontinuities is in no case
worse than the size of the region separator over
all $\N_i$. 

\end{proof}

\begin{figure}
  \centering
  \includegraphics[scale=1.0]{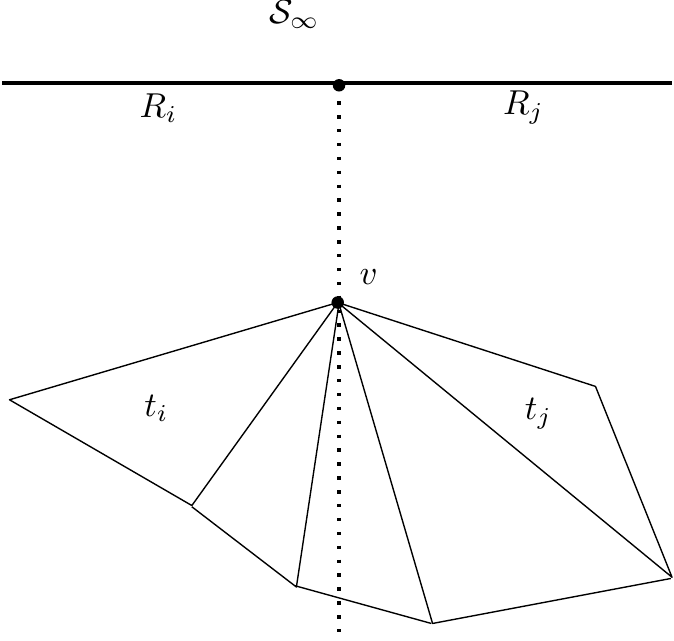}
  \caption[Bound on size of $\mathcal{S}_{\infty}$]{
  Illustration of proof that region boundary triangle must
  exists at a break in $\mathcal{S}_{\infty}$.
  }
  \label{fig:pl_s_infty}
\end{figure}

\begin{lemma}\label{lem:point_location_space}
By augmenting our terrain data structure with a structure using an additional 
$o(N\bitsPerPoint)$ bits, point location queries can be answered in 
$\OhOf{\log_{\succblksize}{N}}$ I/Os.
\end{lemma}

\begin{proof}
In the top level, region finding, data structure there are 
$\BigOh{ \frac{N}{\sqrt{\succblksize \lg^3 N}} }$ region boundary vertices. 
For each such vertex we insert at most six edges into our search structure;
this includes three edges for each triangle in the separator, plus any edges
added to $S_{\infty}$ by Lemma~\ref{lem:size_s_infty}.
Associated with each edge we must store $2\bitsPerPoint$ bits for the endpoints, 
plus a region label requiring $\lg{ \left( \frac{N}{\succblksize \lg^3 N} \right) }$ 
bits. 
Thus the total space, in bits, required by this structure will be:

\begin{eqnarray}
O\left( \frac{N}{\sqrt{\succblksize \lg^3{N}}} \right) \cdot \left( 2\bitsPerPoint + 
\lg{\left( \frac{N}{\succblksize \lg^3 N} \right)} \right) &=&  
O\left( \frac{N\bitsPerPoint}{\sqrt{\succblksize \lg^3 N}} \right) + 
\OhOf{\frac{N}{\sqrt{\succblksize \lg^3 N}}} \cdot 
\lg{ \left( \frac{N}{\succblksize \lg^3 N} \right) } \nonumber \\
	&=& \ohOf{N\bitsPerPoint} + \ohOf{N} \label{eqn:pl_region_space}
\end{eqnarray}

For the subregion search structures we will consider their space cumulatively.
We add $\OhOf{N/\sqrt{\succblksize}}$ edges, each requiring $2\bitsPerPoint$ bits 
for the endpoints,
and $\lg{( \lg^3 N )}$ bit references to the $\lg^3{N}$ sub-regions within the region. 
Again, by Lemma~\ref{lem:size_s_infty}, the addition of edges in $S_{\infty}$ does not 
asymptotically increase the size of the $\N_i$.
The space required for this structure is then:

\begin{equation}\label{eqn:pl_subregion_space}
\BigOh{\frac{N}{\sqrt{\succblksize}}} \cdot ( 2\bitsPerPoint + \lg{(\lg^3{N})} ) = 
\BigOh{ \frac{N\bitsPerPoint}{\sqrt{\succblksize}}} + \BigOh{ \frac{N}{\sqrt{\succblksize}} \cdot \lg{\lg{(N)}}}
\end{equation}

The first term of this equation is clearly $\ohOf{ N \bitsPerPoint }$. 
Recall that $\succblksize = (B \lg N)/(c + \bitsPerPoint)$ and that $B = \OmegaOf{\lg{N}}$. 
Thus, for the second term of Eq.(\ref{eqn:pl_subregion_space}) we have:

\begin{eqnarray}
\BigOh{ \frac{N}{\sqrt{\succblksize}} \cdot 3 \lg{\lg{N}} } &=& 
\BigOh{ \frac{N}{\sqrt{\frac{B \lg N}{c+\bitsPerPoint}}} \cdot \lg{\lg{(N)}} } \nonumber\\
&=& \ohOf{N\bitsPerPoint}\label{eqn:pl_subregion_space_final}
\end{eqnarray} 

The total space we use is thus bounded by $\ohOf{ N\bitsPerPoint}$ bits. 
The I/O complexity stems from the fact that we perform two point location 
queries on data structures that answer the query in $\OhOf{\log_B N}$ I/Os.
\end{proof}

\begin{theorem}\label{thm:terrain_with_point_location}
Given a terrain $\triang$, where each point coordinate may be stored in $\bitsPerPoint$ 
bits, there is a data structure that represents $\triang$ in 
$N\bitsPerPoint + O(N) +o(N\bitsPerPoint)$ bits that permits traversal of a path crossing 
$K$ faces in $\triang$ with $O \left( \frac{K}{ \lg{B} } \right)$ I/Os, and which 
supports point location queries with $\OhOf{\log_B{N}}$ I/Os.
\end{theorem} 

\begin{proof}
The terrain data structure requires $N\bitsPerPoint + O(N) +o(N\bitsPerPoint)$ bits by 
Theorem \ref{thm:terrain_traversal}. 
By Lemma \ref{lem:point_location_space}, adding point location requires only 
$\ohOf{N\bitsPerPoint}$ bits, thus the overall space bound is the same and queries 
can be answered in $\OhOf{\log_B{N}}$ I/Os.
\end{proof}

  \section{Applications}
  \label{sec:applications}

  In this section, we present a number of applications involving path traversal
  in triangulations.
  These applications employ the data structures described in Sections \ref{sec:tins}
  and \ref{sec:point_location}.
  The applications, in particular those described in 
  Section~\ref{ssec:ter_profile_tpath}, deal with triangulations used to model 
  terrains.
  A terrain model is a digital model of some surface.
  Most often the surface being modeled is some region of the earth, but any surface
  with relief can be modeled.
  Triangulations, commonly refered to as Triangular Irregular Networks, or TINs 
  in this context, are one means of modeling a digital terrain.
  In a TIN, the point set $\pointset$ is a set of points of known elevation. 
  Given a triangulation, $\triang$, built on $\pointset$, the elevation for any
  point, say $q$, on the surface can be readily estimated from the vertices
  of the triangle containing $q$.
  In addition to $x$ and $y$ coordinates, each point in $\pointset$ stores a $z$
  coordinate which records the elevation.
  Such models are sometimes refered to as 2.5 dimensional, since they cannot 
  properly represent a truly three-dimensional surface (for example there is no
  way for a planar triangulation to represent an overhanging cliff). 

  We begin by describing two simple and closely related queries, 
  reporting terrain profiles, and trickle paths, in 
  Section~\ref{ssec:ter_profile_tpath}.  
  In Section~\ref{sec:con_comp_queries}, we present a slightly more complex 
  application of our data structures, 
  reporting connected components of a triangulation (or terrain). 
  We start with the simpler case of reporting a component which is a convex
  subregion of the triangulation, and then describe methods for the more 
  general case where the connected component may represent a non-convex region.
  As all the queries we describe in this section require that a starting 
  triangle in the mesh be identified, we will assume that this triangle
  is given as part of the query.
  In practice, most of the queries we describe here can use a start triangle,
  identified by answering a
  single planar point location query using the technique outlined in 
  Section~\ref{sec:point_location}.

  \subsection{Terrain Profiles and Trickle Paths}\label{ssec:ter_profile_tpath}

  Terrain profiles are a common tool in GIS visualization. 
  The input is a line segment, or chain of line segments possibly forming a polygon, 
  and the output is a profile of the elevation along the line segment(s). 
  The trickle path, or path of steepest descent, from a point $p$ is the path on
  the terrain, represented by $\triang$, 
  that begins at $p$ and follows the direction of steepest descent until it reaches 
  a local minimum or the boundary of $\triang$ \cite{DBLP:conf/cccg/BergBDKOGRSY96}. 
  Both queries involve traversing a path over $\triang$, with the fundamental 
  difference being that in the terrain profile the path is given, whereas in 
  reporting the trickle path, the path is unknown beforehand and must be 
  determined based on local terrain characteristics.

  In analyzing these algorithms, we measure the complexity of a path by length 
  of the sequence of triangles visited, which we denote by $K$. 
  When a path intersects a vertex, we consider all triangles adjacent to that 
  vertex to have been visited. 
  Given this definition, we have the following result for terrain profile and 
  trickle path queries:

  \begin{lemma}\label{lem:tprofile_tpath}
  Let $\triang$ be a terrain stored using the representation described
  in Theorem~\ref{thm:terrain_traversal}, then:
  \begin{enumerate}
  \item{ Given a chain of line segments, $S$, the profile of the intersection 
  of $S$ with $\triang$ can be reported with $\BigOh{ \frac{K}{\lg{B}} }$ I/Os.}
  \item{ Given a point $p$ the trickle path from $p$ can be reported with 
  $\BigOh{ \frac{K}{\lg{B}} }$ I/Os.}
  \end{enumerate}
  \end{lemma}

  \begin{proof}
  Let $S$ be a chain of $i$ segments, denoted $s_0, s_1, \ldots s_i$. 
  In order to report an elevation profile, start at the endpoint of $s_0$, and 
  let $t \in \triang$ be the triangle which contains this endpoint. 
  We calculate the intersection of $s_0$ with the boundary of $t$ in order to determine 
  which triangle to visit next. 
  If at any point in reporting the query the next endpoint of the current segment 
  $s_j$ falls within the current triangle, we advance to the next segment 
  $s_{j+1}$. 
  This procedure is repeated until the closing endpoint of $s_i$ is reached.
  This query requires walking a path through $T$ that visits exactly $K$ triangles.

  For the trickle path, we are given triangle $t$ and some point interior to $t$. 
  The trickle path from $p$, and its intersection with the boundary of $t$, 
  can be calculated from the coordinates of the point vertices adjacent to $t$. 
  If the path crosses only the faces of triangles, we can simply report the 
  intersection of the path with those triangles.
  By Theorem \ref{thm:terrain_traversal}, it 
  requires $O(K/\lg{B})$ I/Os to report a path crossing $K$ triangles. 
  When an edge is followed, visiting both faces adjacent to that edge no more 
  than doubles the number of triangles visited. 
  The only possible problem occurs when the path intersects a point vertex. 
  Such cases require a walk around the point vertex in order to determine 
  through which triangle 
  (or edge) the path exits, or if the point vertex is a local minima. 
  In this case, the path must visit each triangle adjacent to the point vertex 
  through which it passes, thus all triangles visited during the cycle around a 
  point vertex are accounted for in the path length $K$.
  \end{proof}

  \subsection{Connected Component Queries}\label{sec:con_comp_queries}

  In this section, we describe how connected component queries can be reported 
  using our data structures. 
  We begin by defining \emph{connected component} and \emph{connected
  component query} in this setting.
  Let $\triang$ be a triangulation, and let $t$ be a triangle in $\triang$. 
  We denote by $\mathcal{P}(t)$ some property that is true for $t$.
  This property may be some value stored for each triangle, or some value that
  can be computed locally, such as slope or aspect if $\triang$ is a terrain.
  The connected component of $t$ with respect to $\mathcal{P}$ is the 
  subset $\concomp \subset \triang$, such that following conditions hold for
  all $t_c \in \concomp$:

  \begin{enumerate}
    \item $\mathcal{P}(t_c)$ is true, and
    \item there exists a path $t_1, t_2, \ldots t_n$ where $t = t_1$ and 
    $t_c = t_n$ such that $\mathcal{P}(t_i)$ is true for all $i$, and $t_i$ 
    and $t_{i+1}$ are adjacent.
  \end{enumerate}

  In order to simplify the discussion that follows, we adopt two conventions that we use
  throughout this section. 
  Firstly, while we assume that $\triang$ is represented by its augmented dual
  graph $\augdual{\triang}$, we will describe the various algorithms in terms
  of operations on the dual graph $\dual{\triang}$. 
  Furthermore, when we can do so without confusion, we use a single identifier
  to represent both a triangle and its corresponding node in the dual.
  For example, if we have triangle $t \in \triang$ corresponding to  
  $\dual{t} \in \dual{\triang}$, we may use $t$ in reference to both the
  triangle and the dual node.

  The second convention we will adopt relates to $\mathcal{P}$.
  If $\mathcal{P}(t)$ is true, then we call $t$ a \emph{red} node (triangle).  
  If $\mathcal{P}(t)$ is false, then we call $t$ a \emph{black} node (triangle).



  In Section~\ref{ssec:conv_con_comp} we address queries 
  where the component being reported is convex. 
  This constraint is not as limiting as it may seem.
  For example, rectangular window queries, a very common type of query, 
  may be viewed as a specialized case of reporting a convex 
  component.
  In Section~\ref{ssec:gen_con_comp}, we deal with the more challenging 
  problem of reporting
  connected components where the component may be non-convex.

  \subsubsection{Convex Connected Components}\label{ssec:conv_con_comp}

  Before describing the reporting of convex connected components, we show how 
  to solve a related problem. 
  Given the dual of a convex triangulation, $\dual{\triang}$, and a 
  vertex $\dual{s} \in \dual{\triang}$, 
  perform a depth-first traversal of $\dual{\triang}$, starting at $s$, without 
  using mark bits.
  Mark bits are commonly used in depth-first traversal to record which vertices 
  have already been visited. 
  However, their use is infeasible with our data structures, due to the 
  duplication of the vertices of $\dual{\triang}$.
  Mark bits would require that we find and mark all duplicates of 
  a vertex whenever it is visited.
  This entirely defeats the purpose of our data structure.

  In order to avoid the use of mark bits, we select from among the edges incident 
  to each node a single \emph{entry} edge for that node. 
  While performing a depth-first traversal, we extend the traversal 
  to a node only along its entry edge.
  This leaves us with the problem of determining how to select the entry edge
  for a node.  
  Gold and Maydell~\cite{gold_maydell_1978} and Gold and 
  Cormack~\cite{gold_cormack86} demonstrated how such an entry edge could 
  be selected to enable depth-first traversal on a triangulation.
  Their approach identified one of the three edges of a triangle as its 
  entry edge.
  Since there is a one-to-one correspondence between the edges in $\triang$ 
  and $\dual{\triang}$, we can apply such a rule to identifying entry
  edges for nodes in $\dual{\triang}$. 
  De~Berg~\etal~\cite{DBLP:journals/gis/BergKOO97} showed that the same technique can 
  be applied to the more general case of planar subdivisions, and 
  gave an entry-edge selection rule for polygons that, naturally, also works 
  for triangles.
  De~Berg's selection rule operates as follows\footnote{We 
  omit some minor details that are not relevant to triangles.}:

  \begin{enumerate}
  \item Let $t_s$ be a triangle in $\triang$, and let $s$ be a point interior
  to $t_s$.
  \item Let $t$ be the triangle for which we wish to identify an entry edge.
  Calculate the distance from $s$ to the closures of all edges of $t$.
  Let $s'$ be the point on $t$ that minimizes $\texttt{dist}(s,s')$. 
  \item If $s'$ is interior to an edge of $t$; select this edge as the
  entry edge, otherwise, $s'$ must be a vertex of $t$. In this case, let
  $e$ and $e'$ be the edges adjacent to $s'$. If $e$ is \emph{exposed}
  to $s$, such that the line $\vec{\ell}$ induced by $e$ has
  $s$ strictly to its right, then select $e$ as the entry edge, otherwise
  select $e'$. 
  \end{enumerate}

  Now consider the graph $\dual{\triang}$.
  If we remove all edges from $\dual{\triang}$ that
  do not correspond to entry edges in $\triang$, we are left with a tree, 
  rooted at $t_s$ (see Figure \ref{fig:imp_tree_entry}).
  Proof of this claim can be found in~\cite{DBLP:journals/gis/BergKOO97}.
  The triangle $t_s$, which contains $s$, has no entry edges itself.

  \begin{figure}[th]
	  \centering
		  \includegraphics[width=0.6\textwidth]{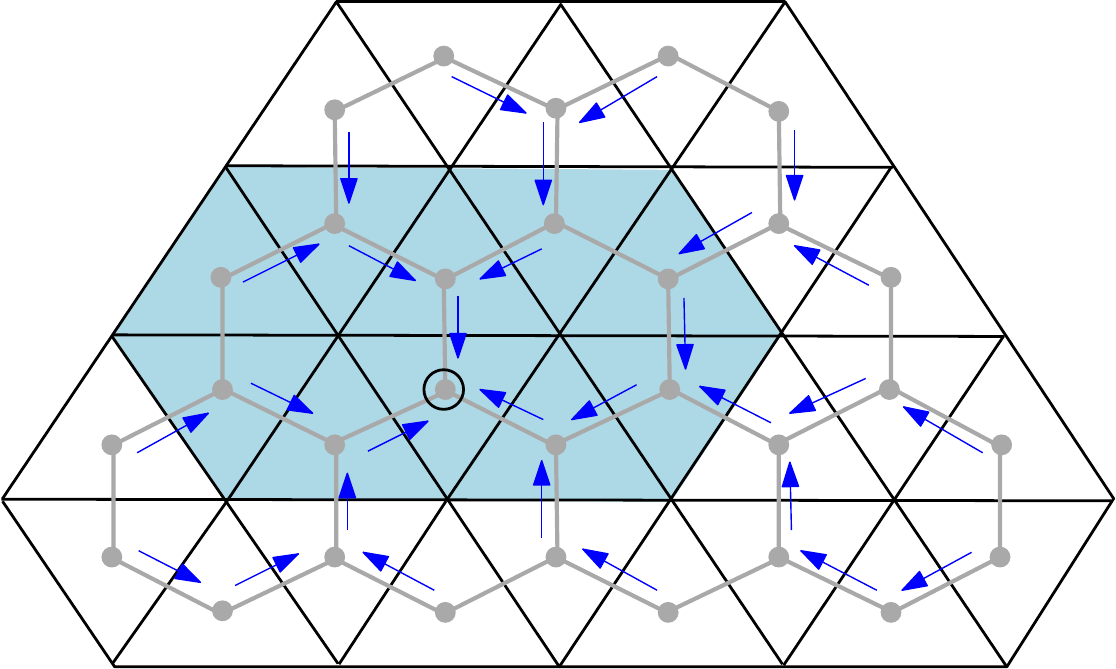}
	  \caption[Dual graph with non-entry edges removed forms a tree]{ 
	    A triangulation is shown with its dual graph. 
	    The blue arrows indicate the entry edges selected for each triangle when
	    triangle $t_s$ is selected as the root ($\dual{t}_s$ is circled). 
	    The arrows are directed from child to parent.
	    Removing the non-entry edges results in a tree in the dual, rooted at 
	    $\dual{t}_s$.
	    So long as a connected component is convex (e.g., the shaded area) this 
	    tree property holds for the subgraph consisting of only nodes in the
	    component.} \label{fig:imp_tree_entry}
  \end{figure}
	  
  \begin{figure}[th]
	  \centering 
	  \includegraphics[width=0.6\textwidth]{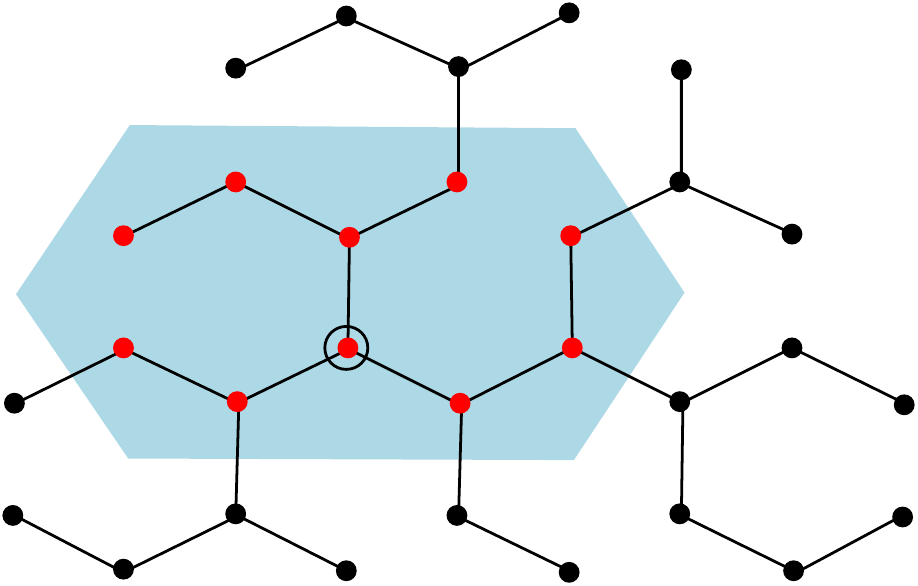}
	  \caption[Tree rooted at $t_s$]{ The tree rooted at $t_s$. 
	    The red nodes indicate nodes of the subtree contained within the 
	    connected component; this subtree is still connected even when all 
	    (black) nodes - which are not part of the component - 
	    are removed.} \label{fig:implicit_tree}
  \end{figure}

  This result leads to a simple traversal algorithm for connected 
  components, in the case where we know that the region to be reported 
  is convex.
  We summarize with the following lemma.

  \begin{lemma}\label{lem:report_convex_component}
  Let $\triang_C \subset \triang$ be a connected, convex component.
  Given $t_s \in \triang_C$ we can report all triangles in $\triang_C$ with 
  $\BigOh{\frac{|\triang_C|}{\lg{B}}}$ I/Os.
  \end{lemma}

  \begin{proof}
  Performing a depth-first search on a tree on $|\triang_C|$ vertices is equivalent 
  to walking a path of length at most $4 \cdot |\triang_C|$.
  By Theorem \ref{thm:terrain_traversal}, our data structures permit a walk along 
  such a path with a total I/O cost of $\BigOh{ \frac{|\triang_C|}{\lg{B}} }$.

  Let $s$ be some point interior to $\triang_C$, and let $t_s \in \triang_C$ be
  the triangle containing $s$.
  Starting with $t_s$, we implicitly construct a tree in $\dual{\triang}$ using the 
  selection method of~\cite{DBLP:journals/gis/BergKOO97}.
  Since $\triang_C$ is convex, removing all subtrees rooted at black nodes results 
  in an implicit tree, rooted at $t_s$, which includes all red nodes in $\triang_C$
  (see Figures \ref{fig:imp_tree_entry} and \ref{fig:implicit_tree}).
  It remains to be proven that the removal of all subtrees rooted at black nodes
  will not result in the removal of any red nodes from the tree.
  Assume that removing the subtree rooted at some black node, $b$ removes some
  red node $r$. 
  Let $q$ be the parent of $b$ in $\triang_C$. 
  The entry edge corresponding to dual edge $(q,b)$ must be on the convex
  hull of $\triang_C$.
  However, this means that $r$ must be outside the convex hull of $\triang_C$, which
  is a contradition.
  Therefore, no such triangle (node) $r$ can be removed.
  \end{proof}
  
  We give the following corollary to Lemma \ref{lem:report_convex_component}
  with application to rectangular window queries.
  Such queries are a specialization of reporting a convex region 
  and occur commonly in practice.

  \begin{corollary}\label{cor:window_query}
  Given a terrain $\triang$, and a query window $W$, the set of 
  triangles which intersect the query window can be 
  reported with $O \left( \frac{|\triang_W|}{\lg{B}} \right)$ I/Os. 
  \end{corollary}

  \begin{proof}
  The query window problem is equivalent to reporting 
  a connected component where the selection property is that a triangle 
  intersects the query window $W$. 
  The region to be reported is convex, but care must be taken with 
  triangles that intersect the query window boundary. 
  In \cite{DBLP:journals/gis/BergKOO97} it is shown that the entry edge selection 
  rule can be modified to disqualify any edge (or portion thereof) that 
  is outside $W$. 
  Using this modification, an implicit rooted tree is still formed on 
  the triangles of $\triang_W$, and we can perform a memoryless 
  depth-first traversal.
  \end{proof}

  \subsubsection{General Connected Components}\label{ssec:gen_con_comp}

  In this section, we present an algorithm that reports connected components 
  that may be non-convex and/or contain holes.
  Using the terminology developed in Section~\ref{ssec:conv_con_comp}, we 
  say $\concomp$ is a connected component in $\triang$, corresponding 
  to $\dual{\concomp}$ in $\dual{\triang}$. 
  Recall that $\dual{\concomp}$ forms a connected component of red nodes 
  in the subgraph formed by all red nodes in $\dual{\triang}$.
  In $\triang$, we can think of the boundary of a connected component being 
  formed by all edges along the outer perimeter of the component 
  in addition to any edges along the perimeters of holes. 
  In the dual, the boundary corresponds to the set of edges connecting 
  red vertices in $\dual{\concomp}$, with black vertices in $\dual{\triang}$. 

  In order to report $\concomp$, we select some triangle $s \in \concomp$ as our
  start triangle. 
  We select entry edges relative to an arbitrary point interior to $s$.
  Since $\triang$ itself is convex, we can construct an implicit tree
  on the nodes of $\dual{\triang}$, rooted at $s$.
  We denote this tree $\dual{\triang}_T$.
  
  Let $r$ and $b$ be red and black vertices in
  $\dual{\concomp}$, corresponding to adjacent triangles in $\triang$.
  Let $e$ be the edge of $\triang$ separating triangles $r$ and $b$.
  In $\dual{\triang}$, $e$ corresponds to the edge $(r,b)$.
   While this edge is undirected in $\dual{\triang}$, for the sake of the
  discussion that follows we will assume that in $\dual{\triang}_T$, it
  is directed towards the root, $s$.
  We call such edges \emph{boundary} edges, and classify them according
  to their role in the implicit tree $\dual{\triang}_T$, as follows 
  (see Fig.\ref{fig:imp_tree_non-convex}):
  
  \begin{enumerate}
    \item if $(r,b) \notin \dual{\triang}_T$, then $e$ is a \emph{wall}\footnote{
    All edges on the convex hull of $\triang$ are considered wall edges.} edge,
    \item if $r$ is the parent of $b$, then $e$ is an \emph{access} edge, and 
    \item if $b$ is the parent of $r$, then $e$ is an \emph{exit} edge.
  \end{enumerate}

  \begin{figure}[th]
	  \centering
		  \includegraphics[width=0.6\textwidth]{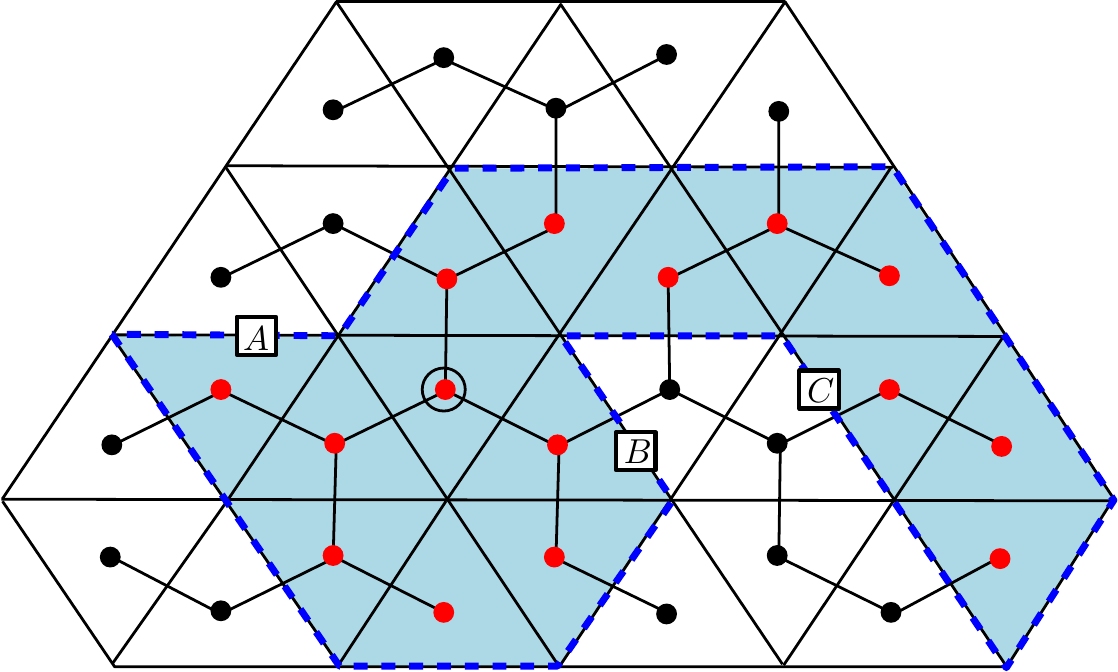}
	  \caption[Boundary edge definitions]{A triangulation and the implicit tree 
	  formed by the entry edges selected for the circled vertex $\dual{t}_s$. 
	  In this case, the connected component of $t_s$ is non-convex so that the tree 
	  formed by only red nodes is disconnected. 
	  The boundary edges of the triangulation are indicated by the heavy dashed 
	  blue line. 
	  Using our notation for boundary edges, A is a \emph{wall} edge, B is an 
	  \emph{access} edge where the tree enters the component, and C is an 
	  \emph{exit} edge where the tree leaves the component. }
	  \label{fig:imp_tree_non-convex}
  \end{figure}

  We report the triangles in $\concomp$ using the algorithm 
  $\funccase{DepthFirstTraversal}$ listed in Fig. \ref{fig:traverse_graph_algorithm}.
  $\funccase{DepthFirstTraversal}$ performs a standard depth-first traversal in 
  $\dual{\triang}_T$, with one important modification. 
  If a branch of the algorithm's execution terminates at an access boundary 
  edge, then the function $\funccase{ScanBoundary}$ (Fig. \ref{fig:scanhandrail_alg}) 
  is invoked. 
  $\funccase{ScanBoundary}$ traverses the chain of boundary edges, and recursively calls 
  $\funccase{DepthFirstSearch}$ at each exit edge. 
  A search structure, $\mathcal{V}$, is maintained to ensure that no boundary edge is 
  scanned more than once.  
  Whenever an access edge is visited during the $\funccase{ScanBoundary}$ 
  or $\funccase{DepthFirstSearch}$ 
  processes, it is added to the search structure $\mathcal{V}$ if it is not already 
  present. 
  If an access edge is encountered that is already in $\mathcal{V}$, then execution 
  of $\funccase{ScanBoundary}$ is halted. 
  Likewise, when $\funccase{DepthFirstTraversal}$ encounters an access edge already in 
  $\mathcal{V}$, it does not invoke $\funccase{ScanBoundary}$. 
  Figure \ref{fig:cmp_with_holes} demonstrates the operation of
  $\funccase{DepthFirstSearch}$ and $\funccase{ScanBoundary}$.

  \begin{figure}[h]
  \protect \framebox[0.95\linewidth]
  {
  \begin{minipage}{0.90\linewidth}

  \textbf{Algorithm $\textsc{\funccase{DepthFirstTraversal}}(t_s,s)$} \\
  1. \hspace{3pt} $t \leftarrow t_s$ \\
  2. \hspace{3pt} \textbf{do} \\
  3. \hspace{18pt} \textbf{if}($t$ has unvisited children) \\
  4. \hspace{33pt} $c \leftarrow $ next unvisited child of $t$ \\
  5. \hspace{33pt} \textbf{if}($\mathcal{P}(c) \ne \mathcal{P}(t)$) \\
  6. \hspace{48pt} $e \leftarrow$ boundary edge corresponding to $(c,t)$ \\
  7. \hspace{48pt} \textbf{if}($e$ is an access edge \textbf{and} $e \notin \mathcal{V}$) \\
  8. \hspace{63pt} $\funccase{ScanBoundary(e, t)}$ \\
  9. \hspace{48pt} \textbf{endif} \\
  10. \hspace{30pt} \textbf{else} \\
  11. \hspace{45pt} $t \leftarrow c$ \\
  12. \hspace{30pt} \textbf{endif} \\
  13. \hspace{15pt} \textbf{else} \\
  14. \hspace{30pt} $t \leftarrow$ parent of $t$ \\
  15. \hspace{15pt} \textbf{endif} \\
  16. \textbf{until}($t$ has no more unvisited children \textbf{and} $t=t_s$) \\
  \end{minipage}
  }
  \caption{Algorithm $\funccase{DepthFirstTraversal}$. }
  \label{fig:traverse_graph_algorithm}
  \end{figure}

  \begin{figure}[h]
  \protect \framebox[0.95\linewidth]
  {
  \begin{minipage}{0.90\linewidth}

  \textbf{Algorithm $\textsc{\funccase{ScanBoundary}}(e,t,s)$} \\
  1. \hspace{3pt} $e' \leftarrow$ next boundary edge in counterclockwise direction from $e$ \\
  2. \hspace{3pt} \textbf{do} \\
  3. \hspace{18pt} \textbf{if}($e'$ is an access edge) \\
  4. \hspace{33pt} \textbf{if}($e' \in \mathcal{V}$) \\
  5. \hspace{48pt} \textbf{break} \\
  6. \hspace{33pt} \textbf{else} \\
  7. \hspace{48pt} add $e'$ to $\mathcal{V}$ \\
  8. \hspace{33pt} \textbf{endif} \\
  9. \hspace{18pt} \textbf{endif} \\
  10. \hspace{15pt} \textbf{if}($e'$ is an exit edge) \\
  11. \hspace{30pt} select triangle $t$ adjacent to $e'$ \\
  12. \hspace{30pt} $\funccase{DepthFirstTraversal(t,s)}$ \\
  13. \hspace{15pt} \textbf{endif} \\
  14. \textbf{until}(e' == e) \\
  \end{minipage}
  }
  \caption{Algorithm $\funccase{ScanBoundary}$. }
  \label{fig:scanhandrail_alg}
  \end{figure}

  \begin{figure}[th]
	  \centering
		  \includegraphics[width=0.8\textwidth]{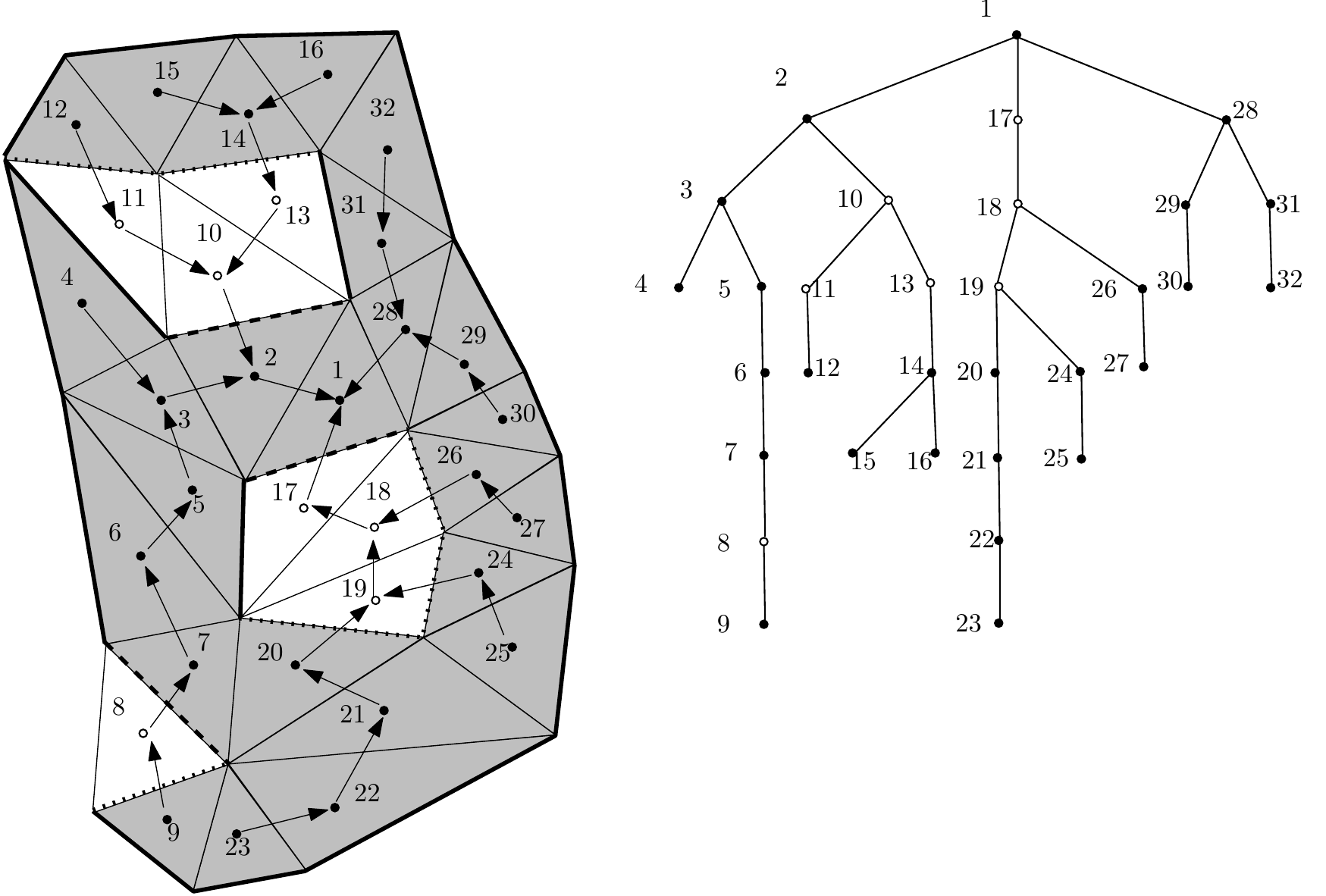}
	  \caption[Traversal example on non-convex region with holes]{On the left is shown a 
	  triangulation $\triang$ and connected component $\triang_C$ (shaded) 
	  with a non-convex boundary and holes. 
	  The original tree $G_T$ is shown on the right, with vertices labeled by preorder 
	  number for the entire triangulation $\triang$. 
	  Hallow vertices correspond to vertices removed from $G_T$ by holes and/or 
	  concavities in the boundary of $\triang_C$. 
	  The parent-child relationships, as well as the vertex labels are also shown on $\triang$. 
	  In $\triang$ boundary edges are thick lines, with entry edges being dashed 
	  and exit edges dotted. 
	  The original call to $\funccase{DepthFirstTraversal}$ first encounters the 
	  entry edge 
	  (8,7), from which $\funccase{ScanBoundary}$ will visit exit edges (9,8), 
	  (12,11), and 
	  (14,13), in that order, before terminating back at $(8,7)$. 
	  This initial call to $\funccase{ScanBoundary}$ results in subtrees rooted
	  at vertices 
	  $9$, $12$, and $14$ being reported. 
	  $\funccase{ScanBoundary}$ is not invoked from entry edge (10,2), as this 
	  edge is added 
	  to $\mathcal{V}$ during the invocation of $\funccase{ScanBoundary}$ from 
	  $(8,7)$. 
	  $\funccase{ScanBoundary}$ is, however, called from $(17,1)$, which results 
	  in entry 
	  edges $(26,18)$, $(24,19)$ and $(20,19)$ being visited and the subtrees 
	  rooted at vertices $26$, $24$, and $20$ being reported.
	  The dashed arrows represent cases where $\funccase{ScanBoundary}$ jumps 
	  between
	  \emph{access} edges and \emph{exit} edges.}
	  \label{fig:cmp_with_holes}
  \end{figure}

  \begin{lemma}\label{lem:cc_holes_alg_correctness}
  Given a triangle $t_s \in \concomp$, the algorithms 
  $\funccase{DepthFirstTraversal}$ and $\funccase{ScanBoundary}$ report all
  triangles in $\concomp$. 
  \end{lemma}

  \begin{proof}
  Let $s$ be the start triangle selected for $\funccase{DepthFirstTraversal}$,
  and let $t \in \concomp$ be a node in $\dual{\triang}$ not reported by 
  $\funccase{DepthFirstTraversal}$.
  Since $s$ is the root of $\dual{\triang}_T$, there is a path connecting $s$ 
  to $t$ in  $\dual{\triang}_T$. 
  If all nodes on the path from $s$ to $t$ in $\dual{\triang}_T$ are red, then
  clearly $\funccase{DepthFirstTraversal}$ reports $t$.
  Thus, there must be black nodes on the path $s \leadsto t$, corresponding to
  the path leaving $\concomp$.
  Let $w$ and $u$ be the first, and last, nodes encountered on the first such
  black subpath on $s \leadsto t$ (it may be the case that $w = u$).
  Let $w'$ be the red node preceeding $w$ on $s \leadsto t$, and let $u'$ be
  the red node following $u$ on the same path (see 
  Figure~\ref{fig:component_subpaths}).
  The edge $(w'w)$ is an exit edge, while $(u,u')$ is an access edge under
  our definitions. 
  The subpath $s \leadsto w'$ is wholly contained in $\concomp$, thus the call
  to $\funccase{DepthFirstTraversal}$ from $s$ reaches $w'$.
  If the access edge corresponding to $(w',w)$ has not 
  been previously visited, the $\funccase{ScanBoundary}$ algorithm is invoked.
  The call to $\funccase{ScanBoundary}$ visits the edge $(u,u')$, unless it
  is blocked when it encounters a previously visisted access edge in 
  $\mathcal{V}$.
  This may  occur if one of the recursive 
  calls to $\funccase{DepthFirstSearch}$ made at exit edges during 
  $\funccase{ScanBoundary}$ 
  encounters the same boundary. 
  However, if $\funccase{ScanBoundary}$ is blocked from visiting $(u,u')$ in 
  such a fashion, then $\funccase{ScanBoundary}$ must have been called
  from the blocking access edge. 
  Let $s_0, s_1, \ldots, s_i$ be a sequence of such blocking calls to 
  $\funccase{ScanBoundary}$ along the boundary chain between $(w',w)$ 
  and $(u,u')$. 
  Clearly, the last such call, $s_i$, will result in $(u,u')$ being visited. 
  This same argument can be applied to any other subpaths leaving 
  $\concomp$ on $s \leadsto t$, and as such $\dual{t}$ is visited by 
  $\funccase{DepthFirstTraversal}$.
  \end{proof}

  \begin{figure}[th]
	  \centering
		  \includegraphics[width=0.6\textwidth]{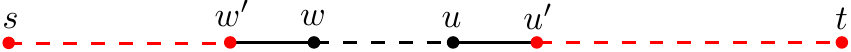}
	  \caption[General component subpaths]{The path through the tree 
	  $\dual{\triang}_T$ connecting nodes $s$ and $t$, which are both
	  part of the connected component $\concomp$. The path
	  leaves $\concomp$ at $(w',w)$ and re-enters at $(u,u')$.}
	  \label{fig:component_subpaths}
  \end{figure}
  
  The I/O cost of reporting a general connected component includes
  the cost to walk the triangles of $\concomp$, plus the cost
  of walking the boundary of $\concomp$.
  We denote by $h$ the size of the boundary, and summarize our I/O
  costs with the following lemma.

  \begin{lemma}\label{lem:io_eff_ccomp}
  The algorithm $\funccase{DepthFirstTraversal}$ reports a connected 
  component 
  $\concomp$ with $h$ boundary edges using 
  $O\left( \frac{|\concomp|}{\lg{\succblksize}} \right) + O(h \log_{B}{h})$ I/Os.
  \end{lemma} 

  \begin{proof}
  Performing depth-first traversal is equivalent to a walk of 
  length at most $4|\concomp|$, which can be performed in 
  $O\left( \frac{|\concomp|}{\lg{\succblksize}} \right)$ I/Os. 
  We must also account for the length of the paths traversed by calls 
  to $\funccase{ScanBoundary}$. 
  Any triangle in $\concomp$ may be visited at most three times, since 
  it may be adjacent to no more than three different boundary chains. 
  Thus, the total additional length of the walks associated with calls 
  to $\funccase{ScanBoundary}$ is bounded by $3|\concomp|$, which 
  increases the length of the path traversed by a constant factor. 

  We must also account for the number of I/Os required to maintain 
  and query the boundary edge structure $\mathcal{V}$. 
  There are $h$ boundary edges, and in the worst case most of these 
  edges may be access edges. 
  Using a B-tree to store $\mathcal{V}$ 
  supports insertions and queries in $O(\log_B{h})$ time.  
  An access edge is added to $\mathcal{V}$ only once, and is visited at 
  most one additional time. 
  Thus the total cost to maintain and query $\mathcal{V}$ is $O(h \log_B{h})$.  
  \end{proof}

  Finally, we account for the space used to store $\mathcal{V}$. 
  Since our triangulation does not store edges, we must have some
  way of uniquely identifying the edges in $\mathcal{V}$.
  We do so by storing the coordinates of the midpoint of each edge, which
  serves as our search key.
  For $\mathcal{V}$ we use a B-tree, and for each edge store a $\bitsPerPoint$-bit 
  key plus a $\lg{N}$ bit pointer. 
  Thus the space for this structure is $\OhOf{h \cdot(\bitsPerPoint + \lg{N})}$ 
  bits. 
  In theory, the size of this structure could be as large as $\concomp$,
  but in many scenarios it will be significantly smaller.  
  The following theorem summarizes our results 
  for convex and general connected components. 
  The space bound adds the space for $\mathcal{V}$ to the bound from
  Theorem \ref{thm:terrain_traversal}.
  The I/O bounds are obtained from Lemmas \ref{lem:report_convex_component} 
  and \ref{lem:io_eff_ccomp}.

  \begin{theorem}\label{thm:conn_comp}
  A triangulation $T$, with $\bitsPerKey$-bit keys per triangle, may be 
  stored using 
  $N(\bitsPerPoint + \bitsPerKey) + O(N) + o(N(\bitsPerPoint + \bitsPerKey)$ 
  bits such that a 
  connected component $\concomp$ may be reported using 
  $O \left( \frac{|\concomp|}{\lg{B}} \right)$ I/Os if $\concomp$
  is convex.  
  If $\concomp$ may be non-convex or have holes, then the query requires 
  $O\left( \frac{|\concomp|}{\lg{B}}  + h \log_B{h}\right)$ 
  I/Os, plus an additional $\OhOf{h \cdot(\bitsPerPoint + \lg{N})}$ bits of storage, 
  where $h$ is the number of boundary edges of 
  $\concomp$.
  \end{theorem}

  \subsection{Connected Components Without Additional Storage}
  \label{ssec:no_add_storage}

  One drawback with our technique for reporting connected components 
  in Section~\ref{ssec:gen_con_comp} is that 
  we must store the search structure $\mathcal{V}$ in order to complete the 
  traversal. 
  In this section, we present a revised version of the algorithm that removes 
  the need for this additional data structure, at the cost of performing 
  additional I/O operations.

  Our technique is based on the algorithm of Bose and 
  Morin~\cite{DBLP:conf/isaac/BoseM00}, for the more general case of subdivision 
  traversal in planar subdivisions. 
  That paper, in turn, is a refinement of the algorithm presented by 
  de~Berg~\etal~\cite{DBLP:journals/gis/BergKOO97}. 
  The strategy in both papers is to identify a single entry edge on each face. 
  When an edge is visited while reporting the edges of a face, a check is made 
  to determine if it is the (unique) entry edge for an adjacent face. 
  If the edge proves to be an entry edge, then the adjacent face is entered. 
  The resulting traversal is a depth-first traversal of the faces of the 
  subdivision. 
  In both papers, (\cite{DBLP:conf/isaac/BoseM00} and \cite{DBLP:journals/gis/BergKOO97}),
  the subdivision is assumed 
  to be represented as a doubly-connected edge list, or a similar 
  topological structure. 

  Bose and Morin~\cite{DBLP:conf/isaac/BoseM00} select entry edges based 
  on a total order $\preceq_p$ on the edges of the subdivision. 
  The position of each edge, $e$ in $\preceq_p$, is determined based on 
  the edge's key. 
  The key is a 4-tuple of properties that can be calculated locally for 
  each edge, based on the edge's geometry. 
  As with our previous entry edge selection rule, the keys are calculated 
  with respect to a known point interior to the start triangle.
  To determine if an edge is the entry point of a face, the following 
  \emph{both-ways} search is performed.
  Starting at edge $e$, the face is scanned in both directions 
  until either (a) an edge $e'$ is encountered with 
  a lower key than $e$ in $\preceq_p$, or (b) the scans meet without 
  having found any such edge.  
  In case (b), the edge $e$ is then selected as the entry edge for 
  the face.

  The main result of Bose and Morin is that for a subdivision 
  with $N$ vertices, all faces (including edges and vertices) can be 
  reported in $\OhOf{N \log{N}}$ steps. 
  Their approach can also be applied to reporting a connected 
  component of the subdivision in $\OhOf{h \log{h}}$ time, where the 
  $h$ is the number of vertices in the component.  

  In this chapter, all faces are triangles, thus to find the entry edge
  at the triangle level is a constant-time operation using any
  selection technique.
  Where the 'search both ways' technique proves useful in our setting 
  is in dealing with the boundary of the component. 
  A hole may consist of 
  one or many faces (triangles), but we treat a hole as if it were a 
  single face.
  For each hole, we want to identify a unique entry edge from among
  the access edges on its boundary.
  Since edges are not explicitly stored in our construction, walking
  the boundary involves walking the set of triangles that touch the boundary. 
  This includes all triangles in $\concomp$ that:
  
  \begin{enumerate}
  \item have an edge on the boundary, or
  \item have an adjacent vertex, $p \in P$, that lies on the boundary.
  \end{enumerate}
  
  Let $h'$ denote the number of triangles touching the boundary of
  $\concomp$. 
  This includes both holes in $\concomp$ and its outer boundary.
  A triangle can be adjacent to the boundary at no more than three 
  points (or edges), therefore 
  $h' = \OhOf{|\concomp|}$. 
  
  Assume $H$ is some hole in our component. 
  In order to apply the analysis of Bose and Morin directly to our results, 
  we conceptually add zero length
  \emph{pseudo-edges} to $H$ at any point on the boundary of $H$ that 
  is adjacent to a triangle $t$ which touches $H$ at a point, but which 
  does not share an edge with $H$ (see Fig. \ref{fig:pseudo-face} ). 
  With respect to the key values in $\preceq_p$, we set the value of 
  a key for a pseudo-edge to $\infty$. 
  Since the entry edge in any hole is the edge of minimum key value, 
  no pseudo-edge will ever be selected. 
  Given this definition of a pseudo-edge, the value $h'$ can also 
  be considered the sum of real and pseudo-edges over all holes 
  and the exterior boundary of the component $\concomp$.

  \begin{figure}[th]
	  \centering
		  \includegraphics[width=0.8\textwidth]{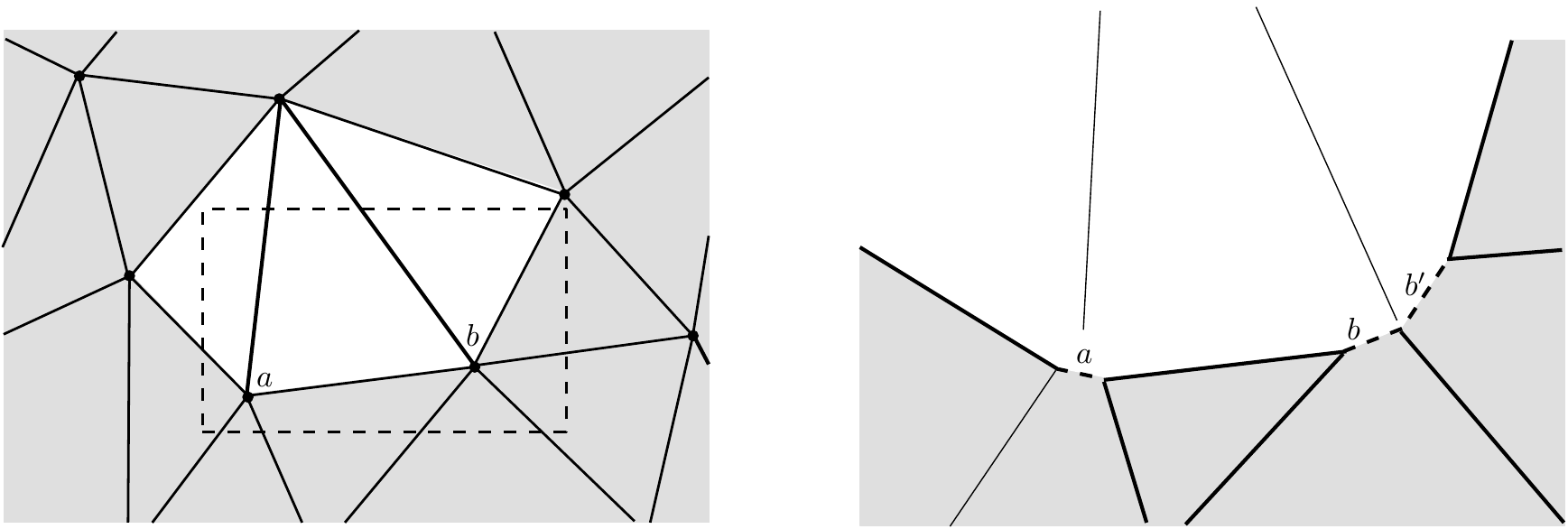}
	  \caption[Pseudo-edges in region boundaries]{The figure on the left 
    shows a portion of a component (gray) with a hole (white). 
	  The dashed box indicates the detailed area shown in the figure on the
	  right. 
	  The right-hand figure shows conceptually how additional edges are 
	  added to the hole boundary, corresponding to triangles in the component.
	  At the point marked $a$, a single pseudo-edge is added as there is 
	  only one non-edge adjacent triangle adjacent to this point. 
	  The point $b$ has two non-edge adjacent triangles, therefore two pseudo-edges 
	  $b$ and $b'$ are added to the hole boundary. }
	  \label{fig:pseudo-face}
  \end{figure}

  The following lemma summarizes results for Bose and Morin that can now be 
  applied directly to our setting.

  \begin{lemma}\label{lem:bose_morin_time}
  For a connected component $\concomp$ requiring $h'$ triangles to be visited 
  in order to walk the boundary of all holes plus the exterior boundary of 
  $\concomp$, the both-ways search technique can find entry edges for all 
  holes (and the perimeter) in at most $\OhOf{h' \log{h'}}$ steps, 
  or $\OhOf{h' \log_{\succblksize}{h'}}$ I/Os. 
  \end{lemma}

  \begin{proof}
  Theorem 1 in Bose and Morin \cite{DBLP:conf/isaac/BoseM00} states that 
  a planar subdivision of faces with $n$ vertices can be traversed in 
  $\OhOf{n \log{n}}$ time. 
  By adding zero-length pseudo-edges, we effectively make 
  the set of holes and the exterior boundary equivalent to faces with 
  a total of $h'$ edges. 
  Thus, using the both-ways search, we can locate entry edges with 
  $\OhOf{h' \log{h'}}$ steps. 
  Since we can travel $\OhOf{\log{\succblksize}}$ steps during such 
  searches before incurring an I/O, we can perform all such searches in 
  $\OhOf{h' \log_{\succblksize}{h'}}$ I/Os.
  \end{proof}

  Applying the both-ways search technique presented above requires only 
  minor modifications to our algorithms. 
  In the $DepthFirstTraversal$ algorithm (Fig. 
  \ref{fig:traverse_graph_algorithm}) at line 7, rather than check 
  if $e \in \mathcal{V}$, we perform the both-ways search to determine 
  if $e$ is the unique entry edge. 
  If this is true, we then perform $ScanBoundary$ starting with $e$. 
  The only alteration to the $ScanBoundary$ algorithm 
  (Fig. \ref{fig:scanhandrail_alg}) is that we can omit steps 3 
  through 9, since following the both-ways search we know that $e$ 
  is the  unique entry edge for the boundary or hole.  

  To summarize we have the following theorem.

  \begin{theorem}\label{thm:conn_comp_without_add_storage}
  A triangulation $\triang$, with $\bitsPerKey$-bit keys per triangle, 
  may be stored using $N\bitsPerKey + O(N) + o(N\bitsPerKey)$ bits such 
  that a connected component $\concomp$ may be reported using 
  $\BigOh{\frac{|\concomp|}{\lg{B}} + h' \log_B{h'}}$ I/Os, where 
  $h'$ is the total number of triangles that touch all holes in, 
  plus the boundary of, $\concomp$.
  \end{theorem}
  
  \section{Conclusions}\label{sec:graph_conclusions}
  
  In this chapter we have presented succinct structures supporting
  I/O efficient path traversal in planar graphs.
  In fact, our results are somewhat more general than this.
  The key features of planar graphs which our data structures rely
  on are; first, planar graphs are $k$-page embeddable, and second, 
  planar graphs of bounded degree have an $r$-partition with
  a suitably small separator. 
  Any class of graphs for which these two properties hold should 
  permit construction of the succinct data structures.

\chapter[Path Traversal in Well-Shaped Meshes]{I/O Efficient Path Traversal in 
  Well-Shaped Tetrahedral Meshes}\label{chp:mesh_trav}
\chaptermark{Traversal in Well-Shaped Meshes}

\section{Introduction}
Traversal of a path visiting the elements of a data structure is 
a situation that commonly arises in computing. 
For example, reporting an elevation contour on a digital terrain model, 
finding a path between nodes in a network, and even searching in a tree 
structure can all be viewed as traversing a path. 
In representing very large graphs in the external memory setting, 
the graph is often partitioned into \emph{blocks} which correspond to the 
disk blocks in memory. 
An effective partitioning of the graph permits the traversal of a path without
incurring a large cost in terms of I/Os (input-output operations that occur each
time a new block is visited).
Nodine~\emph{et al.}~\cite{ngv_1996} first studied the problem of graph blocking, 
and derived bounds for path traversal in several classes of graphs. 
Agarwal~\etal~\cite{DBLP:conf/soda/AgarwalAMVV98} describe an 
$O \left( N \right)$ space data 
structure for planar graphs of bounded degree $d$ that permits traversal of a path 
visiting a sequence of $K$ vertices in $O\left( K/\log_d B \right)$ I/Os. 
Our own research in Chapter~\ref{chp:succinct_graphs} examined succinct 
representions of such graphs while maintaining the $O\left( K/\log_d B \right)$ 
traversal bound.

In this chapter, we explore representing a tetrahedral mesh, $\myMesh$, in a 
manner that permits I/O efficient path traversal. 
We assume that $\myMesh$ is \emph{well-shaped}, by which we mean that the aspect 
ratio of each tetrahedron is bounded by a constant. 
This assumption is valid, for instance, for mesh generation algorithms that enforce 
the well-shaped property on their output meshes~\cite{mttv_1998}.

\subsection{Results}
We present a representation for well-shaped convex tetrahedral meshes in 
$\mathbb{R}^3$, permitting efficient path traversal in the external memory setting.  
We partition the dual graph of $\myMesh$ into regions, and store these regions 
in disk blocks. 
We also construct fixed-size neighbourhoods around the vertices forming the boundaries 
between the regions, and store a collection of these neighbourhoods in blocks. 
I/O efficiency of path traversal is guaranteed by the fact that the number of vertices 
we can traverse every time we encounter a boundary vertex (tetrahedron in the primal) 
is bounded from below. 
Our mesh representation has two key requirements. 
Firstly, we must be able to partion $\myMesh$'s dual into block-sized regions, 
and secondly, the number of \emph{boundary} vertices occurring at the intersection 
of the regions must be small. We will demonstrate that these requirements can be met 
for well-shaped meshes.

We give a representation of a well-shaped convex tetrahedral mesh that can be stored 
in $O\left( N \right)$, space and that permits traversal of a path visiting a sequence 
of $K$ tetrahedra in $O\left( K / \lg{B} \right)$ I/Os. 
We have not completed a detailed analysis of the pre-processing time, but the geometric 
separator algorithm we employ \cite{DBLP:journals/jacm/MillerTTV97} can partition a mesh 
in $Scan(N)$ I/Os. 
Since the recursive partitioning will require no more than $O(\log{N})$ steps, the preprocessing 
requires at worst $O(\log{N} \cdot Scan(N))$ I/Os. 
As applications of our technique, we demonstrate how our structure can be used to report the 
intersection of $\myMesh$ with an axis-parallel box, or an arbitrarily oriented plane 
$\mathcal{H}$, in $O\left( K / \lg{B} \right)$ I/Os. 
For axis-parallel box queries, the Priority R-Tree \cite{DBLP:journals/talg/ArgeBHY08} 
permits reporting in $O( (N/B)^{2/3} + K/ B )$ I/Os, but cannot efficiently handle 
arbitrarily-oriented plane queries. 
In each case, $K$ is the total number of tetrahedra intersecting the query box or plane, 
respectively. 
Further, we note that the $(N/B)^{2/3}$ term for the Priority R-Tree is the search time 
in $\mathbb{R}^3$, for which we do not account in our analysis.

\section{Mesh Partitioning}\label{sec:Prelim}

The tetrahedral mesh $\myMesh$ is composed of vertices, edges, faces, and tetrahedra. 
The vertices, edges, and faces form the \emph{skeleton} of the mesh, and we say a tetrahedron 
is \emph{adjacent} to elements of the skeleton that compose its boundary. 
Two tetrahedra are adjacent if and only if they share a face. 
Let $\dual{\myMesh}$ be the dual graph of $\myMesh$. 
For each tetrahedron $t \in \myMesh$, $\dual{\myMesh}$ contains a vertex $v(t)$ corresponding 
to $t$. 
Two vertices, $v(t)$ and $v(t')$, are connected by an edge in $\dual{\myMesh}$ if the 
corresponding tetrahedra share a face. 
Exclusive of the outer face (which we do not represent in $\dual{\myMesh}$), the degree of any 
vertex in $\dual{\myMesh}$ is at most $4$.

For tetrahedron $t$, we let $R(t)$ be the radius of the smallest enclosing sphere of $t$, 
and $r(t)$ be the radius of the largest sphere that can be inscribed in $t$. 
The \emph{aspect ratio} of $t$, denoted $\rho_t$, is defined as  the ratio of these two radii, 
$\rho_t = R(t)/r(t)$. 
If the aspect ratio of all tetrahedra in $\myMesh$ is bounded by a constant, 
$\rho$, then we say $\myMesh$ is \emph{well-shaped}.

In a \emph{geometric graph}, a $d$ dimensional coordinate is associated with each vertex. 
These coordinates yield an embedding of the graph in $\mathbb{R}^d$.  
Miller~\emph{et al.}~\cite{DBLP:journals/jacm/MillerTTV97} studied the problem of finding 
separators for geometric graphs. 
More specifically, they developed a technique for finding separators for $k$-ply neighbourhood 
systems defined as follows.

\begin{definition}[\cite{DBLP:journals/jacm/MillerTTV97}] 
A \emph{$k$-ply neighbourhood system} in $d$ dimensions is a set ${B_1, \ldots, B_n}$ of $n$ 
closed balls in $\mathbb{R}^d$, such that no point in $\mathbb{R}^d$ is strictly interior to 
more than $k$ balls.
\end{definition}

The authors were able show that separators could be found on several concrete classes of graphs 
that can be represented as $k$-ply \emph{neighbourhood systems}. 
The authors extend this work to the class of \emph{$\alpha$-overlap graphs};  they are defined
as follows:

\begin{definition}[\cite{mttv_1998}]
Let $\alpha \ge 1$ be given, and let $(B_1,\ldots,B_n)$ be a 1-ply neighbourhood system. 
The \emph{ $\alpha$-overlap graph} for this neighbourhood system is the undirected graph 
with vertices $V={1,\ldots,n}$ and edges:
\begin{equation*}
E = \{(i,j): B_i \cap (\alpha \cdot B_j) \ne \emptyset \land (\alpha \cdot B_i) \cap B_j \ne \emptyset \}
\end{equation*}
\end{definition}

For the class of $\alpha$-overlap graphs, the paper's main result is summarized in 
Lemma~\ref{lem:alpha_overlap_separator}.

\begin{lemma}[\cite{mttv_1998}]\label{lem:alpha_overlap_separator}
Let $G$ be an  $\alpha$-overlap graph in a fixed dimension $d$;
then $G$ has an \\ $O\left( \alpha \cdot n^{(d-1)/d}+q(\alpha,d)) \right)$ 
separator that $(d+1)/(d+2)$-splits.
\end{lemma}

In our setting, $d=3$, and $\alpha$ is fixed, thus the $q(\alpha,d)$ term, which 
is a function of two constants (representing the maximum degree of a vertex), is
constant, and we are left with a separator of size $O\left( n^{2/3} \right) $.
The geometric separator can also be used to find a vertex and edge separator on 
a well-shaped tetrahedral mesh \cite{mttv_1998}. 
Next we show that if $\myMesh$ is well-shaped, then the separator can likewise 
be applied to partition $\dual{\myMesh}$.

\begin{lemma}\label{lem:dual_is_overlap}
Let $\myMesh$ be a tetrahedral mesh where each tetrahedron 
$t \in \myMesh$ has aspect ratio bounded by $\rho$. 
Let $\dual{\myMesh}$ be the dual graph of $\myMesh$; then $\dual{\myMesh}$ is a 
subgraph of an $\alpha$-overlap graph for $\alpha = 2 \rho$.
\end{lemma}

\begin{proof}
Consider the following neighbourhood system. 
Let $v(t) \in \mathcal{M^*}$ be a vertex corresponding to tetrahedron 
$t \in \myMesh$. 
Let $b(t)$ be the largest inscribed ball in $t$, with radius $r(t)$ and centred
at point $p(t)$. 
The collection of balls $\mathcal{B}=\{b(1), \ldots, b(n)\}$ corresponds to the 
$n$ tetrahedra of $\myMesh$ and the vertices of $\mathcal{M^*}$, and form a 
$1$-ply system. 
No two balls can overlap, since each is inscribed within a single tetrahedron.  
We show that $\mathcal{B}$ forms an $\alpha$-overlap graph, for $\alpha = 2 \rho$.

Since each $b(t)$ is maximal, it touches each of the four faces of $t$ at a point.
Let $B(t)$ be the smallest ball that wholly contains $t$; recall that this ball 
has radius $R(t)$. 
Furthermore, let the point $P(t)$ be the centre of $B(t)$. 
Since $B(t)$ wholly contains $t$, it also contains $b(t)$ and its centre $p(t)$.
Therefore, the distance between the centres of these two balls, $P(t)$ and 
$p(t)$, is less than the radius, $R(t)$, of the larger ball. 
Consider the ball of radius $2R(t)$ centered at $p(t)$. 
Since its distance to $P(t)$ is less than $R(t)$, this ball contains both 
$P(t)$ and the entire ball $B(t)$, which wholly contains $t$.

For $i=1, \ldots, 4$, let $b(i) \in \mathcal{B}$ represent the inscribed balls 
associated with each neighbour of $v(t)$ in $\mathcal{M^*}$. 
Each $b(i)$ touches a face of $t$, since the tetrahedra in which they are
inscribed are neighbours ot $t$.
The ball centered at $p(t)$ of radius $2R(t)$ 
fully contains $t$ (and all its faces), and thus intersects each of these balls.

The aspect ratio of all $t \in \myMesh$ is bounded by $\rho$; therefore, for the aspect 
ratio, $\rho_t$, of any tetrahedron, we have that $\rho_t \le \rho$. 
Increasing the radius of a ball $b(t)$ to $2 R(t)$ is equivalent to 
$2 \frac{R(t)}{r(t)} r(t) = 2 \rho_t \cdot r(t) \le 2 \rho \cdot r(t)$. 
Thus, the overlap graph is an $\alpha$-overlap graph with $\alpha = 2 \rho$.
\end{proof}

\begin{lemma}\label{lem:dual_separator}
Given the dual graph of a tetrahedral mesh, with bounded aspect ratio $c$, there
is a partitioning algorithm that produces an $O\left( n^{\frac{2}{3}} \right)$ 
separator that $\frac{4}{5}$-splits $\dual{\myMesh}$.
\end{lemma}

\begin{proof}
By Lemma \ref{lem:dual_is_overlap}, we demonstrated that $\dual{\myMesh}$ is a 
subgraph of an $\alpha$-overlap graph for $\alpha = 2c$. 
Applying the separator theorem of Lemma~\ref{lem:alpha_overlap_separator} 
yields an $O\left( n^{\frac{2}{3}} \right)$ separator that 
$\frac{4}{5}$-splits $\dual{\myMesh}$.
\end{proof}


This separator is sufficient to partition $\dual{\myMesh}$, 
but we must still show that it can be recursively applied to split $\dual{\myMesh}$ into 
block-sized regions while bounding the total number of boundary vertices. 
In Lemma~\ref{lem:rec_dual_separator}, we extend the result of \cite{Frederickson87} 
on planar graphs to well-shaped tetrahedral meshes.

\begin{lemma}\label{lem:rec_dual_separator}
An $n$-vertex dual graph, $\dual{\myMesh}$, of a well-shaped tetrahedral mesh can be divided 
into $O\left( n/r \right)$ regions with no more than $r$ vertices each, and 
$O\left( \frac{n}{r^{1/3}} \right)$ boundary vertices in total.
\end{lemma}

\begin{proof}
By Lemma~\ref{lem:dual_separator}, we can subdivide the dual graph 
$\dual{\myMesh}$, on $n$ vertices, into two regions of size $\delta n$ and 
$(1-\delta)n$ for $\frac{1}{5} \le \delta \le \frac{4}{5}$ with a separator of 
size $\beta = O\left( n^{\frac{2}{3}} \right)$. 
Each region retains the vertices in the separator (boundary vertices), and the 
separator is recursively applied to the new regions until we have regions of 
maximum size $r$. 

Let $v$ be a boundary vertex in the resulting subdivided dual graph. 
Let $d(v)$ be one less than the number of regions that contain $v$, and 
let $D(n,r)$ be the sum of the $d(v)$s over all boundary vertices for a graph 
of size $n$ with regions of maximum size $r$. 
$D(n,r)$ can be calculated by the recurrence: $
D(n,r) \le cn^{2/3} + D(\delta n + O\left( n^{2/3} \right),r) + D((1 - \delta)n + O\left( n^{2/3} \right),r)$ 
for $n > r$, and $D(n,r) = 0$ for $n \le r$, where $c>1$ is a constant, and $\frac{1}{5} \le \delta \le \frac{4}{5}$. 
It can be shown by induction that $D(n,r) \in O\left( \frac{n}{r^{1/3}} \right)$. 
Since all boundary vertices have $d(v) \ge 1$, this value also bounds the total number of boundary vertices. 

Our inductive proof that $D(n,r) \in O\left( \frac{n}{r^{1/3}} \right)$ is as 
follows. 
In order to simplify our calculations we will use a modified version of our 
inductive hypothesis $D'(n,r) = D(n,r) \cdot r^{1/3}$ and we have:

\begin{eqnarray*}
D'(n,r) &\le& O\left[\frac{cn}{r^{1/3}} - dn^{2/3}\right] \cdot r^{1/3} \\
&=& cn - d r^{1/3} n^{2/3}
\end{eqnarray*}

We now want to show that $D'(n,r) \le  cn - d r^{1/3} n^{2/3}$ for some 
suitably-selected constant $d$. 
We let $\gamma_1 = (\delta n + c_1n^{2/3})$ and 
$\gamma_2 = (1-\delta)n + c_1n^{2/3}$, where $c_1$ is the constant term in 
the Big-Oh notation of our separator. 
We have:

\begin{eqnarray}
D'(n,r) &\le& O\left [c n^{2/3} + \frac{c( \delta n + c_1n^{2/3} )}{r^{1/3}} - d( \gamma_1^{2/3} ) + \frac{ c((1-\delta)n + c_1n^{2/3})}{r^{1/3}} - d(\gamma_2^{2/3}) \right] r^{1/3} \nonumber \\ 
&=& c r^{1/3} n^{2/3} + c( \delta n + c_1 n^{2/3})  + c( ( 1-\delta ) n + c_1 n^{2/3} ) - d r^{1/3} O\left[ \gamma_1^{2/3} + \gamma_2^{2/3} \right] \nonumber \\
&=& cr^{1/3}n^{2/3} + c\delta n + c c_1 n^{2/3} + cn - c\delta n + c c_1 n^{2/3} - dr^{1/3} O\left[ \gamma_1^{2/3} + \gamma_2^{2/3} \right] \nonumber \\
&=& cr^{1/3}n^{2/3} + cn + 2cc_1n^{2/3} - dr^{1/3} \left[ \gamma_1^{2/3} + \gamma_2^{2/3} \right] \label{eqn:recc1}
\end{eqnarray}

Again, we want to select $d$ so that Equation \ref{eqn:recc1} is less than or 
equal to $cn - dr^{1/3} n^{2/3}$. 
Noting that both equations contain a $cn$ term, we can remove it and we wish to 
prove the following statement true for a value of $d$:

\begin{equation}
c r^{1/3} n^{2/3} + 2cc_1n^{2/3} - dr^{1/3} \left[ \gamma_1^{2/3} + \gamma_2^{2/3} \right] \le -d r^{1/3} n^{2/3} 
\end{equation}

which we can further simplify by dividing each side by $n^{\frac{2}{3}}$ and 
grouping the $d$ terms as follows:

\begin{equation}\label{eqn:recc2}
c r^{1/3} + 2cc_1 \le dr^{1/3} \left[ \frac{\gamma_1^{2/3}}{n^{\frac{2}{3}}} + \frac{\gamma_2^{2/3}}{n^{\frac{2}{3}}} - 1 \right]  
\end{equation}

We next solve for $\gamma_1^{2/3}$ and $\gamma_2^{2/3}$:
\begin{eqnarray*}
\gamma_1^{2/3} &=& (\delta n + c_1n^{2/3})^{2/3} \\
&=& \left( n \left(\delta + \frac{c_1}{n^{1/3}}\right) \right)^{2/3} \\
&=& n^{2/3} \cdot \left( \delta + \frac{c_1}{n^{1/3}} \right)^{2/3}
\end{eqnarray*}

\begin{eqnarray*}
\gamma_2^{2/3} &=& ((1-\delta) n + c_1n^{2/3})^{2/3} \\
&=& \left( n \left((1-\delta) + \frac{c_1}{n^{1/3}}\right) \right)^{2/3} \\
&=& n^{2/3} \cdot \left((1-\delta) + \frac{c_1}{n^{1/3}}\right)^{2/3}
\end{eqnarray*}

We substitute the values for $\gamma_1$ and $\gamma_2$ and simplify the 
right-hand side of Equation \ref{eqn:recc2}.

\begin{eqnarray}
&=& dr^{1/3}[\frac{n^{2/3} \cdot \left(\delta + \frac{c_1}{n^{1/3}}\right)^{2/3}}{n^{2/3}} + \frac{n^{2/3} \cdot \left((1-\delta) + \frac{c_1}{n^{1/3}}\right)^{2/3}}{n^{2/3}} - 1] \nonumber \\
&=& dr^{1/3}\left[ \left(\delta + \frac{c_1}{n^{1/3}}\right)^{2/3} + \left((1-\delta) + \frac{c_1}{n^{1/3}}\right)^{2/3} - 1 \right] \label{eqn:recc3}
\end{eqnarray}

Note that the two $\frac{c_1}{n^{1/3}}$ terms in Eq. \ref{eqn:recc3} are bounded 
by a small constant, and dropping those terms results in a smaller expression 
within the square brackets. 
Furthermore, we minimize the value within the square brackets when we select 
$\gamma = 1/5$. 
Therefore, we can replace the value in the square brackets in Eq. \ref{eqn:recc3} with 
$\lambda$, which we define as follows:

\begin{equation}
\lambda = \left[ \left( \frac{1}{5} \right)^{\frac{2}{3}} + \left( \frac{4}{5} \right)^{\frac{2}{3}} - 1 \right]
\end{equation}

We are left with choosing $d$ for which the following is true:

\begin{eqnarray*}
cr^{1/3} + 2cc_1 &\le& d r^{1/3} \lambda \\
\frac{c}{\lambda} + \frac{2cc_1}{r^{1/3} \lambda} &\le& d
\end{eqnarray*}

Finally, since 
$\frac{c}{\lambda} + \frac{2cc_1}{r^{1/3} \lambda} < \frac{c}{\lambda} + \frac{2cc_1}{\lambda}$,
it is sufficient to select $d > \frac{c}{\lambda} + \frac{2cc_1}{\lambda}$ in 
order to satisfy the recurrence.
\end{proof}

It turns out that our result in Lemma~\ref{lem:rec_dual_separator}
is a specialized case of a more general proof given by Teng~\cite{teng:118}
(see Lemma 4.4 in that paper). 
That paper applies the same geometric separator that we have used to decompose
a well-shaped mesh into regions of size $\BigTheta{ \log^h{n} }$, and results
in a separator with $\OhOf{ n / \log^{(h/d)}{n} }$ boundary vertices,
where $d$ is the dimension and $h$ is a positive integer.
In our case, where $d=3$, if we let $r = \log{n}$, and select $h=1$,
we end up with regions of size $\BigTheta{ r }$ with $\OhOf{ n / r^{1/3}}$
boundary vertices.

\section{Data Structure and Navigation}\label{ssec:ds_navigation}

We apply the recursive partitioning algorithm on $\dual{\myMesh}$ to produce 
regions of at most $B$ vertices, where $B$ is the disk block size. 
Each region is stored in a single block in external memory. 
The block stores both the dual vertices and the corresponding geometry from 
$\myMesh$. 
Next, we identify around each boundary vertex a neighbourhood which we call 
its $\alpha$-neighbourhood. 
The $\alpha$-neighbourhood is selected by performing a breadth-first search, 
starting at the boundary vertex, and retaining the subgraph of $\dual{\myMesh}$ 
induced on the first $\alpha = \sqrt{B}$ vertices encountered. 
We can store the regions in $O\left( N/B \right)$ blocks, and the 
$O\left( N/\sqrt{B} \right)$ $\alpha$-neighbourhoods in $O\left( N/B \right)$ 
blocks. 
Thus the total space is linear.

In order to traverse the data structure, assume that we start with some tetrahedron $t$ 
interior to a region for which the corresponding block is loaded. 
We follow the path to a neighbour of $t$. 
If the neighbour is a boundary vertex, we load the $\alpha$-neighbourhood. 
Since $\dual{\myMesh}$ is of bounded degree, loading an $\alpha$-neighbourhood 
guarantees at least $\log_4{\sqrt{B}} = O\left( \lg B \right)$ progress along 
the path before another I/O is incurred. 
A path of length $K$ can be traversed with $O\left( K/ \lg{B} \right)$ I/Os. 
We summarize our results with the following theorem.

\begin{theorem}\label{thm:mesh_path_traversal}
Given a convex tetrahedral mesh, $\myMesh$, of bounded aspect ratio, there is 
an $O\left( N/B \right)$ block representation of $\myMesh$ that permits 
traversal of a path which visits a sequence of $K$ tetrahedra of $\myMesh$ 
using $O\left( \frac{K}{\lg{B}} \right)$ I/Os.
\end{theorem}

\section{Applications}

We now demonstrate the application of our data structures for reporting the 
results of box and plane intersection queries. 
A box query is the $\threeD$ equivalent of a rectangular window query in 
$\plane$.
For large meshes, such queries allow users to visualize or analyze a subset
of the mesh elements.
Given that the meshes may be very large, restricting access to a subset of the mesh
can be expected to be a common operation.

Plane intersection queries are the $\threeD$ extension of line intersection
queries (or terrain profiles) in $\plane$.
Why might such a query be of interest?
Assume that we have a mesh that models the variation of some property within
a volume. 
Visualization of the volume, or a subset thereof, is achived through somehow 
mapping mesh elements to a $\plane$ display.
Intersecting the mesh with an arbitrary plane is one means of generating a 
$\plane$ display from the mesh data.

\subsection{Depth-First Traversal}

To begin, we demonstrate how we can perform a depth-first traversal on a mesh 
using our structures, as intersection queries can be answered as 
special cases of depth-first traversal. 
A challenge in performing depth-first traversals in $\dual{\myMesh}$ is that 
depth-first traversal algorithms generally mark vertices as visited in order to avoid 
revisiting them from another execution branch. 
In our data structure, vertices may appear in multiple blocks, and loading all 
copies of a vertex to mark them as visited destroys the I/O efficiency of the 
structure. 
Thus, we need a means of determining if a vertex in $\dual{\myMesh}$ 
(tetrahedron in $\myMesh$) has already been visited. 
This can be achieved using the following lemma based on 
De Berg~\emph{et al.}~\cite{DBLP:journals/gis/BergKOO97}.

\begin{lemma}[\cite{DBLP:journals/gis/BergKOO97}]\label{lem:subgraph_mesh_dual_is_tree}
Given a convex tetrahedral mesh, $\myMesh$, and a tetrahedron, $t_s \in \myMesh$, 
then for every tetrahedron $(t \ne t_s) \in \myMesh$ a unique \emph{entry face} 
can be selected such that the edges in the dual $\dual{\myMesh}$, corresponding
to the entry faces, implicitly form a tree rooted at $t_s^* \in \dual{\myMesh}$.
\end{lemma}

The entry faces are selected by picking a special tetrahedron $t_s$ and a 
reference point $p_s$ interior to $t_s$. 
We perform a depth-first traversal of $\dual{\myMesh}$ on the implicit tree 
rooted at $t_s^*$. 
The selection of entry faces is based entirely on $p_s$ and the geometry of the 
current vertex of $\dual{\myMesh}$ 
(recall that each dual vertex stores the corresponding tetrahedron geometry). 
Depth-first traversal in a tree does not require the use of mark bits.

As part of a mesh traversal, we may be interested in reporting the features
of the mesh intersecting the query region.
In this setting, we will refer to the skeleton of the mesh as the edges and faces
of the tetrahedra visited.
Assume we wish to report the features of the mesh skeleton during 
traversal without double-reporting.
We define the \emph{interior} of an element on the mesh skeleton to be the 
closed set of points included in that element, exclusive of the endpoints for an
edge, and exclusive of the edges in the case of a face.

Consider the point $p_s$, and assume that the endpoints of no line segment are 
collinear with $p_s$ and that no three points defining a face are co-planar with 
$p_s$ (we remove these assumptions later). 
With respect to a tetrahedron, $t$, we say that any vertex, line or face 
adjacent to $t$ is \emph{invisible} if the line connecting $p_s$ with any point 
on that feature intersects the interior of $t$, otherwise we say the feature is 
$\emph{visible}$. 
For any element $k$ on the skeleton of $\myMesh$, the neighbourhood of $k$ is 
the collection of tetrahedra adjacent to $k$. 
The lemma leads to a technique for reporting the skeleton without duplicates, 
while the theorem summarizes our results for depth-first traversal.

\begin{lemma}\label{lem:tetra_visibility}
Let $p$ be any point interior to a convex tetrahedral mesh $\myMesh$, and let 
$x$ be a point invisible to $p$ and interior to any line segment or face of a 
tetrahedron $t \in \myMesh$. Then with respect to $t$, every interior point on 
the line segment or face containing $x$ is invisible.
\end{lemma}

\begin{proof}
Let $x$ and $y$ be points on the line segment or face under consideration. 
With respect to $p$ and $t$, assume that $x$ is invisible but $y$ is visible. 
Consider the triangle $\triangle pxy$. 
By definition, $\overline{px}$ intersects 
$t$ in its interior, while $\overline{py}$ does not. 
For this to be the case, $t$ must intersect $\overline{xy}$ at some point. 
However, $t$ is convex (as are all features on its skeleton), and all points 
on $\overline{xy}$ are interior to the line (or face) on $t$'s skeleton; therefore 
such an intersection cannot occur, which is a contradiction.
\end{proof}

\begin{lemma}\label{lem:skeleton_visibility}
For any element $k$ on the skeleton of $\myMesh$ there is one, and only one, 
tetrahedron $t$ in the neighbourhood of $k$, for which $k$ is invisible with 
respect to $p_s$.
\end{lemma}

\begin{proof}

The tetrahedra adjacent to any feature are non-overlapping.  
Consider a line segment $\ell$ drawn from $p_s$ to some point $x$ on a feature, 
denoted $z$, in the skeleton of $\myMesh$. If $z$ is a vertex then $z=x$. 
All features adjacent to $t_s$ are invisible from $t_s$ since $p_s$ is interior 
to $t_s$, and any line segment drawn from $p_s$ intersects the interior of $t_s$. 
Assume $z$ is not adjacent to $t_s$.  
As $\myMesh$ is convex, $\ell$ intersects some tetrahedron $t$, adjacent to $z$, 
prior to encountering $x$. 
Since $t$ makes $x$ invisible by Lemma~\ref{lem:tetra_visibility}, all of $z$ is 
invisible with respect to $t$. 
Since tetrahedra are non-overlapping, no tetrahedron other than $t$ can make 
$z$ invisible.
\end{proof}

Finally, we deal with degenerate cases of a face supported by a plane that is 
co-planar with $p_s$, and an edge for which the endpoints and $p_s$ are 
collinear. 
The case of a face is simpler, so we deal with it first. 
We take the plane supporting the face in question $f$, which we call $H_f$. 
$H_f$ contains $p_s$. 
Consider the directed line segment from $p_s$ to an arbitrary point in $f$. 
Since this line is directed and lies on $H_f$, we can use its direction to 
distinguish a left from a right side of $H_f$. 
We arbitrarily select the tetrahedron on the left (or right) side of $f$ as the 
tetrahedron that reports $f$. 
For the case of a line segment collinear with $p_s$, let the segment $s$ have 
endpoints $u$ and $w$, and $\ell$ be the line through $p_s, u, w$. 
Let $H_{\ell}$ be the plane through $\ell$ that is orthogonal to the $X-Y$ plane 
in $\threeD$ space. 
$H_{\ell}$ intersects two of the tetrahedra adjacent to $s$. 
Relative to the $X-Y$ plane, one of these tetrahedra is above and one is below
$\ell$. 
Select the tetrahedra above $\ell$. 
If $\ell$ is parallel to the $Z$ axis, then we can use the $X-Z$ plane in place 
of the $X-Y$ plane.

\begin{theorem}\label{thm:mesh_df_traversal}
Given a convex well-shaped tetrahedral mesh $\myMesh$ of $N$ tetrahedra, 
$\myMesh$ can be represented in linear space such that from an arbitrary 
tetrahedron, $t \in \myMesh$, a depth-first traversal can be performed that 
reports each tetrahedron in $\myMesh$ once, as well as all features on the skeleton of 
$\myMesh$ without duplication in $O\left( N / \lg{B} \right)$ I/Os.
\end{theorem}

\begin{proof}
By Theorem \ref{thm:mesh_path_traversal}, there is a representation of $\myMesh$ 
that permits traversal of a path of length $N$ in $\myMesh$ in 
$\BigOh{ N / \lg{B} }$ I/Os. 
By Lemma~\ref{lem:subgraph_mesh_dual_is_tree}, we can perform a tree traversal in 
$\dual{\myMesh}$ to visit all tetrahedra in $\myMesh$.  
The total length of the path traversed is $\BigOh{ N }$ so we use 
$\BigOh{ N / \lg{B} }$ I/Os. 
By Lemma~\ref{lem:skeleton_visibility}, we can report each feature of the 
skeleton exactly once and, during the traversal the tetrahedron can be reported 
at either its first or last visit.
\end{proof}

\subsection{Box Queries} 
Assume we are given an axis parallel box in $\threeD$. 
Select as a starting point one of the box's corner points, and assume we are 
given the tetrahedron $t \in \myMesh$ containing this point. 
We modify the entry face selection rule so that only faces from among those 
that intersect the interior of the box can be selected. 
We then build an implicit tree in $\dual{\myMesh}$ that visits all tetrahedra
that intersect the box interior. 
Again, using our structure, this allows us to report the intersection in 
$\BigOh{ K/ \lg B }$ I/Os.

\subsection{Plane Intersection Queries}
As a visualization tool, we wish to intersect a convex tetrahedral mesh 
$\myMesh$ by a plane $\mathcal{H}$ and report the intersection of 
$\myMesh$ and $\mathcal{H}$. 
Let $K$ be the number of tetrahedra in $\myMesh$ intersected by $\mathcal{H}$, 
including tetrahedra that intersect $\mathcal{H}$ only in their skeleton. 
Further, assume that we are given some tetrahedra $t_s \in \mathcal{T}$ that 
is part of the intersection set. 
The intersection of $\myMesh$ by $\mathcal{H}$ produces a two-dimensional planar 
subdivision mapped onto $\mathcal{H}$, which we will denote $\myMesh_{\mathcal{H}}$. 
Faces in $\myMesh_{\mathcal{H}}$ correspond to the intersection of $\mathcal{H}$ with 
the interior of a tetrahedron $t \in \myMesh_{\mathcal{H}}$. 
Edges in $\myMesh_{\mathcal{H}}$ correspond to the intersection of $\mathcal{H}$ with a face 
of $t$, while vertices correspond to the intersection of $\mathcal{H}$ with an 
edge of $t$. 
A face of $t \in \myMesh$, which lies directly on $\mathcal{H}$, appears as a face 
in $\myMesh_{\mathcal{H}}$. 
Likewise, edges and vertices from $\myMesh$ that fall exactly on $\mathcal{H}$ 
appear as edges and vertices, respectively, on $\myMesh_{\mathcal{H}}$.

Let $\mathcal{H}^+$ be the open half-space above $\mathcal{H}$, and 
$\mathcal{H}^-$ the open half-space below $\mathcal{H}$. 
Let $\hat{\myMesh}_{\mathcal{H}}$ be the set of tetrahedra in $\myMesh$ which 
intersect both $\mathcal{H}$ and $\mathcal{H}^+$. 
Tetrahedra in $\hat{\myMesh}_{\mathcal{H}}$ have some point in their interior in 
$\mathcal{H}^+$, and intersect $\mathcal{H}$. 
We select entry faces as before with one minor modification; 
we measure the distance between $p_s$ and tetrahedron $t$ on the plane 
$\mathcal{H}$, and disqualify as a candidate entry face any face which does not 
extend into $\mathcal{H}^+$.

\begin{lemma}
It is sufficient to visit the tetrahedra in $\hat{\myMesh}_{\mathcal{H}}$ to 
report all faces, edges, and vertices of $\myMesh_{\mathcal{H}}$.
\end{lemma}

\begin{proof}
We must show that every vertex, face, or edge in $\myMesh \cup \mathcal{H}$ is 
adjacent to a tetrahedron $t \in  \hat{\myMesh}_{\mathcal{H}}$. 
A face that lies exactly on $\mathcal{H}$ separates two tetrahedra, one of which 
must lie in $\mathcal{H}^+$. 
Likewise, any edge or vertex which lies directly on $\mathcal{H}$ is adjacent to 
a tetrahedron in $\mathcal{H}^+$. 
Since these tetrahedron are adjacent to features on $\mathcal{H}$ they are 
included in $\hat{\myMesh}_{\mathcal{H}}$, and therefore are reported.
\end{proof}

\begin{lemma}\label{lem:entry_in_H+}
Let $t$ be a tetrahedron in $\hat{\myMesh}_{\mathcal{H}}$. 
Let $t' \in \myMesh$ be the tetrahedron adjacent to $t'$ across the face 
$\mathtt{entry}(t)$; then $t'$ is also an element of 
$\hat{\myMesh}_{\mathcal{H}}$.
\end{lemma}

\begin{proof}
By definition, the entry face of $t$ contains a point $x_t$ in $\mathcal{H}^+$. 
The tetrahedron $t'$ shares this face with $t$, thus $t'$ must intersect 
$\hat{\myMesh}_{\mathcal{H}}$.
\end{proof}

Now consider the dual graph $\dual{\myMesh}$ of $\myMesh$. 
Selecting the tetrahedra in $\hat{\myMesh}_{\mathcal{H}}$ defines a subset of 
the vertices of $\dual{\myMesh}$. 
Let $\hat{\myMesh}_{\mathcal{H}}^*$ be the subgraph induced on this subset. 

\begin{lemma}\label{lem:dual_tree_plane_intersection}
Defining the set of entry faces as above produces an implicit tree, rooted at 
$v(t_s)$, on $\hat{\myMesh}_{\mathcal{H}}^*$.
\end{lemma}

\begin{proof}
The proof is essentially the same as for 
Lemma~\ref{lem:subgraph_mesh_dual_is_tree}. 
Since the entry faces are still unique, there is no danger of introducing cycles. 
By Lemma \ref{lem:entry_in_H+}, each entry face leads to another tetrahedron in 
$\hat{\myMesh}_{\mathcal{H}}$, therefore connectivity is maintained.
\end{proof}

Lemma \ref{lem:dual_tree_plane_intersection} leads to a simple traversal 
algorithm for reporting the intersection of $\myMesh_{\mathcal{H}}$ with $\mathcal{H}$. 
We perform a depth-first traversal of $\hat{\myMesh}_{\mathcal{H}}^*$ on the 
implicit tree defined by the entry edges, and report the intersection of each tetrahedron 
with $\mathcal{H}$. 
If we visit a total of $K$ tetrahedra, the total path length of the traversal is 
$\BigOh{ K }$ steps, thus if we represent $\hat{\myMesh}_{\mathcal{H}}^*$ with 
the I/O efficient representation of $\myMesh_{\mathcal{H}}$, we can report the intersection 
in $\BigOh{ K/ \lg B }$ I/Os. 

The correctness of the box query algorithm does not rely on the fact that the 
box is axis-parallel.  
In fact, given a starting point interior to any arbitrarily oriented convex shape, 
we can report the intersection in $O\left( K/ \lg B \right)$ I/Os. 
Plane intersection can be viewed as a degenerate case of reporting an 
arbitrarily-oriented box query with a box that has been flattened. 
For both the box and the plane intersection queries we have presented, we have assumed 
that we are given a starting tetrahedron $t_s$, interior to the query region. 
Efficiently finding $t_s$ represents the next step in our resarch, and
 an effective technique in this setting is described in 
Chapter~\ref{chp:jump_and_walk}.

\section{Open Problems}

Our original purpose in conducting this research was to develop a data 
structure that is both succinct and efficient in the EM setting.
In particular we hoped to extend our structure for triangulations in
$\plane$ (see Chapter~\ref{chp:succinct_graphs}) to tetrahedral 
meshes in $\threeD$.
Unfortunately the data structures we use rely on the fact that there
is a $k$-page embedding of the dual graph with constant $k$. 
There is no known bound on the page-thickness of the dual of a
tetrahedral mesh, thus it is uncertain that our representation will
work.
Barat~\etal~\cite{DBLP:journals/combinatorics/BaratMW06} proved 
that bounded-degree planar graphs with maximum degree greater
than nine have arbitrarily large geometric thickness.
They state that finding the page embedding of a degree five graph 
is an open problem.
Meanwhile, Duncan~\etal~\cite{DBLP:journals/corr/cs-CG-0312056}
gave a bound of $k=2$ when the maximum degree is less than four.
With $d=8$, the augmented dual graph of a tetrahedral mesh is in
the gray area between these results.
Thus, a potential approach to solving this problem would be to prove that
the dual of a tetrahedral mesh is $k$-page embeddable, although there
is no guarantee this is possible.

As was the case in the previous chapter, our results are more general
than presented here.
We have stated our results for well-shaped tetrahedra meshes, but they
will hold for any tetrahedral mesh that admits a separator for which
the recursive application results in a suitably small set of separator
vertices.

\chapter[Jump and Walk in Well-Shaped Meshes]{Jump-and-Walk in Well-Shaped 
  Meshes}\label{chp:jump_and_walk}
\chaptermark{Jump-and-Walk}

\section{Introduction}

Point location is a topic that has been extensively studied since the origin 
of computational geometry. 
In $\plane$, the point location problem, typically referred to as planar point 
location, can be defined as follows. 
Preprocess a polygonal subdivision of the plane so that given a query point $q$, 
the polygon containing $q$ can be reported efficiently. 
There are several well-known results showing that a data structure with $\BigOh{n}$ 
space can report such queries in $\BigOh{ \log n }$  time. 
In spatial point location, we have a subdivision of the three-dimensional space into 
polyhedra, and given a query point $q$, we wish to return the polyhedron 
containing $q$.
From a theoretical standpoint, the problem of efficient general spatial point location 
in $\threeD$ is still open~\cite{cg_book_2008}.

One popular method for solving the point location problem originated in the late 1970s with
the works of Green and Sibson~\cite{DBLP:journals/cj/GreenS78} and Bowyer~\cite{bowyer81}, 
in the context of Voronoi diagrams.
This technique, which has come to be commonly known as 
\emph{jump-and-walk}, has shown to be practically efficient and
has attracted theoretical interest.

In this chapter, we consider the application of jump-and-walk to a more
specialized problem; namely, point location queries in two and three 
dimensions for well-shaped triangular and tetrahedral meshes.  
A well-shaped mesh, denoted by $\mathcal{M}$, is one in which all its simplices 
have bounded aspect ratio (see Definition~\ref{def:well-shaped}).
This assumption is valid for mesh generation algorithms that enforce the 
well-shaped property on their output meshes~\cite{mttv_1998}. 
Our motivation comes from an external memory setting, where we have examined 
data structures for representations that permit efficient path traversals in meshes
which are too large to fit in the main memory~\cite{DBLP:conf/cccg/Dillabaugh10}.

\subsection{Jump-and-Walk}\label{sec:jump_and_walk}
Assume we have a subdivision of our search space and a 
query point $q$, and we wish to report the subdivision element
containing $q$.

\begin{enumerate}
  \item  {
  Select a set of possible start (jump) points such that:
  \begin{enumerate}
	\item they can be queried efficiently, and
	\item their location within the subdivision is known.
  \end{enumerate}
  }
  \item From the set of jump points, select a point from which to start, which
 we denote as $p$. 
  \item Starting at $p$, walk along the straight line from $p$ to $q$, passing through 
  the elements of the subdivision until you arrive at the point $q$, at which 
  time you report the current subdivision element.  
\end{enumerate}

M\"ucke \etal~\cite{DBLP:journals/comgeo/MuckeSZ99} formalized the concept of 
jump-and-walk. 
Given a set of points, and a Delaunay triangulation over them, they considered 
a (random) subset of these points as the candidate jump-points. 
Given a query point $q$, they locate the nearest point in this subset (say $p$) 
to $q$ as the jump-point, and then walk along the straight line segment $pq$, 
passing through possibly several triangles until they reach the triangle 
containing $q$. 
The output of the query is the triangle that contains $q$. 
This in turn was motivated by the experiments presented in the Ph.D. thesis of 
M\"ucke~\cite{muckephd}, where the randomized incremental method was used to 
construct 3D-Delaunay triangulations using flips. 
It is shown in \cite{DBLP:journals/comgeo/MuckeSZ99} that the expected running 
time of performing the point location using the jump-and-walk strategy is 
$O(n^{1/4})$, for a Delaunay triangulation of a set of $n$ random, uniformly 
distributed, points over a convex set in $\mathbb{R}^3$. 
M\"ucke provides experiments that confirm his theoretical results. 
Results on point location in planar Delaunay triangulations, under the 
framework of jump-and-walk, has been studied by 
Devroye \etal~\cite{DBLP:journals/algorithmica/DevroyeMZ98}, where 
$O(n^{1/3})$ expected time is spent for a set of
$n$ random points.
With slightly more advanced data structures used in the jump stage, 
Devroye \etal~\cite{Devroye200461} demonstrated that expected search times 
for the jump-and-walk in Delaunay triangulations range from  
$\BigOmega{ \sqrt{n} }$ to $\BigOmega{ \log n }$, depending on the specific 
data structure employed for the jump step.
Devillers~\cite{DBLP:journals/ijfcs/Devillers02} discusses a data structure,
the Delaunay hierarchy, which is used to guide the jump-and-walk point location during
Delaunay triangulation construction.  They provide an expected-time analysis and show that the cost of inserting the
$n$\textsuperscript{th} point in the triangulation is
$\BigOh{ \alpha \log_{\alpha} n}$, where $\alpha$ is a tunable constant that 
reflects the number of vertices retained at each higher level in the 
hierarchy. Furthermore, Devillers \etal~\cite{DBLP:journals/ijfcs/DevillersPT02}
describe an implementation study of
several point location strategies using the jump-and-walk paradigm.
Haran and Halperin~\cite{DBLP:journals/jea/HaranH08} provide a  detailed experimental study of the performance of several point location methods in planar arrangements, including that of jump-and-walk. 
For more recent implementation studies of various point location strategies, including jump-and-walk,
see~\cite{DBLP:conf/compgeom/DevillersC11, ZhuISVD2012}.

\subsection{Our Results}

Suppose that we are given a query point $q$ in a 
well-shaped mesh $\mathcal{M}$, with vertex set $P$ and aspect ratio $\rho$
(refer to Definition~\ref{def:well-shaped}).
Let $p\in P$ be a nearest neighbour of $q$,
and $\hat{p} \in P$ be an \emph{approximate} nearest neighbour (ANN) of $q$
(refer to Definition~\ref{def:ANN}).
We show that,
once $p$ (respectively $\hat{p}$) has been located,
the number 
of tetrahedra intersected by the line segment $pq$
(respectively $\hat{p}q$) is bounded 
by a constant.
More precisely,
$pq$ intersects at most $\frac{\sqrt{3}\,\pi}{486}\rho^3(\rho^2+3)^3$ tetrahedra
(refer to Theorem~\ref{thm:exact_3d}),
and $\hat{p}q$ intersects at most
$\frac{\sqrt{3}\,\pi}{486}\rho^3(\rho^2+3)^3+\left\lfloor \frac{\pi}{\arcsin\left(\frac{3\sqrt{3}}{8\rho^2}\right)}\right\rfloor$ tetrahedra (refer to Theorem~\ref{thm:3d_approximate}).

Therefore,
given a well-shaped mesh $\mathcal{M}$,
a jump-and-walk search can be performed in the time required to perform a 
nearest-neighbour search on the vertices of $\mathcal{M}$, plus 
$O(1)$ time for the walk-step to find the tetrahedron 
containing the query point.
Also, a jump-and-walk search can be performed in the time required to perform an
approximate nearest-neighbour search on the vertices of $\mathcal{M}$, plus $O(1)$ time
for the walk-step to find the tetrahedron containing the query point.

All the previous theoretical results for the jump-and-walk paradigm  resulted in expected 
search times (see Section \ref{sec:jump_and_walk}).
Our research is unique in that it gives the first (that we know of) provable 
result for the exact number of steps required  for the walk step in well-shaped meshes. As a consequence of guaranteeing constant time for the walk step, and given
existing techniques to perform the jump step efficiently, our results lead to
a provably efficient point location strategy in the setting of well-shaped
meshes.

Our motivation to study came from the data structues presented in 
Chapters~\ref{chp:succinct_graphs} and \ref{chp:mesh_trav}, which
permit efficient path traversals in meshes. 
One of the challenges for the applications we presented there was how
to resolve point location queries efficiently. 
Our initial impression was that since the meshes were well-shaped (i.e.,
each tetrahedron is ``fat''),
and the degree of each vertex is constant, proving constant time for
the walk step from a nearest
neighbour would be trivial.
However, when we tried to prove this formally, we ran into 
several intricate issues and technical difficulties, as the reader will see in 
Section \ref{sec:j-w-3D}, especially when wanting to prove tight bounds with 
respect to the aspect ratio. 
Moreover, when we jump to an approximate nearest neighbour, the problem 
becomes even more challenging. 

\subsection{Organization}
The remainder of this chapter is organized as follows. 
In Section \ref{sec:Background}, we define a number of terms related to 
well-shaped meshes and present some observations on such meshes.
This section also defines the nearest-neighbour and approximate nearest-neighbour
problems. 
In Section \ref{sec:j-w-2D}, we present our jump-and-walk results in the
$\CDTwoD$ setting. 
Although several results in $\plane$ are considered to be common knowledge
(see Lemma~\ref{lem:min_angle}, for instance),
we could not find any proofs in the literature for these results.
Also,
the proof schemes in $\plane$
set the table for the proofs in $\threeD$.
For a given aspect ratio $\rho$, 
we wish to describe in terms of $\rho$
the number of triangles intersected by $pq$ and $\hat{p}q$.
Surprisingly, this seems to be unknown in the literature.
This too is an appropriate preparation for the $\threeD$ setting.
In Section \ref{sec:j-w-3D} we present results in the $\threeD$ setting.
In Section \ref{sec:walk-experiments} we describe the results of experiments
based on the implementation of our technique in $\plane$.
Section \ref{sec:summary} concludes this chapter, and discusses possible future
work.

\section{Background}\label{sec:Background}

We consider jump-and-walk in triangular and tetrahedral 
meshes in two and three dimensions, respectively. 
We use $\mathcal{M}$ to denote a mesh in either $\plane$ or $\threeD$, 
and if we want to be specific we use $\meshTwoD$ to denote a triangular mesh
 in $\plane$, and $\meshThreeD$ to denote a tetrahedral mesh in $\threeD$. 
We assume that the triangles and tetrahedra, which are collectively referred 
as simplices, are \emph{well-shaped}, a term which will be defined shortly. 
If all simplices of a mesh $\mathcal{M}$ are well-shaped, then $\mathcal{M}$ 
itself is said to be a {\em well-shaped mesh}. 
Triangles in $\meshTwoD$ are considered adjacent if and only if they share an 
edge.
Similarly, tetrahedra in $\meshThreeD$ are considered adjacent if and only 
if they share a face. 

\subsection{Well-Shaped Meshes}

We begin by stating the \emph{well-shaped} property. 
In this chapter we use a slightly different definition for well-shapedness
than that given in Section~\ref{ssec:geom_seperators} and in~\cite{mttv_1998}.
The definition is conceptually the same, but for $R(t)$ we use the circumsphere
of the mesh simplex rather than the smallest enclosing sphere.
We have made this change because it simplifies some of the calculations
we use to arrive at precise bounds for walk step.
Formally, we use the following definition:

\begin{definition} \label{def:well-shaped}
We say that a mesh $\meshTwoD$ ($\meshThreeD$)
is \emph{well-shaped} if for any triangle (tetrahedron) $t \in \meshTwoD$ 
($t \in \meshThreeD$), the ratio formed by the radius $r(t)$ of the incircle
 (insphere) of $t$ and the radius $R(t)$ of the circumcircle (circumpshere) 
 of $t$ is bounded by a constant $\rho$, i.e., $\frac{R(t)}{r(t)} < \rho$.
\end{definition}

In this work, all meshes and simplicies (triangles and tetrahedra)
are assumed to be well-shaped.
Note that in two dimensions,
$\rho\geq 2$ where $\rho = 2$ corresponds to the equilateral triangle,
and in three dimensions,
$\rho\geq 3$ where $\rho = 3$, corresponds to the regular tetrahedron.
We make the following observations related to Definition~\ref{def:well-shaped}.

\begin{observation}\label{obs:ins_circ_phere}
Let $t$ be a triangle (tetrahedron).
\begin{enumerate}
\item\label{obs:insphere_bound}
Let $v$ be any vertex of $t$.
Denote by $e_v$ ($f_v$)
the opposite edge (face) of $v$ in $t$.
Let $$\mathtt{mdist}(v,t) = \min_{x\in e_v}|xv| \qquad (\mathtt{mdist}(v,t) = \min_{x\in f_v}|xv|),$$
where the minimum is taken over all points $x$ on $e_v$ ($f_v$).
Then, $\mathtt{mdist}(v,t)$
is an upper-bound on the diameter of the incircle (insphere of) $t$.
Formally, $2r(t) \le \mathtt{mdist}(v,t)$.

\item\label{obs:circumsphere_bound}
Let $e$ be the longest edge (or in fact any edge) of $t$.
The diameter of the circumcircle (circumsphere) of $t$ is at least as long as $e$.
 Formally, $2R(t) \ge |e|$.
\end{enumerate}
\end{observation}

\begin{lemma}\label{lem:min_angle}
There is a lower bound of $\alpha$ for each of the angles in any triangle 
of $\meshTwoD$.
There is a lower bound of $\Omega$ for each of the solid angles in any 
tetrahedron of $\meshThreeD$.
In particular, we have $\alpha \leq \frac{\pi}{3}$ and
$\cos(\alpha) = \frac{1+\sqrt{\rho(\rho-2)}}{\rho}$
for the two-dimensional case.
For the three-dimensional case,
$\Omega \leq 3\arccos\left(\frac{1}{3}\right) - \pi$
and $\sin\left(\frac{\Omega}{2}\right) = \frac{3\sqrt{3}}{8\rho^2}$.
\end{lemma}

\begin{proof}
For the three-dimensional case, refer to~\cite{springerlink:10.1007/BF01955874}.

For the two-dimensional case, 
let $\gamma$ and $\Gamma$ be two circles
such that $\gamma$ is included in $\Gamma$
(refer to Fig.~\ref{fig:proof_lem_two_dim}).
\begin{figure}
	\centering
 	\includegraphics[scale=0.875]{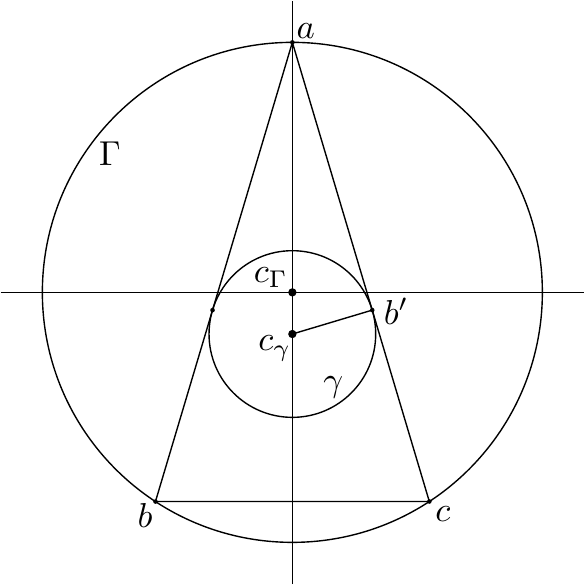}
	\caption{Proof of Lemma~\ref{lem:min_angle} for the two-dimensional case}
	\label{fig:proof_lem_two_dim}
\end{figure}
Denote by $r$ (respectively by $R$)
the radius of $\gamma$ (respectively of $\Gamma$),
and by $c_{\gamma}$ (respectively by $c_{\Gamma}$)
the center of $\gamma$ (respectively of $\Gamma$).
Let $d = |c_{\gamma}c_{\Gamma}|$.
Without loss of generality,
suppose $c_{\Gamma} = (0,0)$
and $c_{\gamma} = (0,-d)$.
Consider all triangles $\triangle abc$ such that $\Gamma$ contains $\triangle abc$
and $\triangle abc$ contains $\gamma$.
We will prove that:

\begin{itemize}
\item among all these triangles,
there exists a unique one that minimizes $\angle bac$.

\item This triangle is such that $\gamma$ and $\Gamma$
are, respectively, the incircle and the circumcircle of $\triangle abc$.
\end{itemize}

Therefore, $\alpha = \angle bac$.

In order to prove the existence of $\triangle abc$, suppose without loss 
of generality that:

\begin{itemize}
\item $a,b,c\in\Gamma$,

\item $a=(0,R)$,

\item $ab$ and $ac$ are tangent to $\gamma$
(denote by $b'$ the point where $ac$ and $\gamma$ are tangent),

\item and $bc$ and $\gamma$ have at most one intersection point.
\end{itemize}
From the theory of incircles and circumcircles,
$bc$ and $\gamma$ have exactly one intersection point
if and only if $d = \sqrt{R(R-2r)}$.
Hence,
$\triangle abc$ is well-defined if and only if $0\leq d\leq\sqrt{R
(R-2r)}$ and $\rho=\frac{R}{r} \geq 2$.
Therefore,
\begin{eqnarray*}
\cos(\angle bac) &=& \cos(2\angle c_{\gamma}ac)\\
&=& 2\cos^2\!(\angle c_{\gamma}ac) - 1 \\
&=& 2\left(\frac{|ab'|}{|ac_{\gamma}|}\right)^2 - 1 \\
&=& 2\left(\frac{\sqrt{(R+d)^2-r^2}}{R+d}\right)^2 - 1 \\
&=& \frac{(R+d)^2-2r^2}{(R+d)^2}\enspace.
\end{eqnarray*}

Using calculus, one can verify that
\begin{eqnarray*}
\max_{0\leq d\leq\sqrt{R(R-2r)}} \cos(\angle bac) &=& \max_{0\leq d\leq\sqrt{R(R-2r)}} \frac{(R+d)^2-2r^2}{(R+d)^2} \\
&=& \frac{\left(R+\sqrt{R(R-2r)}\right)^2-2r^2}{\left(R+\sqrt{R(R-2r)}\right)^2} \\
&=& \frac{r+\sqrt{R(R-2r)}}{R} \\
&=& \frac{1+\sqrt{\rho(\rho-2)}}{\rho}
\end{eqnarray*}
for any $r$, $R$ such that $\frac{R}{r} = \rho \geq 2$.
Therefore, the maximum of $\cos(\angle bac)$ corresponds to the case 
where $\gamma$ and $\Gamma$ are respectively the incircle 
and the circumcircle of $\triangle abc$.
Thus, $\cos(\alpha) = \frac{1+\sqrt{\rho(\rho-2)}}{\rho}$.

Finally,
$$\min_{\rho \geq 2} \frac{1+\sqrt{\rho(\rho-2)}}{\rho} = \frac{1}{2} \enspace,$$
hence $\alpha \leq \frac{\pi}{3}$.
This is consistent with the fact that in a triangle,
it is impossible for all angles to be strictly greater than $\frac{\pi}{3}$.
\end{proof}

For the rest of the chapter,
$\mathcal{C}(c,r)$
(respectively $\mathcal{S}(c,r)$)
denotes the circle
(respectively the sphere)
with centre $c$ and radius $r$. 

\subsection{Nearest Neighbour Queries}

The nearest neighbour query works as follows:
given a point set $P$ and a query point $q$, return a point $p \in P$ nearest 
to $q$, i.e., for all $v\in P$, $v \ne p$, $|pq| \leq |vq|$.
A closely-related query is the approximate nearest neighbour (ANN) query 
defined as follows \cite{DBLP:journals/jacm/AryaMNSW98}.

\begin{definition} \label{def:ANN}
Let $P$ be a point set in $\mathbb{R}^d$,
$q$ be a query point, and $p \in P$ be a nearest neighbour of $q$.
Given an $\epsilon \ge 0$, we say that a point $\hat{p} \in P$ is 
an $(1 + \epsilon)$-\emph{approximate nearest neighbour} of $q$ 
if $|\hat{p}q| \le (1 + \epsilon)|pq|$. 
\end{definition}

Arya~\etal~\cite{DBLP:journals/jacm/AryaMNSW98} show that a linear-space search 
structure can be efficiently computed that permits ANN queries to be answered in 
$\BigOh{c_{d,\epsilon} \log{n}}$ time, where $c_{d,\epsilon}$ is a constant that
depends on the dimension $d$ and the selected $\epsilon$.

\begin{figure}
	\centering
 	\includegraphics[scale=0.95]{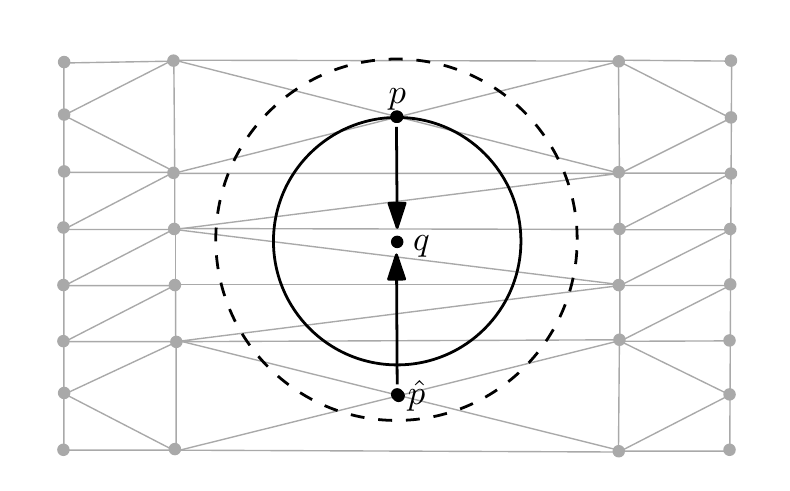}
	\caption[Jump-and-walk]{Demonstration of jump-and-walk either from nearest neighbour 
	$p$ to $q$ or from approximate nearest neighbour $\hat{p}$ to $q$. 
	The solid circle is centred on $q$ and has radius $|pq|$, while the 
	dashed circle is also centered on $q$ but has radius $(1 + \epsilon)|pq|$. 
	The arrows from $p$ and $\hat{p}$, respectively, show the sequence of triangles
	visited on the walk from $p$ ($\hat{p}$) to the triangle containing $q$.}
	\label{fig:jump_and_walk}
\end{figure}

\section{Jump-and-Walk in $\meshTwoD$}\label{sec:j-w-2D}

\subsection{Jump to Nearest Neighbour}
Let $P$ be the set of vertices of a well-shaped mesh $\meshTwoD$, and $q$ be 
a query point lying in a triangle of $\meshTwoD$.
Denote by $\alpha$ the lower bound on the angle of any triangle of $\meshTwoD$
(refer to Lemma~\ref{lem:min_angle}).
Let $p$ be a nearest neighbour of $q$. 
Consider the set of triangles encountered in a straight-line walk from 
$p$ to $q$ in $\meshTwoD$.

\begin{proposition}\label{prop:const_bound_2d}
The walk step along $pq$ visits at most 
$\left\lfloor \frac{\pi}{\alpha}\right\rfloor$ triangles.
\end{proposition}

\begin{proof}
Without loss of generality, suppose $|pq| = 1$.
Since $p$ is a nearest neighbour of $q$, there is no vertex of $\meshTwoD$ 
in the interior of $\mathcal{C}(q,|pq|)$. 
Denote by $\ell$ the line through $pq$, and let $p' \neq p$ be the intersection
of $\ell$ with $\mathcal{C}(q,|pq|)$ (see 
Fig.~\ref{fig:triangle_covers_fixed_arc_length}).

\begin{figure}
	\centering
	  \includegraphics[scale=1]{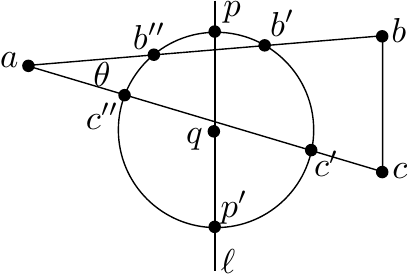}
	\caption[Triangle covers arc of bounded minimum length]{A triangle with a
	fixed minimum angle covers an arc bounded by a minimum fixed length 
	on $\mathcal{C}$.}
	\label{fig:triangle_covers_fixed_arc_length}
\end{figure}

The only triangles that matter are the ones intersecting $pq \setminus\{p\}$.
All such triangles have one vertex to the left of $\ell$, and two vertices to 
the right of $\ell$ or vice versa. 
We separate the triangles intersecting $pq\setminus\{p\}$ into two sets, 
$Left$ and $Right$, containing the triangles with exactly one vertex to the left 
and to the right of $\ell$, respectively.  
Consider an arbitrary triangle $t \in Left$ and let the vertices of 
$t$ be $a$, $b$, and $c$. 
The edge $ab$(respectively $ac$)intersects $\mathcal{C}(q,|pq|)$ at $b'$ and $b''$
(respectively at $c'$ and $c''$).
Let $\theta = \angle bac$. Since $ab$ and $ac$ are two secants which intersect
 $\mathcal{C}(q,|pq|)$, we have that $\theta = \frac{1}{2} 
 (\MyOverarc{b'c'} - \MyOverarc{c''b''})$, from which we conclude that 
 $\MyOverarc{b'c'} \ge 2 \theta$.
Hence, $\MyOverarc{b'c'} \ge 2 \theta \geq 2\alpha$ by Lemma~\ref{lem:min_angle}.
Therefore, since the arc $\MyOverarc{pp'}$ has length $\pi$,
we can conclude that the set $Left$ contains at most 
$\left\lfloor \frac{\pi}{2 \alpha}\right\rfloor$ triangles. 
The same bound holds for triangles in $Right$. 
Thus, the number of triangles intersecting $pq\setminus\{p\}$ is at 
most $\left\lfloor \frac{\pi}{\alpha}\right\rfloor$.
\end{proof}

\subsection{Jump to Approximate Nearest Neighbour}
\label{subsection approximate 2 D}

We now consider the scenario where $\hat{p} \in P$ is an approximate nearest 
neighbor of the query point $q$. 
Note that $\mathcal{C}(q,|\hat{p}q|)$ may contain vertices of $\meshTwoD$.
Therefore,
the proof of Proposition~\ref{prop:const_bound_2d}
does not apply for the walk-step along $\hat{p}q$.

We adopt the following strategy.
The walk along $\hat{p}q$ intersects triangles in 
$\mathcal{C}(q,|\hat{p}q|)\setminus\mathcal{C}(q,|pq|)$
and in $\mathcal{C}(q,|pq|)$.
By Proposition~\ref{prop:const_bound_2d},
we have a bound on the number of triangles intersected in $\mathcal{C}(q,|pq|)$.
To bound the triangles in $\mathcal{C}(q,|\hat{p}q|)\setminus\mathcal{C}(q,|pq|)$,
we do the following:
given an edge $ab$, there is a bound on the length of the edges of the triangles 
adjacent to $ab$ (see Observation~\ref{obs:triangle_adjacent}).
This bound depends on $|ab|$.
There is also a bound on the length of an edge intersecting $\mathcal{C}(q,|pq|)$
(see Lemma~\ref{lem:tri_interC_have_long_edge}).
This bound depends on $|pq|$.
These two results together
will lead to a bound on the length of the walk along $\hat{p}q$ inside $\mathcal{C}(q,|\hat{p}q|)\setminus\mathcal{C}(q,|pq|)$
(see Lemma~\ref{lem:shortest_escape}),
and this in turn will bound the number of triangles encountered in $\mathcal{C}(q,|\hat{p}q|)\setminus\mathcal{C}(q,|pq|)$.
This is all summarized in Theorem~\ref{thm:2d_approximate}.

We begin with the following observation.
\begin{observation}\label{obs:triangle_adjacent}
Let $t_i=\triangle abc$ be a well-shaped triangle with $a,b,c\in\meshTwoD$.
\begin{enumerate}
\item\label{obs:adjacent_edge}
Let $t_{i+1}$ be the well-shaped triangle adjacent to $t_i$ at edge $ab$.
The edges of $t_{i+1}$ have length at least $|ab|\sin\alpha$.

\item\label{obs:adjacent_vertex}
The edges of any triangle incident to $a$
have length at least $|ab|\sin^{\left\lfloor\frac{\pi}{\alpha} \right\rfloor}\alpha$.
\end{enumerate}
\end{observation}

\begin{proof}\
\begin{enumerate}
\item Follows from Lemma~\ref{lem:min_angle}.

\item When we traverse all the triangles incident to $a$,
e.g., in clockwise direction starting at $ab$,
each next edge we encounter can be smaller by a factor of not less than $\sin\alpha$.
Moreover, each next edge we encounter can be longer by a factor of not more 
than $\frac{1}{\sin\alpha}$.
This all follows from Observation~\ref{obs:triangle_adjacent}-\ref{obs:adjacent_edge}.
There are at most $\left\lfloor\frac{2\pi}{\alpha} \right\rfloor$ triangles incident 
to $a$ by Lemma~\ref{lem:min_angle}.
However, if more than $\left\lfloor\frac{\pi}{\alpha} \right\rfloor$ edges
get smaller by a factor of $\sin\alpha$,
then the last edge we encounter 
before going back to $ab$ will be smaller than $|ab|\sin\alpha$,
contradicting Observation~\ref{obs:triangle_adjacent}-\ref{obs:adjacent_edge}.
Therefore, the edges of any triangle incident to $a$
have length at least $|ab|\sin^{\left\lfloor\frac{\pi}{\alpha} \right\rfloor}\alpha$.
\end{enumerate}
\end{proof}

\begin{lemma}\label{lem:tri_interC_have_long_edge}
Let $t$ be a triangle in a well-shaped triangular mesh $\meshTwoD$,
and $\mathcal{C}(q,r)$ be a circle
such that none of the vertices of $t$ are in the interior of $\mathcal{C}$.
If edge $ab \in t$ intersects $\mathcal{C}$,
then $|ab| \geq r\tan\alpha$.
\end{lemma}

\begin{proof}
If $q\in ab$,
then $|ab|\geq 2r > r\tan\alpha$ because $\alpha\leq\frac{\pi}{3}$
by Lemma~\ref{lem:min_angle}.
For the rest of the proof,
we assume that $q\not\in ab$.

Let $t=\triangle abc$ be the triangle adjacent to $ab$ such that $c$ and $q$ are on the same side of $ab$.
By Observation~\ref{obs:triangle_adjacent}.\ref{obs:adjacent_edge},
the smaller the edges of $t$,
the smaller is $ab$.
Therefore,
since we want to minimize $|ab|$,
we assume that $c$ is on the boundary of $\mathcal{C}$.
There are two cases to consider:
(1) $ab$ intersects $\mathcal{C}$ at two different points
or (2) $ab$ is tangent to $\mathcal{C}$.
\begin{enumerate}
\item We are looking for a lower bound on $|ab|$.
Therefore,
suppose that $a$ and $b$ are on the boundary of $\mathcal{C}$.
By Lemma~\ref{lem:min_angle},
$\angle acb \geq \alpha$
(refer to Figure~\ref{fig:tri_interC_have_long_edge}(a)).

\begin{figure}
  \captionsetup[subfigure]{labelformat=parens}
    \centering
        \subfloat[Illustration of Case (1)]{
            \label{fig:2d_edge_bound}
            \includegraphics{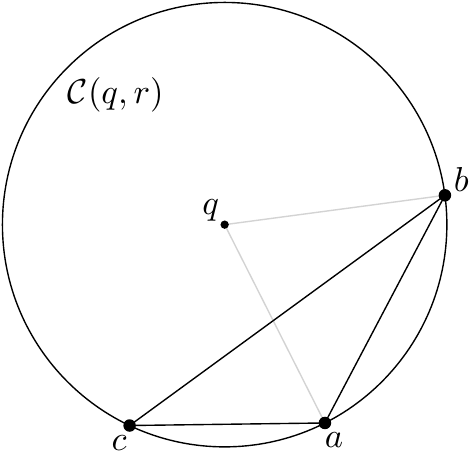}
        }\qquad
        \subfloat[Illustration of Case (2)]{
            \label{fig:tangent_edge}
            \includegraphics{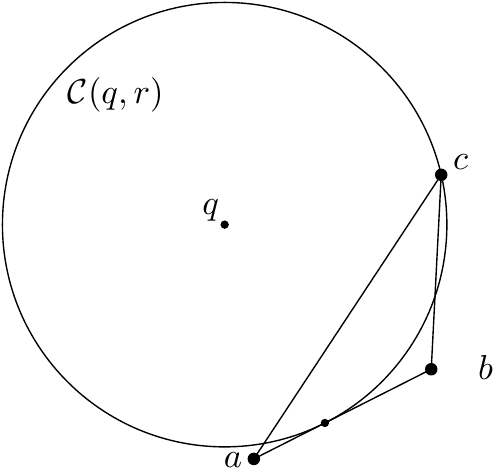}
        }%
    \caption{Illustration of the proof of Lemma~\ref{lem:tri_interC_have_long_edge}.}
    \label{fig:tri_interC_have_long_edge}
\end{figure}

Therefore,
$\angle aqb \geq 2\alpha$,
from which $|ab| \geq 2r\sin\alpha \geq r\tan\alpha$
because $\alpha\leq\frac{\pi}{3}$ by Lemma~\ref{lem:min_angle}.

\item Since $c$ is on the boundary of $\mathcal{C}$,
then $ac$ satisfies the hypothesis of Case (1)
(refer to Figure~\ref{fig:tri_interC_have_long_edge}(b)).
Therefore,
$|ac| \geq 2r\sin\alpha$.
By the law of sines,
we have $\frac{|ab|}{\sin(\angle acb)} = \frac{|ac|}{\sin(\angle abc)}$.
By Lemma~\ref{lem:min_angle},
$\angle acb \geq \alpha$ and $\angle cab \geq \alpha$,
hence $\angle abc \leq \pi-2\alpha$.
Consequently,
\begin{eqnarray*}
|ab| &=& \frac{|ac|\sin(\angle acb)}{\sin(\angle abc)}\\
&\geq& \frac{(2r\sin\alpha)\sin\alpha}{\sin(\pi-2\alpha)}\\
&=& 2r\tan\alpha \enspace.
\end{eqnarray*}
\end{enumerate}
\end{proof}

Consider the walk from $\hat{p}$ to $q$ in $\meshTwoD$;
it intersects the boundary of $\mathcal{C}(q,|pq|)$ at a point $x$.
Let $t_i$ be the last triangle\footnote{If $qx$ is contained in an edge of the mesh,
 there are two such triangles.
In this case,
the following lemma together with its proof stand for any choice of $t_i$.
In any other situation,
$t_i$ is uniquely defined.}
we traverse in the walk from $\hat{p}$ to $q$
that contains $x$.
Denote the vertices of $t_i$ by $a$, $b$, and $c$,
where $ab$ is the edge of $t_i$ that is the closest to $q$.
Let $\mathcal{G}$ be the union of all the triangles incident to $a$, $b$, and $c$
(see Fig.~\ref{fig:tricycle}).

\begin{figure}
  \captionsetup[subfigure]{labelformat=parens}
    \centering
        \subfloat[Neighbourhood $\mathcal{G}$ of $t_i$.]{
            \label{fig:tricycle}
            \includegraphics{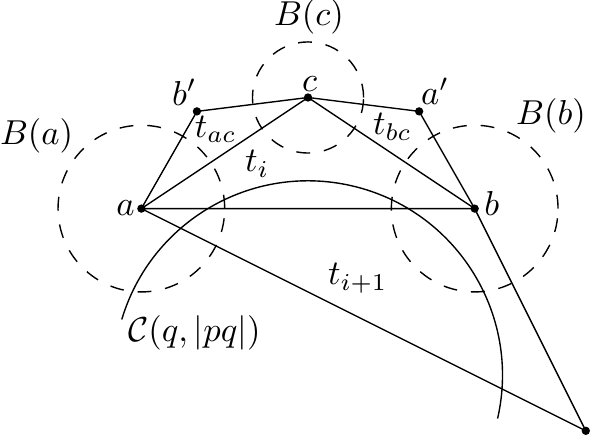}
        }\qquad
        \subfloat[Zoom in $t_i$ and $t_{ac}$.]{
            \label{fig:zoom-tricycle}
            \includegraphics{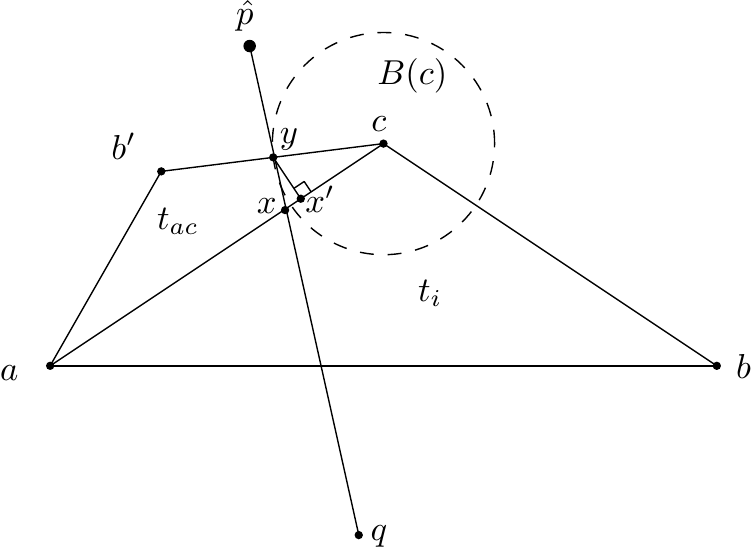}
        }%
    \caption{ Illustration of the proof of Lemma~\ref{lem:shortest_escape}.}
    \label{fig:tri_interC_have_long_edge}
\end{figure}


In Lemma~\ref{lem:shortest_escape},
we give a lower bound on the length of a walk in $\mathcal{G}$.
Thus,
if we choose $\epsilon$ with respect to this bound,
we can ensure that $\hat{p}\in\mathcal{G}$.
This way,
we have a grip on the number of tetrahedra 
we encounter in the walk from $\hat{p}$ to $q$.
\begin{lemma}\label{lem:shortest_escape}
Let $x \in t_i$ be the intersection of $\hat{p}q$ with the boundary of 
$\mathcal{C}(q,|pq|)$.
Let $y\in \mathcal{G}$ be the intersection of the line through $\hat{p}q$
with the boundary of $\mathcal{G}$
such that $x$ is between $q$ and $y$.
Then $|xy| \ge |pq|\tan\alpha\sin^{\left \lfloor \frac{\pi}{\alpha} \right \rfloor+3} \alpha$.
\end{lemma}

\begin{proof}
Denote by $t_{ac} = \triangle ab'c$
(respectively by $t_{bc} = \triangle a'bc$) the triangle adjacent to $t_i$ at $ac$
(respectively at $bc$) (see Fig.~\ref{fig:tricycle}).
Note that $t_{ac}$ and $t_{bc}$ are in $\mathcal{G}$.

By Observation~\ref{obs:triangle_adjacent}-\ref{obs:adjacent_vertex}
and Lemma~\ref{lem:tri_interC_have_long_edge},
the length of all edges incident to $a$
(respectively to $b$ and to $c$)
is at least 
$|pq|\tan\alpha\sin^{\left\lfloor\frac{\pi}{\alpha} \right\rfloor}\alpha $
(respectively at least $|pq|\tan\alpha\sin^{\left\lfloor\frac{\pi}{\alpha}\right\rfloor}\alpha $
and at least $|pq|\tan\alpha\sin^{\left\lfloor\frac{\pi}{\alpha}\right\rfloor+1}\alpha$ 
by Observation~\ref{obs:triangle_adjacent}-\ref{obs:adjacent_edge}).
Therefore, $\mathcal{G}$ contains a ball $B(a)$
(respectively $B(b)$ and $B(c)$)
with centre $a$
(respectively with center $b$ and with center $c$)
and radius\footnote{We add an extra $\sin{(\alpha)}$ factor
in order to ensure that the balls $B(a)$, $B(b)$, and $B(c)$
are contained entirely within the neighbourhood 
of $a$, $b$, and $c$,
respectively.}
$|pq|\tan\alpha\sin^{\left\lfloor\frac{\pi}{\alpha} \right\rfloor +1}\alpha$
(respectively with radius $|pq|\tan\alpha\sin^{\left\lfloor\frac{\pi}{\alpha}\right\rfloor+1}\alpha$
and with radius \\ $|pq|\tan\alpha\sin^{\left\lfloor\frac{\pi}{\alpha}\right\rfloor+2}\alpha$),
which does not contain any other vertices of $\meshTwoD$ in its interior.

We bound the length of $|xy|$ as follows:
assume without loss of generality that $\hat{p}q$ intersects $t_{ac}$ through
$ac$ and $b'c$. 
By defintion, $\hat{p} \ne c$ and $B(c)$ contains no other vertex of 
$\meshTwoD$.
Therefore, the smallest possible value for $|xy|$ is achieved when
$x \in ac$ and $y$ is at the intersection of $b'c$ with the boundary of
$B(c)$. 
Let $x'$ be the orthogonal projection of $y$ on $ac$. 
Then $|x'y| \le |xy|$ and the smallest possible length for $|xy|$ is 
$|pq|\tan\alpha\sin^{\left\lfloor\frac{\pi}{\alpha} \right\rfloor+3}\alpha$
by the well-shaped property of $t_{ac}$ and Lemma~\ref{lem:min_angle}.

\end{proof}

We can now state our main result.
\begin{theorem}\label{thm:2d_approximate}
Let $\meshTwoD$ be a well-shaped triangular mesh in $\CDTwoD$. 
Take $\epsilon \leq \tan\alpha\sin^{\left \lfloor \frac{\pi}{\alpha} \right \rfloor +3} \alpha$,
and let $\hat{p}$ be an $(1 + \epsilon)$-approximate nearest neighbour of a 
query point $q$ from among the vertices of $\meshTwoD$.
The straight line walk from $\hat{p}$ to $q$ visits at most 
$2\left\lfloor \frac{\pi}{\alpha}\right\rfloor$
triangles.
\end{theorem}

\begin{proof}
Following the notation of Lemma~\ref{lem:shortest_escape},
if $\hat{p} \in \mathcal{G}$,
then the straight line walk from $\hat{p}$ to $q$ visits
at most $2\left\lfloor \frac{\pi}{\alpha}\right\rfloor$ triangles.
There are $\left\lfloor \frac{\pi}{\alpha}\right\rfloor$ triangles
for the part of the walk inside $\mathcal{C}(q,|pq|)$
(see Proposition~\ref{prop:const_bound_2d})
and $\left\lfloor \frac{\pi}{\alpha}\right\rfloor$ triangles
for the part of the walk inside $\mathcal{G}$.
Indeed,
in the worst case,
the walk inside $\mathcal{G}$
will either cross $ab'$, $b'c$, $ca'$ or $a'b$.
Therefore, this walk will either cross the triangles incident to $a$,
or the triangles incident to $b$,
or the triangles incident to $c$.

We can ensure that $\hat{p} \in \mathcal{G}$
by building an ANN search structure with 
$\epsilon \leq \tan\alpha\sin^{\left \lfloor \frac{\pi}{\alpha} \right \rfloor +3} \alpha$ on the vertices of $\meshTwoD$.
Indeed, in this case:
\begin{eqnarray*}
|\hat{p}q| & \leq & \left(1+\tan\alpha\sin^{\left \lfloor \frac{\pi}{\alpha} \right \rfloor +3} \alpha\right)|pq|\\
&=& |pq| + |pq|\tan\alpha\sin^{\left \lfloor \frac{\pi}{\alpha} \right \rfloor +3} \alpha\\
&\leq& |pq| + |xy| \qquad\qquad\textrm{by Lemma~\ref{lem:shortest_escape},}\\
&=& |qx| + |xy|\\
&=& |qy|
\end{eqnarray*}
because $q$, $x$, and $y$ are aligned in this order.
So $\hat{p}$ must be in $\mathcal{G}$.
\end{proof}

\section{Jump-and-Walk in $\meshThreeD$}\label{sec:j-w-3D}

Searching in a well-shaped three-dimensional mesh can be performed using 
essentially the same technique as outlined for the two-dimensional case
 (refer to Section~\ref{sec:j-w-2D}).
Let $P$ be the set of vertices of a well-shaped mesh $\meshThreeD$, and 
$q$ be a query point lying in a triangle of $\meshThreeD$.
We first study the walk step given an exact nearest neighbour,
and later we address the walk step given an approximate nearest neighbour.

\subsection{Jump to Nearest Neighbour}

Let $p$ be a nearest neighbour of $q$.
We perform the walk-step starting at $p$, and walk towards $q$ on a straight line.
Theorem~\ref{thm:exact_3d} states that $pq$ intersects only a constant number 
of tetrahedra.
It is the three-dimensional version of Proposition~\ref{prop:const_bound_2d}.
\begin{theorem}\label{thm:exact_3d}
Let $\meshThreeD$ be a well-shaped tetrahedral mesh in $\threeD$. 
Given $p$, a nearest neighbour of a query point $q$ from among the vertices 
of $\meshThreeD$, the walk from $p$ to $q$ visits at most 
$\frac{\sqrt{3}\,\pi}{486}\rho^3(\rho^2+3)^3$ tetrahedra.
\end{theorem}

The proof of Theorem~\ref{thm:exact_3d} is not as simple as the one in the 
two-dimensional case.
We begin by showing that the first tetrahedron intersected by $pq$, when 
we travel from $p$ towards $q$, covers a constant fraction of $|pq|$. 
This is done step by step in Lemmas~\ref{lem:minmax_of_circle},
\ref{lem:minmax_psi},
and~\ref{lem:pp_prime_covers_const_frac}.
Next,
using this constraint on the first tetrahedron intersected by $pq$,
we find a lower bound on the volume of any tetrahedron intersected by $pq$.
This lower bound compared to the volume of $\mathcal{S}(q,|pq|)$
leads to an upper bound on the number of tetrahedra crossing $pq$.

Let $\tau$ be the first tetrahedron intersected by $pq$ when we travel from 
$p$ towards $q$.
Denote by $f$ the face opposite to $p$ in $\tau$.
If $q$ is interior to $\tau$, then clearly $\tau$ 
covers all of $pq$.
Therefore,
suppose $q$ is not interior to $\tau$, and let $pq$ intersects $f$ at a point  $p'$ (see 
Figure~\ref{fig:fface_and_circle_abc}(a)). 

\begin{figure}
  \captionsetup[subfigure]{labelformat=parens}
    \centering
        \subfloat[The segment $pq$.]{
            \label{fig:first_face}
            \includegraphics[scale=.55]{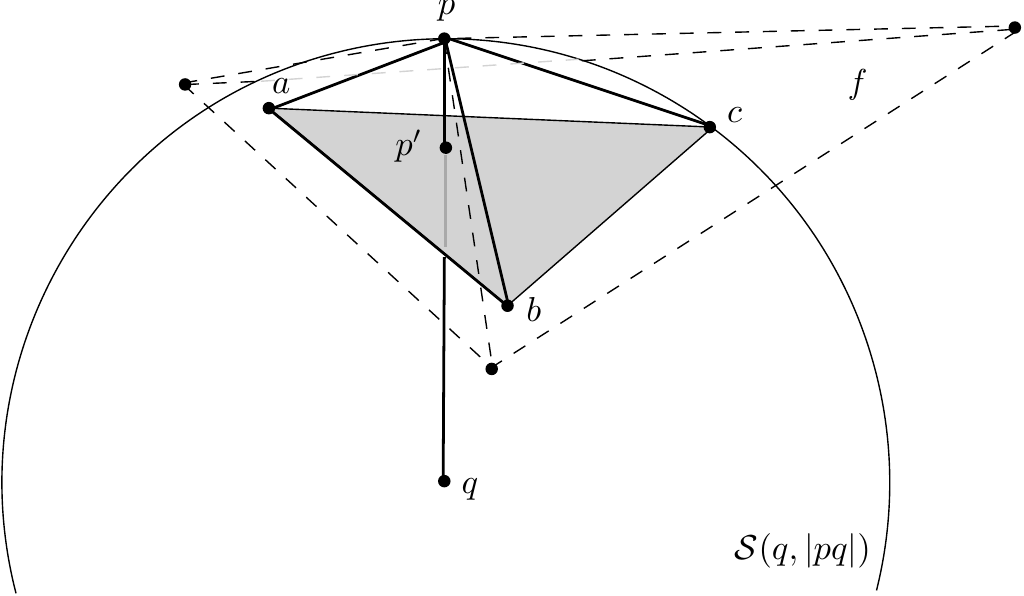}
        }\qquad
        \subfloat[Circle formed by intersection of $\mathcal{S}(q,|pq|)$ and $\triangle abc$.]{
            \label{fig:circle_abc}
            \includegraphics[scale=0.7]{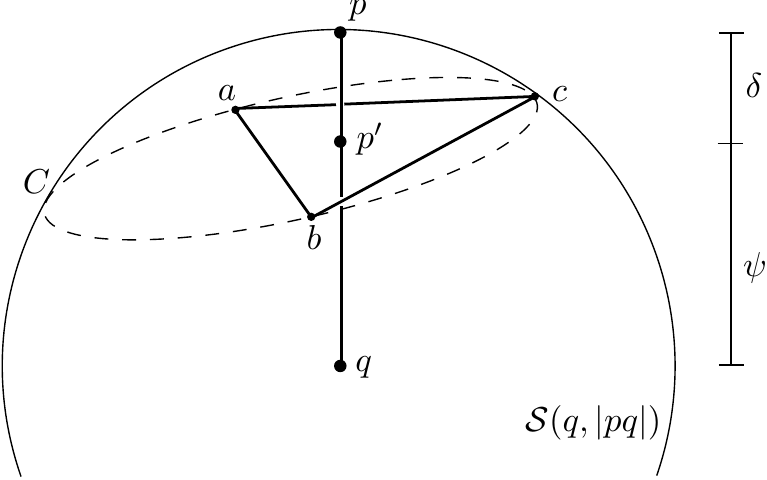}
        }%
    \caption[First tetrahedron on walk from $p$ to $q$]{Illustration of the first 
      tetrahedron encountered in walk from $p$ to $q$.}\label{fig:fface_and_circle_abc}
\end{figure}


Without loss of generality, and for the proof of Lemmas~\ref{lem:minmax_of_circle},
\ref{lem:minmax_psi},
and~\ref{lem:pp_prime_covers_const_frac}, suppose that $|pq| = 1$. 
Let $C=\mathcal{C}(o,r)$ be the circle formed by the intersection of the supporting plane of $f$ with $\mathcal{S}(q,|pq|)$.
Consider all triangles $\triangle abc$ satisfying
\begin{enumerate}
\item $\triangle abc \subseteq f$,
\item $p'\in \triangle abc$, and
\item the vertices $a$, $b$, and $c$ are on the boundary of $C$.
\end{enumerate}
(see Figure~\ref{fig:fface_and_circle_abc}(b)).
Let $\CDMinMax(\triangle abc)$ denote the shortest maximum edge length of $\triangle abc$
over all such possible triangles.
By construction and Observation~\ref{obs:ins_circ_phere}-\ref{obs:circumsphere_bound},
$\CDMinMax(\triangle abc)$ is a lower bound on the diameter of the circumsphere of $\tau$. 
The following lemma expresses $\CDMinMax(\triangle abc)$ in terms of $|op'|$ and $r$.

\begin{lemma}\label{lem:minmax_of_circle}
Let $\lambda=|op'|$.
\begin{enumerate}
\item If $\frac{1}{2}r \le \lambda < r$, then $\CDMinMax(\triangle abc) = 2 \sqrt{ r^2 - \lambda^2 }$.
\item If $0 \leq \lambda \leq \frac{1}{2}r$, then $\CDMinMax(\triangle abc)=\sqrt{3}\,r$.
\end{enumerate}
\end{lemma}

\begin{proof}
Without loss of generality,
the equation of $C$ is $x^2+y^2=r^2$,
and $p'=(0,\lambda)$
(refer to Fig.~\ref{fig:minmax_parts}(b)).

\begin{figure}
  \captionsetup[subfigure]{labelformat=parens}
    \centering
        \subfloat[]{
            \label{fig:minmax1}
            \includegraphics{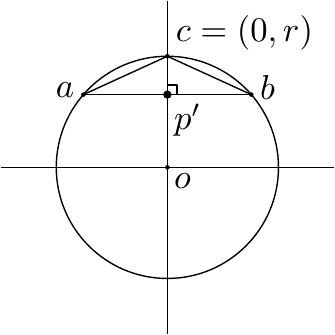}
        } \qquad
        \subfloat[]{
            \label{fig:minmax2}
            \includegraphics{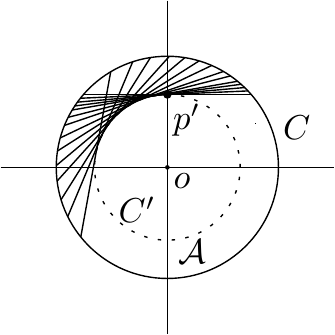}
        }%
    \caption{Illustration of the proof of Lemma~\ref{lem:minmax_of_circle}}
    \label{fig:minmax_parts}
\end{figure}

\begin{enumerate}
\item Suppose $\frac{1}{2}r \leq \lambda<r$.
Let $a$, $b$, and $c$ be three points on the boundary of $C$
such that $c=(0,r)$,
$p'\in ab$ and $ab$ is perpendicular to $oc$.
By the Pythagorean theorem, $|ab| = 2\sqrt{r^2-\lambda^2}$.
By the hypothesis,
the longest side in $\triangle abc$ is $ab$,
so $\CDMinMax(\triangle abc) \leq 2 \sqrt{r^2 - \lambda^2}$.
Let us now prove the equality.

Rotate $ab$ around $o$ as shown in Figure~\ref{fig:minmax_parts}(b).
This rotation traces a new circle $C'=\mathcal{C}(o,2\sqrt{r^2-\lambda^2})$ 
interior to $C$ such that $p'$ is on the boundary of $C'$. 
By elementary geometry, a line segment inscribed in $C$
\begin{itemize}
\item is tangent to $C'$ if and only if it is of length $2\sqrt{r^2-\lambda^2}$,

\item intersects the interior of $C'$ if and only if it is of length greater than $2\sqrt{r^2-\lambda^2}$,

\item and is entirely within the annulus $\mathcal{A}$ defined by $C$ and $C'$ if and only if it is of length less than $2\sqrt{r^2-\lambda^2}$.
\end{itemize}

Suppose there exists a triangle $t$ inscribed in $C$ and containing $p'$ such that all sides are shorter than $2\sqrt{r^2-\lambda^2}$; then the edges of $t$ are entirely inside $\mathcal{A}$. Therefore, since $\frac{1}{2}r \leq \lambda<r$, $t$ does not contain $C'$ and hence does not contain $p'$.
This is a contradiction.

\item Suppose $0 < \lambda \leq \frac{1}{2}r$. 
In this case, it is possible to construct an equilateral triangle inscribed in $C$ that includes $p'$. 
The length of all sides of this triangle are $\sqrt{3}\,r$. 
Therefore, $\CDMinMax(\triangle abc) \leq \sqrt{3}\,r$. 
Let us now prove the equality.

Define $C'$ and $\mathcal{A}$ as in the previous case,
by rotating a line segment of length $\sqrt{3}\,r$ around $o$.
Suppose there exists a triangle $t$ inscribed in $C$ and containing $p'$,
such that all sides are shorter than $\sqrt{3}\,r$;
then all edges of $t$ are entirely inside $\mathcal{A}$.
However,
since $0 < \lambda \leq \frac{1}{2}r$,
$t$ does not contain $C'$.
Therefore,
$t$ does not contain $p'$,
which is a contradiction.
\end{enumerate}
\end{proof}

The next lemma expresses $\CDMinMax(\triangle abc)$ in terms of $\psi$, the distance from $q$ to $p'$.
\begin{lemma}\label{lem:minmax_psi}
$\CDMinMax(\triangle abc) \geq \sqrt{3} \sqrt{ 1 - \psi^2 }$, where $\psi = |qp'|$.
\end{lemma}

\begin{proof}
Locate $\mathcal{S}(q,|pq|)$ in the usual system of axes $XYZ$ in such a way that
$q$ is on the origin of the system of axes,
and the line through $o$ and $p'$ lies in the $YZ$ plane.
Denote by $\theta$ the smallest angle between the $Z$, axis and the line through $o$ and $p'$
(see Figure~\ref{fig:minmax}).

\begin{figure}
	\centering
 	\includegraphics[scale=0.75]{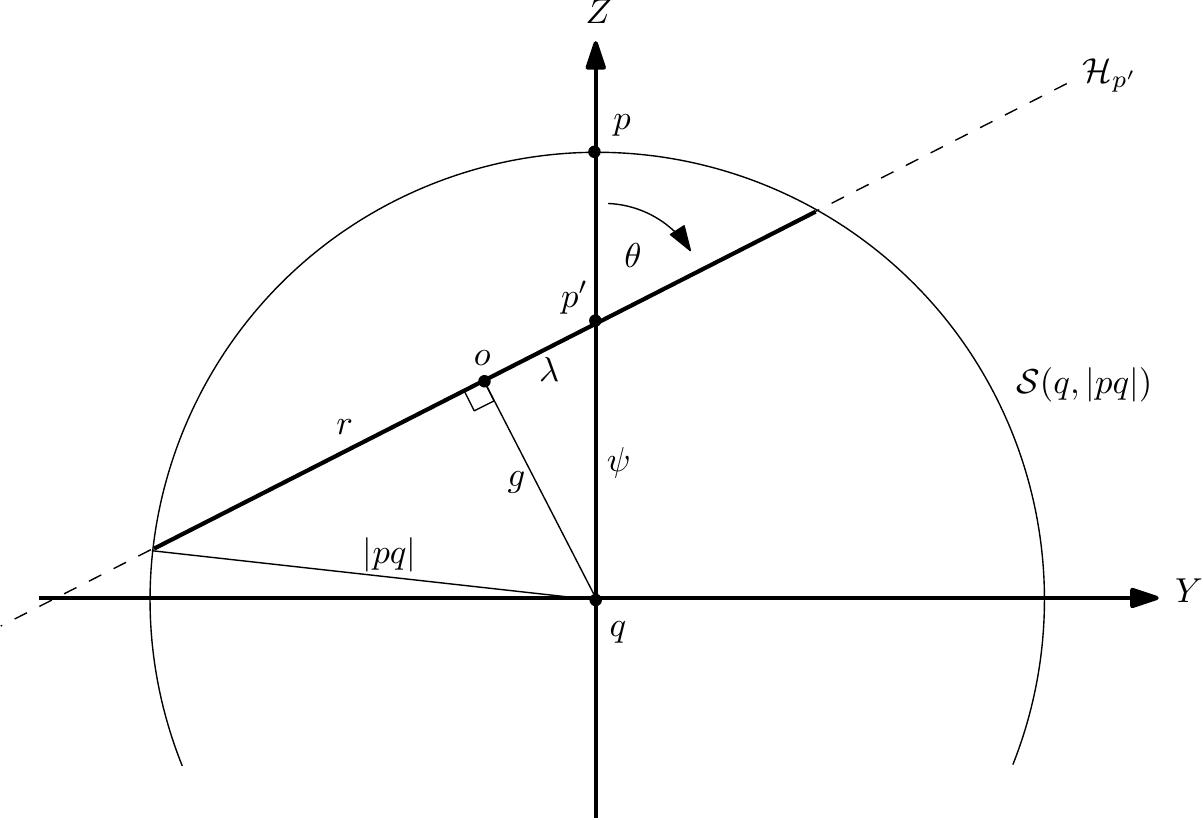}
	\caption[Cross-section of $\mathcal{S}$ and $C$]{A cross-section of 
	$\S(q,|pq|)$ and $C$ in the $YZ$ plane as viewed 
	from the direction of the positive $X$ axis.}
	\label{fig:minmax}
\end{figure}

First we determine $\lambda$ and $r$. 
Since $\angle op'q = \theta$, we have 
$\lambda = \psi \cos{\theta}$ and $r =  \sqrt{ 1 - \psi^2 \sin^2\theta}$ by 
the Pythagorean theorem.
If $\frac{1}{2}r \le \lambda < r$,
then by Lemma~\ref{lem:minmax_of_circle}
we have

\begin{eqnarray*}
\CDMinMax(\triangle abc) &=& 2 \sqrt{ r^2 - \lambda^2 } \\
 &=& 2 \sqrt{ 1 - \psi^2 \sin^2{\theta} - \psi^2 \cos^2{\theta} } \\
 &=& 2 \sqrt{ 1 - \psi^2 } \\
 &>& \sqrt{3} \sqrt{ 1 - \psi^2 } \enspace.
\end{eqnarray*}
If $0 \le \lambda \leq \frac{1}{2}r$,
then we have
\begin{eqnarray*}
\CDMinMax(\triangle abc) &=& \sqrt{3} r\\
 &=& \sqrt{3} \sqrt{ 1 - \psi^2 \sin^2\theta}\\
 &\geq& \sqrt{3} \sqrt{ 1 - \psi^2}
\end{eqnarray*}
for any $\psi$.
\end{proof}

We can now prove the following lemma.
\begin{lemma}\label{lem:pp_prime_covers_const_frac}
The segment $pp'$ covers at least a constant fraction of $|pq|$.
\end{lemma}

\begin{proof}
Let $e$ be the longest edge of $\tau$.
\begin{eqnarray*}
\frac{R(\tau)}{r(\tau)} &\geq& \frac{\frac{|e|}{2}}{\frac{\mathtt{mdist}(p,\tau)}{2}} \qquad\textrm{by Observation~\ref{obs:ins_circ_phere}}\\
&\geq& \frac{\frac{\CDMinMax (\triangle abc)}{2}}{\frac{|pp'|}{2}} \qquad\textrm{by construction}\\
&\geq& \frac{\sqrt{3} \sqrt{ 1 - \psi^2}}{|pp'|} \qquad\textrm{by Lemma~\ref{lem:minmax_psi}}\\
&=& \frac{\sqrt{3} \sqrt{ 1 - \psi^2}}{1 - \psi}
\end{eqnarray*}

We now solve
\begin{eqnarray*}
\frac{\sqrt{3} \sqrt{ 1 - \psi^2}}{1 - \psi} &=& \rho \\
\frac{3(1 - \psi^2)}{(1 - \psi)^2} &=& \rho^2 \\
\frac{3(1 + \psi)}{1 - \psi} &=& \rho^2 \\
(\rho^2+3)\psi &=& \rho^2-3 \\
\psi &=& \frac{\rho^2-3}{\rho^2+3} \enspace.
\end{eqnarray*}
Since $\frac{\sqrt{3} \sqrt{ 1 - \psi^2}}{1 - \psi}$ is an increasing function of $\psi$
for $0 \leq \psi < 1$,
then we must have $\psi < \frac{\rho^2-3}{\rho^2+3}$,
otherwise the aspect ratio of $\tau$ is bigger than $\rho$.
Therefore,
$|pp'| = 1-\psi > \frac{6}{\rho^2+3}$.
In general,
$|pp'| > \frac{6}{\rho^2+3}|pq|$.
\end{proof}

We are now in a position to prove Theorem~\ref{thm:exact_3d}.

\begin{proof}[Proof of Theorem~\ref{thm:exact_3d}]
The aim is to construct a tetrahedron $t^*$
that satisfies the aspect ratio hypothesis,
and that is smaller in volume than any tetrahedron crossing $pq$.
The volume of $\mathcal{S}(q,|pq|)$ divided by the volume of $t^*$
is therefore an upper bound on the number of tetrahedra crossing $pq$.
Let $t^* = abcd$ be a tetrahedron satisfying the aspect ratio hypothesis
and such that $|ab|=|ac|=|ad|=|pp'|$.
By the construction of $p'$,
$t^*$ is smaller in volume than any tetrahedron crossing $pq$.
Let $e^*$ be the longest edge of $t^*$.
By Observation~\ref{obs:ins_circ_phere}-\ref{obs:circumsphere_bound},

\begin{eqnarray*}
R(t^*) &\geq& \frac{|e^*|}{2} \\
R(t^*) &\geq& \frac{|pp'|}{2}\qquad\textrm{by construction,}\\
\frac{R(t^*)}{r(t^*)} &\geq& \frac{|pp'|}{2r(t^*)} \\
\rho &>& \frac{|pp'|}{2r(t^*)}\qquad\textrm{by the hypothesis,}\\
r(t^*) &>& \frac{|pp'|}{2\rho} \\
r(t^*) &>& \frac{3}{\rho(\rho^2+3)}|pq| \qquad\textrm{by Lemma~\ref{lem:pp_prime_covers_const_frac}.}\\
\end{eqnarray*}

Therefore,
the volume of the insphere of $t^*$ is at least $\frac{36\pi}{\rho^3(\rho^2+3)^3}|pq|^3$.
By elementary geometry,
the smallest possible ratio comparing the volume of a tetrahedron and the volume 
of its insphere is $\frac{6\sqrt{3}}{\pi}$, in the case of a regular tetrahedron.
Therefore, the volume of $t^*$ is at least 
$\frac{6\sqrt{3}}{\pi}\times\frac{36\pi}{\rho^3(\rho^2+3)^3}|pq|^3 = 
\frac{216\sqrt{3}}{\rho^3(\rho^2+3)^3}|pq|^3$.
Hence there are at most
$$\frac{\frac{4}{3}\pi|pq|^3}{\frac{216\sqrt{3}}{\rho^3(\rho^2+3)^3}|pq|^3}=\frac{\sqrt{3}\,\pi}{486}\rho^3(\rho^2+3)^3$$
tetrahedra crossing $pq$.
\end{proof}

\subsection{Jump to Approximate Nearest Neighbour}\label{3d-ann}

We now examine the case where we jump to $\hat{p}$,
an approximate nearest neighbouqr of $q$ in $\CDThreeD$.
The following observation is the three-dimensional version
of Observation~\ref{obs:triangle_adjacent}.
\begin{observation}\label{obs:tetrahedron_adjacent}
Let $t_i = abcd$ be a well-shaped tetrahedron with
$a,b,c,d \in \meshThreeD$.
\begin{enumerate}
\item\label{obs:adjacent_edge3D}
Let $t_{i+1}$ be a well-shaped tetrahedron adjacent to $t_i$ at edge $ab$.
The edges of $t_{i+1}$ have length of at least $\frac{1}{\rho}|ab|$.

\item\label{obs:adjacent_vertex3D}
The edges of any tetrahedron incident to $a$
have length of at least $|ab|\left(\frac{1}{\rho}\right)^{\left\lfloor\frac{2\pi}{\Omega} \right\rfloor}$.
\end{enumerate}
\end{observation}

\begin{proof}\
\begin{enumerate}
\item Let $e$ be any edge of $t_{i+1}$,
and let $v$ be any of the two vertices of $e$.
By Observation~\ref{obs:ins_circ_phere}-\ref{obs:insphere_bound},
$2r(t_{i+1}) \leq \mathtt{mdist}(v,t_{i+1}) \leq |e|$.
By Observation~\ref{obs:ins_circ_phere}-\ref{obs:circumsphere_bound},
$2R(t_{i+1}) \geq |ab|$.
Therefore,
$$\frac{|ab|}{|e|} \leq \frac{R(t_{i+1})}{r(t_{i+1})} < \rho \enspace,$$
from which $|e| > \frac{1}{\rho}|ab|$.

\item The proof is similar to the one of Observation~\ref{obs:triangle_adjacent}-\ref{obs:adjacent_vertex}.
It uses Observation~\ref{obs:tetrahedron_adjacent}-\ref{obs:adjacent_edge3D}
and the fact that the full solid angle is $4\pi$.
\end{enumerate}
\end{proof}

Consider an arbitrary ball $\mathcal{S}$
that contains no vertex of $\meshThreeD$.
In order to obtain a three-dimensional version of Lemma~\ref{lem:tri_interC_have_long_edge},
we need to distinguish between two cases.
If an edge $e$ of the mesh intersects $\mathcal{S}$,
then there is a lower bound on $|e|$.
This is the subject of Lemma~\ref{lem:ab_is_long}.
If a face $f$ of the mesh intersects $\mathcal{S}$
without having any of its edges intersecting $\mathcal{S}$,
then there is a lower bound on the length of at least one of the edges of $f$.
This is the subject of Lemma~\ref{lem:face_abc_has_long_edge}.

\begin{lemma}
\label{lem:ab_is_long}
Let $ab$ be an edge in a well-shaped tetrahedral mesh $\meshThreeD$, and 
$\mathcal{S}(q,r)$ be a sphere intersecting $\mathcal{M}$ but containing no vertex 
of $\mathcal{M}$ in its interior. 
If $ab$ intersects $\mathcal{S}$, then $|ab|\geq\frac{2}{\rho^2+1}r$.
\end{lemma}

\begin{proof}
We are looking for a lower bound on $|ab|$.
Therefore,
we assume that $a$ and $b$ are on the surface of $\mathcal{S}$.
Denote by $\pi_{qab}$ the plane defined by $q$, $a$, and $b$.
Locate $S$ and $ab$ in the usual system of axes $XYZ$ in such a way that
\begin{itemize}
\item $q$ is on the origin of the system of axes,

\item $\pi_{qab}$ corresponds to the $YZ$ plane, and

\item $ab$ intersects perpendicularly the positive $Z$ axis.
\end{itemize}
Let $\mathcal{H}_{ab}$ be the  plane containing $ab$ that is parallel to the $XY$ plane,
and consider the tetrahedra that are adjacent to $ab$.
Without loss of generality,
let $t=abcd$ be a tetrahedron such that
\begin{itemize}
\item $c$ is below $\mathcal{H}_{ab}$,

\item $d$ is not lower than $c$,

\item the $X$ coordinate of $c$ is non-negative, and

\item the $X$ coordinate of $d$ is non-positive.
\end{itemize}
Such a tetrahedron always exists.
Indeed,
take a tetrahedron $t'$ adjacent to $ab$
such that $c$ is below $\mathcal{H}_{ab}$
and $d$ is not lower than $c$.
If $t'$ does not satisfy the two remaining constraints,
one of the neighbouring tetrahedra of $t'$ does.
Let $t$ be this tetrahedron.

By Observation~\ref{obs:tetrahedron_adjacent}-\ref{obs:adjacent_edge3D},
the smaller the edges of $t$ are,
the smaller $ab$ is.
Therefore,
since we want to minimize $|ab|$,
we assume that $c$ and $d$ are on the surface of $\mathcal{S}$.
By the hypotheses,
if $c$ is below the $XY$ plane,
then $|ac|\geq r$ or $|bc|\geq r$.
Thus, $|ab| \geq \frac{1}{\rho}r$ 
by Observation~\ref{obs:tetrahedron_adjacent}-\ref{obs:adjacent_edge3D},
from which $|ab| \geq \frac{1}{\rho}r > \frac{2}{\rho^2+1}r$
because $\rho \geq 3 > 1$.

For the rest of the proof,
suppose $c$ is inclusively between the $XY$ plane and $\mathcal{H}_{ab}$.
Let $\mathcal{H}_{c}$ be the plane containing $c$ that is parallel to the $XY$ plane.
Since $d$ is not lower than $c$,
there are two cases to consider:
either (1) $d$ is inclusively between $\mathcal{H}_{c}$ and $\mathcal{H}_{ab}$,
or (2) $d$ is above $\mathcal{H}_{ab}$.
(Note that in the proof of Case~\ref{case inclusively between},
we do not need the hypotheses
about the $X$ coordinates of $c$ and $d$.
However,
these hypotheses are indispensable for the proof of Case~\ref{case above}.)
\begin{enumerate}
\item\label{case inclusively between}
Denote by $a'$ the intersection point of $ab$ with the positive $Z$ axis
(refer to Fig.~\ref{fig:t_between_planes}).
\begin{figure}
	\centering
	  \includegraphics[scale=1.00]{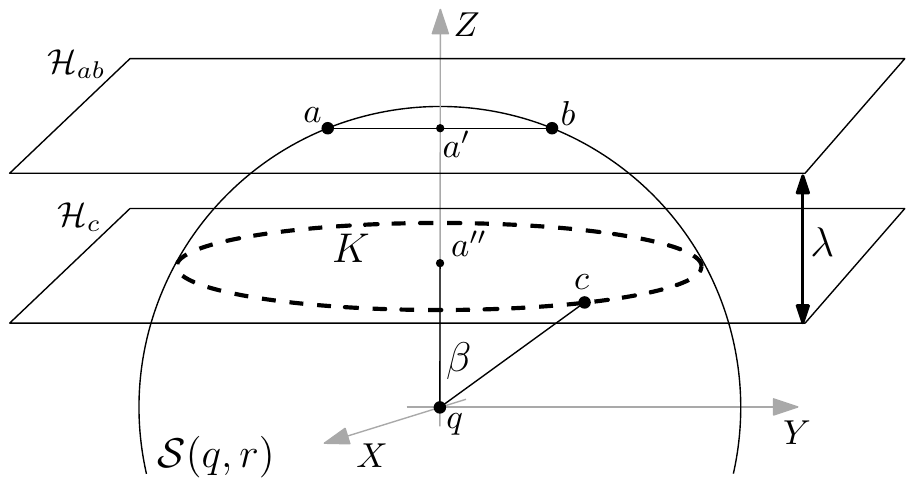}
	\caption{By the hypotheses, $t$ is bounded by $\mathcal{H}_{c}$ and $\mathcal{H}_{ab}$}
	\label{fig:t_between_planes}
\end{figure}
Let $a''$ be the intersection point of $\mathcal{H}_{c}$ with the $Z$ axis.
Denote by $K$ the circle defined by the intersection of $\mathcal{S}$
and $\mathcal{H}_{c}$.
Note that $a''$ is the center of $K$,
and that $c$ is on the boundary of $K$.
Finally,
let $\lambda$ be the distance between $\mathcal{H}_{c}$ and $\mathcal{H}_{ab}$.

Since $t$ lies between $\mathcal{H}_{c}$ and $\mathcal{H}_{ab}$,
$r(t) \leq \frac{1}{2}\lambda$.
Furthermore,
between the two edges connecting $ab$ and $c$,
there is at least one edge of length greater than $a''c$.
Without loss of generality,
let $ac$ be this edge.
Therefore,
$R(t) \geq \frac{1}{2}|ac|$
by Observation~\ref{obs:ins_circ_phere}-\ref{obs:circumsphere_bound}.
Thus,
since $t$ is well-shaped,
\begin{equation}
	\rho > \frac{R(t)}{r(t)} \ge \frac{|ac|}{\lambda} \ge \frac{|a''c|}{\lambda} \enspace.
\end{equation}

We also have
$|a''c| = r\sin{\beta}$ 
and $\lambda = |a''a'| \leq (1 - \cos{\beta})r$,
where $\beta = \angle a''qc$. 
Thus,
$\rho > \frac{|a''c|}{\lambda} \geq \frac{\sin{\beta}}{1 - \cos{\beta}} = \cot\left(\frac{1}{2}\beta\right)$,
which implies $\beta > 2\,\arccot(\rho)$.
Hence,
$|ac| \ge |a''c| > r\sin(2\,\arccot(\rho)) = \frac{2\rho}{\rho^2+1}r$,
from which
$|ab| \geq \frac{2}{\rho^2+1}r$ 
by Observation~\ref{obs:tetrahedron_adjacent}-\ref{obs:adjacent_edge3D}.

\item\label{case above}
Denote by $\pi_{qad}$ the plane defined by $q$, $a$, and $d$.
Relocate $\mathcal{S}$ and $t$ such that
\begin{itemize}
\item $q$ is on the origin of the system of axes,

\item $\pi_{qad}$ corresponds to the $YZ$ plane, and

\item $ad$ intersects perpendicularly the positive $Z$ axis.
\end{itemize}
Denote by $\mathcal{H}_{ad}$ the plane containing $ad$ that is parallel to the $XY$ plane.
We claim that $b$ is below $\mathcal{H}_{ad}$,
and that $c$ is not above $\mathcal{H}_{ad}$.
Therefore,
we can apply the proof of Case (1).
Consequently,
if $b$ or $c$ is below the $XY$ plane,
there is an edge $e$ of $t$ such that $|e|\geq r$.
Thus, $|ab| \geq \frac{1}{\rho}r > \frac{2}{\rho^2+1}r$ 
by Observation~\ref{obs:tetrahedron_adjacent}-\ref{obs:adjacent_edge3D}.
If $b$ and $c$ are both inclusively between the $XY$ plane and $\mathcal{H}_{ad}$,
there is an edge $e$ of $t$ with $|e|\geq\frac{2\rho}{\rho^2+1}r$,
from which
$|ab| \geq \frac{2}{\rho^2+1}r$ 
by Observation~\ref{obs:tetrahedron_adjacent}-\ref{obs:adjacent_edge3D}.
\end{enumerate}

In order to complete this proof, one detail remains.
We must prove our claim that 
$b$ is below $\mathcal{H}_{ad}$,
and that $c$ is not above $\mathcal{H}_{ad}$.
To prove our claim,
we stay with the orientation
of the previous case
(see Figure~\ref{fig:t_between_planes}).
We suppose without loss of generality that $r=1$.
Therefore,
the coordinates $(X,Y,Z)$ of $a$, $b$, $c$ and $d$ are
\begin{eqnarray*}
a&=&\left(0,-s,\sqrt{1-s^2}\right) \enspace,\\
b&=&\left(0,s,\sqrt{1-s^2}\right) \enspace,\\
c&=&\left(u,v,\sqrt{1-u^2-v^2}\right) \enspace,\\
d&=&\left(m,n,\sqrt{1-m^2-n^2}\right) \enspace,
\end{eqnarray*}
where
\begin{eqnarray}
\label{eqn s positive}
&&\phantom{-}0 \leq s <  1 \enspace,\\
\label{eqn u positive}
&&\phantom{-}0\leq u \leq 1 \qquad\textrm{since the $X$ coordinate of $c$ is non-negative,}\\
\label{eqn range for v}
&&-1\leq v \leq 1 \enspace,\\
\label{eqn m negative}
&&-1\leq m \leq 0 \qquad\textrm{since the $X$ coordinate of $d$ is non-positive,}\\
\label{eqn range for n}
&&-1\leq n \leq 1 \enspace.
\end{eqnarray}
Since $c$ is below $\mathcal{H}_{ab}$,
then
\begin{eqnarray}
\label{eqn c below Hab}
\sqrt{1-u^2-v^2} &<& \sqrt{1-s^2} \enspace,\\
\label{eqn s^2 < u^2+v^2}
s^2 &<& u^2+v^2 \enspace.
\end{eqnarray}
Since $d$ is above $\mathcal{H}_{ab}$,
then
\begin{eqnarray}
\label{eqn d above Hab}
\sqrt{1-s^2} &<& \sqrt{1-m^2-n^2} \enspace,\\
m^2+n^2 &<& s^2 \enspace,\\
\pm\,m &<& s \enspace,\\
\label{eqn s bigger than n}
\pm\,n &<& s \enspace.
\end{eqnarray}

The vector from $q$ to the midpoint of $ad$ is
$$\left( \frac{1}{2}m,\frac{1}{2}(n-s),\frac{1}{2}\left(\sqrt{1-m^2-n^2}+\sqrt{1-s^2}\right) \right) \enspace,$$
hence the equation of $\mathcal{H}_{ad}$ is
$$\mathcal{H}_{ad}:f(x,y)=\frac{-mx-(n-s)y-sn+\sqrt{1-m^2-n^2}\sqrt{1-s^2}+1}{\sqrt{1-m^2-n^2}+\sqrt{1-s^2}} \enspace.$$
In order to prove our claim,
we need to show that $\sqrt{1-s^2} < f(0,s)$,
and $\sqrt{1-u^2-v^2} \leq f(u,v)$.

We start with the first inequality.
\begin{eqnarray*}
0 &<& s-n \qquad\textrm{by (\ref{eqn s bigger than n})}\\
0 &\leq& 2(s-n)s \qquad\textrm{by (\ref{eqn s positive})}\\
-s^2 &<& -2sn+s^2 \\
\sqrt{1-m^2-n^2}\sqrt{1-s^2}+(1-s^2) &<& -2sn+s^2+\sqrt{1-m^2-n^2}\sqrt{1-s^2}+1 \\
\sqrt{1-s^2} &<& \frac{-2sn+s^2+\sqrt{1-m^2-n^2}\sqrt{1-s^2}+1}{\sqrt{1-m^2-n^2}+\sqrt{1-s^2}} \\
\sqrt{1-s^2} &<& f(0,s)
\end{eqnarray*}

In order to prove the second inequality,
we distinguish between two cases:
according to (\ref{eqn range for v}),
either (a) $0 \leq v \leq 1$
or (b) $-1\leq v < 0$.
\begin{enumerate}
\item\label{case v positive}
Suppose
\begin{eqnarray}
\label{eqn v positive}
0 \leq v \leq 1 \enspace.
\end{eqnarray}
We want to prove that $\sqrt{1-u^2-v^2} < f(u,v)$.
By (\ref{eqn c below Hab}),
it is sufficient to prove that $\sqrt{1-s^2} < f(u,v)$.
Since
\begin{eqnarray*}
&&\sqrt{1-s^2} < f(u,v) \enspace,\\
&\Leftrightarrow& \sqrt{1-s^2} < \frac{-mu-(n-s)v-sn+\sqrt{1-m^2-n^2}\sqrt{1-s^2}+1}{\sqrt{1-m^2-n^2}+\sqrt{1-s^2}} \enspace,\\
&\Leftrightarrow& \sqrt{1-m^2-n^2}\sqrt{1-s^2}+(1-s^2) < -mu-(n-s)v-sn+\sqrt{1-m^2-n^2}\sqrt{1-s^2}+1 \enspace,\\
&\Leftrightarrow& 0< s^2-mu-nv+sv-sn \enspace,
\end{eqnarray*}
we will prove that $s^2-mu-nv+sv-sn$ is positive.

We have
\begin{eqnarray*}
0 &<& s-n \qquad\textrm{by (\ref{eqn s bigger than n}),}\\
0 &<& (s+v)(s-n) \qquad\textrm{by (\ref{eqn s positive}) and (\ref{eqn v positive}),}\\
0 &<& s^2-nv+sv-sn \enspace,\\
0 &<& s^2-mu-nv+sv-sn \qquad\textrm{by (\ref{eqn m negative}) and (\ref{eqn u positive}).}
\end{eqnarray*}

\item Suppose
\begin{eqnarray}
\label{eqn v negative}
-1\leq v < 0 \enspace.
\end{eqnarray}
We need to distinguish between two subcases:
according to (\ref{eqn v negative}),
either (i) $-s \leq v < 0$
or (ii) $-1 \leq v < -s$.
\begin{enumerate}
\item Suppose
\begin{eqnarray}
\label{eqn subrange for v}
-s \leq v < 0 \enspace.
\end{eqnarray}
Therefore,
by (\ref{eqn s^2 < u^2+v^2}) and (\ref{eqn u positive}),
\begin{eqnarray}
\label{eqn subrange for u}
\sqrt{s^2-v^2} < u \leq 1 \enspace.
\end{eqnarray}

As we explained at the beginning of Case~\ref{case v positive},
it is sufficient to prove that $s^2-mu-nv+sv-sn$ is positive.
By (\ref{eqn m negative}) and (\ref{eqn subrange for u}),
$s^2-mu-nv+sv-sn > s^2-m\sqrt{s^2-v^2}-nv+sv-sn$.
We will prove that
$$\min_{-s \leq v \leq 0} \phi(v) \geq 0 \enspace,$$
where $\phi(v) = s^2-m\sqrt{s^2-v^2}-nv+sv-sn$,
which completes the proof of this subcase.

If $v=-s$,
then $\phi(v) = 0$.
If $v=0$,
then
\begin{eqnarray*}
\phi(v) &=& s^2-ms-sn \qquad\textrm{by (\ref{eqn s positive}),}\\
&=& -ms+s(s-n) \\
&\geq& 0 \qquad\textrm{by (\ref{eqn s positive}), (\ref{eqn m negative}) and (\ref{eqn s bigger than n}).}
\end{eqnarray*}
Moreover,
\begin{eqnarray*}
\frac{d}{dv}\phi(v) &=& 0 \enspace,\\
\frac{mv}{\sqrt{s^2-v^2}}-n+s &=& 0 \enspace,\\
v &=& \pm \frac{s(s-n)}{\sqrt{m^2+(s-n)^2}} \enspace.
\end{eqnarray*}
We reject
$$\frac{s(s-n)}{\sqrt{m^2+(s-n)^2}}$$
by (\ref{eqn subrange for v}),
(\ref{eqn s positive}), and (\ref{eqn s bigger than n}).
It remains to prove that 
$$\phi\left(-\frac{s(s-n)}{\sqrt{m^2+(s-n)^2}}\right) \geq 0 \enspace.$$
\begin{eqnarray}
\nonumber
&& \phi\left(-\frac{s(s-n)}{\sqrt{m^2+(s-n)^2}}\right) \geq 0\\
\nonumber
&\Leftrightarrow& \frac{s\left((s-n)\sqrt{m^2+(s-n)^2}+m^2-(s-n)^2\right)}{\sqrt{m^2+(s-n)^2}} \geq 0 \qquad\textrm{by (\ref{eqn s positive}) and (\ref{eqn m negative}),}\\
\nonumber
&\Leftrightarrow& (s-n)\sqrt{m^2+(s-n)^2}+m^2-(s-n)^2 \geq 0 \qquad\textrm{by (\ref{eqn s positive}),}\\
\label{inequality to prove}
&\Leftrightarrow& (s-n)\sqrt{m^2+(s-n)^2}+m^2 \geq (s-n)^2 \enspace.
\end{eqnarray}
If $-m \geq s-n$,
then $m^2 \geq (s-n)^2$
by (\ref{eqn m negative}) and (\ref{eqn s bigger than n});
hence, (\ref{inequality to prove}) is true by (\ref{eqn s bigger than n}).

Otherwise,
if $s-n > -m$,
then
\begin{eqnarray*}
s-n &>& -m \enspace,\\
\sqrt{3}(s-n) &>& -m \qquad\textrm{by (\ref{eqn s bigger than n}),}\\
3(s-n)^2 &>& m^2 \qquad\textrm{by (\ref{eqn m negative}) and (\ref{eqn s bigger than n}),}\\
3(s-n)^2m^2 &>& m^4 \enspace,\\
(s-n)^2m^2+(s-n)^4 &>& (s-n)^4-2(s-n)^2m^2+m^4 \enspace,\\
(s-n)^2\left(m^2+(s-n)^2\right) &>& \left((s-n)^2-m^2\right)^2 \enspace,\\
(s-n)\sqrt{m^2+(s-n)^2} &>& (s-n)^2-m^2 \qquad\textrm{because $s-n > -m$,}\\
\phantom{(s-n)\sqrt{m^2+(s-n)^2}} &\phantom{>}& \phantom{(s-n)^2-m^2} \qquad\textrm{and by (\ref{eqn m negative}) and (\ref{eqn s bigger than n}),}\\
(s-n)\sqrt{m^2+(s-n)^2} + m^2&>& (s-n)^2 \enspace,\\
\end{eqnarray*}
so (\ref{inequality to prove}) is true.

\item Suppose
\begin{eqnarray}
\label{eqn second subrange for v}
-1 \leq v < -s \enspace.
\end{eqnarray}
We want to prove that $\sqrt{1-u^2-v^2} < f(u,v)$.
Since
\begin{eqnarray*}
&&\sqrt{1-u^2-v^2} < f(u,v) \enspace,\\
&\Leftrightarrow& \sqrt{1-u^2-v^2} < \frac{-mu-(n-s)v-sn+\sqrt{1-m^2-n^2}\sqrt{1-s^2}+1}{\sqrt{1-m^2-n^2}+\sqrt{1-s^2}} \enspace,\\
&\Leftrightarrow& \phantom{<} \sqrt{1-m^2-n^2}\sqrt{1-u^2-v^2}+\sqrt{1-s^2}\sqrt{1-u^2-v^2}\\
&& < -mu-(n-s)v-sn+\sqrt{1-m^2-n^2}\sqrt{1-s^2}+1 \enspace,\\
&\Leftrightarrow& -mu-(n-s)v-sn+\sqrt{1-m^2-n^2}\sqrt{1-s^2}+1 \\
&& -\sqrt{1-m^2-n^2}\sqrt{1-u^2-v^2}-\sqrt{1-s^2}\sqrt{1-u^2-v^2} > 0 \enspace,
\end{eqnarray*}
we will work on this last inequality.

By (\ref{eqn m negative}) and (\ref{eqn u positive}),
$-mu-(n-s)v-sn+\sqrt{1-m^2-n^2}\sqrt{1-s^2}+1-\sqrt{1-m^2-n^2}\sqrt{1-u^2-v^2}-\sqrt{1-s^2}\sqrt{1-u^2-v^2}
\geq -(n-s)v-sn+\sqrt{1-m^2-n^2}\sqrt{1-s^2}+1-\sqrt{1-m^2-n^2}\sqrt{1-v^2}-\sqrt{1-s^2}\sqrt{1-v^2}$.
We will prove that
\begin{eqnarray}
\label{eqn min to prove}
\min_{-1 \leq v \leq -s} \phi(v) \geq 0 \enspace,
\end{eqnarray}
where $\phi(v) = -(n-s)v-sn+\sqrt{1-m^2-n^2}\sqrt{1-s^2}+1-\sqrt{1-m^2-n^2}\sqrt{1-v^2}-\sqrt{1-s^2}\sqrt{1-v^2}$,
which completes the proof of this subcase.

If $v=-1$, then 
\begin{eqnarray*}
\phi(v) &=& n-s-sn+\sqrt{1-m^2-n^2}\sqrt{1-s^2}+1 \\
&=& (1-s)(n+1)+\sqrt{1-m^2-n^2}\sqrt{1-s^2} \\
& \geq & 0 \qquad\textrm{by (\ref{eqn s positive}) and (\ref{eqn range for n}).}
\end{eqnarray*}
If $v=-s$,
then $\phi(v)=0$.
Moreover,
\begin{eqnarray*}
\frac{d}{dv}\phi(v) &=& 0 \enspace,\\
\frac{\left(\sqrt{1-m^2-n^2}+\sqrt{1-s^2}\right)v}{\sqrt{1-v^2}}-n+s &=& 0 \enspace,\\
v &=& \pm \frac{s-n}{\sqrt{(s-n)^2+\left(\sqrt{1-m^2-n^2}+\sqrt{1-s^2}\right)^2}} \enspace.
\end{eqnarray*}
We reject
$$\frac{s-n}{\sqrt{(s-n)^2+\left(\sqrt{1-m^2-n^2}+\sqrt{1-s^2}\right)^2}}$$
by (\ref{eqn second subrange for v}) and (\ref{eqn s bigger than n}).
As for
$$-\frac{s-n}{\sqrt{(s-n)^2+\left(\sqrt{1-m^2-n^2}+\sqrt{1-s^2}\right)^2}} \enspace,$$
if it is not between $-1$ and $-s$,
then it contradicts (\ref{eqn second subrange for v}).
Consequently,
(\ref{eqn min to prove}) is true
since $\phi(-1)\geq 0$,
$\phi(-s) \geq 0$,
and there is no local minimum between $-1$ and $-s$.

For the rest of the proof,
suppose that
$$-1 \leq -\frac{s-n}{\sqrt{(s-n)^2+\left(\sqrt{1-m^2-n^2}+\sqrt{1-s^2}\right)^2}} < -s \enspace.$$
Therefore,
\begin{eqnarray}
\label{inequality useful}
-\sqrt{(s-n)^2+\left(\sqrt{1-m^2-n^2}+\sqrt{1-s^2}\right)^2} &>& -1+\frac{n}{s}\qquad\qquad
\end{eqnarray}
by (\ref{eqn s positive}).
Hence,
\begin{eqnarray*}
&& \phi\left(-\frac{s-n}{\sqrt{(s-n)^2+\left(\sqrt{1-m^2-n^2}+\sqrt{1-s^2}\right)^2}}\right) \\
&=& 1-sn+\sqrt{1-m^2-n^2}\sqrt{1-s^2}-\sqrt{(s-n)^2+\left(\sqrt{1-m^2-n^2}+\sqrt{1-s^2}\right)^2} \\
&>& 1-sn+\sqrt{1-m^2-n^2}\sqrt{1-s^2}-1+\frac{n}{s} \qquad\textrm{by (\ref{inequality useful}),} \\
&=& n\left(\frac{1}{s}-s\right)+\sqrt{1-m^2-n^2}\sqrt{1-s^2} \\
&\geq& -s\left(\frac{1}{s}-s\right)+\sqrt{1-m^2-n^2}\sqrt{1-s^2} \qquad \textrm{by (\ref{eqn s bigger than n}) and (\ref{eqn s positive}),}\\
&=& -1+s^2+\sqrt{1-m^2-n^2}\sqrt{1-s^2} \\
&>& -1+s^2+\sqrt{1-s^2}\sqrt{1-s^2} \qquad\textrm{by (\ref{eqn d above Hab})}\\
&=& 0 \enspace. \qedhere
\end{eqnarray*}
\end{enumerate}
\end{enumerate}

\end{proof}

\begin{lemma}\label{lem:face_abc_has_long_edge}
Let $f = \triangle abc$ be a face in a well-shaped tetrahedral mesh $\meshThreeD$, and 
$\mathcal{S}(q,r)$ a sphere
intersecting $\meshThreeD$ but containing no vertex 
of $\meshThreeD$ in its interior. 
If $f$ intersects $\mathcal{S}$,
but no edge of $f$ intersects $\mathcal{S}$,
then $f$ has an edge of length at least $\frac{2}{\rho^2+1}r$.
\end{lemma}

\begin{proof}
We use a strategy similar to the one for proof of Lemma~\ref{lem:ab_is_long}.
Locate $S$ and $f$ in the usual system of axes $XYZ$ in such a way that
\begin{itemize}
\item $q$ is on the origin of the system of axes,

\item $f$ is parallel to the $XY$ plane, and

\item $f$ intersects perpendicularly the positive $Z$ axis.
\end{itemize}
Let $\mathcal{H}_{f}$ be the  plane containing $f$.
Let $t=abcd$ be the tetrahedron adjacent to $f$ such that $d$ is below $\mathcal{H}_{f}$.
By Observation~\ref{obs:tetrahedron_adjacent}-\ref{obs:adjacent_edge3D},
the smaller the edges of $t$,
the smaller the edges of $f$.
Therefore,
as we want to minimize $ab$, $ac$, and $bc$,
we assume that $d$ is on the surface of $\mathcal{S}$.
By the hypotheses,
if $d$ is below the $XY$ plane,
then $|ad|\geq r$, $|bd|\geq r$, or $|cd|\geq r$.
Thus, $|ab| \geq \frac{1}{\rho}r$ 
by Observation~\ref{obs:tetrahedron_adjacent}-\ref{obs:adjacent_edge3D},
from which $|ab| \geq \frac{1}{\rho}r > \frac{2}{\rho^2+1}r$
because $\rho \geq 3 > 1$.

For the rest of the proof,
suppose $d$ is below $\mathcal{H}_{f}$ but not below the $XY$ plane.
Let $\mathcal{H}_{d}$ be the plane containing $d$ that is parallel to the $XY$ plane.
Denote by $a'$ the intersection point of $f$ with the positive $Z$ axis
(refer to Fig.~\ref{fig:t_bound_tangent_face}).

\begin{figure}
	\centering
	\includegraphics[scale=1.00]{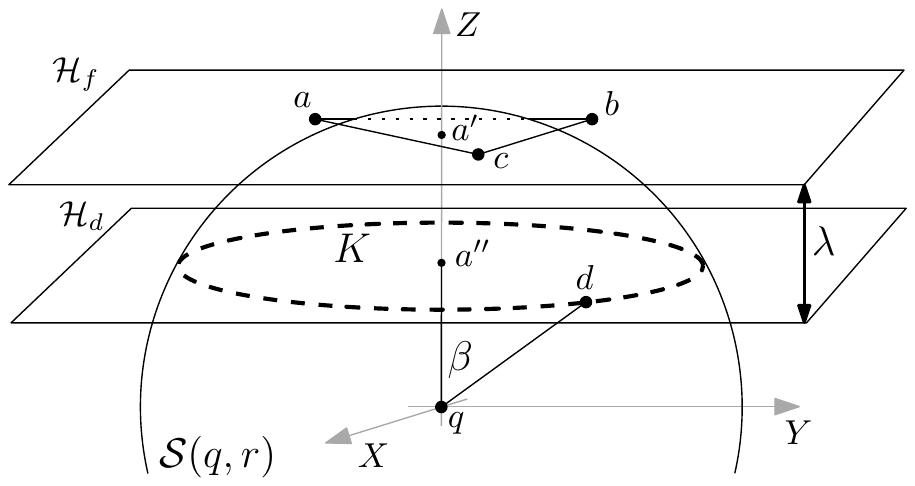}
	\caption{
	  By the hypotheses, $t$ is bounded by $\mathcal{H}_{d}$ and 
	  $\mathcal{H}_{f}$.
	}
	\label{fig:t_bound_tangent_face}
\end{figure}

Let $a''$ be the intersection point of $\mathcal{H}_{d}$ with the $Z$ axis.
Denote by $K$ the circle defined by the intersection of $\mathcal{S}$
and $\mathcal{H}_{d}$.
Note that $a''$ is the center of $K$,
and that $d$ is on the boundary of $K$.
Finally,
let $\lambda$ be the distance between $\mathcal{H}_{d}$ and $\mathcal{H}_{f}$.

Since $t$ lies between $\mathcal{H}_{d}$ and $\mathcal{H}_{f}$,
$r(t) \leq \frac{1}{2}\lambda$.
Furthermore,
from among the three edges connecting $f$ and $d$,
there is at least one edge of length greater than $a''d$.
Without loss of generality,
let $ad$ be this edge.
Therefore,
$R(t) \geq \frac{1}{2}|ad|$
by Observation~\ref{obs:ins_circ_phere}-\ref{obs:circumsphere_bound}.
Thus,
since $t$ is well-shaped,
\begin{equation}
	\rho > \frac{R(t)}{r(t)} \ge \frac{|ad|}{\lambda} \ge \frac{|a''d|}{\lambda} \enspace.
\end{equation}
We also have
$|a''d| = r\sin{\beta}$ 
and $\lambda = |a'a''| \leq (1 - \cos{\beta})r$,
where $\beta = \angle a''qd$. 
Thus,
$\rho > \frac{|a''d|}{\lambda} \geq \frac{\sin{\beta}}{1 - \cos{\beta}} = \cot\left(\frac{1}{2}\beta\right)$,
which implies $\beta > 2\,\arccot(\rho)$.
Hence,
$|ad| \ge |a''d| > r \sin(2\,\arccot(\rho)) = \frac{2\rho}{\rho^2+1}r$,
from which
$|ab| \geq \frac{2}{\rho^2+1}r$ 
by Observation~\ref{obs:tetrahedron_adjacent}-\ref{obs:adjacent_edge3D}.
\end{proof}

We combine the previous two lemmas in the following corollary
that stands as the three-dimensional version of Lemma~\ref{lem:tri_interC_have_long_edge}.

\begin{corollary}\label{cor:3d-tetra}
Let $t$ be a tetrahedron in a well-shaped tetrahedral mesh $\meshThreeD$,
and let $\mathcal{S}(q,r)$ be an empty sphere wholly contained in $\meshThreeD$
but not containing any vertex of $\meshThreeD$.
If $t$ intersects $\mathcal{S}$,
then $t$ has an edge of length at least $\frac{2}{\rho^2+1}r$.
\end{corollary}

Consider the walk from $\hat{p}$ to $q$ in $\meshThreeD$.
It intersects the boundary of $\mathcal{S}(q,|pq|)$ at a point $x$.
Let $t_i$ be the last tetrahedron
we encounter in the walk from $\hat{p}$ to $q$
that contains $x$.
We can now apply the same approach as we used in $\meshTwoD$
(see Subsection~\ref{subsection approximate 2 D})
to show that the number of tetrahedron visited along $\hat{p}q$ is a constant.
Denote the vertices of $t_i$ by $a$, $b$, $c$, and $d$,
where $\triangle abc$ is the face of $t_i$ that is the closest to $q$.
Let $\mathcal{G}$ be the union of all the tetrahedra incident to $a$, $b$, $c$, and $d$.
The following lemma is the three-dimensional version of Lemma~\ref{lem:shortest_escape}.

\begin{lemma}\label{lem:shortest_escape_3d}
Let $x \in t_i$ be the intersection of $\hat{p}q$ with the boundary of $\mathcal{S}(q,|pq|)$.
Let $y\in\mathcal{G}$ be the intersection of the line through $\hat{p}q$,
with the boundary of $\mathcal{G}$
such that $x$ is between $q$ and $y$.
Then 
$$|xy| \ge \frac{1}{\left(\rho^2+1\right)}\left(\frac{1}{\rho}\right)^{\left\lfloor\frac{2\pi}{\Omega}\right\rfloor+3}|pq| \enspace.$$
\end{lemma}

\begin{proof}
In the proof of Lemma~\ref{lem:shortest_escape},
the lower bound on $|xy|$ was computed by
calculating the shortest exit out of a well-shaped triangle $t$.
This shortest exit is perpendicular to an edge of $t$,
and constrained by the radius of the ball $B(c)$.
The ball $B(c)$ was defined so that it is guaranteed
to fit entirely inside $\mathcal{G}$.

In three dimensions,
we consider the balls $B(a)$, $B(b)$, $B(c)$, and $B(d)$ as in two dimensions,
but we also need to look at the constellation of tetrahedra adjacent to each edge of $t$.
By Observation~\ref{obs:tetrahedron_adjacent}-\ref{obs:adjacent_vertex3D}
and Corollary~\ref{cor:3d-tetra},
the edges incident to $a$, $b$, and $c$ have length of at least
$\frac{2}{\rho^2+1}\left(\frac{1}{\rho}\right)^{\left\lfloor\frac{2\pi}{\Omega}\right\rfloor}|pq|$.
We now calculate a lower bound on the radius of $B(a)$
(the reasoning is identical for $B(b)$ and $B(c)$).
Let $t'=a'b'c'd'$ be a tetrahedron
with $|a'b'|=|a'c'|=|a'd'|=\frac{2}{\rho^2+1}\left(\frac{1}{\rho}\right)^{\left\lfloor\frac{2\pi}{\Omega}\right\rfloor}|pq|$.
A lower bound for the radius of $B(a)$ can be obtained by computing the height $h$ of $t'$ relative to $a'$.
We have
$$\rho > \frac{R(t')}{r(t')} \geq \frac{\frac{1}{2}\frac{2}{\rho^2+1}\left(\frac{1}{\rho}\right)^{\left\lfloor\frac{2\pi}{\Omega}\right\rfloor}|pq|}{\frac{1}{2}h} $$
by Observation~\ref{obs:ins_circ_phere},
from which
$h > \frac{2}{\rho^2+1}\left(\frac{1}{\rho}\right)^{\left\lfloor\frac{2\pi}{\Omega}\right\rfloor+1}|pq|$.
Hence,
$B(a)$, $B(b)$, and $B(c)$ have radius 
\begin{eqnarray}
\label{eqn radius B(a)}
\frac{2}{\rho^2+1}\left(\frac{1}{\rho}\right)^{\left\lfloor\frac{2\pi}{\Omega}\right\rfloor+1}|pq| \enspace.
\end{eqnarray}
Since $ad$ is adjacent to $ab$,
then $B(d)$ has radius
\begin{eqnarray}
\label{eqn radius B(d)}
\frac{2}{\rho^2+1}\left(\frac{1}{\rho}\right)^{\left\lfloor\frac{2\pi}{\Omega}\right\rfloor+2}|pq| 
\end{eqnarray}
by Observation~\ref{obs:tetrahedron_adjacent}-\ref{obs:adjacent_edge3D}.

We now turn to the constellations of tetrahedra adjacent to the edges of $t$.
First we look at the constellation $\Gamma \subset \mathcal{G}$ of tetrahedra adjacent to the edge $ad$.
The proof is identical for $bd$ and $cd$.
As for,
$ab$, $bc$, and $ac$,
a similar argument applies since the balls $B(a)$, $B(b)$, and $B(c)$
have a larger radius than $B(d)$
(compare (\ref{eqn radius B(a)}) and (\ref{eqn radius B(d)})).

Note that $\Gamma$
occupies some volume around $ad$.
We need to describe precisely what this volume is.
For each tetrahedron $t'$ adjacent to $ad$,
denote by $S'$ the insphere of $t'$.
By Observation~\ref{obs:ins_circ_phere}-\ref{obs:circumsphere_bound},
we know that
$$\rho > \frac{R(t')}{r(t')} \geq \frac{\frac{1}{2}|ad|}{r(t')} \enspace.$$
Therefore,
$r(t') > \frac{|ad|}{2\rho}$.
Thus,
there is a constellation of spheres with radius at least $\frac{|ad|}{2\rho}$ around $ad$.

For each tetrahedron $t'$ in $\Gamma$,
do the following:
consider the set of tangents from $d$ to $S'$.
This set of tangents defines a cone $C'$ that is entirely contained in $t'$.


Reduce the radius $r(t')$ of $S'$ until it is equal to $\frac{|ad|}{2\rho}$
(refer to Figure~\ref{fig:3d_jump_walk_tetra_bound}(a)).

\begin{figure}
  \captionsetup[subfigure]{labelformat=parens}
    \centering
        \subfloat[Constellation of inspheres adjacent to $ad$. 
        The dotted lines correspond to the lower bound on the radii 
        of the inspheres.]{
	  \label{fig:constellationGen}
	  \includegraphics{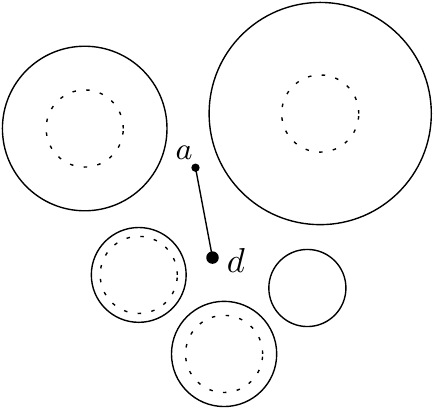}
        } \qquad
        \subfloat[Constellation of inspheres around $ad$ after stretching and translation.]{
	  \label{fig:constellationBound}
	  \includegraphics{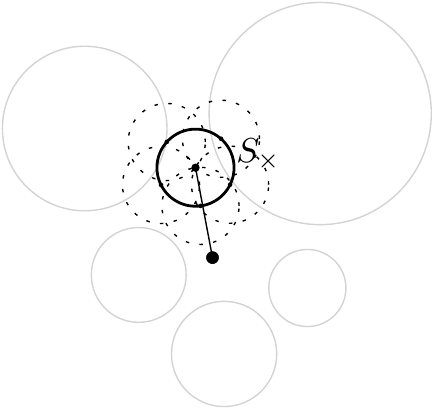}
	}%
	\\
	\subfloat[Three-dimensional view of $C_1 \cap C_2$.]{
        \label{fig:cones}
        \includegraphics{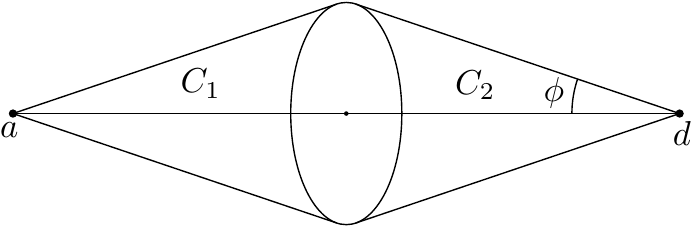}
	} \qquad
        \subfloat[Tangents with bigger elevations with respect to $ad$.
	The part of $C'$ that is furthest from $ad$ appears in bold.]{
	  \label{fig:aperture}
	  \includegraphics{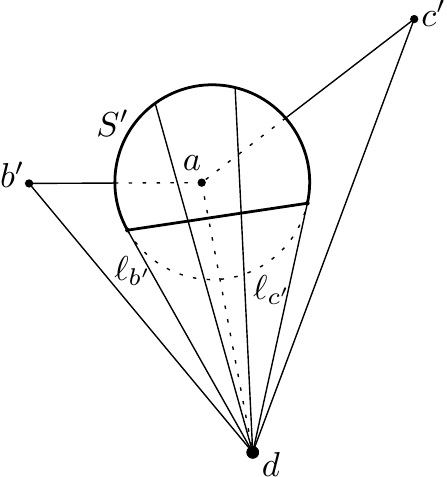}
        }%
        \\ \vspace{0.8cm}
	\subfloat[Two-dimensional view of $C_1 \cap C_2$ and 	$B(d)$.
	  A lower bound on the shortest exit out of $\mathcal{G}$ appears in bold.]{
         \label{fig:exit3D}
         \includegraphics{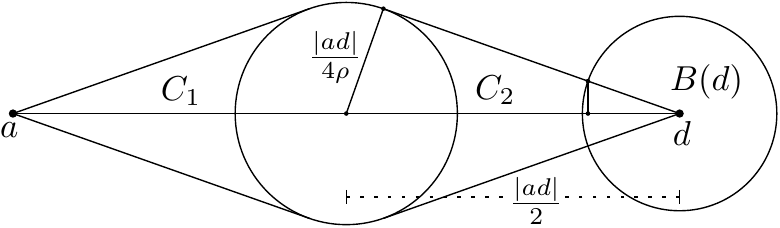}
      }
    \caption{Illustration of the proof of Lemma~\ref{lem:shortest_escape_3d}}
    \label{fig:3d_jump_walk_tetra_bound}
\end{figure}

Translate $S'$ towards $a$ until $S'$ is tangent to the line supporting $da$ at $a$
(refer to Figure~\ref{fig:3d_jump_walk_tetra_bound}(b)).
During these transformations,
maintain $C'$ accordingly.
Do these transformations on all inspheres of all tetrahedra in $\Gamma$.
Consider the sphere $S^{\times}$
with center $a$ and radius $\frac{|ad|}{2\rho}$.
Let $C_1$ be the cone with apex $d$ that is tangent to $S^{\times}$,
and define $C_2$ symmetrically
(refer to Figure~\ref{fig:3d_jump_walk_tetra_bound}(c)).
We claim that $C_1 \cap C_2 \subset \Gamma$.

We now prove our claim.
Each tetrahedron $t'$ has $a$ and $d$ as vertices.
Denote the other two vertices of $t'$ by $b'$ and $c'$.
Let $\ell_{b'}$ (respectively $\ell_{c'}$) be the tangent from $d$ to $S'$ that belongs to the face $adb'$ (respectively to the face $adc'$)
of $t'$
(refer to Figure~\ref{fig:3d_jump_walk_tetra_bound}(d)).
The tangents $\ell_{b'}$ and $\ell_{c'}$
splits $C'$ into two parts,
one closest to $ad$ and one furthest from $ad$
(refer to Figure~\ref{fig:3d_jump_walk_tetra_bound}(d)).
Focus on the part of $C'$
that is furthest from $ad$.
The tangents in this subset have different elevations with respect to $ad$.
Let $\phi$ be a lower bound on the elevations
of all tangents of all upper subsets of all cones $C'$ in $\Gamma$.
The cones $C_1$ and $C_2$
have an elevation of less than $\phi$
by construction.
This proves our claim.

A lower bound on the shortest exit out of $\mathcal{G}$
can be computed by considering the line segment perpendicular to $ad$ that intersects
both the boundary of $B(d)$ and the boundary of $C_2$
(refer to Figure~\ref{fig:3d_jump_walk_tetra_bound}(e) where this line segment appears in bold).
The length of this line segment is
$$\frac{1}{\left(\rho^2+1\right)}\left(\frac{1}{\rho}\right)^{\left\lfloor\frac{2\pi}{\Omega}\right\rfloor+3}|pq| \enspace.$$
\end{proof}

The proof of the following theorem is identical to the proof of Theorem~\ref{thm:2d_approximate};
it uses Lemma~\ref{lem:shortest_escape_3d} instead of Lemma~\ref{lem:shortest_escape}.

\begin{theorem}\label{thm:3d_approximate}
Let $\meshThreeD$ be a well-shaped tetrahedral mesh in $\CDThreeD$.
Take
$$\epsilon < \frac{1}{\left(\rho^2+1\right)}\left(\frac{1}{\rho}\right)^{\left\lfloor\frac{2\pi}{\Omega}\right\rfloor+3}$$
and let $\hat{p}$ be an $(1 + \epsilon)$-approximate nearest neighbour of a 
query point $q$ from among the vertices of $\meshThreeD$.
The straight line walk from $\hat{p}$ to $q$ visits at most 
$\frac{\sqrt{3}\,\pi}{486}\rho^3(\rho^2+3)^3+\left\lfloor \frac{2\pi}{\Omega}\right\rfloor$
tetrahedra.
\end{theorem}

\section{Experimental Results}
\label{sec:walk-experiments}

In this section, we present experimental results on jump-and-walk for 
well-shaped triangular meshes. 
The purpose of performing these experiments is to verify that, in the $\plane$ 
setting at least, the walk step requires constant time.
Furthermore, since our theoretical bounds are pessimistic, these results are 
intended to give insight into the actual expected walk lengths in this setting.

\subsection{Implementation Details}

Our motivation for investigating jump-and-walk was its natural
applicability to the mesh structures we describe in 
Chapters~\ref{chp:succinct_graphs} and \ref{chp:mesh_trav}.
These data structures, particularly in Chapter~\ref{chp:succinct_graphs},
rely on a very simple representation of the simplices (triangle or
tetrahedron).
Thus, we wish to use as simple a representation as possible 
in our experimental work, so that the algorithms will be applicable to the data
structures that we have developed.
In our $\plane$ implementation, we used precisely such a simple 
data structure, which records for each triangle only its vertices 
and its three neighbouring triangles. 
Edges are not explicitly represented.

In developing the walk algorithm, we also sought to keep the 
algorithm itself as simple as possible. 
For example, we determine the sequence of triangles visited as we 
walk from $p$ to $q$ without performing any exact calculations of
intersections between $pq$ and triangles in the mesh.
Rather, we perform boolean tests for intersection,
and use \emph{memory} of the last visited triangle in order to ensure the
walk progresses from $p$ towards $q$.
This approach has the added advantage of providing a means of 
dealing with degeneracies, such as $pq$ passing directly through
a vertex on the mesh.

\begin{figure}
    \centering
    \includegraphics{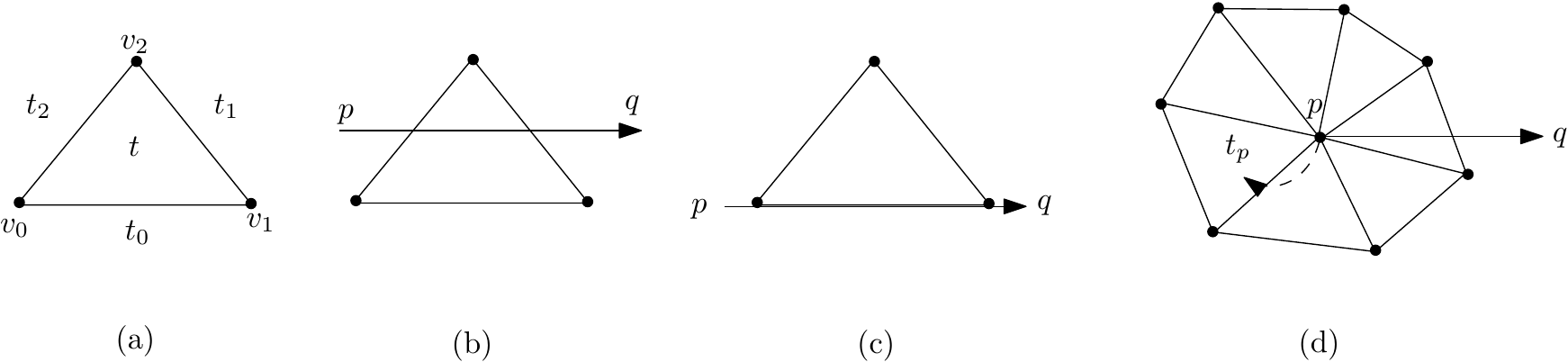}
    \caption[Walk primitives in $\plane$]{Primitive representation and possible
	degeneracies for $\plane$ walk. 
	Triangle representation is shown in $(a)$, while $(b)$ shows the most common
	segment triangle intersection expected.
	However, degeneracies such as $(c)$ and $(d)$ must also be dealt with.
	In particular, since the start of each walk is from an arbitrary triangle
	adjacent to vertex $p$ (e.g., $t_p$ in this figure), case $(d)$ will
	arise frequently in practice.
	}
    \label{fig:walk_primatives_2d}
\end{figure}

The basic steps of the walk procedure, and some of the degeneracies that must
be dealt with, can be 
explained with reference to Figure~\ref{fig:walk_primatives_2d}.
Assume we are are at some triangle $t_i$ having arrived on the walk along $pq$ from
triangle $t_{i-1}$.
We calculate the intersection of $pq$ with the three \emph{closed} edges that form
the triangle. 
We also test the vertices of the $t_i$ to determine if any are on $pq$.
If no vertex of $t_i$ is on $pq$, then we know we are the situation depicted in 
(b).
The line $pq$ intersects edges\footnote{While edges are not explicitly represented
in our structure, we can still reason about them since they are defined by
the vertices that are their endpoints.} seperating $t_i$ from two triangles, 
one of which is $t_{i-1}$.
We select the other triangle as $t_{i+1}$.
This process is repeated until the triangle containing $q$ is reached.
Degeneracies like those shown in (c) are problematic, since according to our
definition of edges, $pq$ intersects all three edges of the triangle.
Thus, even if we know from which triangle we entered, we must still choose 
between two possible triangles for our next step. 
However, we have developed simple rules that handle all such degeneracies
(for example, in instances like that shown in (c), we never choose to cross an
edge collinear to $pq$).

Of particular interest is what to do when starting the walk from point $p$.
In this case, we do not have a \emph{previous} triangle to guide our search.
Furthermore, since we start all walks at a vertex in the mesh, walks will
always start with the situation depicted in Figure~\ref{fig:walk_primatives_2d}(d).
We must choose one neighbour of $p$ from which to start the walk; let $t_p$ be
the selected triangle.
As a preprocessing step, we assign to each vertex a start triangle $t_p$ from
which any walk will start.
The next question is: from $t_p$, which way do we start to walk?
It turns out that choosing either neighbour of $t_p$ adjacent to $p$ will lead 
to a correct result, as we will eventually come around to the triangle that 
contains $q$, or from  which $pq$ exits the neighbourhood of $p$ 
(see Fig.\ref{fig:walk_primatives_2d}(d)). 
However, we can start the walk in a slightly smarter manner.
We perform a test to determine which of the edges of $t_p$ adjacent to $p$ has 
$q$ to the left of its supporting line\footnote{The vertices of triangles are ordered
counterclockwise. It may be the case that $q$ is to the left of both edges 
adjacent to $p$, in which case either can be selected.}, and start the walk 
across this edge. 
Our \emph{remember the last triangle visited} rule ensures that we keep walking in the 
same direction.
This strategy is adequate for our purpose of achieving a constant walk length, as 
the neighbourhood around $p$ is of constant size.

\subsection{Experiments}

\noindent \textbf{System: } 
All tests were executed on a Hewlett-Packard G60
Laptop with a dual-core Pentium T4300 processor, running at 2.1 GHz with
3.8 GB of RAM.
Each experiment involved a two-step process; firstly, locating the 
(approximate) nearest neighbour from among the vertices of the mesh
and, secondly, performing the walk step.
The first step was executed using a C++ program built on the 
ANN library \cite{ann_website, DBLP:journals/jacm/AryaMNSW98}.
The walk step was implemented in the D programming language, using the Digital
Mars D compiler (DMD). 

\noindent \textbf{Datasets: }
The datasets used for these experiments were based on random point 
sets.
Input points sets for both triangulations and queries were generated
using the Python random module, which uses the Mersenne Twister
pseudo-random number generator~\cite{Matsumoto:1998:MTE:272991.272995}.
Well-shaped triangular meshes were generated using the program
Triangle, by Johnathan Shewchuk \cite{triangle_site}.
Triangle generates Delaunay triangulations, but uses mesh refinement to produce 
meshes with a user-specified minimum angle.
The mesh-refinement algorithm adds additional points to the input point set.
As a consequence, it is difficult to precisely control the sizes of the
triangulations generated, thus all sizes given are approximate.

Triangulations were generated with 100, 1K, 10K, 25K, 50K, 100K, 250K, 500K,
1000K, and 2000K triangles. 
Three triangulations were generated at each size with minimum angles of 2, 5,
10, and 20 degrees, respectively.
In total, 120 triangulations were generated from random point sets.

\noindent \textbf{Results: }
Each walk was performed from an input point selected at random.
For each triangulation, a total of 25,000 walks were executed.
A nearest neighbour search was performed over the triangulation
vertices in order to locate the start point for each walk.
Tables \ref{tbl:alpha2}, \ref{tbl:alpha5}, \ref{tbl:alpha10}, and \ref{tbl:alpha20} 
give average and worst case walks for triangulations built with minimum 
angles of $2$, $5$, $10$, and $20$ degrees, respectively. 
Also shown are the walk times for walks starting with an approximate 
nearest neighbour. 
The epsilon values used in the approximate nearest neighbour searches
were set significantly larger than the minimum epsilon values required
by Theorem~\ref{thm:2d_approximate}.
Epsilon values of $\epsilon = 0.2$, $\epsilon = 0.5$, $\epsilon = 1.0$, and
$\epsilon = 2.0$ were selected for the triangulations with minimum angles
of $2$, $5$, $10$, and $20$ degrees respectively.
If $\epsilon = \tan\alpha\sin^{\left \lfloor \frac{\pi}{\alpha} 
\right \rfloor +3} \alpha$ were used, these values would have been
$\epsilon = 4.9 x 10^{-324}$, $\epsilon = 1.1 x 10^{-137}$, 
$\epsilon = 4.1 x 10^{-43}$, and $\epsilon = 1.9 x 10^{-17}$, respectively.
In the exact nearest neighbour case, minimum angles of $2$, $5$, $10$, and 
$20$ degrees corresponded to theoretical worst case walks of length
$90$, $36$, $18$, and $9$.
In the approximate nearest neighbour case, these theoretical bounds
became $180$, $72$, $36$, and $18$.

\begin{table}
\centering
\caption[Walk results for mesh with $\rho=29.2$]{
Average and worst-case walk steps over 25,000 walks 
for $\rho = 29.2$, or minimum angle $\alpha = 2^\circ$. 
Theoretical worst case 90 steps for nearest neighbour,
and 180 steps for approximate nearest neighbour.}\label{tbl:alpha2}
\begin{tabular}{ |l|l|l|l|l|l|}
\hline
\multicolumn{2}{ |l| }{  } &
\multicolumn{2}{ |l| }{NN Search} &
\multicolumn{2}{ |l| }{ANN Search}
\\ \hline 
\textbf{N} & & \textbf{Avg.} & \textbf{Long} & \textbf{Avg.}& \textbf{Long} \\ \hline
 100 & 1 &  2.65  &  7  & 2.63 &  7 \\
 100 & 2 &  2.77  &  7  & 2.78 &  7 \\
 100 & 3 &  2.68  &  6  & 2.67 &  6 \\ \hline
 1K & 1 &  2.67  &  9  &  2.66 &  9 \\
 1K & 2 &  2.74  &  8  & 2.74 &  8 \\
 1K & 3 &  2.65  &  9  & 2.64 &  9  \\ \hline
 10K & 1 &  2.67  &  9  & 2.67 &  8 \\
 10K & 2 &  2.70  &  9  & 2.69 &  9 \\
 10K & 3 &  2.69  &  8  & 2.67 &  8 \\ \hline
 25K & 1 &  2.68  &  8  & 2.67 &  9 \\
 25K & 2 &  2.68  &  8  & 2.68 &  8 \\
 25K & 3 &  2.68  &  9  & 2.67 &  8 \\ \hline
 50K & 1 &  2.67  &  8  & 2.69 &  8 \\
 50K & 2 &  2.67  &  8  &  2.67 &  8 \\
 50K & 3 &  2.67  &  10  & 2.67 &  9 \\ \hline
 100K & 1 & 2.67  &  8  & 2.66 &  9 \\
 100K & 2 &  2.67  &  8  & 2.68 &  8 \\
 100K & 3 &  2.69  &  8  & 2.67 &  8 \\ \hline
 250K & 1 &  2.66  &  8  & 2.68 &  9 \\
 250K & 2 & 2.66  &  9  & 2.66 &  8 \\
 250K & 3 &  2.66  &  9  & 2.68 &  8 \\ \hline
 500K & 1 &  2.66  &  9  & 2.69 &  8 \\
 500K & 2 &  2.69  &  8  & 2.67 &  9 \\
 500K & 3 &  2.67  &  8  & 2.67 &  9 \\ \hline
 1000K & 1 &  2.67  &  8 & 2.67 &  8 \\
 1000K & 2 &  2.67  &  8  &2.67 &  8  \\
 1000K & 3 &  2.67  &  8  & 2.67 &  8\\ \hline
 2000K & 1 &  2.66  &  8  & 2.69 &  8 \\
 2000K & 2 &  2.67  &  9  & 2.67 &  8  \\
 2000K & 3 &  2.68  &  9  & 2.67 &  9 \\  \hline
\end{tabular}
\end{table} 

\begin{table}
\centering
\caption[Walk results for mesh with $\rho=12.0$]{ 
Average and worst-case walk steps over 25,000 walks
for $\rho = 12.0$, or minimum angle $\alpha = 5^\circ$. 
Theoretical worst case 36 steps for nearest neighbour,
and 72 steps for approximate nearest neighbour.
}\label{tbl:alpha5}
\begin{tabular}{ |l|l|l|l|l|l| }
\hline
\multicolumn{2}{ |l| }{  } &
\multicolumn{2}{ |l| }{NN Search} &
\multicolumn{2}{ |l| }{ANN Search}
\\ \hline 
\textbf{N} & & \textbf{Avg.} & \textbf{Long} & \textbf{Avg.}& \textbf{Long} \\ \hline
 100 & 1 & 2.54 & 7 & 2.55 &  7 \\
 100 & 2 & 2.59 & 7 & 2.59 &  7 \\
 100 & 3 & 2.59 & 6 & 2.57 &  6 \\ \hline
 1K & 1 & 2.65 & 8 & 2.65 &  8 \\
 1K & 2 & 2.62 & 8 & 2.62 &  7 \\
 1K & 3 & 2.68 & 7 & 2.70 &  7 \\ \hline
 10K & 1 & 2.66 & 8 & 2.65 &  8  \\
 10K & 2 & 2.66 & 8 & 2.68 &  8 \\
 10K & 3 & 2.67 & 9 & 2.66 &  9 \\ \hline
 25K & 1 & 2.67& 8 & 2.68 &  9 \\
 25K & 2 & 2.67 & 9 &  2.69 &  8\\
 25K & 3 & 2.68 & 8 & 2.69 &  8 \\ \hline
 50K & 1 & 2.65 & 9 & 2.68 &  9 \\
 50K & 2 & 2.66 & 8 & 2.66 &  8 \\
 50K & 3 & 2.67 & 10 & 2.69 &  9  \\ \hline
 100K & 1 & 2.66 & 8 & 2.67 &  8  \\
 100K & 2 & 2.66  & 8 & 2.68 &  9 \\
 100K & 3 & 2.66 & 8 & 2.67 &  9 \\ \hline
 250K & 1 & 2.67 & 8 & 2.68 &  8 \\
 250K & 2 & 2.67 & 8 & 2.67 &  8  \\
 250K & 3 & 2.65 & 9 & 2.68 &  8 \\ \hline
 500K & 1 & 2.66 & 8 & 2.67 &  9 \\
 500K & 2 & 2.65 & 8 &2.68 &  8 \\
 500K & 3 & 2.64 & 8 &  2.67 &  9 \\ \hline
 1000K & 1 & 2.66& 8 & 2.68 &  9 \\
 1000K & 2 & 2.66 & 8 & 2.67 &  8 \\
 1000K & 3 & 2.68 & 8 &2.67 &  8 \\ \hline
 2000K & 1 & 2.67 & 8 & 2.69 &  8 \\
 2000K & 2 & 2.67 & 8 & 2.67 &  9 \\
 2000K & 3 & 2.66 & 8 & 2.67 &  9 \\ \hline
\end{tabular}
\end{table}

\begin{table}
\centering
\caption[Walk results for mesh with $\rho=6.3$]{
Average and worst-case walk steps over 25,000 walks 
for $\rho = 6.3$, or minimum angle $\alpha = 10^\circ$. 
Theoretical worst case 18 steps for nearest neighbour,
and 36 steps for approximate nearest neighbour.
}\label{tbl:alpha10}
\begin{tabular}{ |l|l|l|l|l|l| }
\hline
\multicolumn{2}{ |l| }{  } &
\multicolumn{2}{ |l| }{NN Search} &
\multicolumn{2}{ |l| }{ANN Search}
\\ \hline 
\textbf{N} & & \textbf{Avg.} & \textbf{Long} & \textbf{Avg.}& \textbf{Long} \\ \hline
 100 & 1 & 2.61 & 8 & 2.69 &  8 \\
 100 & 2 & 2.80 & 7 & 2.82 &  7 \\
 100 & 3 & 2.52 & 6 & 2.52 &  7 \\ \hline
 1K & 1 & 2.61 & 7 & 2.63 &  7 \\
 1K & 2 & 2.62 & 8 & 2.64 &  8 \\
 1K & 3 & 2.65 & 8 & 2.67 &  8 \\ \hline
 10K & 1 & 2.63 & 9 & 2.66 &  8 \\
 10K & 2 & 2.64 & 8 & 2.67 &  8 \\
 10K & 3 & 2.65 & 8 & 2.66 &  8 \\ \hline
 25K & 1 & 2.64 & 8 & 2.66 &  8 \\
 25K & 2 & 2.65 & 8 & 2.67 &  8 \\
 25K & 3 & 2.64 & 8 & 2.66 &  8 \\ \hline
 50K & 1 & 2.65 & 8 & 2.66 &  8 \\
 50K & 2 & 2.65 & 8 & 2.67 &  8 \\
 50K & 3 & 2.66 & 8 & 2.66 &  8 \\ \hline
 100K & 1 & 2.65 & 8 & 2.66 &  10 \\
 100K & 2 & 2.63 & 8 & 2.67 &  8 \\
 100K & 3 & 2.65 & 8 & 2.67 &  8 \\ \hline
 250K & 1 & 2.65 & 8 & 2.67 &  8 \\
 250K & 2 & 2.65 & 8 & 2.67 &  8 \\
 250K & 3 & 2.65 & 8 & 2.66 &  8 \\ \hline
 500K & 1 & 2.64 & 9 & 2.67 &  8 \\
 500K & 2 & 2.64 & 8 & 2.67 &  8 \\
 500K & 3 & 2.64 & 8 & 2.67 &  8 \\ \hline
 1000K & 1 & 2.65 & 8 & 2.67 &  8 \\
 1000K & 2 & 2.65 & 8 & 2.66 &  8 \\
 1000K & 3 & 2.65 & 8 & 2.67 &  8 \\ \hline
 2000K & 1 & 2.65 & 8 & 2.66 &  8 \\
 2000K & 2 & 2.64 & 8 & 2.66 &  8 \\
 2000K & 3 & 2.64 & 8 & 2.66 &  8 \\ \hline
\end{tabular}
\end{table}

\begin{table}
\centering
\caption[Walk results for mesh with $\rho=3.5$]{
Average and worst-case walk steps over 25,000 walks 
for $\rho = 3.5$, or minimum angle $\alpha = 20^\circ$. 
Theoretical worst case 9 steps for nearest neighbour,
and 18 steps for approximate nearest neighbour.
}\label{tbl:alpha20}
\begin{tabular}{ |l|l|l|l|l|l| }
\hline
\multicolumn{2}{ |l| }{  } &
\multicolumn{2}{ |l| }{NN Search} &
\multicolumn{2}{ |l| }{ANN Search}
\\ \hline 
\textbf{N} & & \textbf{Avg.} & \textbf{Long} & \textbf{Avg.}& \textbf{Long} \\ \hline
 100 & 1 & 2.31 & 6  & 2.31 &  6 \\
 100 & 2 & 2.25 & 7  & 2.18 &  7 \\
 100 & 3 & 2.39 & 8  & 2.34&  6\\ \hline
 1K & 1 & 2.52 & 7  & 2.53 &  7 \\
 1K & 2 & 2.53 & 7  & 2.53 &  7 \\
 1K & 3 & 2.48 & 7  & 2.49 &  7 \\ \hline
 10K & 1 & 2.55 & 7  & 2.56 &  8  \\
 10K & 2 & 2.53 & 7  & 2.55 &  8\\
 10K & 3 & 2.54 & 8  & 2.56 &  7 \\ \hline
 25K & 1 & 2.55 & 7  & 2.57 &  7 \\
 25K & 2 & 2.55 & 7  & 2.58 &  7 \\
 25K & 3 & 2.55 & 7  & 2.57 &  8 \\ \hline
 50K & 1 & 2.55 & 7  & 2.58 &  7 \\
 50K & 2 & 2.54 & 7  & 2.58 &  7 \\
 50K & 3 & 2.56 & 7  & 2.56 &  7 \\ \hline
 100K & 1 & 2.55 & 8  & 2.59 &  7 \\
 100K & 2 & 2.56 & 7  & 2.60 &  8 \\
 100K & 3 & 2.56 & 7  & 2.58 &  8 \\ \hline
 250K & 1 & 2.55 & 7  & 2.59 &  8\\
 250K & 2 & 2.56 & 7  & 2.59 &  8 \\
 250K & 3 & 2.55 & 8  & 2.57 &  7 \\ \hline
 500K & 1 & 2.57 & 8  & 2.58 &  8 \\
 500K & 2 & 2.55 & 7  & 2.59 &  7 \\
 500K & 3 & 2.56 & 8  & 2.60 &  7\\ \hline
 1000K & 1 & 2.57 & 7  & 2.59 &  8 \\
 1000K & 2 & 2.56 & 7  & 2.59 &  8  \\
 1000K & 3 & 2.55 & 7  & 2.58 &  7\\ \hline
 2000K & 1 & 2.55 & 7  & 2.57 &  7 \\
 2000K & 2 & 2.56 & 9  & 2.58 &  7 \\
 2000K & 3 &  2.56 & 7  & 2.59 &  7 \\ \hline
\end{tabular}
\end{table}

The most striking trend that the tables reveaedl is that, in this setting,
changing the dataset size, the minimum angle, and from nearest neighbour
to approximate nearest neighbour appear to have little impact on either
the average or longest walks.
In fact, the reported values were remarkably consistent.
In Table \ref{tbl:alpha2}, average walks for all triangulation sizes were,
with a few very minor outliers, almost identical for all tests. 
The longest walks for the smallest triangulations, $N=100$, were smaller than
for the others, but for all other values of $N$, longest walk lengths were 
consistent.
This trend again held with $\alpha=5$ in Table \ref{tbl:alpha5}.
In this case, however, average and longest walks were shorter for 
$N=100$, and slightly shorter for $N=1000$.
Oddly, the trend for shorter walks for the smallest values of $N$ was not
as apparent for $\alpha=10$ (Table \ref{tbl:alpha10}), but appeared again
for $\alpha=20$ in Table \ref{tbl:alpha20}.

There are a number of conclusions that can be drawn from these results.  
The first is that the results do indeed support the claim that the walk
distances are constant, as  neither the average nor the worst-case walks 
changed with increasing $N$. 
The most surprising result is that even using large values for $\epsilon$,
there was no noticeable difference between the nearest neighbour and approximate
nearest neighbour walks. 
These result suggests that the walk cost is, in the worst case, dominated by the
search around the start point (see Fig. \ref{fig:walk_primatives_2d}(d)).
This cost is more or less random, depending on the query point and the 
pointer saved for the selected vertex, and as such is not likely to be 
influenced by whether we search from an exact or from an approximate nearest
neigbhour.

A second surprising result was that minimum angle of the mesh did not
have any noticable impact on either the average or the longest walk in these
experiments.
The probable cause of this was the nature of triangulations, which were 
generated from random point sets.
Thus, even though the mesh generator added points more aggresively with some
datasets (in order to remove small-angle triangles), the random distribution 
of points produced short walks in all cases.

To verify this hypothesis, three additional triangulations were generated.
These triangulations were small,  about 150 triangles apiece, but were
generated from a set of regularly-spaced points. 
Essentially, three vertical lines of evenly-spaced points were generated;
one down the left-side, one down the centre, and one along the 
right-side of the domain (which was a unit square).
A pair of points were placed between the left and the centre lines, and between
the right and the centre lines, so that the triangulations were not completely regular.
During triangulation, no minimum angle was specified; therefore, the mesh generator
did not add additional points.
However, while generating the mesh vertex set, no two or more points were placed 
too close, thereby ensuring a degree of well-shapedness.
The triangulalation generated with minimum angle $\alpha = 2.86$ is shown in 
Figure~\ref{fig:tri_regular}.
The results of 25,000 walks in these three triangulations are presented in 
Table \ref{tbl:non-random}.
As expected, the average and longest walks were significantly longer than with 
the randomly-generated point set. 
They were also more strongly influenced by changes in the well-shapedness of 
the triangulation, ranging from $22$ with $\alpha=2.86$, to $34$ for 
$\alpha=1.29$.
Finally, walking from an approximate nearest neighbour increased average walk
lengths by between $0.4$ and $0.9$ steps. 
This is in contrast to meshes generated from random points, where selecting
the nearest neighbour or the exact nearest neighbour did not change average walk
lengths significantly.

\begin{figure}
	\centering
 	\includegraphics[scale=0.4]{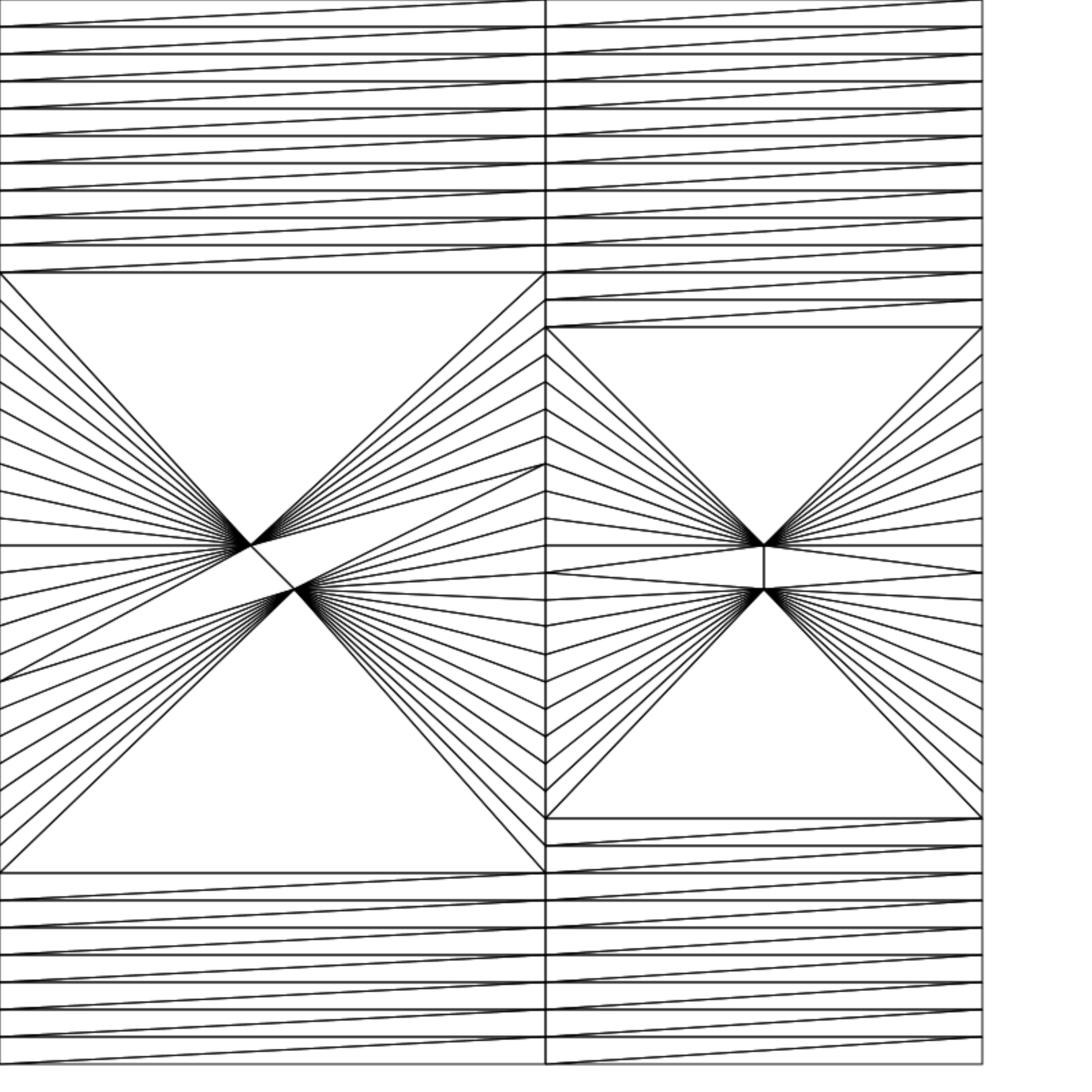}
	\caption[Regular triangulation]{Triangulation on regularly-spaced point set}
	\label{fig:tri_regular}
\end{figure}

\begin{table}
\centering
\caption[Walk results for mesh on regularly-spaced point set]{
Average and longest walks measured 25,000 walks, 
for specified minimum angles, based on triangulations generated 
from regularly-spaced point sets.
Also shown is the theoretical worst-case walk length for
each minimum angle.
For approximate nearest neighbour walks, $\epsilon=0.2$.
}\label{tbl:non-random}
\begin{tabular}{ |l|l|l|l|l|l|l| }
\hline
\multicolumn{1}{ |l| }{} &
\multicolumn{3}{ |c| }{NN Search} &
\multicolumn{3}{ |c| }{ANN Search}
\\ \hline 
\textbf{$\alpha$} & \textbf{Avg.} & \textbf{Long} & \textbf{Worst-case} & 
  \textbf{Avg.}& \textbf{Long} & \textbf{Worst-case} \\ \hline
 2.86 & 3.20 & 22 & 62 & 3.59 & 22 & 124 \\
 2.29 & 3.41 & 25 & 78 & 4.03 & 25 & 156 \\
 1.29 & 4.00 & 34 & 139 & 4.91 & 34 & 278 \\ \hline
\end{tabular}
\end{table}

\section{Summary and Open Problems}
\label{sec:summary}
 
In this chapter, we have shown that the straight line walk traverses a constant number 
of simplices when we walk from a nearest neighbour, or an approximate nearest
neighbour, among the vertices of a well-shaped mesh to a query point.
This holds in both $\plane$ and $\threeD$.
Our bounds, though constant in terms of the aspect ratio, are pessimistic,
and we believe that there is room for improvement.
The key result of our work is that the approach works with respect to an 
approximate nearest neighbor. 
Our implementation of the walk step in the $\plane$ setting demonstrates
the validity of our proposed approach.

In addition to tightening the bounds that we have found, a number of 
possible extensions of this work can be pursued.
The first of these would centre around the search structures used to perform 
the initial nearest neighbour search.
For practical purposes, one of the widely-used data structures for point
location is the kd-tree \cite{DBLP:journals/cacm/Bentley75},
\cite{DBLP:journals/toms/FriedmanBF77}.
For randomly distributed points, kd-trees perform well with expected $\BigOh{\log n}$
searching.
However, the worst-case query times for point location using kd-trees in $\plane$
can  be $\BigOh{n}$.
The vertex set of a well-shaped mesh need not be randomly distributed, but it should
also avoid the worst-case scenario.
Thus, a possible study of interest would be to determining the how the value of $\rho$ for
a well-shaped mesh infleunces worst-case search times in $\plane$ and $\threeD$ kd-trees 
built on the mesh vertex set.
 
A second improvement to this work would be to proving that the walk-step
remains constant time if the jump set is some subset of the vertices of the
mesh. 
If this is the case, then what is the smallest possible subset that
maintains this guarantee of constant walk time.
Furthermore, can the size of this subset be linked to $\rho$?

Finally, we have implemented walk in $\plane$.
The technique that we have used should be extendable to $\threeD$ without
tremendous difficulty.
Thus, implementing this jump-and-walk technique in well-shaped meshes
in $\threeD$ would also be an interesting extension of this work.

\chapter[Conclusions]{Conclusions}\label{chp:conclusions}

The chapter begins with some concluding remarks on the research
presented in this thesis, and ends by describing open problems 
and possible future extensions of this work, in 
Section~\ref{sec:open_problems}.

The common thread running through the results
presented herein is path traversal.
In Chapter~\ref{chp:leaf-root}, data structures are described
that support efficient path traversal in rooted trees in the external
memory setting. 
In addition to being I/O-efficient, our structures for
trees are also succinct, and this research is among the first
work on designing succinct structures that are efficient
in the external memory setting.
The primary challenge addressed in developing these 
structures was tuning existing non-succinct tree blockings to
achieve the desired space bounds, and using succinct structures
to map between blocks.

In Chapter~\ref{chp:succinct_graphs} we present similar data structures
for bounded-degree planar graphs, these structures are both succinct and 
efficient in the external memory setting.
Again, much of the challenge in developing these structures was 
determining suitable block representations, and defining the structures
and methods needed to navigate between blocks.
In this chapter on planar graphs we also present a number of applications
on triangulations presented using our structure. 

Chapter~\ref{chp:mesh_trav} presents a partitioning scheme
for well-shaped tetrahedral meshes in $\threeD$, which permits
efficient path traversal in the external memory setting.
The initial goal of this work was to extend our planar graph techniques
to $\threeD$, and develop succinct structures. 
However, we have not yet been able to identify, or develop, any suitable
succinct encodings for tetrahedral meshes in either the external
memory or RAM memory model.
This is a possible area of continued research.
This chapter does however demonstrate how tetrahedral meshes can be
blocked to support I/O-efficient path traversal, and we present a number of
applications on a mesh blocked in this manner.

Finally, in Chapter~\ref{chp:jump_and_walk}, results were presented
on the walk step of jump-and-walk point location in 
well-shaped meshes.
The key contribution from this work is the proof that in a 
well-shaped mesh, the walk step for point location can be proven 
to take constant time, given an approximate nearest neighbour in
$\plane$ or $\threeD$.
This research is not specifically related to either external memory
algorithms, or succinct structures, however, it is applicable to the
applications presented in both Chapter~\ref{chp:succinct_graphs} and 
Chapter~\ref{chp:mesh_trav}.
Most of the applications we presented rely on point location as an 
initial step. 
Since the structures support efficient path traversal, and the walk 
step is effectively a path traversal operation; the jump-and-walk
technique gives an effective means of performing the point location 
operation in an I/O-efficient manner.

\section{Open Problems and Future Work}\label{sec:open_problems}

\subsection{Succinct and I/O-Efficient Tree Traversal}

Chapter~\ref{chp:leaf-root} presented two new data structures that 
are both I/O-efficient 
and succinct for bottom-up and top-down traversal in trees. 
Our bottom-up result applies to trees of arbitrary degree, while our 
top-down result applies to binary trees.  
In both cases, the number of I/Os is asymptotically optimal.
 
Our results lead to several open problems.
Our top-down technique is valid for binary trees only. 
Whether this approach can be extended to trees of larger degrees is an 
open problem.  
For the bottom-up case, it would be interesting to see whether the asymptotic 
bound on I/Os can be improved from $O(K/B+1)$ to something closer to 
$O(K/\succblksize+1)$ I/Os, where $\succblksize$ is the number of nodes 
that can be represented succinctly in a single disk block.  
In both the top-down and bottom-up cases, several $\rankop$ and 
$\selop$ operations are required to navigate between blocks. 
These operations use only a constant number of I/Os, and it would be 
useful to reduce this constant factor. 
This might be achieved by reducing the number of $\rankop$ and 
$\selop$ operations used in the algorithms, or by demonstrating how 
the bit arrays could be interleaved to guarantee a low number of I/Os 
per block.

\subsection{Succinct and I/O-Efficient Bounded Degree Planar Graphs}

Chapter~\ref{chp:succinct_graphs} presented succinct structures supporting
I/O efficient path traversal in planar graphs.
In fact, our results are somewhat more general than this.
The key features of planar graphs which our data structures rely
on are; first, planar graphs are $k$-page embeddable, and second, 
planar graphs of bounded degree have an $r$-partition with
a suitably small separator. 
Any class of graphs for which these two properties hold should 
permit construction of the succinct data structures.

\subsection{Tetrahedral Meshes Supporting I/O-Efficient Path Traversal}

The original purpose of the work presented in Chapter~\ref{chp:mesh_trav} 
was to develop a data 
structure that is both succinct and efficient in the EM setting.
In particular, we hoped to extend our structure for triangulations in
$\plane$ (see Chapter~\ref{chp:succinct_graphs}) to tetrahedral 
meshes in $\threeD$.
Unfortunately, the data structures we use rely on the fact that there
is a $k$-page embedding of the dual graph with constant $k$. 
There is no known bound on the page-thickness of the dual of a
tetrahedral mesh, thus it is uncertain whether our representation will
work.
Barat~\etal~\cite{DBLP:journals/combinatorics/BaratMW06} proved 
that bounded-degree planar graphs with maximum degree greater
than nine have arbitrarily large geometric thickness.
They claim that finding the page embedding of a degree five graph 
is an open problem.
Meanwhile, Duncan~\etal~\cite{DBLP:journals/corr/cs-CG-0312056}
gave a bound of $k=2$, when maximum degree is less than four.
With $d=8$, the augmented dual graph of a tetrahedral mesh is in
the grey area between these results.
Thus, a potential approach to solving this problem would be to prove that
the dual of a tetrahedral mesh is $k$-page embeddable, although there
is no guarantee that this is possible.

\subsection{Jump-and-Walk}

Chapter~\ref{chp:jump_and_walk} demonstrates that the straight line walk 
traverses a constant number 
of simplices when we walk from a nearest neighbour, or from an approximate 
nearest neighbour, among the vertices of a well-shaped mesh to a query point.
This holds in both $\plane$ and $\threeD$.
Our bounds, though constant in terms of the aspect ratio, are pessimistic,
and we believe that there is room for improvement.
The key result of our work is that the approach works with respect to an 
approximate nearest neighbour. 
Our implementation of the walk step in the $\plane$ setting demonstrates
the validity of our proposed approach.

In addition to tightening the bounds we have found, there are a number of 
possible extensions of this work that could be pursued.
The first of these centres around the search structures used to perform 
the initial nearest neighbour search.
For practical purposes, one of the widely-used data structures for point
location is the kd-tree \cite{DBLP:journals/cacm/Bentley75},
\cite{DBLP:journals/toms/FriedmanBF77}.
For randomly distributed points, kd-trees perform well with expected $\BigOh{\log n}$
searching.
However, the worst-case query times for point location using kd-trees in $\plane$
can  be $\BigOh{n}$.
The vertex set of a well-shaped mesh need not be randomly distributed, but it should
also avoid the worst-case scenario.
Thus, a possible study of interest would be determining how the value of $\rho$ for
a well-shaped mesh influences worst-case search times in $\plane$ and $\threeD$ kd-trees 
built on the mesh vertex set.

A second improvement to this work would be to proving that the walk step
remains constant time if the jump set is some subset of the vertices of the
mesh. 
If this is the case, then what is the smallest possible subset that
maintains this guarantee of constant walk time?
Furthermore, can the size of this subset be linked to $\rho$?

Finally, we have implemented walk in $\plane$.
The techniques we have used should be extendable to $\threeD$ without tremendous 
difficulty.
Thus, implementing this jump-and-walk technique in well-shaped meshes in $\threeD$
would also be an interesting extension of this work.

\cleardoublepage

\bibliography{cdillabaugh-thesis}
\bibliographystyle{plain}
\end{document}